\documentclass[11pt]{article}
\usepackage[T1]{fontenc}
\usepackage[a4paper, textwidth = 16cm, textheight = 22.5cm, top = 3cm]{geometry}
\usepackage{xcolor,amsmath,amssymb,amsthm,bbm,enumitem,graphicx,mathtools,enumitem,float}
\usepackage{hyperref}
\hypersetup{colorlinks=true, linkcolor=darkred, urlcolor=darkred, citecolor=darkblue}
\usepackage[round]{natbib}
\usepackage[ruled, lined, linesnumbered]{algorithm2e}
\usepackage{mathrsfs}

\newtheorem{thm}{Theorem}   
\newtheorem{prop}[thm]{Proposition}

\newtheorem{lemma}[thm]{Lemma}
\newtheorem{cor}[thm]{Corollary}
\newtheorem{defn}{Definition}

\definecolor{darkred}{RGB}{179,0,0}
\definecolor{darkblue}{RGB}{0,0,179}
\definecolor{darkgreen}{RGB}{0,179,0}

\DeclareMathOperator*{\argmax}{arg\,max}
\DeclareMathOperator*{\argmin}{arg\,min}
\DeclareMathOperator{\Prob}{\mathbb{P}}

\newcommand{\indep}{\perp\!\!\!\perp} 
\newcommand{\E}{\mathbb{E}}
\newcommand{\mQ}{\mathcal{Q}}
\newcommand{\cM}{\mathcal{M}}
\newcommand{\vol}{\mathrm{vol}}
\newcommand{\abs}[1]{| #1|}
\newcommand{\var}{\mathrm{Var}}
\newcommand{\cov}{\mathrm{Cov}}
\newcommand{\bone}{\mathbbm{1}}
\newcommand{\cp}{\mathcal{P}}
\newcommand{\pto}{\xrightarrow{\,\mathrm{p}\,} }
\newcommand{\dto}{\xrightarrow{\,\mathrm{d}\,} }
\newcommand{\iid}{\stackrel{\mathrm{iid}}\sim}
\newcommand{\bs}[1]{\boldsymbol{#1}}

\title{Coverage correlation: detecting singular dependencies between random variables}
\author{Xuzhi Yang, Mona Azadkia and Tengyao Wang\\[3pt]
London School of Economics and Political Science}
\date{(\today)}

\begin{document}

\maketitle

\begin{abstract}
We introduce the coverage correlation coefficient, a novel nonparametric measure of statistical association designed to quantify the extent to which two random variables have a joint distribution concentrated on a singular subset with respect to the product of the marginals. Our correlation statistic consistently estimates an $f$-divergence between the joint distribution and the product of the marginals, which is 0 if and only if the variables are independent and 1 if and only if the copula is singular. Using Monge--Kantorovich ranks, the coverage correlation naturally extends to measure association between random vectors. It is distribution-free, admits an analytically tractable asymptotic null distribution, and can be computed efficiently, making it well-suited for detecting complex, potentially nonlinear associations in large-scale pairwise testing.
\end{abstract}

\section{Introduction}
Correlation is a measure of statistical association that quantifies how two variables tend to vary together. Classically, Pearson's correlation, $r^{X,Y}$, captures linear relationships between real-valued variables $X$ and $Y$ \citep{pearsonhistory}. In contrast, Spearman's rank correlation, $\rho^{X,Y}$, and Kendall's rank correlation, $\tau^{X,Y}$, capture monotonic associations between $X$ and $Y$, using different definitions of rank concordance \citep{spearman1904proof,kendall1938new}. A key limitation of these classical measures is their poor performance in detecting non-monotonic associations, even in noise-free data.

To overcome this limitation, numerous approaches have been proposed, including the maximal correlation coefficient~\citep{hirschfeld1935connection, gebelein1941statistische, renyi1959measures, breiman1985estimating}, various methods based on joint cumulative distribution functions and
ranks 
%\citep[e.g.][]{hoeffding1948non,blum1961distribution,bergsma2014consistent,drton2018high,deb2023multivariate}
\citep{hoeffding1948non, blum1961distribution, yanagimoto1970measures, puri1971nonparametric, rosenblatt1975quadratic, csorgo1985testing, romano1988bootstrap, bergsma2014consistent, nandy2016large, weihs2016efficient, weihs2018symmetric, han2017distribution, wang2017generalized, drton2018high, gamboa2018sensitivity,shi2022distribution, deb2023multivariate, shi2025distribution, zhou2025association},
kernel-based methods~\citep{gretton2005measuring, gretton2008kernel, sen2014testing, pfister2018kernel, zhang2018large}, information-theoretic coefficients~\citep{linfoot1957informational, kraskov2004estimating, reshef2011detecting, berrett2019nonparametric,berrett2021optimal}, copula-based coefficients %\citep[e.g.][]{sklar1959fonctions, schweizer1981nonparametric, zhang2019bet}
\citep{sklar1959fonctions, schweizer1981nonparametric, dette2013copula, lopez2013randomized, zhang2019bet, griessenberger2022multivariate}, 
and coefficients based on pairwise
distances %\citep[e.g.][]{friedman1983graph, szekely2007measuring, heller2013consistent}.
\citep{friedman1983graph, heller2013consistent, lyons2013distance, szekely2009brownian, szekely2007measuring, pan2020ball}. 

Although many of these coefficients are commonly applied, two significant drawbacks remain. Most are constructed with the primary goal of testing for independence, offering little direct information about the magnitude of the underlying dependence. Moreover, their null distributions are often analytically intractable, so p-values often must be obtained through computationally intensive permutation procedures.

Recently, there has been renewed interest in developing nonparametric measures of statistical
association, driven in part by the need to identify relevant features and interactions in large
datasets. Existing contributions may be broadly divided into several strands. For scalar
variables, \citet{dette2013copula} introduced a copula-based measure of regression dependence
quantifying the extent to which $Y$ is a measurable function of $X$, and
\citet{chatterjee2021new} proposed a novel, simple rank-based estimator for it. Later,
\citet{azadkia2021simple} extended it to multivariate predictors, which was subsequently
followed by \citet{azadkia2025new}, who introduced another measure of statistical association of
a similar nature. Other works extend the framework to broader spaces, including kernel- and
graph-based measures on metric spaces \citep{deb2020measuring, wiesel2022measuring, roudaki2026kernel}. Other notable extensions of \citet{chatterjee2021new} to the multivariate-response setting have
also been proposed by \citet{ansari2022simple} and \citet{huang2025multivariate}. While these
methods provide important extensions of the Chatterjee framework, they largely retain a directed
interpretation.

% Recently, there has been renewed interest in developing nonparametric measures of statistical association, driven in part by the need to identify relevant features and interactions in large datasets. Several new statistics have been proposed to capture the extent to which $Y$ can be expressed as a deterministic (measurable) function of $X$ \citep{dette2013copula, chatterjee2021new, azadkia2021simple, deb2020measuring, wiesel2022measuring, ansari2022simple, azadkia2025new, huang2025multivariate, roudaki2026kernel}.

The correlation coefficient of~\citet{chatterjee2021new} has seen remarkably rapid adoption in practice, particularly in fields such as bioinformatics, where uncovering complex and potentially nonlinear associations in high-dimensional data is a central challenge \citep[e.g.][]{dong2023causal,suo2024dandelion,sansalone2024unexpectedly}. Its popularity arises from several appealing properties: the statistic is distribution-free under the null hypothesis of independence, allowing precise characterisation of its asymptotic null distribution, and it is computationally efficient, scaling well to large datasets. 

To illustrate, consider the task of detecting covariation in gene expression levels in a single-cell RNA sequencing experiment, where thousands of genes are measured across tens of thousands of cells. The scale of this problem creates a substantial multiple testing burden, often requiring raw p-values on the order of $10^{-8}$ or smaller to declare significance for any gene pair. In such settings, resampling-based tests such as permutation become computationally prohibitive, making access to an accurate asymptotic null distribution essential.

Chatterjee's correlation and the corresponding extensions are specifically designed to capture the extent to which $Y$ can be expressed as a measurable function of $X$ and is therefore inherently asymmetric. While this asymmetry can be advantageous in certain contexts, such as when a clear predictor-response relationship is present, it may be less suitable in others. In the genetic association example above, the direction of dependence between gene expression levels is often not known \emph{a priori}, and in fact, the relationship may not be directional at all. For instance, two genes, A and B, may exhibit strong statistical dependence simply because they are both downstream of a common regulator gene C, rather than one being a function of the other. Motivated by such considerations, we propose a new measure of statistical association, the \emph{coverage correlation coefficient}, which is symmetric and designed to capture more general implicit functional relationships of the form $f(X, Y) = 0$. More precisely, the coverage correlation coefficient quantifies the extent to which the joint distribution $P^{(X,Y)}$ is singular with respect to the product of marginals  $P^X \otimes P^Y$, thereby detecting dependencies that may lie on low-dimensional structures within the joint space.

\subsection{Coverage correlation coefficient}\label{sec:CoveCorrConstruction}
To provide intuition for the coverage correlation coefficient, consider a simplified setting in which the random variables $X$ and $Y$ follow the uniform distribution over $[0, 1]$. Given an i.i.d.\ sample $\{(X_i, Y_i)\}_{i=1}^n$ drawn from the joint distribution $P^{(X, Y)}$, we examine two extreme cases: 
\begin{enumerate}[noitemsep, label=(\roman*)]
    \item $X$ and $Y$ are independent;
    \item $P^{(X,Y)}$ is singular with respect to $\mathrm{Unif}([0, 1]^2)$. 
\end{enumerate}
In the first case, the sample points $(X_i, Y_i)$ are uniformly distributed over the unit square $[0,1]^2$. In contrast, under the second scenario, the points $(X_i, Y_i)$ are concentrated on a subset of $[0,1]^2$ with Lebesgue measure zero. These two cases yield qualitatively distinct scatter plots, reflecting the presence or absence of dependence (see the bottom row of Figure~\ref{Fig:ChatterjeePlot} for an illustration).

To quantify this difference, consider the area of the region in $[0, 1]^2$ not covered by the union of $\ell_\infty$-balls (squares) centred at each sample point $(X_i, Y_i)$, where the area of each ball is fixed to be $1/n$. The behaviour of this uncovered area has been extensively studied in the context of \emph{coverage processes} \citep[see, e.g.][]{hall1988introduction}. 

As $n \to \infty$, the limiting uncovered area exhibits fundamentally different behaviour in the two cases considered above: when $X$ and $Y$ are independent, it converges to $e^{-1}$, whereas when $P^{(X,Y)}$ is singular with respect to $\mathrm{Unif}([0,1]^2)$, it converges to $1$ in probability.

Building on this geometric intuition, we propose the \emph{coverage correlation coefficient} of random vectors $X \in \mathbb{R}^{d_X}$ and $Y \in \mathbb{R}^{d_Y}$ with $d_X, d_Y \in \mathbb{N}$, based on the (multivariate) ranks of $\bs X := (X_i)_{1\leq i\leq n}$ and $\bs Y := (Y_i)_{1\leq i\leq n}$. This statistic is powerful against the alternative where $X$ and $Y$ possess a singular dependence. Intuitively speaking, the proposed correlation coefficient measures the uncovered volume in $[0, 1]^{d_X + d_Y}$ when small cubes of volume $1/n$, centred at the multivariate ranks of $(X_i, Y_i)$, are used to cover the space. 

Given two sets of reference points $\bs{U} = (U_1, \ldots, U_n)$ in $[0,1]^{d_X}$ and $\bs{V} = (V_1, \ldots, V_n)$ in $[0,1]^{d_Y}$, let 
\[
\pi^{X} := \argmin_{\pi\in\mathcal{S}_n} \frac{1}{n}\sum_{i=1}^n\|U_{\pi(i)} - X_i\|_2^2 \quad \text{and} \quad \pi^{Y} := \argmin_{\pi\in\mathcal{S}_n} \frac{1}{n}\sum_{i=1}^n\|V_{\pi(i)} - Y_i\|_2^2 
\]
be the optimal assignments from $\bs{X}$ to $\bs{U}$ and from $\bs{Y}$ to $\bs{V}$ respectively, where $\mathcal{S}_n$ is the set of all permutations on $[n]:=\{1,\ldots,n\}$. If the minimizer is not unique, we choose one uniformly at random among all minimizers, independently for the two marginal assignments. For each $i\in[n]$, the empirical multivariate ranks for $X_i$ and $Y_i$ are 
\begin{equation}
R_i^X:=U_{\pi^{X}(i)} \quad \text{and} \quad R_i^Y:=V_{\pi^{Y}(i)}, \label{def:EmpRankGen}
\end{equation} 
respectively, and we write
\begin{equation}
    \label{eq:R_iGen}
R_i:= (R_i^X, R_i^Y)\in [0,1]^d,
\end{equation}
for their joint rank, where $d:= d_X+d_Y$. We remark that $R_i^X$ and $R_i^Y$ are known as the \emph{Monge--Kantorovich ranks} in the literature \citep{chernozhukov2017monge,hallin2021distribution}.

We also define a $d$-dimensional $\ell_\infty$ neighbourhood of radius $r$ centred at $w\in[0,1]^d$ as 
\[  
B(w, r) := \{z\in [0,1]^d: \inf_{\ell\in \{-1,0,1\}^d} \|z - w - \ell\|_\infty \leq r\}.
\]
Here, the infimum over $\ell$ defines the $\ell_\infty$ distance with a periodic boundary condition that identifies opposite sides of the unit cube $[0,1]^d$. This improves the empirical performance of the correlation coefficient statistic in small sample size (see Section~\ref{sec:no-wrapping}). Writing $\mathrm{vol}(\cdot)$ for the $d$-dimensional Lebesgue measure, for $\gamma\in(0,1)$, we define
\begin{equation}
\label{eq:VnGen}
\mathcal{V}(\bs{X},\bs{Y},\bs{U},\bs{V}; \gamma) := 1 - \vol\Bigl(\bigcup_{i=1}^n B(R_i, \gamma)\Bigr)
\end{equation}
to be the uncovered volume in the $d$-dimensional unit cube outside subcubes of radius $\gamma$ centred at the empirical ranks. 

The use of reference points in the unit cube is the most natural choice to construct a coverage volume: each point contributes equally to the total covered volume. It also mimics the traditional notion of rank on the real line (see Section~\ref{Sec:Discussion} for more discussion of the choice of reference points). In what follows, we will mostly be working with random reference points uniformly distributed over $[0, 1]^d$, i.e. $(U_1,V_1),\ldots,(U_n,V_n)\iid \mathrm{Unif}([0,1]^d)$, and $\gamma = \frac{1}{2n^{1/d}}$ so that each subcube has volume $1/n$. We write $\mathcal{V}_n = \mathcal{V}(\bs{X},\bs{Y},\bs{U},\bs{V}; \gamma)$ for this specific choice of reference points and $\gamma$. 
\begin{defn}
\label{def:covcorr}
The \emph{empirical coverage correlation coefficient} between $X_1,\ldots,X_n$ and $Y_1,\ldots,Y_n$ is defined as
\[
\kappa_n^{X, Y}:= \frac{\mathcal{V}_n - e^{-1}}{1 - e^{-1}}.
\]
\end{defn}
We note that $\kappa_n^{X,Y}$ is random due to the uniformly random reference points, even when conditioning on $\bs{X}$ and $\bs{Y}$. If a non-random correlation statistic is preferred, it is possible to replace the randomly generated $\bs{U}$ and $\bs{V}$ by a set of fixed reference points that are sufficiently `spread out' in $[0,1]^d$. For instance, when $d_X=d_Y=1$, a natural choice could be $\bs{U} = \bs{V} = (1/n,\ldots, (n-1)/n, 1)$. We will discuss this further in Section~\ref{sec: derandomising}. 

We summarise several key features of the coverage correlation coefficient $\kappa_n^{X, Y}$ as follows:
\begin{enumerate}
    \item When $d_X = d_Y = 1$, the coverage correlation statistic converges to $0$ in probability if and only if $X$ is independent of $Y$, and to $1$ in probability if and only if $P^{(X, Y)}$ is singular with respect to the product of the marginals. 
    \item More generally, $\kappa_n^{X,Y}$ converges to a population quantity that measures an $f$-divergence between the joint distribution and the product of the marginals with respect to the divergence generator function $f_{\mathrm{cov}}(x) = \frac{e^{-x}-e^{-1}}{1-e^{-1}}$. Notably, $\kappa_n^{X,Y}$ circumvents the need for density estimation, distinguishing it from numerous existing divergence estimators \citep[e.g.][]{rubenstein2019practical}.
    \item For any $d_X, d_Y$, under the null hypothesis of independence between $X$ and $Y$, $\kappa_n^{X,Y}$ is asymptotically normally distributed, which allows us to construct asymptotically valid p-values. 
    \item The coverage correlation statistic is distribution-free, thanks to Monge--Kantorovich ranks used in~\eqref{def:EmpRankGen}. This is in contrast to other possible multivariate ranks such as depth-based ranks \citep{tukey1975mathematics, liu1993quality, zuo2000general}, spatial ranks \citep{mottonen1995multivariate, chaudhuri1996geometric, koltchinskii1997m}, componentwise ranks \citep{hodges1955bivariate, bickel1965some}, and Mahalanobis ranks \citep{hallin2002optimal, hallin2002multivariate}. The distribution-free property of the coverage correlation statistic yields a pivotal null distribution, enabling easy computation of p-values. 
    \item For univariate marginal distributions, we develop an algorithm with $O(n \log n)$ time complexity (see Section~\ref{Sec:Algorithm}). The method is implemented as R and Python packages \texttt{covercorr}. Both the package and simulation code for reproducing figures and tables in the paper can be found at \url{https://github.com/wangtengyao/covercorr}.
\end{enumerate}

\subsection{Connection to Chatterjee's correlation}
For random variables $X$ and $Y$ and given a sample $(X_i,Y_i)_{i\in [n]}$,  let $X_{(1)}\leq \cdots\leq X_{(n)}$ denote the order statistics of the $X_i$'s (break ties uniformly at random) and let $(Y_{(i)})_{i\in[n]}$ denote the corresponding concomitants. Assuming for simplicity that there are no ties among $Y_i$'s, the empirical Chatterjee's correlation is defined as 
\begin{equation}
\label{Eq:ChatterjeeCorrelation}
\xi_n^{X,Y}:=1 - \frac{\sum_{i=1}^{n-1}|r_{i+1}-r_i|}{(n^2-1)/3},
\end{equation}
where $r_i := \#\{j: Y_{(j)}\leq Y_{(i)}\}$ is the rank of $Y_{(i)}$. \citet[Theorem~1.1]{chatterjee2021new} shows that $\xi_n^{X,Y}$ converges stochastically to 0 when $X$ and $Y$ are independent, and to 1 when $Y$ is a function of $X$. This convergence is interpreted through the population statistic
\[  
\xi^{X,Y}:= \frac{\int_{\mathbb{R}}\var(\mathbb{E}[\mathbbm{1}\{Y\geq t\}\mid X])\,dP^Y(t)}{\int_{\mathbb{R}}\var(\mathbb{E}[\mathbbm{1}\{Y\geq t\}])\,dP^Y(t)}.
\]
We remark that~\eqref{Eq:ChatterjeeCorrelation} can also be interpreted as a measure of `excess vacancy', similar to our coverage correlation coefficient, in the following sense. Writing $\tilde{\mathcal{V}} := 1 - \sum_{i=1}^{n-1}\bigl|\frac{r_{i+1}}{n}-\frac{r_i}{n}\bigr|\cdot\frac{1}{n}$ for the total area in $[0,1]^2$ not covered by the union of rectangles
\[
\bigcup_{i=1}^{n-1} \Bigl(\Bigl[\frac{i-1}{n},\frac{i}{n}\Bigr] \times \Bigl[\min\Bigl\{\frac{r_i}{n},\frac{r_{i+1}}n\Bigr\}, \max\Bigl\{\frac{r_i}{n},\frac{r_{i+1}}n\Bigr\}\Bigr]\Bigr),
\]
we have 
\[  
\xi_n^{X,Y} = \frac{\tilde{\mathcal{V}} - 2/3}{1/3} + O(n^{-2}).
\]
In other words, if we draw a line plot of the ordered normalised $X_i$ ranks (which are $1/n, 2/n,\ldots, 1$) against the corresponding normalised $Y_i$ ranks $(r_1/n, \ldots, r_n/n)$, with a line `thickness' of $1/n$, then $\tilde{\mathcal{V}}$ approximates the area in $[0,1]^2$ that remains uncovered by this thickened line plot. 

Figure~\ref{Fig:ChatterjeePlot} illustrates this for several $(X,Y)$ distributions. In the case of independence (first column), the normalised $Y_i$ ranks resemble independent $\mathrm{Unif}[0,1]$ values, and successive differences in rank are approximately $1/3$ on average, yielding $\tilde{\mathcal{V}} \approx 2/3$. On the other hand, when $Y$ is a function of $X$, the normalised $Y_i$ ranks become a deterministic function of the normalised $X_i$ ranks (modulo discretisation), which under mild conditions causes the uncovered area $\tilde{\mathcal{V}}$ to shrink towards 0. As seen in the second and third columns, when a fraction of the $Y_i$'s can be well-approximated by a function of the $X_i$'s, the area uncovered by the line plot decreases accordingly. Chatterjee's correlation captures this structure.

The last column of Figure~\ref{Fig:ChatterjeePlot} presents an interesting example in which both $X$ and $Y$ are generated as functions of a third latent variable $U$, with added noise. In this setting, even though a visually striking relationship exists between the $X_i$'s and $Y_i$'s, Chatterjee’s correlation remains close to zero. This highlights a key limitation of the statistic: it is specifically designed to detect directional functional dependence, but it may fail to capture more symmetric or indirect relationships.

\begin{figure}[htbp]
\begin{center}
    \begin{tabular}{cccc}
    \includegraphics[width=0.18\textwidth]{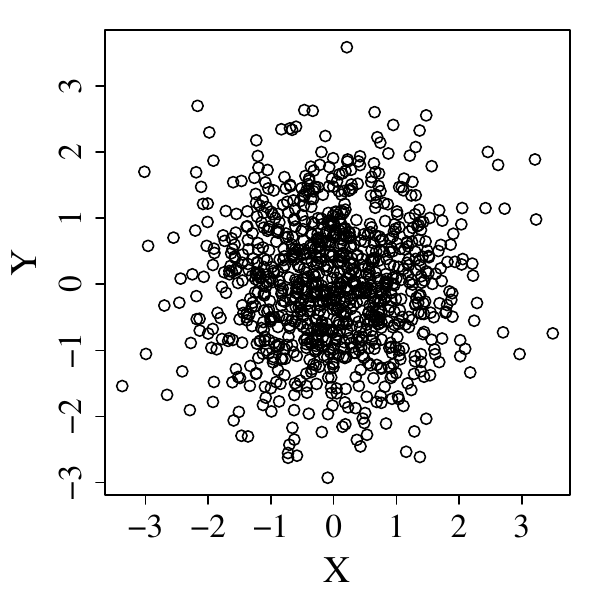}
    &
    \includegraphics[width=0.18\textwidth]{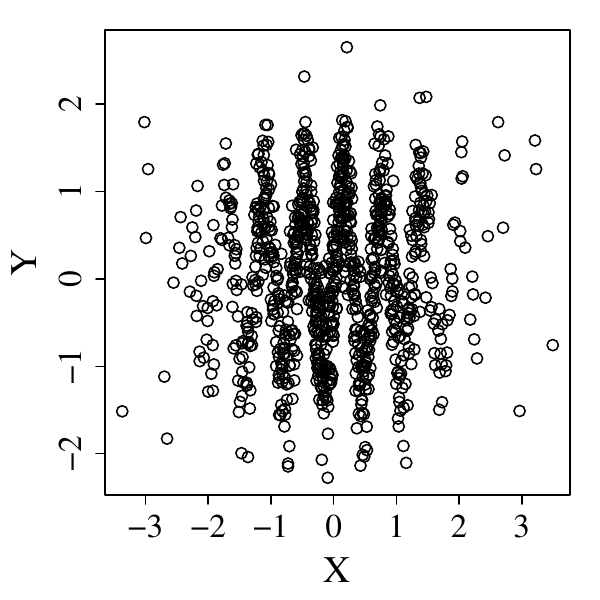}
    & 
    \includegraphics[width=0.18\textwidth]{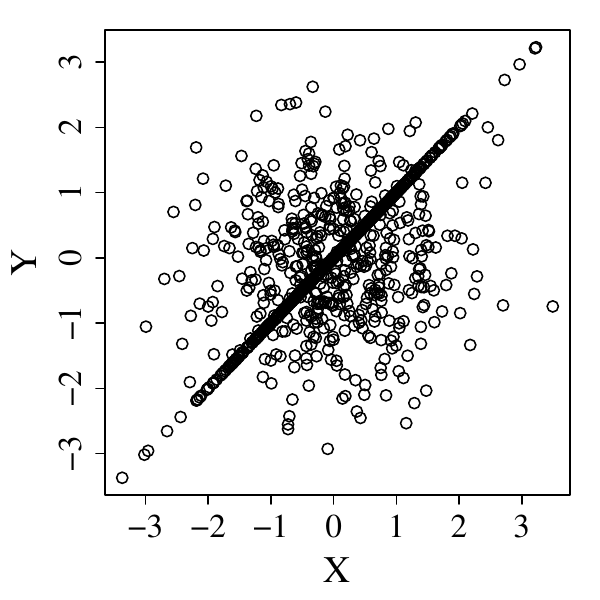}
    &
    \includegraphics[width=0.18\textwidth]{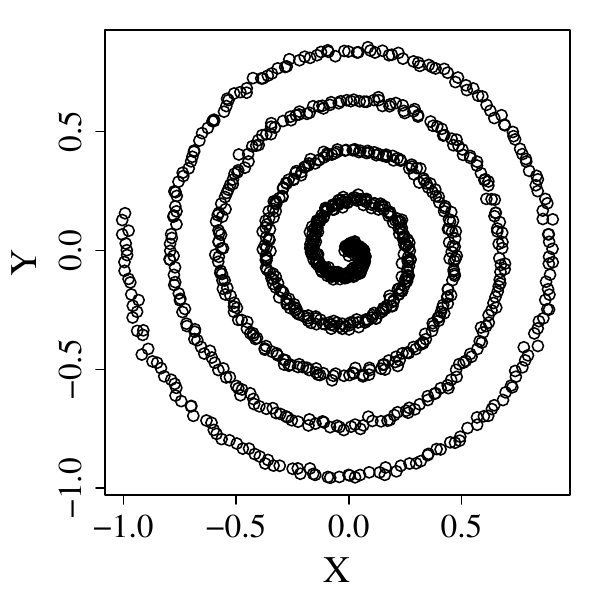}
    \\
    \includegraphics[width=0.18\textwidth]{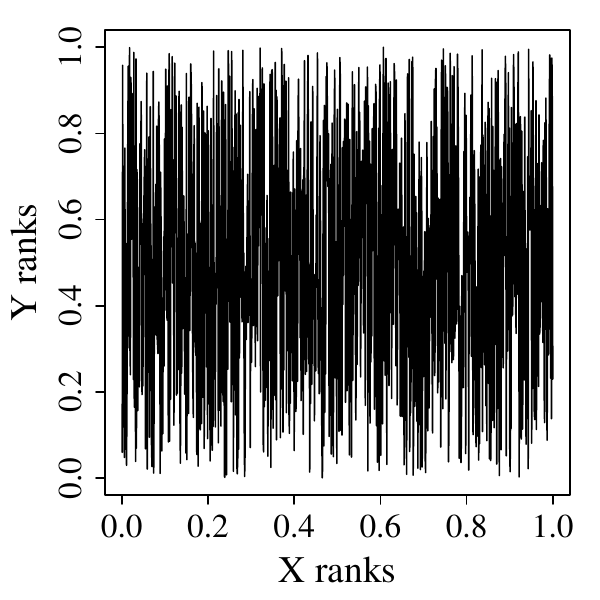}
    &
    \includegraphics[width=0.18\textwidth]{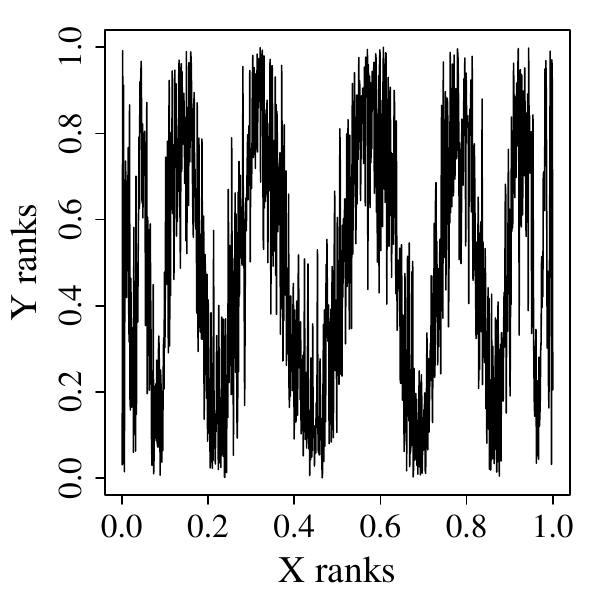}
    & 
    \includegraphics[width=0.18\textwidth]{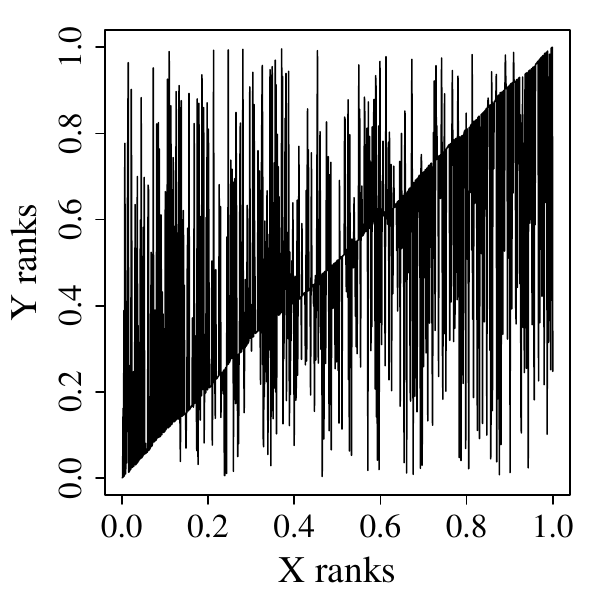}
    &
    \includegraphics[width=0.18\textwidth]{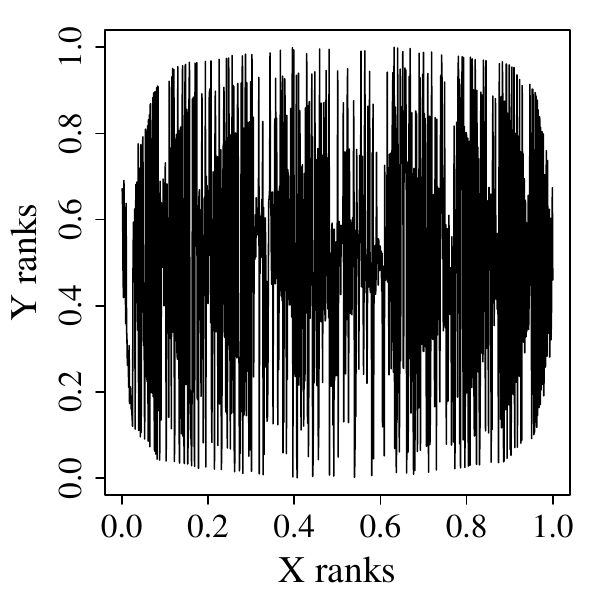}
    \\
    $\xi_n^{X,Y} = -0.019$ & $\xi_n^{X,Y} = 0.46$ & $\xi_n^{X,Y} = 0.23$ & $\xi_n^{X,Y} = 0.0094$\\
    $p_{\mathrm{\xi}} = 0.83$ & $p_{\mathrm{\xi}} < 10^{-16}$ & $p_{\mathrm{\xi}} < 10^{-16}$ & $p_{\mathrm{\xi}} =0.32$\\
    \includegraphics[width=0.18\textwidth]{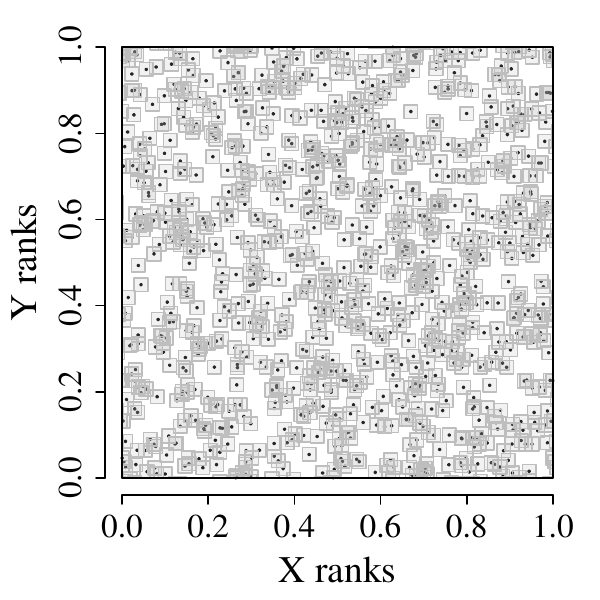}
    &
    \includegraphics[width=0.18\textwidth]{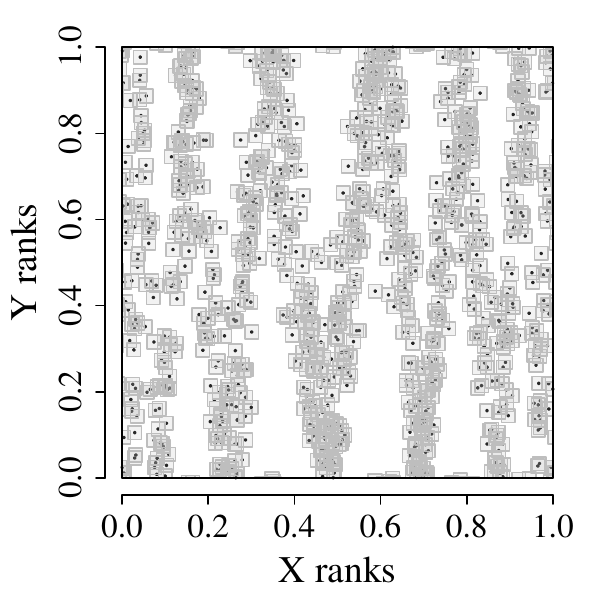}
    & 
    \includegraphics[width=0.18\textwidth]{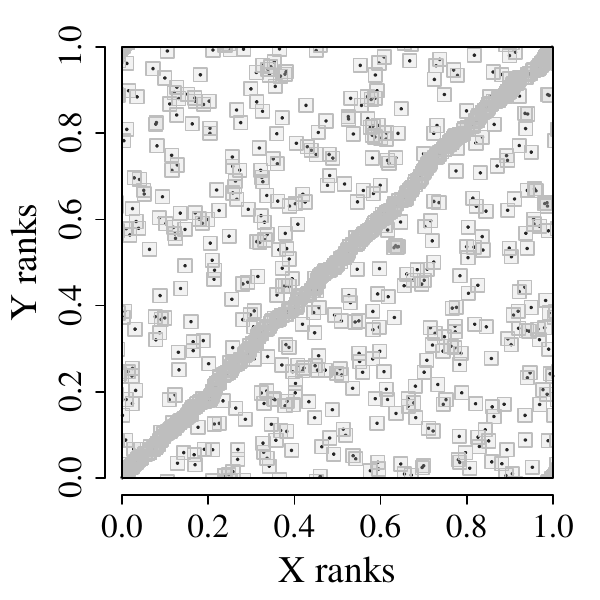}
    &
    \includegraphics[width=0.18\textwidth]{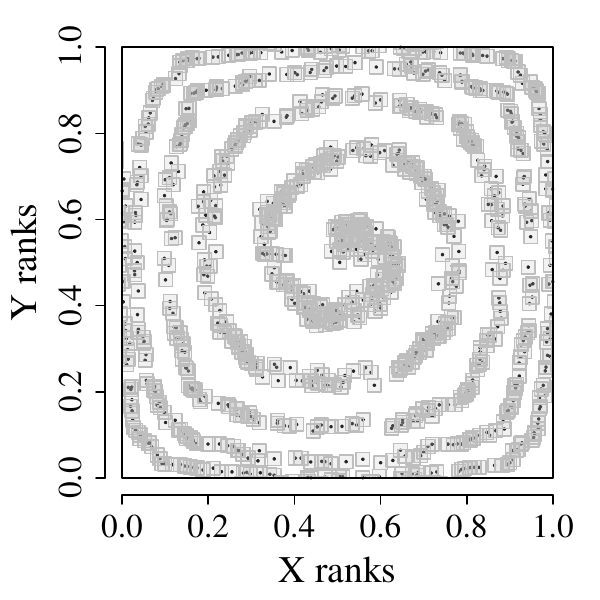}
    \\
    $\kappa_n^{X,Y} = -0.003$ & $\kappa_n^{X,Y} = 0.12$ & $\kappa_n^{X,Y} = 0.22$ & $\kappa_n^{X,Y} = 0.21$\\
    $p_{\kappa} = 0.69$ & $p_{\kappa} < 10^{-16}$ &$p_{\kappa} < 10^{-16}$ &$p_{\kappa} < 10^{-16}$ 
    \end{tabular}
\end{center}
\caption{\label{Fig:ChatterjeePlot}Chatterjee's correlation and coverage correlation for various joint distributions using a sample of $n=1000$ observation pairs.  Data generating mechanisms are as follows --- first column: $X,Y\iid \mathcal{N}(0,1)$; second column: $X\sim \mathcal{N}(0,1)$ and $Y=\sin(10X) + 0.5\epsilon$ where $\epsilon\sim \mathcal{N}(0,1)\indep X$; third column: $X\sim \mathcal{N}(0,1)$ and $Y = XB + \epsilon(1-B)$, where $(B,\epsilon)\sim \mathrm{Bernoulli}(1/2)\otimes \mathcal{N}(0,1) \indep (X,Y)$; fourth column: $X = U\sin(10\pi U) + 0.01\epsilon_X$ and $Y=U \cos(10\pi U) + 0.01\epsilon_Y$, where $(U,\epsilon_X,\epsilon_Y)\sim \mathrm{Unif}[0,1]\otimes \mathcal{N}(0,1)\otimes \mathcal{N}(0,1)\indep (X,Y)$. For each column, the top panel shows the scatter plot, the middle panel shows the line plot of ordered normalised $X$ ranks against the corresponding $Y$ ranks, and the bottom panel shows the union of small squares of area $1/n$, centred at joint normalised ranks of $X$ and $Y$. Test statistics and p-values for both correlation measures are also displayed.}
\end{figure}

\section{Theoretical guarantees}
\label{sec:theory}
We first show in Section~\ref{sec:consistency} that the proposed coverage correlation coefficient consistently estimates an $f$-divergence between the joint distribution and the product of its marginals. Section~\ref{sec:nullCLT} then derives its null asymptotic distribution. We discuss derandomising the coverage coefficient using grid reference points in Section~\ref{sec: derandomising}. Finally, Section~\ref{sec:DependencyAmongKRV} extends both the consistency and the null asymptotic theory for dependence among several random vectors.
\subsection{Consistency of the coverage correlation}\label{sec:consistency}
\subsubsection{Univariate consistency}
\label{Sec:UnivariateConsistency}
% TODO [Xuzhi]: rearrange Sec 2.1 as follows
% 1. Consistency for univariate random ref (Thm 1)
% 2. Properties of kappa (Prop 5)
% 3. Discussion about viewing empirical coverage correlation as estimating a population f-divergence between joint and product of marginals; relate to general f-divergence literature
% 4. Grid ref vs random ref (pros and cons)
% 5. In general dimensions, we have Thm 4 (under alpha-holder of the OT map)
% 6. Discuss that the proof techniques for Thm 4 and Thm 1 are completely disjoint. What happens in the dimenion in between is still open but likely beyond the scope of the current work. But nontheless, we are able to characterise what happens at the extremes for the dependence statistic (Prop 2)
% TODO [Xuzhi]: revert statement and proof of Thm 3 (consistency of fixed grid ref in 1d) back to the original version.
% TODO [Xuzhi]: start a response file, copy editor/AE/reviewers comments
% TODO [Tengyao]: replace the yeast real data example with the hormonal change one
% TODO [Mona]: Add discussion about relationship to f-divergence estimation in general
% TODO [non-urgent]: add simulation results about general reference distribution (multivariate standard Gaussian, then mapped back to square)
% TODO [non-urgent]: update R package, i) offer option for grid/random design and default to grid; ii) give the option to visualise dependence as in fig 3; iii) general dependence between k random vectors

We first show that for univariate $X$ and $Y$, $\kappa_n^{X, Y}$, as defined in Definition~\ref{def:covcorr}, converges in probability to a population quantity that measures an $f$-divergence between $P^{(X, Y)}$ and $P^X \otimes P^Y$.  We recall the definition of $f$-divergence \citep[see, e.g.][Definition~8.2]{samworth2025statistics}.

\begin{defn}\label{def:f-divergence}
Let $f: \mathbb{R} \to \mathbb{R} \cup \{+\infty\}$ be a convex function with $f(1)=0$. For any two probability measures $\mu, \nu$ on a space $\mathcal{S}$, let $\mathrm{d} \mu = h \, \mathrm{d}\nu + \mathrm{d} \nu^{\perp}$ be the Lebesgue--Radon--Nikodym decomposition of $\mu$ with respect to $\nu$, where $h$ is $\nu$-integrable and $\nu^{\perp}$ is singular with respect to $\nu$. The \emph{$f$-divergence} between $\mu$ and $\nu$ is defined as 
\begin{align}
    D_f(\mu \,\|\, \nu) = \int_{\mathcal{S}} f\circ h \, d\nu + f'(\infty) \nu^{\perp}(\mathcal{S}),  \label{eq: fDivergenceDefinition}
\end{align}
where $f'(\infty) : = \lim_{t \to \infty} t^{-1}f(t)$ is the asymptotic slope of $f$ at infinity.
\end{defn}

The following theorem shows that when $d_X=d_Y=1$, the correlation coefficient $\kappa^{X,Y}_n$ converges stochastically to a population limit.
\begin{thm}
\label{Thm:populationlimit}
Let $P^{(X,Y)}$ be a Borel probability measure on $\mathbb{R}^{2}$ with marginals $P^X$ and $P^Y$. Define $f_{\mathrm{cov}}:\mathbb{R}\to\mathbb{R}$ by
\begin{equation}
\label{Eq:fcov}
f_{\mathrm{cov}}(x) := \frac{e^{-x}-e^{-1}}{1-e^{-1}}.
\end{equation}
Given $(X_1,Y_1),\ldots,(X_n,Y_n)\iid P^{(X,Y)}$, where $X_i, Y_i\in\mathbb{R}$, we have
\begin{equation}
\label{eq:f-divergence}
\kappa_n^{X,Y} \pto \kappa^{X,Y} := D_{f_{\mathrm{cov}}}(P^{(X,Y)}\,\|\,P^X\otimes P^Y), 
\qquad \text{as } n\to\infty .
\end{equation}
\end{thm}

We refer to the $f$-divergence with generator function $f_{\mathrm{cov}}$ as the \emph{coverage divergence} between distributions. Theorem~\ref{Thm:populationlimit} shows that for univariate data $X$ and $Y$, the empirical coverage correlation converges in probability to the coverage divergence between the joint distribution and the product of marginals between $X$ and $Y$.

One striking fact of Theorem~\ref{Thm:populationlimit} is that it does not require any assumptions on the joint distribution $P^{(X,Y)}$, thus the variables $X$ and $Y$ may be continuous, discrete, or a mixture of both. To understand intuitively why such a limit makes sense, consider the simple setting where $X$ and $Y$ are independent. In this case, it can be established relatively easily, without invoking Theorem~\ref{Thm:populationlimit}, that 
$\mathcal{V}_n \pto e^{-1}$ and hence $\kappa^{X,Y}_n\pto  0$,  by observing that for an independent point $W\sim \mathrm{Unif}([0,1]^2)$ we have 
\begin{align}
\mathbb{E}(\mathcal{V}_n) &= \mathbb{E}\Bigl[\mathbb{P}\Bigl\{W\notin\bigcup_{i=1}^n B\Bigl(R_i, \frac{1}{2\sqrt{n}}\Bigr)\Bigm | R_1,\ldots, R_n\Bigr\}\Bigr]\nonumber\\
&=\mathbb{E}\Bigl[\mathbb{P}\Bigl\{R_i\notin B\Bigl(W, \frac{1}{2\sqrt{n}}\Bigr)\;\forall\, i\in[n]\Bigm | W\Bigr\}\Bigr]=(1-1/n)^n\to e^{-1},\label{Eq:NullLimit}
\end{align}
and similarly (through a second moment calculation) $\var(\mathcal{V}_n) \to 0$.  However, this argument does not extend to the general case, because the joint ranks $R_i$ are no longer independent, so the final equality above does not hold. This subtle dependence structure among the $R_i$'s is the main technical obstacle in proving Theorem~\ref{Thm:populationlimit}. 

The following proposition shows that the coverage divergence between the joint distribution and the product of the marginals, $\kappa^{X,Y}$, as defined in~\eqref{eq:f-divergence}, can indeed serve as a correlation measure between distributions in any dimension, as it satisfies some of the key desired properties for correlation statistics discussed in the literature. 

\begin{prop}\label{Prop:KappaProperty}
Suppose ${P}^{(X, Y)}$ is a Borel probability measure on $\mathbb{R}^{d_X + d_Y}$ with marginals $P^X$ and $P^Y$ on $\mathbb{R}^{d_X}$ and $\mathbb{R}^{d_Y}$, respectively. Let $(X, Y) \sim P^{(X, Y)}$, and let $\kappa^{X, Y}$ be defined as in~\eqref{eq:f-divergence}. Then
    \begin{enumerate}[label = (\roman*), noitemsep]
        \item $0 \leq \kappa^{X, Y} \leq 1$; 
        \item $\kappa^{X, Y} = 0$ if and only if $X$ is independent of $Y$; 
        \item $\kappa^{X,Y} = 1$ if and only if $P^{(X,Y)}$ is singular with respect to $P^{X}\otimes P^{Y}$.
        \item For a random variable $Z$, if $X \indep Y \mid Z$, then we have $\kappa^{X, Z} \geq \kappa^{X, Y}$;
        \item For any sequence of random variables $X^{(n)}$ and $Y^{(n)}$ such that $P^{(X^{(n)}, Y^{(n)})} \dto P^{(X, Y)}$, we have $\liminf_{n \to \infty} \kappa^{X^{(n)}, Y^{(n)}} \geq \kappa^{X, Y}$;
        \item $\kappa^{X, Y} = \kappa^{Y, X}$. 
    \end{enumerate}
\end{prop}
We remark that parts (ii), (iv), (v) and (vi) demonstrate that $\kappa^{X, Y}$, as a measure of statistical association, satisfies the `zero-independence', `information-monotonicity', `lower semicontinuity' and `symmetry axioms' considered in \citet{borgonovo2025convexity} \citep[see also][]{mori2019four, renyi1959measures}. Part (iii) is related to the `max-functionality' axiom, though the coverage correlation measures a more general statistical association between $X$ and $Y$ than a purely directional functional relationship. Also, the following information gain inequality is an immediate consequence of part (iv).
\begin{enumerate}
    \item[(iv')] for any random variables $X$, $X'$ and $Y$, we have $\kappa^{(X, X'), Y} \geq \kappa^{X, Y}$.
\end{enumerate}
This inequality is not mentioned in \citet{borgonovo2025convexity}, but it appears as an axiom for dependency measurement in \citet{griessenberger2022multivariate}.

\subsubsection{Multivariate consistency}
We now extend the consistency result to the case where $X$ and $Y$ are multivariate random vectors. The proof strategy underlying Theorems \ref{Thm:populationlimit} does not extend to this setting, as it relies heavily on the notion of order statistics, which is absent for multivariate marginals. Instead, the multivariate result relies on the convergence of the empirical optimal transport map to the population version. Hence, we first present a few notions from the optimal transport theory; for a comprehensive introduction to the subject, see~\citet{villani2009optimal, villani2021topics}. 

For probability measures $\mu$ and $\nu$ on $\mathbb{R}^d$, Monge's problem \citep{monge1781memoire} seeks a measurable map $T : \mathbb{R}^d \to \mathbb{R}^d$ solving
\begin{align}
\inf_{\tilde T} \int\|u- {\tilde T}(u)\|_2^2 \, \mathrm{d} \mu \quad \text{subject to $\tilde T{\#} \mu=\nu$}, \label{eq:Monge}
\end{align}
where $(\tilde T\#\mu) (B) := \mu (\tilde T^{-1}(B))$ for any Borel set $B \subseteq \mathbb{R}^d$. We call $\mu$ the source distribution and $\nu$ the target distribution. If $\mu$ is absolutely continuous with respect to the Lebesgue measure and both $\mu$ and $\nu$ have finite second moments, Brenier's theorem~\citep{brenier1991polar} ensures the existence of a convex function $\phi : \mathbb{R}^d \to \mathbb{R}$, unique up to an additive constant, such that $T:=\nabla\phi$ is the $\mu$-almost surely unique solution of~\eqref{eq:Monge}. We refer to $\nabla\phi$ as the \emph{optimal transport map} from $\mu$ to $\nu$, and $\phi$ the \emph{Brenier potential}. In the literature, much effort has been dedicated to estimate the optimal transport map given samples from the source and the target distributions; see, among many others, \cite{hutter2021minimax, Deb2021RatesOE, divol2025optimal, manole2024plugin, balakrishnan2025stability}. The following regularity property lies at the core of obtaining the desired convergence rates for the empirical optimal transport map.

%Instead, by employing the convergence of the empirical optimal transport map, we show in Lemma~\ref{lem: VonRawData} that the consistency of~\eqref{eq:VnGen} can be reduced to the case where the joint distribution $P^{(X, Y)}$ is a copula and the vacancy is constructed directly from the i.i.d.\! samples $(X_1, Y_1), \ldots, (X_n, Y_n)$. 

%Theoretical guarantees on plug-in estimators of the optimal transport map typically requires the population optimal transport map to be the gradient of a strongly convex function with Lipschitz gradient \citep{hutter2021minimax, manole2024plugin, divol2025optimal, balakrishnan2025stability}. These assumptions often verified using the Caffarelli regularity theory \citep[see e.g.][Theorem 4.14(ii)]{villani2021topics} when the source and target distributions are supported on bounded uniformly convex sets with smooth boundaries, which does not hold for the uniform distribution on the unit cube. For our purposes, however, Lipschitz continuity is stronger than necessary. The same theory still yields H\"older continuity of the population optimal transport map under only bounded convexity assumptions on the supports. We therefore work with the following \textit{$(\theta, L)$-H\"older condition}, which entails a sub-optimal convergence rate of empirical optimal transport map but suffices to show the consistency of the coverage correlation coefficient. 

\begin{defn}\label{def:HolderConti}
Let $ \Omega_1, \Omega_2 \subseteq \mathbb R^d$, and let $T: \Omega_1\to \Omega_2$ be a map. For $\theta\in(0,1]$ and $L>0$, we say that $T$ is $(\theta,L)$-H\"older continuous on $\Omega_1$ if 
\[
\|T(x)-T(x')\|_2 \le L\|x-x'\|_2^\theta
\]
for all $x,x'\in \Omega_1$.
\end{defn} 

For $t>0$ and $\alpha \in (0, 2]$, let $\mathcal{P}_{t, \alpha}(\mathbb{R}^d)$ denote the class of
distributions on $\mathbb{R}^d$ such that for any $P \in \mathcal{P}_{t, \alpha}(\mathbb{R}^d)$ and random variable $X \sim P$, it holds that $\mathbb{E}\exp\{t\|X\|_2^\alpha\}<\infty$. For any integer $K \geq 2$ and $\theta_1, \ldots, \theta_K \in (0, 1]$, define
\begin{align*}
\mathcal{D}(\theta_1,\ldots,\theta_K) := \biggl\{(d_i)_{i = 1}^K\in\mathbb{N}^K:\; & \theta_k > \frac{d_k}{2\sum_{j\in[K]}d_j - d_k}\;\forall k\in[K],
\\
& \sum_{j\in[K]} d_j > \frac{2}{\min_{k\in[K]}\theta_k}+2\biggr\}.
\end{align*}

\begin{thm}\label{thm:ConsistencyGeneralDim}
Let $P^{(X,Y)}$ be a Borel probability measure on $\mathbb{R}^{d}$ with marginal distributions $P^X$ and $P^Y$ absolutely continuous with respect to the Lebesgue measure on $\mathbb{R}^{d_X}$ and $\mathbb{R}^{d_Y}$. Let $T_X$ and $T_Y$ be the almost surely unique optimal transport maps from $P^X$ to $\mathrm{Unif}([0,1]^{d_X})$ and from $P^Y$ to $\mathrm{Unif}([0,1]^{d_Y})$, respectively. Assume that 
\begin{enumerate}[label = (C\arabic*), ref=(C\arabic*)]
    \item\label{ass:subweibull} $P^X\in \mathcal{P}_{t_1, \beta_1}(\mathbb{R}^{d_X})$ and $P^Y\in \mathcal{P}_{t_2, \beta_2}(\mathbb{R}^{d_Y})$ for some $t_1,t_2>0$ and $\beta_1,\beta_2\in (0, 2]$; 
    \item\label{ass:holderotmap} $T_X$ is $(\theta_1, L_1)$-H\"older continuous and $T_Y$ is $(\theta_2, L_2)$-H\"older continuous, with some $\theta_1,\theta_2\in(0,1]$ and $L_1,L_2>0$.
    \item\label{ass:dimensionregime}  $(d_X, d_Y) \in \mathcal{D}(\theta_1, \theta_2)$.
\end{enumerate}
Given $(X_1,Y_1),\ldots,(X_n,Y_n)\stackrel{\mathrm{iid}}{\sim}P^{(X,Y)}$. Let $f_{\mathrm{cov}}$ be defined as in Theorem~\ref{Thm:populationlimit}. We have
\begin{equation*}
\kappa_n^{X,Y} \pto \kappa^{X,Y} := D_{f_\mathrm{cov}}(P^{(X,Y)}\,\|\,P^X\otimes P^Y),
\end{equation*}
as $n\to\infty$.
\end{thm}

We make a few remarks on assumptions of Theorem~\ref{thm:ConsistencyGeneralDim}. The main purpose of Conditions~\ref{ass:subweibull}, \ref{ass:holderotmap} and \ref{ass:dimensionregime} is to guarantee convergence of the empirical optimal transport map to the population counterpart on the data points, as stated in Proposition~\ref{prop:HolderOTConvergence}. Condition~\ref{ass:subweibull} assumes sub-Weibull tails for $X$ and $Y$, which is a common assumption for convergence of optimal transport maps in the literature \citep[see, e.g.][]{Deb2021RatesOE}, and is implied by other related conditions, such as the curvature condition on the population Brenier potentials \citep[Theorem~3]{balakrishnan2025stability}; see Proposition~\ref{prop:strongly-conv-equiv} and the preceding discussion.

If $X$ and $Y$ have $\alpha$-H\"older densities for $\alpha\in(0,1)$, that are bounded away from zero and infinity, and with compact and convex support, then Caffarelli's regularity theorem \citep[][Theorem~4.14]{caffarelli1996boundary, villani2009optimal} implies that \ref{ass:holderotmap} holds for $\theta_1 = \theta_2 = \alpha$. More recently, \citet[Theorem~5.1]{collins2025boundary} showed that under the same conditions on the densities of $X$ and $Y$, \ref{ass:holderotmap} holds for $\theta_1 = \theta_2 = 1-\epsilon$ for \emph{any} $\epsilon > 0$, with H\"older constants $L_1, L_2$ dependent on $\epsilon$. In this case, Condition~\ref{ass:dimensionregime} simplifies to $d_X+d_Y \geq 5$. 

% When $\theta_i = 1$, condition~\ref{ass:holderotmap} essentially requires the population optimal transport maps to be Lipschitz continuous. However, due to the lack of smoothness of the Boundary of unit cube, the Caffarelli's regularity theory does not yield a sufficient condition. Though some developments have been made on the regularity of the optimal transport map beyond the smooth boundary condition \citep{kolesnikov2010global, chen2023regularity}, they 

We also remark that the proof techniques used in Theorems~\ref{Thm:populationlimit} and~\ref{thm:ConsistencyGeneralDim} are completely disjoint. They prove convergence of the empirical coverage correlation to the population limit for $d_X=d_Y=1$ and, under suitable regularity conditions, for $d_X + d_Y \geq 5$. It remains an interesting open question as to what happens when $d_X + d_Y \in\{3, 4\}$, though the resolution is likely beyond the scope of the current work. Nonetheless, we are able to characterise what happens at the extremes for the dependence statistic in the proposition below. 

\begin{prop}
\label{Prop:MultivariateLimit}
    Suppose $(X_1,Y_1),\ldots,(X_n,Y_n)\iid P^{(X,Y)}$ for a Borel probability measure $P^{(X,Y)}$ on $\mathbb{R}^{d_X+d_Y}$ with marginals $P^X$ and $P^Y$ on $\mathbb{R}^{d_X}$ and $\mathbb{R}^{d_Y}$ respectively. 
    \begin{enumerate}[label = (\roman*)]
        \item If $P^{(X,Y)} = P^X\otimes P^Y$, then $\kappa_n^{X,Y}\pto 0$.
        \item If $P^{(X,Y)}$ is singular with respect to $P^X\otimes P^Y$, then $\kappa_n^{X,Y}\pto 1$.
    \end{enumerate}
\end{prop}

\subsection{Statistical test for independence}\label{sec:nullCLT}
Now we turn to establish the asymptotic theory for the coverage correlation coefficient. The following result derives the asymptotic distribution of $\kappa_n^{X,Y}$ under the null. This allows us to use the coverage correlation to perform independence testing between $X$ and $Y$. 

\begin{thm}\label{thm: CLT}
    Suppose that $X$ and $Y$ are independent random vectors on $\mathbb{R}^{d_X}$ and $\mathbb{R}^{d_Y}$ respectively. Define
    \[
        \sigma_n^2: = \frac{1}{(1-e^{-1})^2}\sum_{k=2}^n \binom{n}{k} \Bigl(1 - \frac{2}{n}\Bigr)^{n-k}\biggl\{\Bigl(\frac{2}{k + 1}\Bigr)^d n^{-k -1} - n^{-2k}\biggr\}.
    \]
    Given independent and identically distributed copies $(X_1,Y_1),\ldots,(X_n,Y_n)$ of $(X,Y)$, we have 
    \[
        \frac{\sqrt{n} \kappa_n^{X,Y}}{\sigma_n} \dto \mathcal{N}(0, 1)
    \]
    as $n\to\infty$. 
\end{thm}
Based on the above theorem, we can construct a test for 
\[
    \mathcal{H}_0: P^{(X,Y)} = P^X\otimes P^Y \quad \text{versus} \quad \mathcal{H}_1: P^{(X,Y)}\neq P^X\otimes P^Y
\]
by rejecting the null hypothesis if 
\begin{equation}\label{Eq:IndTest}
    \frac{\sqrt{n}\kappa_n^{X,Y}}{\sigma_n} \geq z_\alpha,
\end{equation}
where $z_\alpha$ is the upper $\alpha$ quantile of the standard normal distribution. Since the limiting distribution is independent of the marginal distributions $P^X$ and $P^Y$, test~\eqref{Eq:IndTest} is distribution-free. 

The normalising factor $\sigma_n/\sqrt{n}$ in Theorem~\ref{thm: CLT} is the exact standard deviation of the empirical coverage correlation under the null. As can be seen in Lemma~\ref{le: VacancyLimitRandomRef}, we have 
\[
\lim_{n\to\infty} \sigma_n^2  = \frac{1}{(e-1)^2}\sum_{k=2}^\infty \frac{1}{k!}\biggl(\frac{2}{k+1}\biggr)^{d_X+d_Y} =:\sigma^2
\]
One can equivalently construct the asymptotic test by replacing $\sigma_n$ with $\sigma$ in~\eqref{Eq:IndTest}, though using the exact variance $\sigma_n^2$ improves the finite sample performance of the test when $n$ is relatively small. When $d_X=d_Y=1$, the expression of the asymptotic variance $\sigma^2$ has an explicit value of $(e-1)^{-2}(4\mathrm{Ei}(1)-4\gamma_0-5)\approx 0.091992$, where $\gamma_0$ is the Euler--Mascheroni constant and $\mathrm{Ei}(1)=\int_{t=0}^1e^{-1/t}/t\,\mathrm{d}t$ is the exponential integral evaluated at $1$.

In settings where the consistency of the empirical coverage correlation can be proved (see Theorems~\ref{Thm:populationlimit} and~\ref{thm:ConsistencyGeneralDim}), the proposed test in~\eqref{Eq:IndTest} is asymptotically consistent under the alternative. Indeed, we have
\begin{align*} 
\mathbb{P}(\kappa_n^{X,Y} < n^{-1/2} \sigma_n z_\alpha) \leq \mathbb{P}(|\kappa_n^{X,Y} - \kappa^{X,Y}| \leq n^{-1/2} \sigma_n z_\alpha - \kappa^{X,Y}) \to 0,
\end{align*}
since $\kappa^{X,Y} > 0$ under the alternative and $\kappa_n^{X,Y} - \kappa^{X,Y} = o_p(1)$. 

While the central limit theorem in Theorem~\ref{thm: CLT} allows us to derive asymptotically valid p-values, the worst-case relative error of such p-values can still be large in the tails of the normal distribution (e.g.\ a bound of order $n^{-1/2}$ from the Berry--Esseen theorem \citep{berry1941accuracy,esseen1942liapunov}). However, using the fact that the coverage correlation coefficient statistic exhibits weak dependence on individual data point $(X_i, Y_i)$, we are able to derive a finite-sample concentration inequality.

\begin{thm} \label{thm: concentration}
Let $P^{(X,Y)}$ be Borel probability measure on $\mathbb{R}^2$. There exists a universal constant $C>0$ such that for any $t > 0$, we have
\[
   \mathbb{P}(|\mathcal{V}_n - \mathbb{E}\mathcal{V}_n| \geq t) \leq 2(n+1)e^{-C\min\{nt^2, (nt^2)^{1/3}\}}.
\]
\end{thm}

\subsection{Derandomising the coverage correlation}
\label{sec: derandomising}
We have defined the coverage correlation using uniformly distributed reference points $\bs{U}$ and $\bs{V}$. In this subsection, we discuss the possibility of replacing them with deterministic reference points to derandomise the coverage correlation coefficient, which can be desirable in practice.  Given fixed reference points $\bs{u} = (u_1,\ldots,u_n)$ and $\bs{v} = (v_1,\ldots,v_n)$, we denote the corresponding vacancy as 
\begin{align}
\mathcal{V}_n^{\mathrm{grid}} \equiv \mathcal{V}_n^{\mathrm{grid}}(\bs{u},\bs{v}) := \mathcal{V}\Bigl((X_i)_{i\in[n]},(Y_i)_{i\in[n]},\bs{u},\bs{v};\frac{1}{2\sqrt{n}}\Bigr). \label{eq:RegularGridVacancy}
\end{align}

When $X$ and $Y$ are univariate, it is also natural to use a regular grid reference $\bs{u} = \bs{v} = (1/n, 2/n, \ldots, 1)$. In this case, the empirical coverage correlation coefficient converges to the same population limit as in Theorem~\ref{Thm:populationlimit} and enjoys the same limiting distribution as in Theorem~\ref{thm: CLT}.

\begin{thm}
\label{Thm:PopulationLimitGrid}
Let $P^{(X,Y)}$ be a Borel probability measure on $\mathbb{R}^2$ with marginals $P^X$ and $P^Y$. Given $(X_1,Y_1),\ldots,(X_n,Y_n)\iid P^{(X,Y)}$, and $\bs{u} = \bs{v} = (1/n,2/n,\ldots,1)$, define
\begin{align}
\kappa_n^{X,Y;\mathrm{grid}} \equiv \kappa_n^{X,Y;\mathrm{grid}}(\bs{u},\bs{v}) := \frac{\mathcal{V}_n^{\mathrm{grid}}(\bs{u},\bs{v})-e^{-1}}{1-e^{-1}}. \label{eq:CovCorrGrid}
\end{align}
We have as $n \to \infty$ that
\[
\kappa_n^{X,Y; \mathrm{grid}} \pto \kappa^{X,Y}.
\]
\end{thm}

When $X$ and $Y$ are independent, the coverage correlation coefficient with the regular grid reference also enjoys a central limit theorem with the same asymptotic variance as in Theorem~\ref{thm: CLT}.

\begin{thm}\label{thm:RegularGridCLT}
Let $P^X$ and $P^Y$ be Borel probability measures on $\mathbb{R}$ and define $P^{(X,Y)} = P^X\otimes P^Y$. For $(X_1,Y_1),\ldots,(X_n,Y_n)\iid P^{(X,Y)}$ and $\kappa^{X, Y; \mathrm{grid}}_n$ defined as in~\eqref{eq:CovCorrGrid} with $\bs{u} = \bs{v} = (1/n, 2/n,\ldots,1)$,  we have as $n\to\infty$ that
\begin{align}
    \sqrt{n}(\kappa^{X, Y; \mathrm{grid}}_n - \E \kappa^{X, Y; \mathrm{grid}}_n) \dto \mathcal{N}\Bigl(0,  \frac{4 \mathrm{Ei}(1) - 4\gamma_0 - 5}{(1-e)^2}\Bigr).\label{eq:GridNullCLT}
\end{align}
\end{thm}
Note that the convergence in~\eqref{eq:GridNullCLT} is around $\mathbb{E}\kappa_{n}^{X, Y; \mathrm{grid}}$ instead of its asymptotic limit. But when $X$ and $Y$ are independent, Lemma~\ref{lem:varVacancyReg} yields the explicit expansion
\begin{align*}
\E \kappa_{n}^{X, Y; \mathrm{grid}}
&= \frac{e^{-1}}{1-e^{-1}}\biggl(-\frac{1}{\sqrt{n}}-\frac{1}{6n}+O(n^{-3/2})\biggr).
\end{align*}
Thus one may safely replace the expectation by its leading term $-e^{-1}n^{-1/2}/(1-e^{-1})$ in practice when $n$ is large.

Another benefit of using the grid reference is a stronger concentration result. The concentration result in Theorem~\ref{thm: concentration} is not sub-Gaussian for large deviations due to the fact that the spacings between consecutive order statistics of the uniformly random reference points $\bs U$ and $\bs V$ have sub-Gamma tails and deviate from the expected spacing of $1/(n+1)$. However, if we use the uniform grid reference, which has constant spacing between consecutive grid points, a sub-Gaussian concentration is available. 

\begin{prop}\label{prop: concentrationGrid}
    Let $P^{(X,Y)}$ be a Borel probability measure on $\mathbb{R}^2$ and let $\kappa^{X, Y; \mathrm{grid}}_n$ be defined as in~\eqref{eq:CovCorrGrid} with $\bs{u} = \bs{v} = (1/n, 2/n,\ldots,1)$.     There exists a universal constant $C>0$ such that for any $t>0$, we have 
    \[
        \mathbb{P}\bigl(|\kappa_n^{X,Y;\mathrm{grid}} - \mathbb{E}(\kappa_n^{X,Y;\mathrm{grid}})| \geq t\bigr) \leq 2e^{-Cnt^2}.
    \]
\end{prop}

We next consider the case where both $X$ and $Y$ are multivariate. Recall that the 2-Wasserstein distance between two probability measures $\mu$ and $\nu$ on $\mathbb{R}^d$ is defined as
\[
\mathcal{W}_2(\mu, \nu) :=
\biggl(\inf_{\gamma\in\Pi(\mu,\nu)}
\int \|x-y\|_2^2 \, \mathrm d\gamma(x,y)\biggr)^{1/2},
\]
where $\Pi(\mu,\nu)$ denotes the set of all couplings of $\mu$ and $\nu$. The following theorem shows that the conclusion of Theorem~\ref{thm:ConsistencyGeneralDim} remains valid for deterministic reference points approximating $\mathrm{Unif}([0, 1]^d)$ in 2-Wasserstein distance, under an additional regularity condition on the corresponding Brenier potential.  

\begin{thm}
    \label{Thm:ConsistencyGeneralDimGrid}
    Let $(X_1,Y_1),\ldots, (X_n,Y_n)\iid P^{(X,Y)}$ for a joint distribution $P^{(X,Y)}$ with $d_X$- and $d_Y$-dimensional marginals that are absolutely continuous satisfying conditions~\ref{ass:subweibull}, \ref{ass:holderotmap} and \ref{ass:dimensionregime} as in Theorem~\ref{thm:ConsistencyGeneralDim}. Let $\bs U=(u_1,\ldots,u_n)$ and $\bs V=(v_1,\ldots,v_n)$ be reference grid points (either deterministic or random) in $[0,1]^{d_X}$ and $[0,1]^{d_Y}$, respectively such that $\mu_n=n^{-1}\sum_{i=1}^n\delta_{u_i}$ and $\nu_n=n^{-1}\sum_{i=1}^n\delta_{v_i}$ satisfy
\begin{equation}\label{eq:UniformApprox}
    \begin{aligned}
    \E\mathcal{W}_2\bigl(\mu_n, \mathrm{Unif}([0,1]^{d_X})\bigr) = O(n^{-\frac{1}{d_X}}\log^2 n), \!\! \quad \E\mathcal{W}_2\bigl(\nu_n, \mathrm{Unif}([0,1]^{d_Y})\bigr) = O(n^{-\frac{1}{d_Y}}\log^2 n). 
\end{aligned}
\end{equation}

Define 
\[  
\kappa_n^{X,Y;\mathrm{grid}}(\bs{U},\bs{V}) := \frac{\mathcal{V}_n^{\mathrm{grid}}(\bs{U},\bs{V}) -e^{-1}}{1-e^{-1}}.
\]
If the Brenier potentials associated with the optimal transport maps from $\mathrm{Unif}([0,1]^{d_X})$ to $P^{X}$ and from $\mathrm{Unif}([0,1]^{d_Y})$ to $P^{Y}$ are both Lipschitz, then we have as $n\to\infty$ that 
\[  
\kappa_n^{X,Y;\mathrm{grid}}(\bs{U},\bs{V}) \pto \kappa^{X,Y}.
\]
\end{thm}

The equation~\eqref{eq:UniformApprox} is readily satisfied by the uniform random reference points $u_i = U_i \stackrel{\mathrm{iid}}{\sim}\mathrm{Unif}([0, 1]^{d_X})$ and $v_i = V_i \stackrel{\mathrm{iid}}{\sim}\mathrm{Unif}([0, 1]^{d_Y})$, $i\in[n]$ \citep{fournier2015rate}. The flexibility of the condition in~\eqref{eq:UniformApprox} also allows us to use deterministic reference points such as the following.
\begin{enumerate}[label = \roman*.]
    \item{(Regular grid)} Suppose $n = m^{d_X} = \ell^{d_Y}$ for $m,\ell\in\mathbb{N}$, then ${\bs U} = \{1/m, 2/m, \ldots, 1\}^{d_X}$ and ${\bs V} =  \{1/\ell, 2/\ell, \ldots, 1\}^{d_Y}$ satisfies~\eqref{eq:UniformApprox} \citep[see e.g.,][Example 4.17]{graf2000foundations}. 
    \item{(Low-discrepancy sequence)} The Halton sequence and the Sobol' sequence, or any other low-discrepancy quasi-random sequences with star discrepancy of order $O((\log n)^d / n)$  also satisfy~\eqref{eq:UniformApprox} (see Proposition~\ref{prop:WassBoundbyStar} and discussion immediately preceding it for definition of star discrepancy and its relationship to the $2$-Wasserstein distance, and \citet[Theorem~3.6 and 4.17]{niederreiter1992random} for references that Halton and Sobol' sequences satisfy the required star discrepancy).
\end{enumerate}

Similarly to Theorem~\ref{thm:ConsistencyGeneralDim}, Caffarelli's regularity theory provides sufficient conditions for Conditions~\ref{ass:subweibull}--\ref{ass:dimensionregime}, as well as for the Lipschitz continuity of the Brenier potentials. Under these conditions, equation~\eqref{eq:UniformApprox} and the dimensional requirement $d_X + d_Y \geq 5$ are the only assumptions needed.

\subsection{Dependency among $K$ random vectors} \label{sec:DependencyAmongKRV}
Beyond the bivariate setting, the proposed correlation coefficient naturally extends to a measure of dependence among $K \geq 2$ random vectors. Let ${Z}=(X^{(1)},\ldots,X^{(K)})$ be a $K$-tuple of random vectors with joint distribution $P$ and marginal distributions $P_1,\ldots,P_K$, where $P_k$ denotes the law of $X^{(k)}\in \mathbb{R}^{d_k}$, $k = 1, \ldots, K$. Given $Z_1, \ldots, Z_n \stackrel{\mathrm{iid}}{\sim} P$ where ${Z}_i = (X_i^{(1)}, \ldots, X_i^{(K)})$, $i = 1, \ldots, n$, we are interested in measuring the departure of $P$ from the product measure $P_1\otimes\cdots\otimes P_K$.

Following a similar construction as in Section~\ref{sec:CoveCorrConstruction}, for each $k \in [K]$, let ${U}^{(k)} = (U^{(k)}_1, \ldots U^{(k)}_n)$ be a tuple of reference points of dimension $d_k$, and write ${Z}^{(k)} := (X_1^{(k)}, \ldots X_n^{(k)})$. Then for each $i \in [n]$ the empirical multivariate rank of $X_i^{(k)}$ is 
\begin{align}
R^{(k)}_i := U^{(k)}_{\pi^{(k)}(i)}, \quad \text{where} \quad \pi^{(k)} := \argmin_{\pi \in \mathcal{S}^n} \frac{1}{n} \sum_{i = 1}^n \|X_i^{(k)} - U^{(k)}_{\pi(i)}\|_2^2. \label{eq:MultiRankKsample}
\end{align}
Define the joint rank of $Z_i$ as 
\begin{align}
R_i  := (R_i^{(1)}, \ldots, R_i^{(K)}) \in [0, 1]^{d}, \label{eq:JointRankKcomponent}
\end{align}
where $d :=\sum_{k = 1}^K d_k$. Analogously to~\eqref{eq:VnGen}, letting $\gamma = \frac{1}{2n^{1/d}}$, we can define the vacancy as
\[
 \mathcal{V}_{n}:= \mathcal{V} \bigl(({Z}^{(k)})_{k = 1}^K, ({U}^{(k)})_{k = 1}^K; \gamma\bigr) = 1 - \mathrm{vol}\Bigl(\bigcup_{i = 1}^n B(R_i, \gamma)\Bigr),
\]
where we draw reference points $U^{(k)}_i \sim \mathrm{Unif}([0,1]^{d_k})$ independently. The empirical coverage correlation coefficient of the $K$-tuple $(Z^{(1)}, \ldots, Z^{(K)})$ is 
\begin{align}
\kappa_n^{Z} := \frac{ \mathcal{V}_{n} - e^{-1}}{1 - e^{-1}}. \label{eq:CovCorrKComponent}
\end{align}

The next theorem gives the $K$-tuple analogue of Theorem~\ref{Thm:populationlimit}. It establishes the consistency of the coverage correlation coefficient for testing mutual independence among the coordinates of a random vector.
\begin{thm}
\label{thm:KVariateCoverageConsistency}
    Let $P$ be a Borel probability measure on $\mathbb{R}^K$ with univariate marginals $P_1, \ldots, P_K$. Given $Z_1, \ldots, Z_n \stackrel{\mathrm{iid}}{\sim} P$, where $Z_i \in \mathbb{R}^K$, we have as $n\to\infty$ that 
    \[
    \kappa_n^Z  \pto \kappa^{Z}:= D_{f_{\mathrm{cov}}}(P \| \otimes_{i = 1}^K P_i).
    \]
\end{thm}

The preceding result deals with the case, where each component of the $K$-tuple is univariate. By a similar argument as Theorem~\ref{thm:ConsistencyGeneralDim}, the following theorem generalises it to obtain the corresponding consistency result for $K$-tuple with possibly multivariate components. 
\begin{thm}\label{thm:ConsistencyKComponent}
    Let $P$ be a Borel probability measure on $\mathbb{R}^{d}$ with absolutely continuous marginals $P_1,\ldots, P_K$. For any $k \in [K]$, let $T_k$ be the a.e.-unique optimal transport map from $P_k$ to $\mathrm{Unif}([0,1]^{d_k})$, and assume that
    \begin{enumerate}[label = (C\arabic*'), ref=(C\arabic*')]
    \item $P_k\in \mathcal{P}_{t_k, \alpha_k}(\mathbb{R}^{d_k})$ for some $t_k>0$ and $\alpha_k\in (0, 2]$; 
    \item $T_k$ is $(\theta_k, L_k)$-H\"older continuous with some $\theta_k\in(0,1]$ and $L_k>0$; 
    \item\label{ass:Kdimensionregime}  $(d_1, \ldots, d_K) \in \mathcal{D}(\theta_1, \ldots, \theta_K)$. 
\end{enumerate}

Given $Z_1,\ldots,Z_n \stackrel{\mathrm{iid}}{\sim}P$, we have as $n\to\infty$ that
\begin{equation*}
\kappa_n^{Z} \pto \kappa^{Z}.
\end{equation*}
\end{thm}

The following central limit theorem allows us to construct asymptotically valid p-values to test mutual independence of the components in the $K$-tuple given a sample of size $n$.
\begin{thm}
\label{Thm:KVariateCoverageCLT}
Let $Z=(X^{(1)},\ldots,X^{(K)})\in
    \mathbb{R}^{d_1}\times\cdots\times\mathbb{R}^{d_K}$ be a random vector such that $X^{(1)},\ldots,X^{(K)}$ are mutually independent. Let $d = \sum_{k = 1}^K d_k$ and
\[
    \sigma_{n}^2 := \frac{1}{(1-e^{-1})^2}\sum_{k=2}^n \binom{n}{k}
    \Bigl(1 - \frac{2}{n}\Bigr)^{n-k}
    \biggl\{\Bigl(\frac{2}{k + 1}\Bigr)^{d} n^{-k -1} - n^{-2k}\biggr\}.
\]
Given independent copies $Z_1,\ldots, Z_n$ of $Z$, let $\kappa_n^Z$ be the coverage correlation coefficient defined as in~\eqref{eq:CovCorrKComponent}. Then we have
\[
    \frac{\sqrt{n} \kappa_n^{Z}}{\sigma_{n}} \dto \mathcal{N}(0, 1)
\]
as $n\to\infty$.
\end{thm}

\section{Empirical studies}
\subsection{Numerical simulations}
\label{Sec:Numerics}
Table~\ref{Tab:Size} summarises the finite-sample sizes of the independence test based on the coverage correlation coefficient $\kappa_n^{X,Y}$ at various nominal levels, for $n \in\{10,100,1000\}$ and $d_X=d_Y\in\{1,2\}$. We see that the test is a bit conservative at relatively low sample sizes and well-calibrated for large sample sizes. 

\begin{table}
\centering
    \begin{tabular}{cccccc}
         $n$ & $d_X$ & $\alpha=1\%$ & $\alpha=2.5\%$ & $\alpha=5\%$ & $\alpha=10\%$\\
         \hline
$10$ & $1$ & $0.69_{(0.03)}$ & $1.54_{(0.04)}$ & $3.03_{(0.05)}$ & $6.02_{(0.08)}$\\
$100$ & $1$ & $0.93_{(0.03)}$ & $2.27_{(0.05)}$ & $4.34_{(0.06)}$ & $8.78_{(0.09)}$\\
$1000$ & $1$ & $0.96_{(0.03)}$ & $2.34_{(0.05)}$ & $4.76_{(0.07)}$ & $9.50_{(0.09)}$\\
$10$ & $2$ & $0.55_{(0.02)}$ & $1.18_{(0.03)}$ & $2.10_{(0.05)}$ & $4.08_{(0.06)}$\\
$100$ & $2$ & $0.94_{(0.03)}$ & $2.11_{(0.05)}$ & $4.12_{(0.06)}$ & $8.03_{(0.09)}$\\
$1000$ & $2$ & $0.97_{(0.03)}$ & $2.36_{(0.05)}$ & $4.66_{(0.07)}$ & $9.30_{(0.09)}$
    \end{tabular}
    \caption{\label{Tab:Size}Empirical sizes (in percentage) of independence test based on coverage correlation coefficient at various nominal levels $\alpha$, estimated over 100000 Monte Carlo repetitions (with standard errors in brackets). }
\end{table}

In Figure~\ref{Fig:Power}, we compare the power of the test based on $\kappa_n^{X,Y}$ to those of Chatterjee's correlation \citep[][implemented in the \texttt{XICOR} R package]{chatterjee2021new}, distance correlation (dCor) \citep[][implemented in the \texttt{Rfast} R package]{szekely2007measuring}, Hilbert--Schmidt Independence Criterion  (HSIC) \citep[][implemented in the \texttt{dHSIC} R package]{gretton2008kernel}, the kernel measure of association (KMAc) \citep[][implemented in the \texttt{KPC} R package]{deb2020measuring} and the U-statistics permutation test (USP) \citep[][implemented in the \texttt{USP} R package]{berrett2021optimal}. For dCor, HSIC, KMAc and USP, we run 100 permutations to obtain p-values. Also, for USP, we set $M=3$ for the maximum frequency to use in the Fourier basis. We generate $n\in\{1000,2000\}$ independent copies of $(X,Y)$ pair in $\mathbb{R}^{d_X}\times \mathbb{R}^{d_Y}$ for $d_X=d_Y\in\{1,2\}$ from one of the six data generating mechanisms described below at different noise levels $\gamma\in\{0,0.2,\ldots,1.8,2\}$ (all functions below are applied componentwise for vector inputs and $(\epsilon_X, \epsilon_Y) \sim \mathcal{N}(0,I_{d_X})\otimes \mathcal{N}(0,I_{d_Y})$ is independent of all other randomness):
\begin{enumerate}[label=(\roman*)]
    \item sinusoidal: $X\sim \mathrm{Unif}([-1,1]^{d_X})$, $Y = \cos(8\pi X) + \gamma \epsilon_Y$
    \item zigzag: $X\sim \mathrm{Unif}([-1,1]^{d_X})$, $Y = |X - 0.5\mathrm{sgn}(X)| + \gamma \epsilon_Y$
    \item circle: $U \sim \mathrm{Unif}([0,2\pi]^{d_X})$, $X = \cos(U) + 0.5\gamma \epsilon_X$, $Y = \sin(U) + 0.5\gamma\epsilon_Y$
    \item spiral: $U \sim \mathrm{Unif}([0,1]^{d_X})$, $X = U \sin(10\pi U) + 0.15\gamma\epsilon_X$, $Y = U \cos(10\pi U) + 0.15\gamma\epsilon_Y$
    \item Lissajous: $U \sim \mathrm{Unif}([0,1]^{d_X})$, $X = \sin(3U + \pi / 2) + 0.1\gamma\epsilon_X$, $Y = \sin(4U) + 0.1\gamma\epsilon_Y$
    \item local: $Z \sim \mathcal{N}(0, I_{d_X})$, $W \sim \mathcal{N}(0, I_{d_Y})$, $X = Z+ 0.8\epsilon_X$, $Y = \mathbbm{1}_{\{Z>0, W>0\}} Z + (1- \mathbbm{1}_{\{Z>0, W>0\}}) W + \epsilon_Y$.
\end{enumerate}

\begin{figure}[htbp]
\begin{center}
    \begin{tabular}{cc}
    \includegraphics[width=0.45\textwidth]{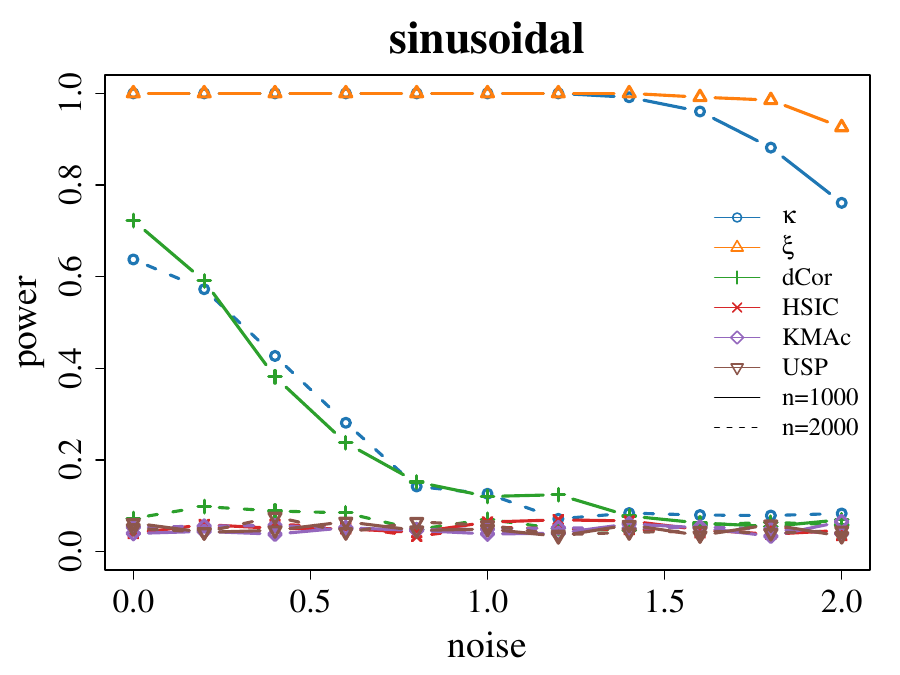}
    &
    \includegraphics[width=0.45\textwidth]{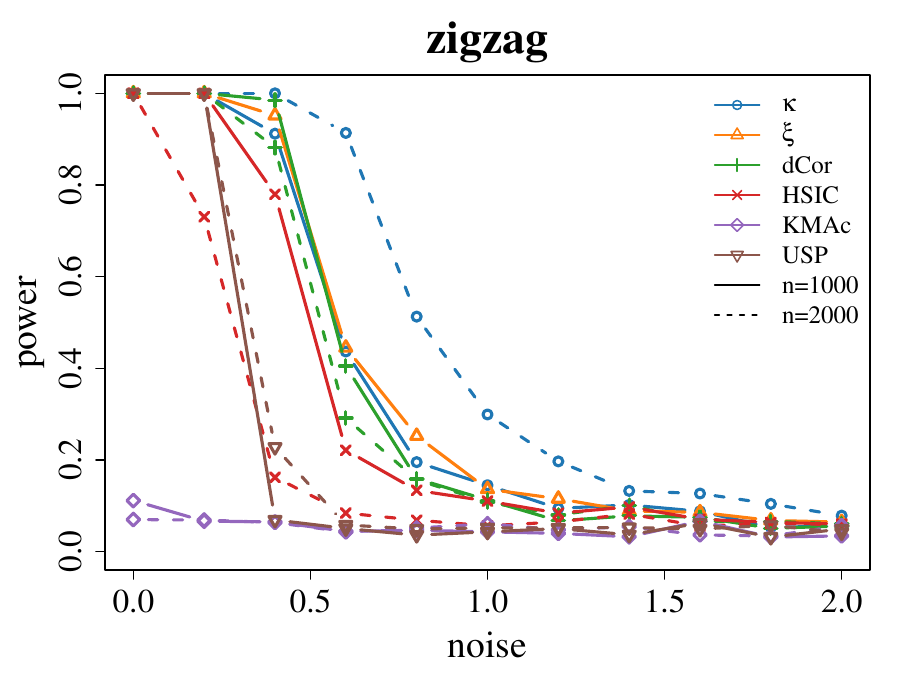}
    \\
    \includegraphics[width=0.45\textwidth]{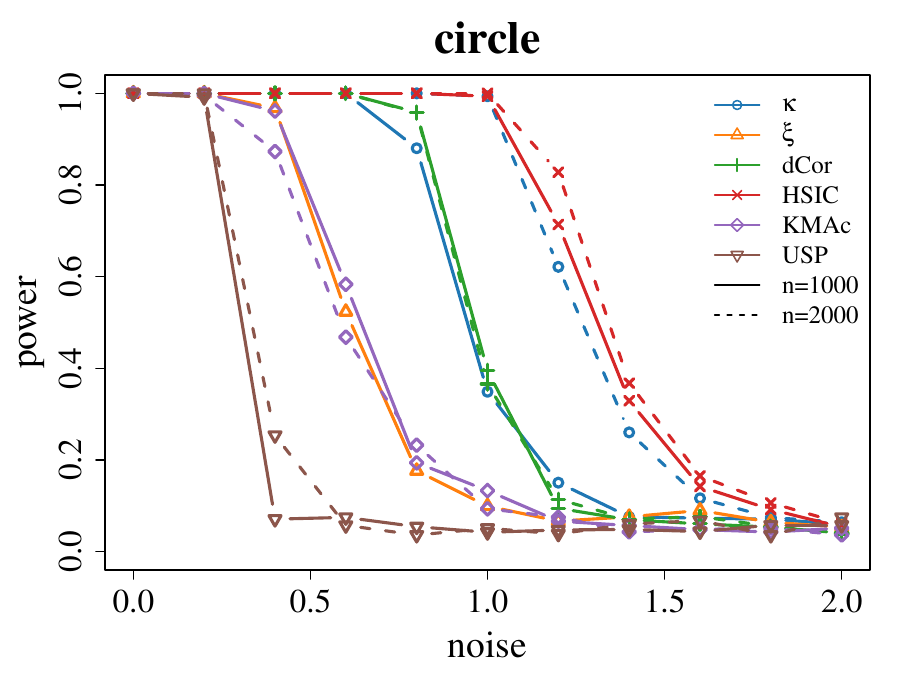}
    &
    \includegraphics[width=0.45\textwidth]{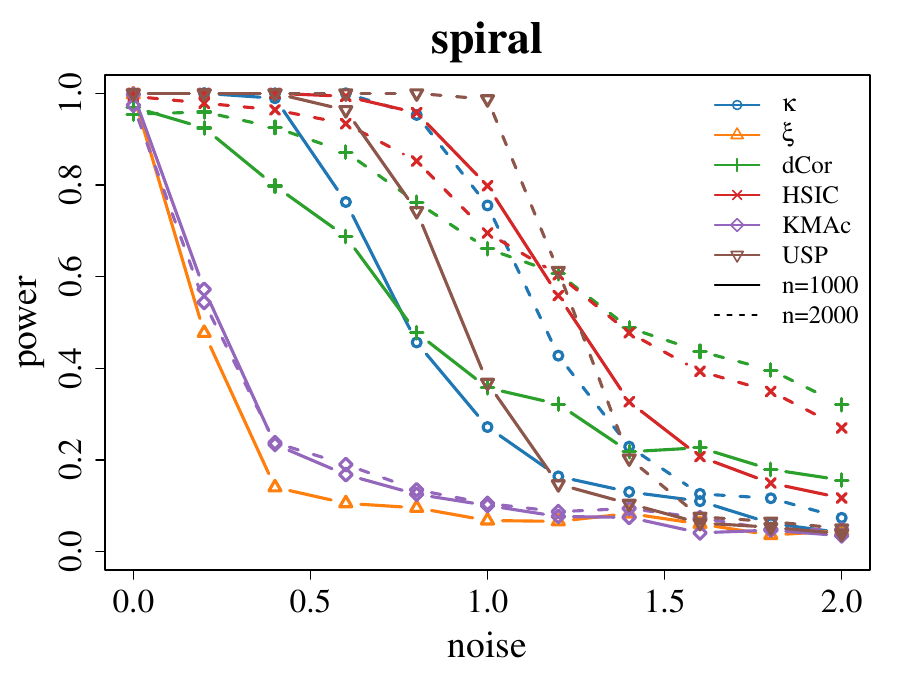}
    \\
    \includegraphics[width=0.45\textwidth]{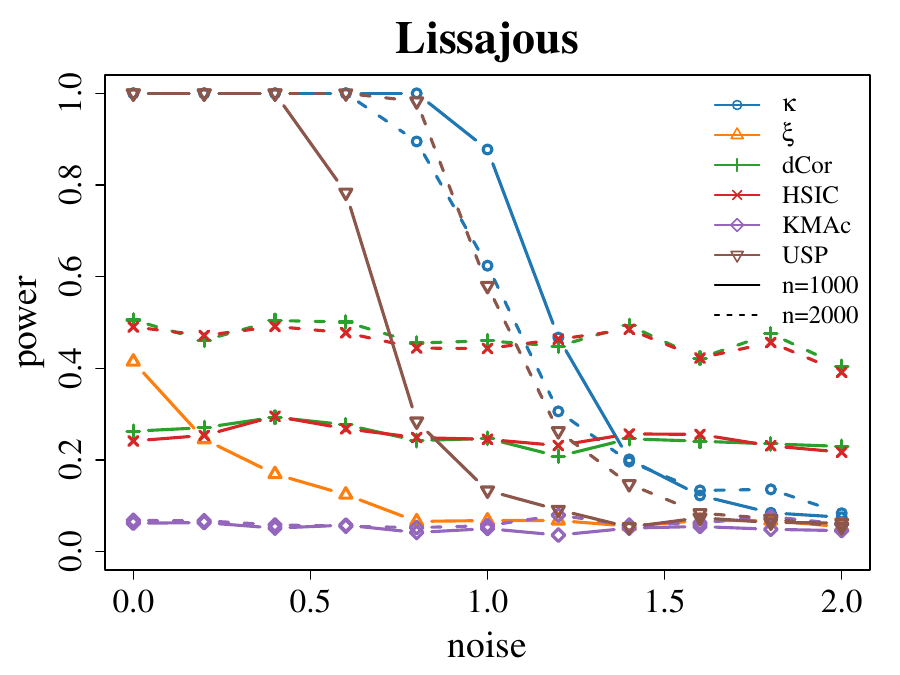}
    &
    \includegraphics[width=0.45\textwidth]{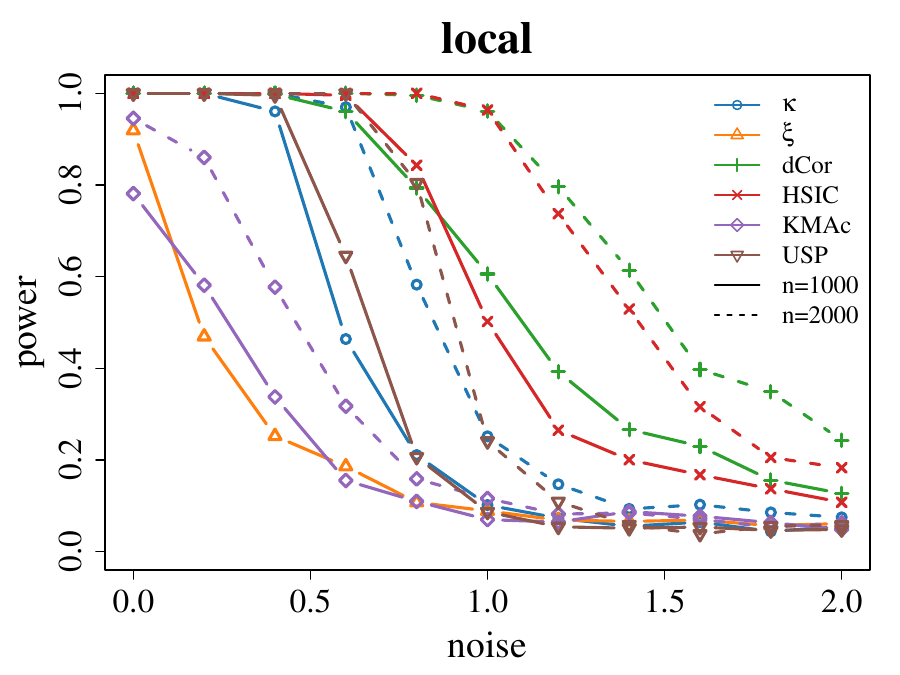}
    \end{tabular}
\end{center}
\caption{\label{Fig:Power}Power curves, estimated over 500 Monte Carlo repetitions, are presented for the coverage correlation, Chatterjee's correlation, distance correlation, HSIC, KMAc and USP in six data-generating scenarios described in Section~\ref{Sec:Numerics} for $(n,d_X,d_Y) \in\{ (1000,1,1), (2000,2,2)\}$ and noise level $\gamma\in\{0,0.2,0.4,\ldots,2\}$. Power is evaluated at the nominal level $0.05$. The solid lines correspond to the setting $(n,d_X,d_Y) = (1000,1,1)$ while the dashed lines represent the setting $(n,d_X,d_Y) = (2000,2,2)$.}
\end{figure}

Appendix~\ref{Subsec:SimulationSetting} includes representative scatter plots across various noise levels and reports the average runtime of the algorithms. We remark that in the `sinusoidal' and `zigzag' settings, $Y$ can be viewed as a function of $X$ with added noise. Here, Chatterjee's correlation performs best (though it is only applicable when $d_X=d_Y=1$), and the coverage correlation, while testing against a wider range of alternatives, has similar power performance. In `circle', `spiral', `Lissajous' settings, $X$ and $Y$ are, up to additive noise, functions of a latent variable $U$, and in the `local' setting, the joint distribution of $(X,Y)$ has a lower-dimensional support in the first quadrant only. In these cases, the coverage correlation shows better power than Chatterjee's correlation and KMAc, which are designed to test functional relationships. HSIC, dCor and USP perform well in a subset of these scenarios, though they all show poor performance in at least one setting. Moreover, since these methods rely on permutation-based p-value computation, they may not scale well, particularly when the multiple testing burden is high.

\subsection{Real data}
We use the coverage correlation coefficient to study two biological datasets. 
We first consider the menstrual-cycle hormone dataset from \citet{stricker2006establishment}, which records the evolution of four key reproductive hormones, estradiol, progesterone, luteinising hormone (LH), and follicle-stimulating hormone (FSH), over the course of a menstrual cycle (see Figure~\ref{Fig:hormone_ts}). We study this dataset because the underlying biological relationships are well established. These hormones form a tightly regulated feedback system governing the menstrual cycle and are therefore known \emph{a priori} to exhibit substantial dependence, making it possible to evaluate the performance of different methods against the biological ground truth. We compute correlation statistics and associated p-values using Pearson, Spearman, Chatterjee, and coverage correlations. Among these methods, only coverage correlation detects statistically significant dependence at the 5\% level for every pair of hormones (see Figure~\ref{Fig:hormone}).
\begin{figure}[htbp]
\centering
 \includegraphics[width=0.5\textwidth]{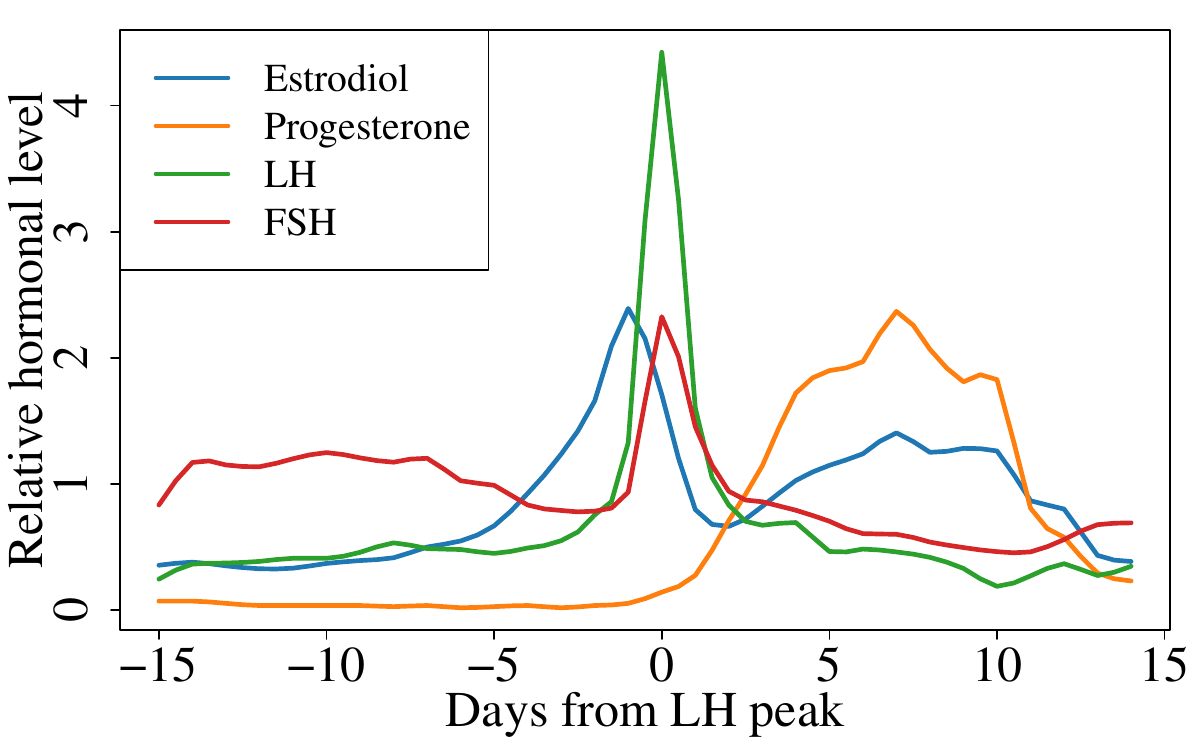}
     \caption{\label{Fig:hormone_ts}Temporal evolution of four reproductive hormones across the menstrual cycle relative to the LH peak. Curves show trajectories of estradiol, progesterone, LH and FSH against days from the LH peak. Hormone levels were normalized to place the four trajectories on a comparable scale.}
\end{figure}
\begin{figure}[htbp]
    \centering
     \includegraphics[width=0.8\textwidth]{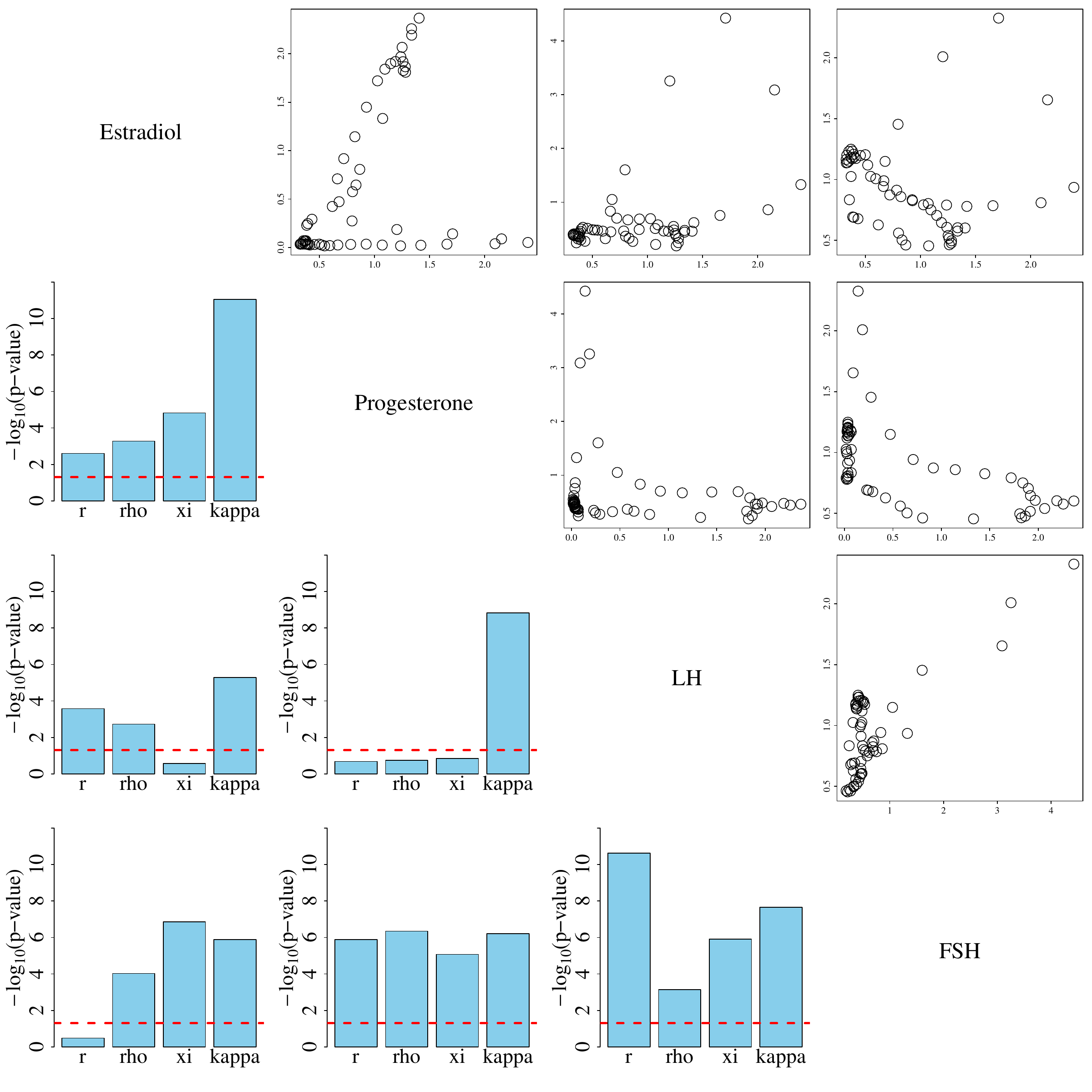}
    \caption{\label{Fig:hormone}Pairwise associations among estradiol, progesterone, LH, and FSH across the menstrual cycle. Upper-triangular panels show scatterplots of hormone pairs. Lower-triangular panels display bar plots of $-\log_{10}(\text{p-value})$ of four dependence measures, namely Pearson correlation ($r$), Spearman correlation ($\rho$), Chatterjee’s rank correlation ($\xi$), and coverage correlation ($\kappa$), with the dashed red line indicating $5\%$ significance level. Diagonal panels label the corresponding hormones.}
\end{figure}

We then look at a single-cell RNA sequencing dataset from \citet{suo2022mapping}. We use a subset of the data used in the paper, consisting of the gene expression levels of top $p=1000$ highly variable genes measured in $n=9369$ CD8\textsuperscript{+} T cells (the processed data are available in the \texttt{covercorr} R package). We compute all $\binom{p}{2}$ pairwise correlations using Pearson's correlation, Spearman's correlation, Chatterjee's correlation and the coverage correlation and adjust the corresponding p-values via Bonferroni correction. We identified 54 gene pairs as significant by coverage correlation but not by any of the other methods. The two most significant pairs are plotted in Figure~\ref{Fig:CD8}. As can be seen, for both pairs, the scatter plots of pairwise gene expression levels exhibit a clear L-shaped relationship suggestive of an implicit functional dependence.

\begin{figure}[htbp]
    \centering
    \includegraphics[width=0.6\textwidth]{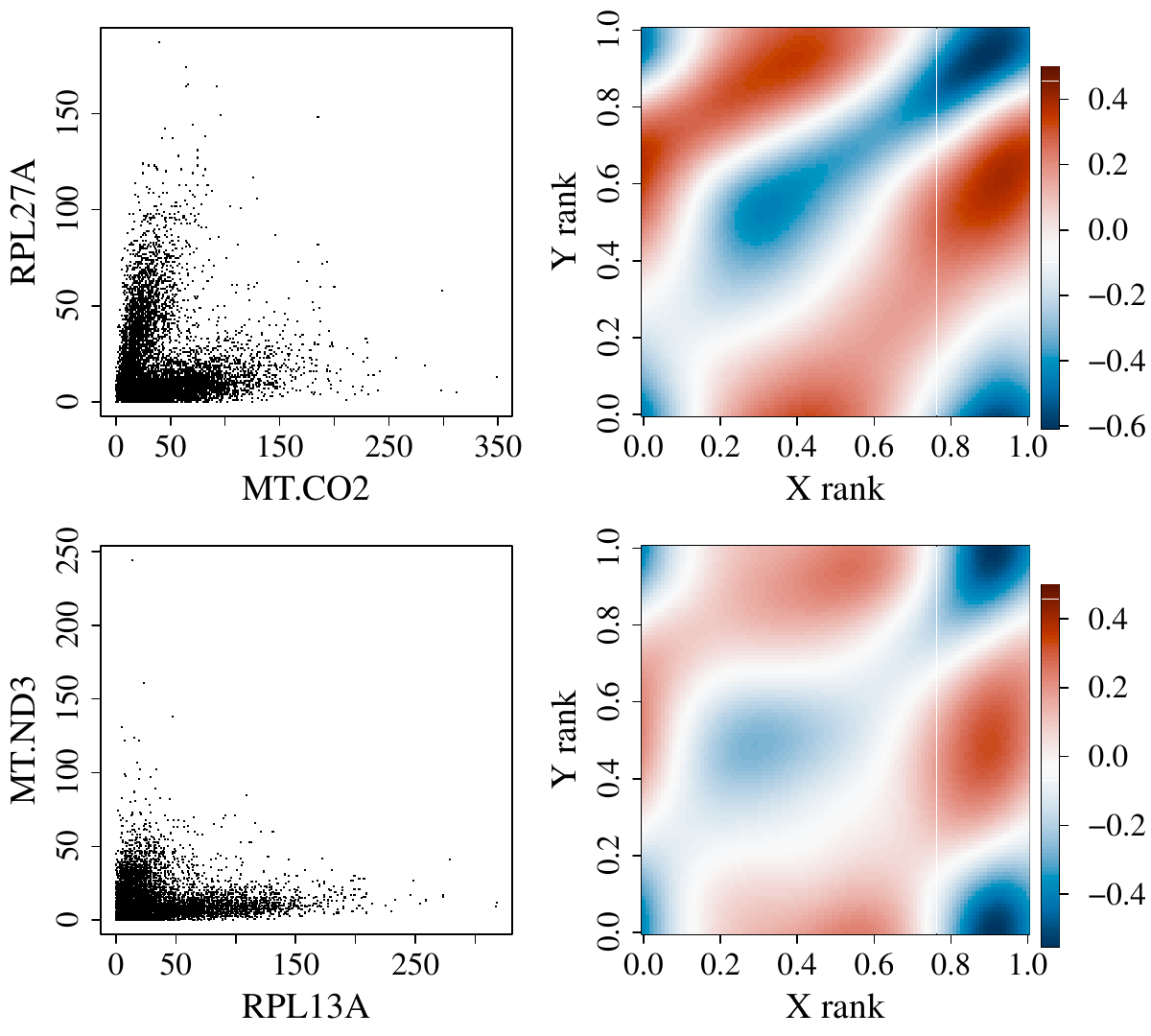}
    \caption{\label{Fig:CD8}The top two gene pairs with significant coverage correlation after Bonferroni correction, but are not significant under Pearson, Spearman, or Chatterjee's correlation tests. The right column shows heatmaps of excess density of covered area in the coverage correlation calculation of the corresponding gene pairs on the left column.}
\end{figure}

\section{Discussion}
\label{Sec:Discussion}

\begin{paragraph}{Choice of reference points}
As mentioned in Section~\ref{sec:CoveCorrConstruction}, choosing the reference distribution of $U$ and $V$ to be uniform over the unit cube is the most natural choice, since the coverage volume computed this way has equal contribution from each of the reference points.  Other reference distributions could in principle also be used. For example, one could draw $U$ and $V$ from $\mathcal{N}_{d_X}(0,I_{d_X})$ and $\mathcal{N}_{d_Y}(0,I_{d_Y})$, respectively. However, before computing the covered volume, it would still be necessary to transform these reference points back to $[0,1]^{d_X}$ and $[0,1]^{d_Y}$ via the corresponding multivariate cumulative distributional maps. The coverage correlation computed after this transformation has the same distribution as the statistic obtained by choosing the reference points directly from uniform distributions. In this sense, the coverage correlation is distribution-free with respect to the choice of reference distributions for $U$ and $V$, provided that the final coverage calculation is carried out in the unit cube after the relevant distributional transformations. 
\end{paragraph}

\begin{paragraph}{Relationship to other $f$-divergences}
The coverage divergence $\kappa^{X,Y}$ is an $f$-divergence between
$P^{(X,Y)}$ and $P^X \otimes P^Y$ with generator $f_{\mathrm{cov}}$ defined in~\eqref{Eq:fcov}. Unlike classical choices such as Kullback--Leibler, total variation,
Hellinger, and Jensen--Shannon divergences,\footnote{For background on $f$-divergences and
related statistical divergences, see, e.g.~\citet{Morimoto1963}, \citet{AliSilvey1966},
\citet{Csiszar1967}, \citet{Sason2022}, and \citet{PolyanskiyWu2025}.} whose $f$-generators
are unbounded, $f_{\mathrm{cov}}$ is uniformly bounded on $[0,\infty)$.
As a consequence, arbitrarily large density ratios cannot produce
arbitrarily large pointwise contributions to the divergence. This gives
the coverage divergence a form of robustness that is absent from
divergences such as mutual information, whose generator
$f_{\mathrm{KL}}(x)=x\log x$ places increasing weight on regions where
the density ratio is large and may therefore take infinite values. The
bounded generator also reflects a different geometric perspective on
dependence: rather than accumulating contributions from extreme
density ratios, $\kappa^{X,Y}$ measures the extent to which the joint
law fails to cover the product space. 

A further advantage of the coverage divergence is its ease of estimation.
Many divergence-based dependence measures, including mutual information,
require nonparametric estimation of joint and marginal densities or their
ratio. The substantial literature on mutual information estimation includes
kernel and histogram plug-in estimators, nearest-neighbour methods such as
the KSG estimator \citep{kraskov2004estimating}, entropy-based estimators based on
$k$-nearest-neighbour distances \citep{KozachenkoLeonenko1987,BerrettSY2019},
and variational neural estimators \citep{BelghaziEtAl2018}. These methods
typically involve smoothing parameters, nearest-neighbour tuning, entropy
estimation, or optimization of a density-ratio variational objective, and
can be particularly delicate near singular alternatives where density
ratios may be large or ill-defined. By contrast, the empirical coverage
divergence is computed directly from marginal ranks and an uncovered-volume
functional in rank space, requiring neither density estimation nor explicit
likelihood-ratio estimation. The resulting rank-based statistic is
distribution-free under independence and admits an analytically tractable
asymptotic null distribution, making it especially attractive for
large-scale testing settings where permutation calibration may be
computationally expensive.
\end{paragraph}

\begin{paragraph}{Gap in the consistency theory for $d_X + d_Y\in\{3,4\}$}
Theorems~\ref{Thm:populationlimit} and~\ref{thm:ConsistencyGeneralDim}
establish consistency of the empirical coverage correlation when
$d_X=d_Y=1$ and when $d_X+d_Y\geq 5$, respectively, with the latter result
requiring suitable regularity conditions. The proof of
Theorem~\ref{Thm:populationlimit} exploits the well-ordering of the real
line, while the proof of Theorem~\ref{thm:ConsistencyGeneralDim} relies on
convergence rates for optimal transport maps. Neither argument extends
readily to the intermediate cases $d_X+d_Y\in\{3,4\}$. A key difficulty is
that, when $d_X\leq 3$, the available convergence rate for the optimal
transport map is slower than $n^{-1/d_X}$, which is insufficient for the
current proof strategy. Closing this gap appears to require new theoretical
ideas and is left as a direction for future research.
\end{paragraph}

\section*{Acknowledgement}
This work has benefited from discussions with many people and the authors would
like to thank Sourav Chatterjee, Marc Hallin, Davy Paindaveine, Qi-Man Shao, Zolt\'an Szab\'o, and the anonymous reviewers for their helpful comments and suggestions.

\appendix
\section{Proof of main results}\label{sec:proof_main}

\subsection{Proof of Theorem~\ref{Thm:populationlimit}}
The proof proceeds by first reducing to the copula case, so that the marginals are both $\mathrm{Unif}[0,1]$. We then approximate the absolutely continuous part of $P^{(X,Y)}$ by a density that is piecewise constant on rectangles, while localising the singular part on a compact Lebesgue null set. The auxiliary perturbation and block-coverage estimates used in this argument are collected in the supplement; the proof below shows how these ingredients combine to identify the limiting coverage contribution, namely $e^{-q}$ on a block with constant density $q$ and zero contribution from the singular component.

\begin{proof}[Proof of Theorem~\ref{Thm:populationlimit}]   
Let $F_X$ and $F_Y$ be the distribution functions of random variables $X$ and $Y$, respectively. By Lemma~\ref{Lemma:DataProcessing} and the fact that the coverage correlation is preserved under monotonic transformation of $X$ and $Y$, we may replace $(X_i, Y_i)_{i\in[n]}$ by $(F_X(X_i), F_Y(Y_i))_{i\in[n]}$ and assume without loss of generality that $P^X = P^Y = \mathrm{Unif}[0,1]$, so that $P^{(X,Y)}$ is simply a copula.

By Lebesgue--Radon--Nikodym decomposition, we can write $\mathrm{d}P^{(X,Y)} = h\, \mathrm{d}\mathrm{vol} + \mathrm{d}\nu$ for some measurable function $h$ and a measure $\nu$ singular with respect to the Lebesgue measure. Fix $\epsilon > 0$. For all sufficiently large integer $N$, we can find a function $\tilde h$ piecewise constant on each $Q_{j,k}:=[(j-1)/N, j/N)\times [(k-1)/N, k/N)$ such that $\int |h-\tilde h|\, d\vol\leq \epsilon$ (for instance, $\tilde h$ can be defined to be equal to the mean value of $h$ in each $Q_{j,k}$ and the claim follows from the Lebesgue Differentiation Theorem and Dominated Convergence Theorem). Since $P^{(X,Y)}$, and hence also $\nu$, is Radon, and in particular inner regular, we can find a compact set $K$ with Lebesgue measure 0 in $[0,1]^2$ such that $\nu(K^{\mathrm{c}}) \leq \epsilon$. By possibly increasing $N$, we may also assume that 
\begin{equation}
\label{Eq:DilatedSupport}
\sum_{(j,k): Q_{j,k}\cap K\neq \varnothing} \vol(Q_{j,k}) \leq \vol\biggl(K + B\Bigl((0,0), \frac{1}{N}\Bigr)\biggr) \leq \epsilon
\end{equation}

We write $\tilde \nu$ for the restriction of $\nu$ on $K$, i.e.\ $\tilde\nu(A) = \nu(A\cap K)$ for all Borel subsets $A$ of $[0,1]^2$. Let $\tilde P^{(X,Y)}$ be defined such that  
\[  
d\tilde{P}^{(X,Y)} = \frac{\tilde h\,d\vol + d\tilde \nu}{\int_{[0,1]^2} \tilde h \, d\vol + \tilde \nu([0,1]^2)}.
\]
By construction, we have $d_{\mathrm{TV}}(P^{(X,Y)},\tilde P^{(X,Y)}) \leq 2\epsilon$. In particular, the marginals $\tilde P^X$ and $\tilde P^Y$ of $\tilde P^{(X,Y)}$ also satisfy $\max\{d_{\mathrm{TV}}(P^{X},\tilde P^{X}), d_{\mathrm{TV}}(P^{Y},\tilde P^{Y})\} \leq 2\epsilon$, so by the triangle inequality
\[
d_{\mathrm{TV}}(P^X\otimes P^Y, \tilde P^X\otimes \tilde P^Y) \leq d_{\mathrm{TV}}(P^X\otimes P^Y,  P^X\otimes \tilde P^Y)+d_{\mathrm{TV}}(P^X\otimes \tilde P^Y, \tilde P^X\otimes \tilde P^Y)\leq 4\epsilon.
\]
By Lemma~\ref{Lemma:DivergenceStability}, we conclude that for $M=1$ and $L=1/(1-e^{-1})$
\begin{equation}
    \label{Eq:fDivApprox}
\bigl|D_{f_\mathrm{cov}}(P^{(X,Y)}\,\|\,P^X\otimes P^Y) - D_{f_\mathrm{cov}}(\tilde P^{(X,Y)}\,\|\,\tilde P^X\otimes \tilde P^Y)\bigr| \leq (8M+6L)\epsilon^{1/2}.
\end{equation}

Let $\kappa_n^{\tilde X, \tilde Y}$ be the empirical coverage correlation of samples $(\tilde X_i, \tilde Y_i)_{i\in[n]}\iid \tilde P^{(X,Y)}$ with respect to reference points $(\tilde U_i,\tilde V_i)\iid \mathrm{Unif}[0,1]^2$. By Proposition~\ref{Prop:TVPerturbation}, there is a coupling between $(X_i,Y_i,U_i,V_i)_{i\in[n]}$ and $(\tilde X_i, \tilde Y_i,\tilde U_i,\tilde V_i)_{i\in[n]}$ such that
\begin{equation}
\label{Eq:EmpiricalCoverageApprox}
|\kappa_n^{X,Y} - \kappa_n^{\tilde X,\tilde Y}|\leq \frac{2\epsilon}{1-e^{-1}} + O_p(n^{-1/2}).
\end{equation}

Now, since $\tilde P^{(X,Y)}$ has piecewise constant density in its absolutely continuous part, we may explicitly control $\kappa_n^{\tilde X, \tilde Y}$. Specifically, writing $\tilde R_i$ for the bivariate ranks of $(\tilde X_i,\tilde Y_i)$ with respect to $(\tilde U_i)_{i\in[n]}$ and $(\tilde V_i)_{i\in[n]}$ defined as in~\eqref{eq:R_iGen}, we have that
\begin{align*}
(1-e^{-1})\kappa_n^{\tilde X, \tilde Y} &= 1-e^{-1} -\vol\biggl(\bigcup_{i=1}^n B\Bigl(\tilde R_i, \frac{1}{2\sqrt{n}}\Bigr)\biggr)\nonumber \\
&\leq 1-e^{-1} - \sum_{(j,k):Q_{j,k}\cap K=\varnothing} \vol\biggl(Q_{j,k}\,\cap \bigcup_{i:(\tilde X_i,\tilde Y_i)\in Q_{j,k}} B\Bigl(\tilde R_i, \frac{1}{2\sqrt{n}}\Bigr)\biggr)\nonumber\\
&\leq 1-e^{-1} - (1+o_p(1))\sum_{j,k=1}^N \int_{(x,y)\in Q_{j,k}} (1-e^{-\tilde h(x,y)}) + \epsilon\nonumber\\
& = (1+o_p(1))\int_{(x,y)\in[0,1]^2} (e^{-\tilde h(x,y)}-e^{-1}) + \epsilon\nonumber\\
&= (1-e^{-1}+o_p(1))D_{f_\mathrm{cov}}(\tilde P^{(X,Y)}\,\| P^X\otimes P^Y)  + \epsilon 
\end{align*}
where the second inequality follows from Proposition~\ref{Prop:PiecewiseConstant} and~\eqref{Eq:DilatedSupport}. By a similar argument, we can bound $\kappa_n^{\tilde X, \tilde Y}$ below as
\begin{align}
    (1&-e^{-1})\kappa_n^{\tilde X, \tilde Y} =  1-e^{-1} - \sum_{(j,k)} \vol\biggl(Q_{j,k}\,\cap \bigcup_{i:(\tilde X_i,\tilde Y_i)\in Q_{j,k}} B\Bigl(\tilde R_i, \frac{1}{2\sqrt{n}}\Bigr)\biggr) \nonumber \\
    &\geq (1+o_p(1))\int_{(x,y)\in[0,1]^2} (e^{-\tilde h(x,y)}-e^{-1}) - \vol\biggl(\bigcup_{\substack{(j,k):Q_{j,k}\cap K\neq\varnothing\\ i:(\tilde X_i,\tilde Y_i)\in Q_{j,k}}} B\Bigl(\tilde R_i,\frac{1}{2\sqrt{n}}\Bigr)\biggr)\nonumber\\
    &\geq (1-e^{-1}+o_p(1))D_{f_\mathrm{cov}}(\tilde P^{(X,Y)}\,\| P^X\otimes P^Y) - \vol\biggl(K + B\Bigl(0,  \frac{1}{N} + \epsilon + O_p(n^{-1/2})\Bigr)\biggr),\label{Eq:PiecewiseConstantLower}
\end{align}
where in the final step, we applied Lemma~\ref{Lemma:DKW} to both $(\tilde X_i)_{i\in[n]}$, $(\tilde Y_i)_{i\in[n]}$ and $(\tilde U_i)_{i\in[n]}$, $(\tilde V_i)_{i\in[n]}$. 

By another application of Lemma~\ref{Lemma:DivergenceStability}, $D_{f_\mathrm{cov}}(\tilde P^{(X,Y)}\,\| P^X\otimes P^Y)$ is at most $(8M+2L)\sqrt{\epsilon}$ away from $D_{f_\mathrm{cov}}(\tilde P^{(X,Y)}\,\| \tilde P^X\otimes \tilde P^Y)$. Also, using the fact that $\vol(K)=0$ and the upper continuity of Lebesgue measure, by choosing $\epsilon$ sufficiently small (and consequently $N$ sufficiently large), the volume of the Minkowski dilation of $K$ of width $1/N+\epsilon+O_p(n^{-1/2})$ on the right-hand side of~\eqref{Eq:PiecewiseConstantLower} can be made smaller than any positive number in probability. Therefore, we conclude that
\begin{equation}
\label{Eq:ConvergenceOfApproximation}
\kappa_n^{\tilde X,\tilde Y} \pto D_{f_\mathrm{cov}}(\tilde P^{(X,Y)}\,\| \tilde P^X\otimes \tilde P^Y).
\end{equation}
The desired result then follows by combining~\eqref{Eq:fDivApprox}, \eqref{Eq:EmpiricalCoverageApprox} and~\eqref{Eq:ConvergenceOfApproximation}, since we can set $\epsilon$ arbitrarily small.
\end{proof}

\subsection{Proof of Proposition~\ref{Prop:KappaProperty}}
\begin{proof}[Proof of Proposition~\ref{Prop:KappaProperty}]
For part (i), we show a general result that for any probability measure $\mu$ and $\nu$ on $\mathbb{R}$, $0 \leq D_{f_\mathrm{cov}}(\mu \| \nu) \leq 1$. First note that $\lim_{t\to\infty} t^{-1}f_{\mathrm{cov}}(t) = 0$ hence the singular part in the definition of $D_{f_{\mathrm{cov}}}$ vanishes. Suppose $\mathrm{d}\mu = g \mathrm{d}\nu + \mathrm{d}\nu^{\perp}$ for some measurable function $g$, then by Definition~\ref{def:f-divergence} and Jensen's inequality we have 
\[
D_{f_{\mathrm{cov}}}(\mu \| \nu) = \int \frac{e^{-g} - e^{-1}}{1 - e^{-1}} \, \mathrm{d} \nu \geq \frac{e^{-\int g \, \mathrm{d}\nu} - e^{-1}}{1-e^{-1}} \geq 0, 
\]
where the final inequality is due to $\int g \, \mathrm{d}\nu \leq \mu(\mathbb{R}) = 1$. Moreover, by $ g\geq 0$, we have $D_{f_{\mathrm{cov}}}(\mu \| \nu) \leq 1$, as desired. 

Part (ii) is true due to the strict convexity of $f_{\mathrm{cov}}$. For part (iii), write $ \mathrm{d}P^{(X,Y)} = h\, \mathrm{d}(P^X\otimes P^Y) + \mathrm{d}\nu$ for $\nu$ singular with respect to $P^X\otimes P^Y$. Since $\lim_{t\to\infty} t^{-1}f_{\mathrm{cov}}(t) = 0$, we have 
\[  
D_{f_\mathrm{cov}}(P^{(X,Y)}\,\|\,\,P^X\otimes P^Y) - 1 = \int_{[0,1]^2} \frac{e^{-h(x)}-1}{1-e^{-1}}\,\mathrm{d}(P^X\otimes P^Y)(x) = 0
\]
if and only if $h(x) = 0$, $P^X\otimes P^Y$-almost everywhere, i.e.\ $P^{(X,Y)} = \nu$ is singular with respect to $P^X\otimes P^Y$. 

For part (iv), observe that conditional independence of $X$ and $Y$ given $Z$ means that we can generate both $P^X\otimes P^Y$ and $P^{(X,Y)}$ from $P^X\otimes P^Z$ and $P^{(X,Z)}$ respectively using the same Markov kernel (channel) $P^{Y\mid Z}$. Hence, we can apply the data processing inequality \citep[Theorem~7.4]{polyanskiy2024information} of the $f$-divergence to obtain that $\kappa(Z,Y)\geq \kappa(X,Y)$ as desired. Part (v) follows from the lower semicontinuity of $f$-divergence \citep[Theorem~4.9]{polyanskiy2024information}. Finally, part (vi) is true since the definition of $\kappa(X,Y)$ is symmetric with respect to the two arguments by Fubini's theorem.
\end{proof}

\subsection{Proof of Theorem~\ref{thm:ConsistencyGeneralDim}}

The key ingredient of the proof is Proposition~\ref{prop: VonRawData}, which reduces the rank-based statistic to the copula case, and then further shows, using a dilation argument, that we may compute the vacancy from $\ell_\infty$ balls centred at the copula points instead of the optimal transport ranks. We are then able to apply Lebesgue--Radon--Nikodym decomposition to give the first-moment limit. Finally, a second-moment calculation shows that the variance vanishes.

\begin{proof}[Proof of Theorem~\ref{thm:ConsistencyGeneralDim}]  
By Proposition~\ref{prop: VonRawData}, we may, without loss of generality, assume that $P^X = \mathrm{Unif}([0,1]^{d_X})$ and $P^Y = \mathrm{Unif}([0,1]^{d_Y})$. Moreover, it also ensures that the vacancy statistic computed from the original data differs negligibly from that computed from the rank-transformed data. Thus we focus on proving the consistency of $\mathcal{V}({\bf X}, {\bf Y}) := \mathcal{V}({\bf X}, {\bf Y}, {\bf X}, {\bf Y}; 1/2 n^{-1/d})$. In the rest of the proof, we write $B := B(0, 1/2 n^{-1/d})$ 

By the Lebesgue-Radon-Nikodym decomposition, we can write $\mathrm{d}P^{(X, Y)} = h \, \mathrm{d}\mathrm{vol} + \mathrm{d}\nu$ for some measurable function $h$ and a measure $\nu$ singular with respect to the Lebesgue measure. Write $Z_i = (X_i, Y_i)$ for $i \in [n]$ and $p_n(u) := n P^{(X, Y)}\bigl(u + B\bigr)$ for $u \in [0, 1]^d$, we have
\begin{align}
    \E \mathcal{V}({\bf X},{\bf Y}) = \int_{[0, 1]^d} \Prob\Bigl(u \not \in \bigcup_{i=1}^n Z_i +B \Bigr) \, \mathrm{d}u &= \int_{[0, 1]^d} \Prob\Bigl(\bigcap_{i = 1}^n \{Z_i \not\in u+B \}\Bigr) \, \mathrm{d}u \notag \\ 
    & = \int_{[0, 1]^d} \Bigl(1 - \frac{p_n(u)}{n}\Bigr)^n \, \mathrm{d}u \longrightarrow \int_{[0, 1]^d}  e^{-h(u)} \, \mathrm{d}u \label{eq: limitExp},
\end{align}
where~\eqref{eq: limitExp} follows by \citet[][Theorem 3.22]{folland1999real}. 

It remains to control the variance of $\mathcal{V}({\bf X}, {\bf Y})$. Let $W_1, W_2 \stackrel{\mathrm{iid}}{\sim}\mathrm{Unif}([0, 1]^d)$, we have 
\begin{align}
    \E \mathcal{V}({\bf X},& {\bf Y})^2 = \Prob\Bigl(W_j \not\in \bigcup_{i = 1}^n Z_i +B, \ j = 1, 2\Bigr) \notag \\
    & =\E\biggl[\Prob\Bigl(\bigcap_{i = 1}^n\bigl\{Z_i \not \in \{W_1 + B\} \cup \{W_2 + B\}\bigr\} \mid W_1, W_2\Bigr)\biggr] \notag \\ 
    & = \E \biggl[\biggl(1 - \Prob\Bigl(Z_1 \in \{W_1 + B\} \cup \{W_2 + B\} \mid W_1, W_2\Bigr) \biggr)^n\biggr] \notag \\ 
    & = \E \biggl[\biggl(1 - \frac{p_n(W_1) + p_n(W_2)}{n} + \Prob\Bigl(Z_1 \in \{W_1 + B\} \cap \{W_2 + B\} \Bigm| W_1, W_2\Bigr) \biggr)^n\biggr]. \label{eq:VacancySecondMoment}
\end{align}
Since as $n \to \infty$, $p_n(W_j) \to h(W_j)$, $j = 1, 2$, and $\{W_1 + B\} \cap \{W_2 + B\}$ eventually becomes an empty set, by the Dominated convergence theorem we have 
\begin{align*}
    \lim_{n \to \infty}\E \mathcal{V}({\bf X}, {\bf Y})^2 = \lim_{n \to \infty} \E \Bigl(1 - \frac{h(W_1)+h(W_2)}{n}\Bigr)^n = (\E e^{-h(W_1)})^2.
\end{align*}
Combining with~\eqref{eq: limitExp}, we have $\mathrm{Var}(\mathcal{V}({\bf X}, {\bf Y})) \to 0$ as $n \to \infty$. Thus $\mathcal{V}({\bf X},{\bf Y})\pto \int e^{-h(u)}\,du$, and hence
\[
\kappa_n^{X,Y}\pto \frac{\int e^{-h(u)}\,du-e^{-1}}{1-e^{-1}}
=D_{f_{\mathrm{cov}}}(P^{(X,Y)}\,\|\,P^X\otimes P^Y),
\]
because the singular part in the Lebesgue--Radon--Nikodym decomposition has zero contribution to $D_{f_{\mathrm{cov}}}$. This concludes the proof. 
\end{proof}

\subsection{Proof of Proposition~\ref{Prop:MultivariateLimit}}
\begin{proof}[Proof of Proposition~\ref{Prop:MultivariateLimit}]
The first part of the proposition follows directly from Lemma~\ref{le: VacancyLimitRandomRef}, which implies that $\mathbb{E}(\mathcal{V}_n)\to e^{-1}$ and $\var(\mathcal{V}_n)\to 0$.

Now we consider the case that $P^{(X,Y)}$ is singular with respect to $P^X\otimes P^Y$. By the same argument as in the proof of Lemma~\ref{Lemma:DataProcessing}, there exist convex functions $\phi:\mathbb{R}^{d_X}\to\mathbb{R}$ and $\psi:\mathbb{R}^{d_Y}\to\mathbb{R}$ and random vectors $U\sim \mathrm{Unif}([0,1]^{d_X})$ and  $V\sim \mathrm{Unif}([0,1]^{d_Y})$ such that $X=\nabla \phi(U)$ and $Y=\nabla\psi(V)$ defines respectively the optimal transport maps from $U$ to $X$ and from $V$ to $Y$. Let $T_X$ and $T_Y$ be the Markov transition kernel from $X$ to $U$ and from $Y$ to $V$, respectively, corresponding to the conditional distribution $U\mid X$ and $V\mid Y$. Then we have $T_X(X_1),\ldots,T_X(X_n)\iid\mathrm{Unif}([0,1]^{d_X})$ and $T_Y(Y_1),\ldots,T_Y(Y_n)\iid\mathrm{Unif}([0,1]^{d_Y})$. Since $(\nabla \phi(T_X(X_1)), \nabla\psi(T_Y(Y_1))) = (X_1, Y_1)$, we have that the joint distribution of $(T_X(X_1),T_Y(Y_1))$ is singular with respect to the product of the marginals (the preimage of the joint distribution of $(X_1,Y_1)$ under the mapping $(a,b)\mapsto (\nabla\phi(a),\nabla\psi(b))$ has measure 1 under the joint distribution of $(T_X(X_1),T_Y(Y_1))$ and measure 0 under the product of the marginals). Fix $\epsilon > 0$. Since the joint distribution of $(T_X(X_1),T_Y(Y_1))$ is a Radon measure, and hence inner regular, we can find a compact subset $K$ of its support (so $K$ has Lebesgue measure 0) such that 
\[  
\mathbb{P}((T_X(X_1), T_Y(Y_1)) \in K) \geq 1 - \epsilon.
\]
Denote $\mathcal{I}:=\{i: (T_X(X_i), T_Y(Y_i)) \notin K\}$. By the multiplicative Chernoff bound \citep[Exercise~10.6.11]{samworth2025statistics}, we have 
\[  
\mathbb{P}(|\mathcal{I}| > 2\epsilon n) \leq e^{-3n\epsilon/8}.
\]
By \citet[Theorem~1]{fournier2015rate}, the empirical distribution of both $(R_i)_{i\in[n]}$ and  $(T_X(X_i), T_Y(Y_i))_{i\in[n]}$ are at most $C_d\rho_n$ away from $\mathrm{Unif}([0,1]^{d})$ in 1-Wasserstein distance where
\[  
\rho_n =\begin{cases}
    n^{-1/2}\log(en) & \text{if $d=2$}\\
    n^{-1/d} & \text{if $d\geq 3$}
\end{cases}
\]
and $C_d$ depends only on $d$. Hence by the triangle inequality and Markov's inequality, the 1-Wasserstein distance between $(R_i)_{i\in[n]}$ and  $(T_X(X_i), T_Y(Y_i))_{i\in[n]}$  is bounded by $2C_d\rho_n^{1/2}$ with probability at least $1-\rho_n^{1/2}$. Define 
\[
\mathcal{J}:=\{i: \min_{j\in[n]} \|R_i - (T_X(X_j), T_Y(Y_j))\|_2 > 2\epsilon^{-1}C_d\rho_n^{1/2}\}.
\]
Then $\mathbb{P}(|\mathcal{J}| > \epsilon n) \leq \rho_n^{1/2}$. Therefore, on an event with probability at least $1-\rho_n^{1/2}-e^{-3n\epsilon/8}$, we have for some $C_d'$ depending only on $d$ that 
\begin{align*}
\vol\biggl(\bigcup_{i\in[n]} B\Bigl(R_i, \frac{1}{2n^{1/d}}\Bigr)\biggr) &\leq \vol\biggl(\bigcup_{i\in\mathcal{I}\cup\mathcal{J}} B\Bigl(R_i, \frac{1}{2n^{1/d}}\Bigr) \biggr) + \vol\Bigl(K + B(0, C_d'\epsilon^{-1}\rho_n^{1/2})\Bigr)\\
&\leq 3\epsilon + \vol\Bigl(K + B(0, C_d'\epsilon^{-1}\rho_n^{1/2})\Bigr)
\end{align*}
For each fixed $\epsilon$, the Minkowski dilation $K + B(0, C_d'\epsilon^{-1}\rho_n^{1/2})$ has Lebesgue measure converging to $0$ (since $K$ is compact and Lebesgue null). Since $\epsilon$ is arbitrary, we must have the left-hand side of the above converging to $0$ in probability, and consequently $\kappa_n^{X,Y}\pto 1$ as desired.
\end{proof}

\subsection{Proof of Theorem~\ref{thm: CLT}}

Our proof strategy is inspired by \citet[Theorem~1]{hall1985three}: we partition the space into local blocks, prove a conditional CLT for the interior contribution, and show that the boundary contribution is negligible.

\begin{proof}[Proof of Theorem~\ref{thm: CLT}]
Write $\gamma:=\frac{1}{2n^{1/d}}$ for notational simplicity and fix $\lambda\in\mathbb{N}$ for now. Define $L_1:=\lfloor\frac{1}{(\lambda+2)\gamma}\rfloor$. We can partition $[0,1)^d$ into $L:=L_1^d$ small cubes $\prod_{j\in[d]} \bigl[\frac{k_j-1}{L_1}, \frac{k_j}{L_1}\bigr)$ for $k_1,\ldots,k_d\in[L_1]$. We call these small cubes $\mathcal{P}_1,\ldots,\mathcal{P}_L$. For each $\ell\in[L]$, we define $\mathcal{Q}_\ell$ to be the concentric cube in $\mathcal{P}_\ell$ with side length $\lambda\gamma$. Define $\mathcal{I}_\ell = \{i: R_i\in \mathcal{P}_\ell\}$ and $N_\ell:=|\mathcal{I}_\ell|$. Writing
\begin{equation}
\begin{aligned}
\mathcal{V}_{n,\ell}^{\mathrm{in}}&:= \vol \Bigl(\mathcal{Q}_\ell \setminus \bigcup_{i\in \mathcal{I}_\ell} B(R_i,\gamma)\Bigr) \; \forall\, \ell\in[L]\\
 \mathcal{V}_n^{\mathrm{out}}&:= \vol\Bigl([0,1]^d \setminus \Bigl\{\bigcup_{\ell\in[L]}\mathcal{Q}_\ell \cup\bigcup_{i\in[n]} B(R_i,\gamma) \Bigr\}\Bigr),
\end{aligned}
\label{eq:VInOut}
\end{equation}
we have 
\[  
\mathcal{V}_n = \mathcal{V}_{n}^{\mathrm{in}}+ \mathcal{V}_n^{\mathrm{out}}, \quad \text{where}\; \mathcal{V}_{n}^{\mathrm{in}} := \sum_{\ell=1}^L \mathcal{V}_{n,\ell}^{\mathrm{in}}.
\]
The key observation here is that $(\mathcal{V}_{n,\ell}^{\mathrm{in}}:\ell\in[L])$ are conditionally independent given $(N_\ell:\ell\in[L])$. Define 
\begin{align}\label{eq:MS}
    M_n:= \mathbb{E}(\mathcal{V}_n^{\mathrm{in}}\mid N_1,\ldots,N_L), \qquad S_n:= \var(\mathcal{V}_n^{\mathrm{in}}\mid N_1,\ldots,N_L).
\end{align}

By Proposition~\ref{Prop:ConditionalBerryEsseen} we have 
\begin{align}\label{eq:CLT1}
    \sup_{t\in\mathbb{R}} \biggl|\mathbb{P}\biggl(\frac{\sqrt{n}(\mathcal{V}_n^{\mathrm{in}} - M_n)}{\sqrt{S_n}} \leq t\biggm| N_1,\ldots,N_L\biggr) - \Phi(t)\biggr| = O_p(n^{-1/2}),
\end{align}
where $\Phi$ is the standard Gaussian distribution function. Additionally Proposition~\ref{Prop:AsymptoticMeanVariance} gives us
\begin{equation}\label{eq:CLT2}
\sqrt{n} (M_n - \mathbb{E}(M_n)) \dto \mathcal{N}(0,\alpha_\lambda^2),\qquad nS_n \pto \beta_\lambda^2.
\end{equation}
Using Lemma~\ref{lmm:ConditionalDistributionCLT},~\eqref{eq:CLT1} and~\eqref{eq:CLT2} gives us 
\begin{align*}
    \sqrt{n}(\mathcal{V}_n^{\mathrm{in}} - \E(\mathcal{V}_n^{\mathrm{in}}))\dto \mathcal{N}(0, \alpha_\lambda^2 + \beta_\lambda^2).
\end{align*}
Finally, using Proposition~\ref{Prop:VoutContribution} and Chebyshev's inequality by $\lambda\rightarrow\infty$ we have 
\begin{align*}
    \sqrt{n}(\mathcal{V}_n^{\mathrm{out}} - \E(\mathcal{V}_n^{\mathrm{out}})) \pto 0.
\end{align*}
By Lemma~\ref{lmm:AlphaBetaLambda}, as $\lambda \to \infty$, we have $\alpha_\lambda^2 \to 0$ and $\beta_\lambda^2 \to \beta^2$ for $\beta^2 > 0$. Consequently, we conclude that
\[
\sqrt{n}(\mathcal{V}_n - \E(\mathcal{V}_n))\dto\mathcal{N}(0, \beta^2).
\]
The proof is complete by combining the above distributional convergence with the variance calculation in Lemma~\ref{le: VacancyLimitRandomRef}.
\end{proof}

\subsection{Proof of Theorem~\ref{thm: concentration}}

The proof combines a bounded-differences argument with the randomised interval construction of Lemma~\ref{Lemma:DeletePoints}, which controls the effect of changing one observation through the random spacings of the reference points.

\begin{proof}[Proof of Theorem~\ref{thm: concentration}]
As in the proof of Lemma~\ref{Lemma:DeletePoints}, we generate $P_1,\ldots,P_{n+1},$ $Q_1,\ldots,Q_{n+1}\iid \mathrm{Beta}(1/2,1/2)$ independent of all other randomness in the problem and define $G_i := U_{(i)}-(U_{(i)}-U_{(i-1)})P_i$ and $H_i := V_{(i)}-(V_{(i)}-V_{(i-1)})Q_i$ for $i\in[n+1]$ (with the convention that $U_{(0)}=V_{(0)}=0$ and $U_{(n+1)}=V_{(n+1)}=1$). 

Given any deterministic $\bs x =(x_i)_{i\in[n]}$ and $\bs y = (y_i)_{i\in[n]}$ and $i_0\in[n]$, define 
\begin{align*}
    \bs x^{-i_0} = (x_1,\ldots,x_{i_0-1},x_{i_0+1},\ldots,x_n), \qquad \bs y^{-i_0} = (y_1,\ldots,y_{i_0-1},y_{i_0+1},\ldots,y_n).
\end{align*}
By the proof of Lemma~\ref{Lemma:DeletePoints}, there exists $\tilde{\bs U}^{[r_x]} = (\tilde U_i^{[r_x]})_{i\in[n-1]}$ defined in terms of $\bs U = (U_i)_{i\in[n]}$, $P_1,\ldots,P_{n+1}$ and $r_x:=\sum_{i\in[n]}\mathbbm{1}\{x_{i_0}\geq x_i\}$, and  $\tilde{\bs V}^{[r_y]} = (\tilde V_i^{[r_y]})_{i\in[n-1]}$ defined in terms of $\bs V = (V_i)_{i\in[n]}$, $Q_1,\ldots,Q_{n+1}$ and $r_y:=\sum_{i\in[n]}\mathbbm{1}\{y_{i_0}\geq y_i\}$, such that $\tilde{\bs U}^{[r_x]}$ and $\tilde{\bs V}^{[r_y]}$ each have iid $\mathrm{Unif}[0,1]$ entries and 
\begin{align*}
\biggl|\mathcal{V}\Bigl(\bs X, \bs Y, \bs U, \bs V;\frac{1}{2\sqrt{n}}\Bigr) &- \mathcal{V}\Bigl(\bs X^{-i_0}, \bs Y^{-i_0}, \tilde{\bs U}^{[r_x]}, \tilde{\bs V}^{[r_y]};\frac{1}{2\sqrt{n-1}}\Bigr)\biggr|\\
&\hspace{2cm}\leq 5 \max\{G_{(r_x+1)}-G_{(r_x)}, H_{(r_y+1)}-H_{(r_y)}, n^{-1}\}.
\end{align*}
Thus, if $\bs X'$ and $\bs Y'$ differ from $\bs X$ and $\bs Y$ only in the $i_0$th entry respectively, we must have 
\begin{align*}
\biggl|\mathcal{V}\Bigl(\bs X, \bs Y, \bs U, \bs V;\frac{1}{2\sqrt{n}}\Bigr) &- \mathcal{V}\Bigl(\bs X', \bs Y', \bs U, \bs V;\frac{1}{2\sqrt{n}}\Bigr)\biggr|\\
&\hspace{2cm}\leq 10 \max\{G_{(r_x+1)}-G_{(r_x)}, H_{(r_y+1)}-H_{(r_y)}, n^{-1}\}.
\end{align*}
Since $(G_{(r+1)}-G_{(r)})_{r\in[n]}$ and $(H_{(r+1)}-H_{(r)})_{r\in[n]}$ have $\mathrm{Beta}(1,n)$ entries, we have by Lemma~\ref{Lemma:Beta} that on an event $\Omega$ with probability at least $1-2ne^{-t^2/(2+4t/3)}$ that 
\[  
\max_{r\in[n]} \max\{G_{(r+1)}-G_{(r)}, H_{(r+1)}-H_{(r)}, n^{-1}\} \leq \frac{1+t}{n}.
\]
Thus, we can apply McDiarmid's inequality conditional on the event $\Omega$ to obtain that 
\begin{align*}
\mathbb{P}(|\mathcal{V}_n - \mathbb{E}(\mathcal{V}_n)| \geq s) &\leq \mathbb{P}\bigl(|\mathcal{V}_n - \mathbb{E}(\mathcal{V}_n)| \geq s \bigm| \Omega\bigr)\mathbb{P}(\Omega) + \mathbb{P}(\Omega^{\mathrm{c}})\\
&\leq 2\exp\biggl\{-\frac{ns^2}{50(1+t)^2}\biggr\} + 2n\exp\biggl\{-\frac{t^2}{2+4t/3}\biggr\}
\end{align*}
Setting $t = \frac{1}{5}\min\{(ns^2)^{1/3}, (ns^2)^{1/2}\}$, we then have 
\[  
\mathbb{P}(|\mathcal{V}_n - \mathbb{E}(\mathcal{V}_n)| \geq s) \leq 2(n+1)\exp\biggl\{-\frac{1}{72}\min\{ns^2, (ns^2)^{1/3}\}\biggr\},
\]
as desired.
\end{proof}

\subsection{Proof of Theorem~\ref{Thm:PopulationLimitGrid}}
The proof strategy for Theorem~\ref{Thm:PopulationLimitGrid} is similar to that of Theorem~\ref{Thm:populationlimit} and is as follows. We first reduce the problem to the case where the marginal distributions $P^X$ and $P^Y$ are both $\mathrm{Unif}[0,1]$, so that the joint distribution $P^{(X,Y)}$ is simply the copula. Next, we approximate the absolutely continuous part of $P^{(X,Y)}$ by a piecewise constant density on rectangular pieces and the singular part of $P^{(X,Y)}$ by a singular measure supported on a Lebesgue null compact set. Finally, we show that the contribution to the coverage from each rectangular piece with constant density $q$ is proportional to $e^{-q}$ and the contribution from the singular measure with Lebesgue null support is 0 to complete the argument. The perturbation and block-coverage estimates invoked in this proof are collected in Section~\ref{sec:additional-results}.
\begin{proof}[Proof of Theorem~\ref{Thm:PopulationLimitGrid}]
    Consider the same set-up as in the proof of Theorem~\ref{Thm:populationlimit}. Let $\kappa_n^{\tilde{X}, \tilde{Y};\mathrm{grid}}$ be the empirical coverage correlation of sample $(\tilde{X}_i, \tilde{Y}_i)_{i\in[n]}\iid \tilde{P}^{(X, Y)}$ with $\tilde{P}^{(X, Y)}$ constructed as in proof of Theorem~\ref{Thm:populationlimit} and using the coupling argument in Proposition~\ref{Prop:TVPerturbationFixedGrid} for $n\epsilon\to \infty$ such that     
    \begin{equation}
        \label{Eq:EmpiricalCoverageApproxFixed}
        |\kappa_n^{X,Y; \mathrm{grid}} - \kappa_n^{\tilde X,\tilde Y; \mathrm{grid}}|\leq O_p(\epsilon\sqrt{n}).
    \end{equation}
    Now, since $\tilde P^{(X,Y)}$ has piecewise constant density in its absolutely continuous part, we may explicitly control $\kappa_n^{\tilde X, \tilde Y; \mathrm{grid}}$. Specifically, writing $\tilde R_i$ for the bivariate ranks of $(\tilde X_i,\tilde Y_i)$ with respect to reference points $\boldsymbol{u} = \boldsymbol{v} = (1/n, \ldots, (n-1)/n, 1)$ defined as in~\eqref{eq:R_iGen}, we have 
    \begin{align*}
(1-e^{-1})\kappa_n^{\tilde X, \tilde Y;\mathrm{grid}} &= 1-e^{-1} -\vol\biggl(\bigcup_{i=1}^n B\Bigl(\tilde R_i, \frac{1}{2\sqrt{n}}\Bigr)\biggr)\nonumber \\
&\leq 1-e^{-1} - \sum_{(j,k):Q_{j,k}\cap K=\varnothing} \vol\biggl(Q_{j,k}\,\cap \bigcup_{i:(\tilde X_i,\tilde Y_i)\in Q_{j,k}} B\Bigl(\tilde R_i, \frac{1}{2\sqrt{n}}\Bigr)\biggr)\nonumber\\
&\leq 1-e^{-1} - (1+o_p(1))\sum_{j,k=1}^N \int_{(x,y)\in Q_{j,k}} (1-e^{-\tilde h(x,y)}) + \epsilon\nonumber\\
& = (1+o_p(1))\int_{(x,y)\in[0,1]^2} (e^{-\tilde h(x,y)}-e^{-1}) + \epsilon\nonumber\\
&= (1-e^{-1}+o_p(1))D_{f_{\mathrm{cov}}}(\tilde P^{(X,Y)}\,\| P^X\otimes P^Y)  + \epsilon  
\end{align*}
where the second inequality follows from Proposition~\ref{Prop:PiecewiseConstantFixed} and~\eqref{Eq:DilatedSupport}. By a similar argument, we can bound $\kappa_n^{\tilde X, \tilde Y; \mathrm{grid}}$ below as
\begin{align}
    &(1-e^{-1})\kappa_n^{\tilde X, \tilde Y; \mathrm{grid}} =  1-e^{-1} - \sum_{(j,k)} \vol\biggl(Q_{j,k}\,\cap \bigcup_{i:(\tilde X_i,\tilde Y_i)\in Q_{j,k}} B\Bigl(\tilde R_i, \frac{1}{2\sqrt{n}}\Bigr)\biggr) \nonumber \\
    &\geq (1+o_p(1))\int_{(x,y)\in[0,1]^2} (e^{-\tilde h(x,y)}-e^{-1}) - \vol\biggl(\bigcup_{\substack{(j,k):Q_{j,k}\cap K\neq\varnothing\\ i:(\tilde X_i,\tilde Y_i)\in Q_{j,k}}} B\Bigl(\tilde R_i,\frac{1}{2\sqrt{n}}\Bigr)\biggr)\nonumber\\
    &\geq (1-e^{-1}+o_p(1))D_{f_{\mathrm{cov}}}(\tilde P^{(X,Y)}\,\| P^X\otimes P^Y) - \vol\biggl(K + B\Bigl(0,  \frac{1}{N} + \epsilon + O_p(n^{-1/2})\Bigr)\biggr),\label{Eq:PiecewiseConstantLowerFixed}
\end{align}
where in the final step, we applied Lemma~\ref{Lemma:DKW} to $(\tilde X_i)_{i\in[n]}$, $(\tilde Y_i)_{i\in[n]}$. 

By Lemma~\ref{Lemma:DivergenceStability}, $D_{f_{\mathrm{cov}}}(\tilde P^{(X,Y)}\,\| P^X\otimes P^Y)$ is at most $(8M+2L)\sqrt{\epsilon}$ away from $D_{f_{\mathrm{cov}}}(\tilde P^{(X,Y)}\,\| \tilde P^X\otimes \tilde P^Y)$. Also, using the fact that $\vol(K)=0$ and the upper continuity of Lebesgue measure, by choosing $\epsilon$ sufficiently small (and consequently $N$ sufficiently large), the volume of the Minkowski dilation of $K$ of width $1/N+\epsilon+O_p(n^{-1/2})$ on the right-hand side of~\eqref{Eq:PiecewiseConstantLowerFixed} can be made smaller than any positive number in probability. Therefore, we conclude that
\begin{equation}
\label{Eq:ConvergenceOfApproximationFixed}
\kappa_n^{\tilde X,\tilde Y; \mathrm{grid}} \pto D_{f_{\mathrm{cov}}}(\tilde P^{(X,Y)}\,\| \tilde P^X\otimes \tilde P^Y).
\end{equation}
The desired result then follows by combining~\eqref{Eq:fDivApprox}, \eqref{Eq:EmpiricalCoverageApproxFixed} and~\eqref{Eq:ConvergenceOfApproximationFixed}, by choosing small $\epsilon$ such that $n\epsilon\to\infty$ and $\sqrt{n}\epsilon\to 0$.
\end{proof}

\subsection{Proof of Theorem~\ref{thm:RegularGridCLT}}

The proof starts by using exchangeability to represent the joint ranks as a uniform random permutation on the grid. The main step is then a cumulant argument: the local coverage indicator is expanded by inclusion--exclusion, truncated to a finite-degree polynomial in permutation-matrix entries, and controlled through F\'eray's weighted dependency graph framework for permutation matrices~\citep{Feray2018WeightedDG}. 

\begin{proof}[Proof of Theorem~\ref{thm:RegularGridCLT}]
Given regular grid reference points $\bs{u}$ and $\bs v$, independence and exchangeability allow us to fix the empirical ranks of one marginal and represent the ranks of the other marginal by an independent uniform random permutation $\pi$ of $[n]=\{1,\dots,n\}$. Thus
\[
\cp_n:=\biggl\{\Bigl(\frac{i}{n},\frac{\pi(i)}{n}\Bigr):1\le i\le n\biggr\} = \bigl\{R_i: i\in[n]\bigr\},
\]
where $R_i$ is the joint rank of $(X_i,Y_i)$ with respect to $\bs u$ and $\bs v$. Let $\gamma=1/(2\sqrt n)$ and define the covered area
\begin{align}
\label{Eq:DefnCn}
\mathscr{C}_n := \mathrm{vol}\biggl(\bigcup_{p \in \cp_n} B(p, \gamma)\biggr).
\end{align}
Then $\mathcal{V}_n^{\mathrm{grid}}(\bs{u},\bs{v})=1-\mathscr{C}_n$, so it suffices to prove a central limit theorem for $\mathscr{C}_n$.

For $z=(x,y)\in[0,1]^2$, let
\[
U_n(z):=\bone\{\cp_n\cap B(z,\gamma)\neq\varnothing\},
\]
so that
\begin{equation}\label{eq:Anintegral}
\mathscr{C}_n=\int_{[0,1]^2}U_n(z)\,dz.
\end{equation}
Define $d_{\mathbb{T}}(u_1,u_2):=\inf_{\ell\in\{-1,0,1\}}|u_1-u_2-\ell|$ on $\mathbb{T}=\mathbb{R}/\mathbb{Z}$, and set
\[
I_n(z) := \Bigl\{1\le i\le n: d_{\mathbb{T}}\Bigl(\frac{i}{n},x\Bigr) \leq \frac{1}{2\sqrt{n}} \Bigr\}, \quad
J_n(z) := \Bigl\{1\le j\le n: d_{\mathbb{T}}\Bigl(\frac{j}{n},y\Bigr) \leq \frac{1}{2\sqrt{n}} \Bigr\}.
\]
There is an absolute constant $C_0>0$ such that
\begin{equation}\label{eq:IJsize}
|I_n(z)|\le C_0\sqrt n,\qquad |J_n(z)|\le C_0\sqrt n.
\end{equation}

Let $X_{ij}:=\bone\{\pi(i)=j\}$ and, for finite $M\subset\{1,\dots,n\}^2$, define
\[
X_M:=\prod_{(i,j)\in M}X_{ij}.
\]
For $m\ge0$, let $\cM_{n,m}(z)$ be the family of subsets $M\subset I_n(z)\times J_n(z)$ of cardinality $m$, and write
\[
N_n(z):=\sum_{(i,j)\in I_n(z)\times J_n(z)}X_{ij},\qquad
V_{n,m}(z):=\binom{N_n(z)}{m}=\sum_{M\in\cM_{n,m}(z)}X_M.
\]
The identity $\bone{\{q\ge1\}}=\sum_{m\ge1}(-1)^{m+1}\binom{q}{m}$ yields
\begin{equation}\label{eq:Uexpand}
U_n(z)=\bone\{N_n(z)\ge1\}=\sum_{m\ge1}(-1)^{m+1}V_{n,m}(z),
\end{equation}
where the sum is finite because $\cM_{n,m}(z)=\varnothing$ for $m>\min(|I_n(z)|,|J_n(z)|)$.

For sets $M,M'\subset\{1,\dots,n\}^2$, say that $M$ shares a row or a column with $M'$ if there exist $(i,j)\in M$ and $(i',j')\in M'$ such that $i=i'$ or $j=j'$. Define
\[
W_n(M;M')=
\begin{cases}
1,&\text{if $M$ shares a row or a column with $M'$},\\
n^{-1},&\text{otherwise},
\end{cases}
\]
and, for finite families $M_1,\dots,M_\ell$,
\[
W_n\bigl(M;\,M_1,\dots,M_\ell\bigr):=
\begin{cases}
1,&\text{if $M$ shares a row or a column with some $M_j$, $j\in[\ell]$},\\
n^{-1},&\text{otherwise}.
\end{cases}
\]
For $L\ge1$, define
\[
U_n^{(L)}(z):=\sum_{m=1}^{L}(-1)^{m+1}V_{n,m}(z),\qquad
\mathscr{C}_n^{(L)}:=\int_{[0,1]^2}U_n^{(L)}(z)\,dz,
\]
and set $\Delta_n^{(L)}:=\mathscr{C}_n-\mathscr{C}_n^{(L)}$.

By Lemma~\ref{lem:varVacancyReg}, we have that $n\var(\mathscr{C}_n) \to \frac{4\mathrm{Ei}(1)-4\gamma_0-5}{e^2} =: \sigma^2$. Hence, it suffices to prove that 
\begin{equation}
\label{Eq:CLT_Cn}
\sqrt n\,(\mathscr{C}_n-\mathbb{E} \mathscr{C}_n) \dto \ \mathcal{N}(0,\sigma^2).
\end{equation}

To this end, set
\[
Z_n := \sqrt n \,(\mathscr{C}_n-\mathbb{E} \mathscr{C}_n),
\qquad
Z_n^{(L)} := \sqrt n \,(\mathscr{C}_n^{(L)}-\mathbb{E} \mathscr{C}_n^{(L)}).
\]
Then
\[
Z_n-Z_n^{(L)}=\sqrt n\,(\Delta_n^{(L)}-\mathbb E\Delta_n^{(L)}).
\]
Let $d_{BL}$ denote the bounded-Lipschitz metric on probability measures on $\mathbb R$. That is, for two probability measures $\nu_1$ and $\nu_2$
\[
d_{BL}(\nu_1,\nu_2)
:=
\sup_{\substack{\|f\|_\infty\le1\\ \mathrm{Lip}(f)\le1}}
\Bigl|\int f\,d\nu_1-\int f\,d\nu_2\Bigr|.
\] 
It is known that this metric metrises weak convergence on $\mathbb R$. 

Since for every Lipschitz and bounded test function $f$,
\begin{align*}
|\mathbb{E}f(Z_n) - \mathbb{E}f(Z_n^{(L)})| \leq \mathbb{E}|Z_n - Z_n^{(L)}|, 
\end{align*}
it follows that
\[
d_{BL}(\mathrm{Law}(Z_n), \mathrm{Law}(Z_n^{(L)})) \leq \mathbb{E}|Z_n - Z_n^{(L)}| \leq \bigl(\mathbb{E}|Z_n - Z_n^{(L)}|^2\bigr)^{1/2} = n^{1/2} \mathrm{Var}^{1/2} \bigl(\Delta_n^{(L)} \bigr).
\]
Therefore, by Proposition~\ref{prop:tail-covariance}, we have 
\begin{align}
\lim_{L \to \infty}\sup_{n \geq 1} d_{BL}(\mathrm{Law}(Z_n), \mathrm{Law}(Z_n^{(L)}))  = 0. \label{eq:BL1}
\end{align}

Next, since
\[
\mathscr{C}_n=\mathscr{C}_n^{(L)}+\Delta_n^{(L)},
\]
we have
\[
\bigl|n\var(\mathscr{C}_n)-n\var(\mathscr{C}_n^{(L)})\bigr|
\le
2 n\var(\mathscr{C}_n)^{1/2}\var(\Delta_n^{(L)})^{1/2}
+
n\var(\Delta_n^{(L)}).
\]
By Proposition~\ref{prop:tail-covariance} and that $n \var(\mathscr{C}_n) \to \sigma^2$, the right-hand side tends to zero as $L \to \infty$ uniformly in $n$. Hence 
\begin{align}
\delta_L := \sup_{n \geq 1}\bigl|n\var(\mathscr{C}_n)-n\var(\mathscr{C}_n^{(L)})\bigr| \to 0 \qquad (L \to \infty). \label{eq:TruncatedVarianceApprox}
\end{align}

Moreover, \eqref{eq:TruncatedVarianceApprox} implies that, for each fixed $L$, $\sup_{n \geq1} n\var(\mathscr{C}_n^{(L)})$ is bounded. Hence, by Chebyshev's inequality, the sequence of laws of $Z_n^{(L)}$ is tight. Let $(n_k)_{k \geq 1}$ be any subsequence such that 
\[
\lim_{k \to \infty} d_{BL}\bigl(\mathrm{Law}(Z_{n_k}^{(L)}), \, \nu\bigr) = 0,
\]
for some probability measure $\nu$ on $\mathbb R$. Passing to a further subsequence if necessary, we may also assume that, as $k \to \infty$, $n_k\var(\mathscr{C}_{n_k}^{(L)}) \to v$ for some $v \geq 0$. Proposition~\ref{prop:truncated-area-cumulants} then yields that $\nu = \mathcal{N}(0,v)$. Moreover, by the definition of $\delta_L$,
\[
|v-\sigma^2| = \lim_{k \to \infty}|n_k \mathrm{Var}(\mathscr{C}_{n_k}^{(L)}) - n_k \mathrm{Var}(\mathscr{C}_{n_k})|\leq \delta_L.
\]
Thus every subsequential limit of the laws of $Z_n^{(L)}$ is Gaussian with variance within $\delta_L$ of $\sigma^2$. Consequently,
\begin{align}
\limsup_{n \to \infty} d_{BL}\bigl(\mathrm{Law}(Z_n^{(L)}), \mathcal{N}(0, \sigma^2)\bigr) 
\leq \sup_{v\geq 0: |v-\sigma^2|\leq \delta_L} d_{BL}(\mathcal{N}(0,v), \mathcal{N}(0, \sigma^2)) 
\xrightarrow{L \to \infty} 0, \label{eq:BL2}
\end{align}
where the convergence follows from $\delta_L\to 0$ and the continuity of the Gaussian law in its variance under $d_{BL}$.

Finally, combining~\eqref{eq:BL2} with \eqref{eq:BL1} and the triangle inequality for $d_{BL}$, we conclude that $Z_n$ converges in distribution to $\mathcal{N}(0, \sigma^2)$ as $n \to \infty$, thus establishing~\eqref{Eq:CLT_Cn}, and consequently the desired result in the theorem as well. 
\end{proof}

\subsection{Proof of Proposition~\ref{prop: concentrationGrid}}
\begin{proof}[Proof of Proposition~\ref{prop: concentrationGrid}] 
The proof of this proposition is an immediate application of Lemma~\ref{lem:VacancyDifference} and McDiarmid's inequality~\citep{mcdiarmid1989method}.
\end{proof}

\subsection{Proof of Theorem~\ref{Thm:ConsistencyGeneralDimGrid}}

\begin{proof}[Proof of Theorem~\ref{Thm:ConsistencyGeneralDimGrid}]  
By Proposition~\ref{prop: VonRawData_deterministic}, we may reduce to the copula case, namely the case where $P^X = \mathrm{Unif}([0,1]^{d_X})$ and $P^Y = \mathrm{Unif}([0,1]^{d_Y})$. The same lemma further shows that, in order to establish the consistency of $\mathcal{V}({\bf X}, {\bf Y}, {\bf U}, {\bf V}; 1/2 n^{-1/d})$, it is enough to prove the corresponding consistency statement for $\mathcal{V}({\bf X}, {\bf Y}, {\bf X}, {\bf Y}; 1/2 n^{-1/d})$. The remaining argument is identical to the proof of Theorem~\ref{thm:ConsistencyGeneralDim}, and is therefore omitted.
\end{proof}

\subsection{Proof of Theorem~\ref{thm:KVariateCoverageConsistency}}

\begin{proof}[Proof of Theorem~\ref{thm:KVariateCoverageConsistency}]   
The proof structure largely mimics the proof of Theorem~\ref{Thm:populationlimit}. We therefore give the argument at a high level, indicating the $K$-tuple analogues of the auxiliary results used in the bivariate proof. By a $K$-tuple variant of Lemma~\ref{Lemma:DataProcessing} (the data-processing inequality there naturally extends to divergences between joint distribution of the $K$-tuple and the product of their marginal), we can again reduce the problem to the case where $P_1=\cdots=P_K = \mathrm{Unif}[0,1]$, so that $P$ is simply a copula. 

By the Lebesgue--Radon--Nikodym decomposition, we can decompose $dP = h\,d\mathrm{vol} + d\nu$ for some measurable function $h$ and a measure $\nu$ singular with respect to the Lebesgue measure. By an almost verbatim argument as in the proof of Theorem~\ref{Thm:populationlimit}, we can use the inner regularity of $\nu$ to show that the singular component does not contribute to the coverage coefficient in the limit. 

We then approximate $h$ by a function $\tilde h$ that is piecewise constant on $Q_{j_1,\ldots,j_K} = \prod_{k\in[K]} [(j_k-1)/N, j_k/N)$, with $N$ chosen large enough so that $\int |h - \tilde h| d\mathrm{vol} \leq \epsilon$. By $K$-tuple versions of Lemma~\ref{Lemma:DivergenceStability} and Proposition~\ref{Prop:TVPerturbation}, both the coverage coefficient and the $f$-divergence limit vary continuously with respect to total-variation perturbations of $P$. Hence, it suffices to only consider $P = \tilde h \,d\mathrm{vol}$. Now, since $\tilde h$ has piecewise constant densities, we can explicitly control the coverage coefficient using a $K$-tuple variant of Proposition~\ref{Prop:PiecewiseConstant} to conclude the proof. 
\end{proof}

\subsection{Proof of Theorem~\ref{thm:ConsistencyKComponent}}
\begin{proof}
For each $i \in [n]$ and $k \in [K]$, define $\widetilde{X}_i^{(k)} := T_k(X_i^{(k)}) \in \mathbb{R}^{d_k}$. For notation simplicity, suppose $Z_i = (X_i^{(1)}, \ldots, X_i^{(K)}) \in \mathbb{R}^{d_1} \times \ldots \times \mathbb{R}^{d_K}$, $Z^{(k)} = (X_1^{(k)}, \ldots, X_n^{(k)}) \in \mathbb{R}^{nd_k}$ and $\widetilde{Z}_i = (\widetilde{X}_i^{(1)},\ldots, \widetilde{X}_i^{(K)})$, $\widetilde{Z}^{(k)}=(\widetilde{X}_1^{(k)},\ldots, \widetilde{X}_n^{(k)})$. Set $d := \sum_{k = 1}^K d_k$. The proof follows the same reduction strategy as Theorem~\ref{thm:ConsistencyGeneralDim}; we first establish the corresponding raw-data-to-copula approximation, now in the $K$-component notation:
\begin{align}
\biggl|\mathcal{V}_{n} - \mathcal{V}\biggl((\widetilde{Z}^{(k)})_{k = 1}^K, (\widetilde{Z}^{(k)})_{k = 1}^K; \frac{1}{2n^{1/d}} \biggr)\biggr| \pto 0. \label{eq:SelfReference} 
\end{align}

For each $k \in [K]$, given the reference points $U^{(k)}$, let $R_i^{(k)}$ and $\widetilde{R}_i^{(k)}$ be the multivariate ranks of $X_i^{(k)}$ and $\widetilde{X}_i^{(k)}$, defined as~\eqref{eq:MultiRankKsample}, respectively, and the induced joint ranks are
\[
R_i = (R_i^{(1)}, \ldots, R_i^{(K)}) \quad \text{and} \quad \widetilde{R}_i = (\widetilde{R}_i^{(1)}, \ldots, \widetilde{R}_i^{(K)}),
\]
for $i = 1, \ldots, n$. Define
\[
\widetilde{A}_1 = \bigcup_{i=1}^n B\biggl(R_i, \frac{1}{2n^{1/d}}\biggr), \quad \widetilde{A}_2 = \bigcup_{i=1}^n B\biggl(\widetilde{R}_i, \frac{1}{2n^{1/d}}\biggr) , \quad \widetilde{A}_3 = \bigcup_{i=1}^n B\biggl(\widetilde{Z}_i, \frac{1}{2n^{1/d}}\biggr). 
\]
Then we get
\begin{align}
  &\biggl|\mathcal{V}_{n} - \mathcal{V}\biggl((\widetilde{Z}^{(k)})_{k = 1}^K, (\widetilde{Z}^{(k)})_{k = 1}^K; \frac{1}{2n^{1/d}} \biggr)\biggr| \notag \\ 
  \leq & \biggl|\mathcal{V}_{n} - \mathcal{V} \biggl((\widetilde{Z}^{(k)})_{k = 1}^K, ({U}^{(k)})_{k = 1}^K; \frac{1}{2n^{1/d}}\biggr)\biggr| \notag \\ 
  &+\biggl|\mathcal{V} \biggl((\widetilde{Z}^{(k)})_{k = 1}^K, ({U}^{(k)})_{k = 1}^K; \frac{1}{2n^{1/d}}\biggr) - \mathcal{V} \biggl((\widetilde{Z}^{(k)})_{k = 1}^K, (\widetilde{Z}^{(k)})_{k = 1}^K; \frac{1}{2n^{1/d}} \biggr)\biggr| \notag \\ 
  \leq & \mathrm{vol}(\widetilde{A}_1 \Delta \widetilde{A}_2) + \mathrm{vol}(\widetilde{A}_2 \Delta \widetilde{A}_3). \label{ineq:SelfReferenceControl}
\end{align}
By Lemma \ref{lem: symdiff}, the final two terms above can be bounded as the following:
\begin{align}
\mathrm{vol}(\widetilde{A}_1 \Delta \widetilde{A}_2) \leq \sum_{i = 1}^n \mathrm{vol}\biggl(B\Bigl(R_i, \frac{1}{2n^{1/d}}\Bigr) \Delta B\Bigl(\widetilde{R}_i, \frac{1}{2n^{1/d}}\Bigr)\biggr) 
\leq2 d n^{-\frac{d-1}{d}} \sum_{i = 1}^n  \|R_i - \widetilde{R}_i\|_2,   \label{ineq:volA1A2}
\end{align}
and
\begin{align}
  \mathrm{vol}(\widetilde{A}_2\Delta \widetilde{A}_3) \leq \sum_{i = 1}^n \mathrm{vol}\biggl(B\Bigl(\widetilde{R}_i, \frac{1}{2n^{1/d}}\Bigr) \Delta B\Bigl(\widetilde{Z}_i, \frac{1}{2n^{1/d}}\Bigr)\biggr) 
\leq2 d n^{-\frac{d-1}{d}} \sum_{i = 1}^n  \| \widetilde{R}_i - \widetilde{Z}_i\|_2. \label{ineq:volA2A3}
\end{align}

Write ${P}_{n, k}^{\widetilde{X}} := \frac{1}{n}\sum_{i = 1}^n \delta_{\widetilde{X}_i^{(k)}}$, $\mu_n^{(k)} := n^{-1}\sum_{i = 1}^n \delta_{\tilde{X}_i^{(k)}}$ and $d_* := \max_{k \in [K]} d_k$. Since $\widetilde{X}_1^{(k)}, \ldots,\widetilde{X}_n^{(k)} \stackrel{\mathrm{iid}}{\sim} \mathrm{Unif}([0, 1]^{d_k})$, by the convergence rate of empirical Wasserstein  distance~\citep{fournier2015rate} we have
\begin{align}
    \frac{1}{n}\sum_{i = 1}^n  \|\widetilde{R}_i - \widetilde{Z}_i\|_2
    &= \frac{1}{n}\sum_{k = 1}^K\sum_{i = 1}^n
    \|\widetilde{R}_i^{(k)} - \widetilde{X}_i^{(k)}\|_2 = \sum_{k = 1}^K\mathcal{W}_1(P_{n, k}^{\widetilde{X}}, \mu_n^{(k)})\notag\\
    &\leq \sum_{k = 1}^K
    \Bigl[\mathcal{W}_1(P_{n, k}^{\widetilde{X}},  \mathrm{Unif}([0, 1]^{d_k}))
    + \mathcal{W}_1( \mathrm{Unif}([0, 1]^{d_k}), \mu_n^{(k)})\Bigr] \notag \\
    &= O(n^{-\frac{1}{d_*}}\log^2 n).
    \label{ineq:A1A2Rate}
\end{align}
Taking this back to~\eqref{ineq:volA2A3}, we have $\mathrm{vol}(\widetilde{A}_2\Delta \widetilde{A}_3) = o_p(1)$. 

To bound~\eqref{ineq:volA1A2}, note that $\|R_i - \widetilde{R}_i\|_2 \leq \|R_i - \widetilde{Z}_i\|_2 + \|\widetilde{Z}_i - \widetilde{R}_i\|_2$, it suffices to bound $n^{-1}\sum_{i = 1}^n \|R_i - \widetilde{Z}_i\|_2$. Note that 
\begin{align}
    \frac{1}{n}\sum_{i = 1}^n \|R_i - \widetilde{Z}_i\|_2 \leq  \sum_{k = 1}^K\Bigl(\frac{1}{n}\sum_{i = 1}^n\|R_i^{(k)} - \widetilde{X}_i^{(k)}\|_2^2\Bigr)^{1/2}. \label{ineq:tildeA2A3Intermed}
\end{align}
By Proposition~\ref{prop:HolderOTConvergence}, for any $k \in [K]$, there exists some $\gamma_k>0$ depending only on $\theta_k, \alpha_k, L_k$ such that  
\begin{align}
\Bigl(\frac{1}{n}\sum_{i = 1}^n\|R_i^{(k)} - \widetilde{X}_i^{(k)}\|_2^2\Bigr)^{1/2} = O\Bigl(r_{n, d_k}^{\frac{\theta_k}{1+\theta_k}} \log^{\gamma_k}n\Bigr) = o_p(n^{-\frac{1}{d}}), \label{eq:TermRate}
\end{align}
where the final equality follows from~\ref{ass:Kdimensionregime}. Consequently, combining~\eqref{ineq:tildeA2A3Intermed} and~\eqref{eq:TermRate} with~\eqref{ineq:A1A2Rate} and~\eqref{ineq:volA1A2}, we get $\mathrm{vol}(\widetilde{A}_1 \Delta \widetilde{A}_2) =o_p(1)$. Hence, the right-hand side of~\eqref{ineq:SelfReferenceControl} converges to 0 in probability, and we obtain \eqref{eq:SelfReference}. 

Therefore, it suffices to consider the copula case. That is, we may assume that
each marginal distribution satisfies $P_k=\mathrm{Unif}([0, 1]^{d_k})$, $k\in[K]$, so that
\[
X_1^{(k)},\ldots,X_n^{(k)}\stackrel{\mathrm{iid}}{\sim}\mathrm{Unif}([0, 1]^{d_k}),
\qquad k\in[K].
\]
We then only need to study the consistency of
\[
\widetilde{\mathcal{V}}_{n}:=\mathcal{V}\biggl((Z^{(k)})_{k = 1}^K,(Z^{(k)})_{k = 1}^K;\frac{1}{2n^{1/d}}\biggr)
\]
under this copula representation. Let $\mathrm{d}P = g\, \mathrm{d}\mathrm{vol} + \mathrm{d}\eta$ be the Lebesgue-Radon-Nikodym decomposition of $P$, where $g$ is a measurable function and measure $\eta$ is singular to the Lebesgue measure. By an identical calculation as~\eqref{eq: limitExp} and~\eqref{eq:VacancySecondMoment}, we can establish that 
\begin{align*}
    \mathbb{E}\widetilde{\mathcal{V}}_{n} \rightarrow \int e^{-g(x)} \, \mathrm{d}x , \qquad \mathrm{Var}(\widetilde{\mathcal{V}}_{n})\rightarrow 0,
\end{align*}
as $n \to \infty$. Hence the result follows as desired. 
\end{proof}

\subsection{Proof of Theorem~\ref{Thm:KVariateCoverageCLT}}
\begin{proof}[Proof of Theorem~\ref{Thm:KVariateCoverageCLT}]
Since $X^{(1)},\ldots,X^{(K)}$ are mutually independent, the random permutations $\pi^{(1)}, \ldots, \pi^{(K)}$ defined in~\eqref{eq:MultiRankKsample} are mutually independent. Consequently, the joint ranks $R_1,\ldots,R_n$ defined in~\eqref{eq:JointRankKcomponent} are independent and identically distributed according to the uniform distribution on $[0,1]^{d}$. The result then follows by repeating the argument used in Theorem~\ref{thm: CLT}.
\end{proof}

\section{Additional results}
\label{sec:additional-results}
We collect here additional results that are used in the proof of the main theorems, as well as to substantiate claims made in the main text. These results supply the technical details deferred from the high-level proofs in Section~\ref{sec:proof_main}.

\subsection{Additional results used in the proof of Theorem~\ref{Thm:populationlimit}}
The first preliminary result shows that the coverage correlation coefficients of samples generated from two probability measures close in total variation distance are (stochastically) close to each other.

\begin{prop}
\label{Prop:TVPerturbation}
Let $\mu$ and $\nu$ be two probability measures on $\mathbb{R}^2$ with $d_{\mathrm{TV}}(\mu, \nu)\leq \epsilon$ and suppose $(X_i,Y_i)_{i\in[n]}\iid \mu$, $(\tilde X_i,\tilde Y_i)\iid \nu$. Also suppose $(U_i,V_i)_{i\in[n]}\iid\mathrm{Unif}[0,1]^2$ and $(\tilde U_i,\tilde V_i)_{i\in[n]}\iid\mathrm{Unif}[0,1]^2$ are independent of $(X_i,Y_i)_{i\in[n]}$ and $(\tilde X_i,\tilde Y_i)_{i\in[n]}$ respectively.  There exists a coupling between $(X_i,Y_i,U_i,V_i)_{i\in[n]}$ and $(\tilde X_i,\tilde Y_i,\tilde U_i,\tilde V_i)_{i\in[n]}$ such that 
\begin{align*}
    \biggl|\mathcal{V}\Bigl((X_i)_{i\in[n]}, &(Y_i)_{i\in[n]}, (U_i)_{i\in[n]},(V_i)_{i\in[n]}; \frac{1}{2\sqrt{n}}\Bigr) \\
    &- \mathcal{V}\Bigl((\tilde X_i)_{i\in[n]}, (\tilde Y_i)_{i\in[n]}, (\tilde U_i)_{i\in[n]},(\tilde V_i)_{i\in[n]}; \frac{1}{2\sqrt{n}}\Bigr)\biggr| \leq \epsilon + O_p(n^{-1/2}).
\end{align*}
\end{prop}
\begin{proof}
For notational simplicity, write $\bs{X}:=(X_i)_{i\in[n]}$, $\bs{Y}:=(Y_i)_{i\in[n]}$, $\tilde {\bs{X}}:=(\tilde X_i)_{i\in[n]}$, $\tilde{\bs Y} :=(\tilde Y_i)_{i\in[n]}$, $\bs{U}:=(U_i)_{i\in[n]}$, $\bs{V}:=(V_i)_{i\in[n]}$, $\tilde {\bs{U}}:=(\tilde U_i)_{i\in[n]}$, $\tilde{\bs V} :=(\tilde V_i)_{i\in[n]}$. Also, let $(\bs X^{\mathrm{Pois}}, \bs Y^{\mathrm{Pois}}) := (X^{\mathrm{Pois}}_i, Y^{\mathrm{Pois}}_i)_{i\in[N]} \sim \mathrm{PP}(n\mu)$ be a Poisson point process with intensity $n\mu$ with $N\sim \mathrm{Poi}(n)$ points. Similarly, define mutually independent Poisson point processes $(\tilde{\bs X}^{\mathrm{Pois}}, \tilde{\bs Y}^{\mathrm{Pois}}) \sim \mathrm{PP}(n\nu)$, $(\bs U^{\mathrm{Pois}}, \bs V^{\mathrm{Pois}}) \sim \mathrm{PP}(n\cdot\mathrm{vol})$ and $(\tilde{\bs U}^{\mathrm{Pois}}, \tilde{\bs V}^{\mathrm{Pois}}) \sim \mathrm{PP}(n\cdot\mathrm{vol})$. 

We define a new measure $\lambda$ that is the maximum of $\mu$ and $\nu$ as follows. Let $f = d\mu / d(\mu+\nu)$ and $g = d\nu/d(\mu+\nu)$ be densities with respect to the common dominating measure $\mu+\nu$, we set $d\lambda:= \max(f,g) \, d(\mu+\nu)$. It is clear that $\lambda\geq \mu$ and $\lambda\geq \nu$, and from the total variation bound between $\mu$ and $\nu$ we have $\lambda(\mathbb{R}^2) - \max\{\mu(\mathbb{R}^2),\nu(\mathbb{R}^2)\}\leq \epsilon$. We define another two independent Poisson point process $(\bs X^{\mathrm{max}}, \bs Y^{\mathrm{max}}) \sim \mathrm{PP}(n\lambda)$ and $(\bs U^{\mathrm{max}}, \bs V^{\mathrm{max}}) \sim \mathrm{PP}(n \cdot\mathrm{vol})$.

We now define a chain of couplings by Lemmas~\ref{Lemma:Poissonisation} and~\ref{Lemma:Thinning} as follows:
\[
    \begin{pmatrix} \bs{X}\\\bs{Y}\\ \bs{U}\\ \bs{V}\end{pmatrix}
    \xleftrightarrow{\ref{Lemma:Poissonisation}}
    \begin{pmatrix} \bs{X}^{\mathrm{Pois}}\\\bs{Y}^{\mathrm{Pois}}\\ \bs{U}^{\mathrm{Pois}}\\ \bs{V}^{\mathrm{Pois}}\end{pmatrix}
    \xleftrightarrow{\ref{Lemma:Thinning}}
    \begin{pmatrix} \bs{X}^{\mathrm{max}}\\\bs{Y}^{\mathrm{max}}\\ \bs{U}^{\mathrm{max}}\\ \bs{V}^{\mathrm{max}}\end{pmatrix} \xleftrightarrow{\ref{Lemma:Thinning}}
    \begin{pmatrix} \tilde{\bs{X}}^{\mathrm{Pois}}\\\tilde{\bs{Y}}^{\mathrm{Pois}}\\ \tilde{\bs{U}}^{\mathrm{Pois}}\\ \tilde{\bs{V}}^{\mathrm{Pois}}\end{pmatrix}
    \xleftrightarrow{\ref{Lemma:Poissonisation}}
    \begin{pmatrix} \tilde{\bs{X}}\\\tilde{\bs{Y}}\\ \tilde{\bs{U}}\\ \tilde{\bs{V}}\end{pmatrix}
\]
In particular, under this coupling, the cardinality of $\bs{X}^{\mathrm{Pois}}$, $\bs{Y}^{\mathrm{Pois}}$, $\bs{U}^{\mathrm{Pois}}$, $\bs{V}^{\mathrm{Pois}}$ and their tilde-ed counterpart are all equal to some $N\sim \mathrm{Poi}(n)$ and the cardinality of $\bs{X}^{\mathrm{max}}$, $\bs{Y}^{\mathrm{max}}$, $\bs{U}^{\mathrm{max}}$, $\bs{V}^{\mathrm{max}}$ are equal to $M\sim\mathrm{Poi}(n\lambda(\mathbb{R}^2))$. By Lemma~\ref{Lemma:Poissonisation}, we have
\begin{align*}
    \biggl|\mathcal{V}\Bigl(\bs{X},\bs{Y},\bs{U}, \bs{V};\frac{1}{2\sqrt{n}}\Bigr) - \mathcal{V}\Bigl(\bs{X}^{\mathrm{Pois}},\bs{Y}^{\mathrm{Pois}},\bs{U}^{\mathrm{Pois}} \bs{V}^{\mathrm{Pois}};\frac{1}{2\sqrt{N}}\Bigr)\biggr| &= O_p(n^{-1/2}) \\
    \biggl|\mathcal{V}\Bigl(\tilde{\bs{X}},\tilde{\bs{Y}},\tilde{\bs{U}}, \tilde{\bs{V}};\frac{1}{2\sqrt{n}}\Bigr) - \mathcal{V}\Bigl(\tilde{\bs{X}}^{\mathrm{Pois}},\tilde{\bs{Y}}^{\mathrm{Pois}},\tilde{\bs{U}}^{\mathrm{Pois}} \tilde{\bs{V}}^{\mathrm{Pois}};\frac{1}{2\sqrt{N}}\Bigr)\biggr| &= O_p(n^{-1/2})
\end{align*}
By Lemma~\ref{Lemma:Thinning}, we have 
\begin{align*}
    \biggl|\mathcal{V}\Bigl(\bs{X}^{\mathrm{max}},\bs{Y}^{\mathrm{max}},\bs{U}^{\mathrm{max}}, \bs{V}^{\mathrm{max}};\frac{1}{2\sqrt{M}}\Bigr) &- \mathcal{V}\Bigl(\bs{X}^{\mathrm{Pois}},\bs{Y}^{\mathrm{Pois}},\bs{U}^{\mathrm{Pois}} \bs{V}^{\mathrm{Pois}};\frac{1}{2\sqrt{N}}\Bigr)\biggr| \\
    &\hspace{4cm}\leq \epsilon+ O_p(n^{-1/2}) \\
    \biggl|\mathcal{V}\Bigl(\bs{X}^{\mathrm{max}},\bs{Y}^{\mathrm{max}},\bs{U}^{\mathrm{max}}, \bs{V}^{\mathrm{max}};\frac{1}{2\sqrt{M}}\Bigr) &- \mathcal{V}\Bigl(\tilde{\bs{X}}^{\mathrm{Pois}},\tilde{\bs{Y}}^{\mathrm{Pois}},\tilde{\bs{U}}^{\mathrm{Pois}} \tilde{\bs{V}}^{\mathrm{Pois}};\frac{1}{2\sqrt{N}}\Bigr)\biggr| \\
    &\hspace{4cm}\leq \epsilon +  O_p(n^{-1/2})
\end{align*}
The desired result follows by combining the above four vacancy difference bounds under the coupling.
\end{proof}

The next result controls the contribution to the coverage from a rectangular region in the domain where the density is constant.

\begin{prop}
\label{Prop:PiecewiseConstant}
Suppose $P^{(X,Y)}$ is a probability measure on $[0,1]^2$ with $\mathrm{Unif}[0,1]$ marginals. Suppose that for some $0\leq a_1 < a_2\leq 1$ and $0\leq b_1 < b_2\leq 1$, $P^{(X,Y)}$ is equal to $q \cdot \mathrm{vol}$ when restricted to $[a_1,a_2]\times [b_1,b_2]$, where $\mathrm{vol}$ denotes the Lebesgue measure. Let $(X_1,Y_1),\ldots,(X_n,Y_n)\iid P^{(X,Y)}$ and let $R_i$ be defined as in~\eqref{eq:R_iGen} with respect to reference points $(U_i,V_i)_{i\in[n]}\iid \mathrm{Unif}[0,1]^2$. Then
\[
   \vol\biggl( \bigcup_{i: (X_i,Y_i)\in(a_1,a_2]\times (b_1,b_2]} B\Bigl(R_i, \frac{1}{2\sqrt{n}}\Bigr)\biggr) = (1-e^{-q})(a_2-a_1)(b_2-b_1)(1+o_p(1)).
\]
\end{prop}
\begin{proof}
Define $\mathcal{I} := \{i\in[n]: a_1 < X_i \leq a_2\}$, $\mathcal{I}_{-}:=\{i\in[n]:X_i \leq a_1\}$, $\mathcal{J} := \{j\in[n]: b_1 < Y_j\leq b_2\}$ and $\mathcal{J}_-:=\{j\in[n]:Y_j\leq b_1\}$. We write $S_0:= |\mathcal{I}_-|$, $S_1:=|\mathcal{I}_-\cup \mathcal{I}|$, $T_0:=|\mathcal{J}_-|$ and $T_1:=|\mathcal{J}_-\cup\mathcal{J}|$. Also, let $M := |\mathcal{I}\cap \mathcal{J}|$.

Given $(U_i,V_i)_{i\in[n]}\iid \mathrm{Unif}[0,1]^2$, observe that for $R_i^X$ and $R_i^Y$ be defined in~\eqref{def:EmpRankGen}, we have 
\begin{align*}
\{R_i^X:i\in\mathcal{I}\} &= \{U_{(S_0 + 1)}, \ldots, U_{(S_1)}\} \quad \text{and} \quad \{R_j^Y:j\in\mathcal{J}\} = \{V_{(T_0 + 1)}, \ldots, V_{(T_1)}\}.
\end{align*}
Let $\mathcal{F}$ be the $\sigma$-algebra generated by $U_{(S_0)}$, $U_{(S_1+1)}$, and $V_{(T_0)}$, $V_{(T_1+1)}$ and $M$. By \citet[Theorem~2.5]{david2004order}, we have 
    \begin{align*}
         (R_i^X: i\in\mathcal{I}) \bigm| \mathcal{F} & \iid \mathrm{Unif}[U_{(S_{0})}, U_{(S_{1} + 1)}] \quad \text{and}\quad 
          (R_j^Y: j\in\mathcal{J}) \bigm| \mathcal{F} \iid \mathrm{Unif}[V_{(T_{0})}, V_{(T_{1} + 1)}].
    \end{align*}
Furthermore, since 
\begin{align}
\bigl((X_i,Y_i):i\in \mathcal{I}\cap\mathcal{J}\bigr)\bigm| M \iid \mathrm{Unif}(a_1,a_2]\otimes \mathrm{Unif}(b_1,b_2], \label{eq:BlockUniformXY}
\end{align}
we have 
\[
    (R_i: i\in\mathcal{I}\cap\mathcal{J})\mid \mathcal{F} \iid \mathrm{Unif}[U_{(S_{0})}, U_{(S_{1} + 1)}]\otimes \mathrm{Unif}[V_{(T_{0})}, V_{(T_{1} + 1)}].
\]

By law of large numbers, there is an event $\Omega$ with probability 1 on which we have $M / n \to q(a_2-a_1)(b_2-b_1)$, $S_0 / n\to a_1$, $S_1/n\to a_2$, $T_0 / n\to b_1$, $T_1/n\to b_2$, $U_{(S_0)}\to a_1$, $U_{(S_1+1)}\to a_2$, $V_{(T_0)}\to b_1$ and $V_{(T_1+1)}\to b_2$.  We will work on this event henceforth. 

As $n\to\infty$, the contribution of the covered area by points near the boundary of any rectangle is negligible (so we may ignore the periodic boundary condition), hence Lemma~\ref{le: VacancyLimitRandomRef} and a linear rescaling, conditional on $\mathcal{F}$, we have 
\begin{align*}
\mathbb{E}\biggl\{\vol\biggl(\bigcup_{i\in\mathcal{I}\cap\mathcal{J}}B\Bigl(R_i, \frac{1}{2\sqrt{n}}\Bigr)\biggr) \biggm | \mathcal{F}\biggr\} &\to (1-e^{-q})(a_2-a_1)(b_2-b_1),\\
\var\biggl\{\vol\biggl(\bigcup_{i\in\mathcal{I}\cap\mathcal{J}}B\Bigl(R_i, \frac{1}{2\sqrt{n}}\Bigr)\biggr) \biggm | \mathcal{F}\biggr\} &\to 0.
\end{align*}
By the Dominated Convergence Theorem, the same result holds unconditionally, which implies the desired result by an application of Chebyshev's inequality.
\end{proof}

\subsection{Additional results used in the proof of Theorem~\ref{thm:ConsistencyGeneralDim}}
We first introduce some notation from the optimal transport theory. 

For any measure $\eta$ and $\zeta$, we have the Kantorovich dual formula for the $2$-Wasserstein distance:
\begin{align*}
 \frac{1}{2}\mathcal{W}_2^2(\eta, \zeta) = \frac{1}{2}\int \|y\|_2^2 \, \mathrm{d}\eta +  \frac{1}{2}\int \|y\|_2^2 \, \mathrm{d}\zeta - \inf_{f\in \mathcal{F}_\eta} \mathcal{S}_{\eta, \zeta}(f)
\end{align*}
where $\mathcal{F}_\eta$ denotes spaces of convex functions on $\mathbb{R}^d$ that are also elements of $L^1(\eta)$ and $\mathcal{S}_{\eta, \zeta}$ denote the Kantorovich functional:
\begin{align}
    S_{\eta, \zeta}(f) := \int f \, \mathrm{d}\eta + \int f^* \, \mathrm{d} \zeta, \label{eq:KantorovichPotentialFunctional}
\end{align}
where $f^*(y):= \sup_x\bigl\{x^Ty - f(x)\bigr\}$ is the Legendre conjugate of $f$.

\begin{prop}\label{prop:HolderOTConvergence}
Let $\mu, \nu$ be absolutely continuous probability measures on $\mathbb{R}^d$ with finite second moment. Suppose $T_0=\nabla \phi_0$ is the (almost surely unique) optimal transport map from $\mu$ to $\nu$, and is $(\theta, L)$-Hölder continuous for some $\theta \in (0, 1]$ and $L>0$ (see Definition~\ref{def:HolderConti}). Set $p:=1+\theta^{-1}$.

Let $\hat\mu_n$ and $\hat\nu_n$ be any estimators of $\mu, \nu$ with finite second moment. Write $\nu_n^\sharp := T_0\# \hat\mu_n$, there exists a constant $C>0$ depending only on $\theta, L$ such that for any optimal transport map $\hat T_n$ from $\hat{\mu}_n$ to $\hat{\nu}_n$, we have
\begin{align}
\int \|T_0(z) -  \hat{T}_n(z)\|_2^{p}  \, \mathrm{d}\hat{\mu}_n(z) \leq C \biggl\{\int \phi_0^*\, \mathrm{d}(\hat \nu_n - \nu_n^\sharp) + \int \hat{\Phi}^* \, \mathrm{d}(\nu_n^\sharp -  \hat\nu_n)\biggr\}, \label{eq:OTEstimationErrors}
\end{align}
where $\hat{\Phi} \in \argmin_{f} \mathcal{S}_{\hat\mu_n, \hat\nu_n}(f)$. 

In particular, given data $X_1, \ldots, X_n \stackrel{i i d}{\sim} \mu$ and $Y_1, \ldots, Y_n \stackrel{i i d}{\sim} \nu$, and let $\hat{\mu}_n= n^{-1} \sum_{i=1}^n \delta_{X_i}$ and $\hat{\nu}_n=n^{-1} \sum_{j=1}^n \delta_{Y_j}$ be the empirical measures. If we further assume that $\nu$ is compactly supported and $\mu \in \mathcal{P}_{t, \alpha}(\mathbb{R}^d)$ for some $t>0$ and $\alpha\in (0, 2]$, then there exists a finite exponent $c_{d, \alpha}>0$ depending on $d, \alpha > 0$ such that 
\begin{align*}
\frac{1}{n} \sum_{i=1}^n\left\|\hat{T}_n\left(X_i\right)-T_0\left(X_i\right)\right\|_2^p=O_p\left(r_{n, d} \log ^{c_{d, \alpha}} n\right)
\end{align*}
where
\begin{align*}
    r_{n, d}:= \begin{cases}
    n^{-1 / 2}, & d = 2, 3 \\ n^{-1 / 2} \log (1+n), & d=4 \\ n^{-2 / d}, & d \geq 5\end{cases}
\end{align*}
Consequently,
\begin{align*}
\frac{1}{n} \sum_{i=1}^n\Bigl\|\hat{T}_n\left(X_i\right)-T_0\left(X_i\right)\Bigr\|_2^2=O_p\bigl(r_{n, d}^{\frac{2 \theta}{1+\theta}} \log ^{\frac{2 c_{d, \alpha}\theta}{1+\theta}} n\bigr)
\end{align*}
\end{prop}
\begin{proof}
Suppose $\mu$ and $\nu$ are supported on open sets $\Omega_1, \Omega_2 \subseteq \mathbb{R}^d$, respectively. Let $\phi_0^*(y):=\sup_{z \in \Omega_1}\{\langle y,z\rangle-\phi_0(z)\}$ denote the Legendre conjugate of $\phi_0$. By the characterisation of optimal transport plans \citep[see, e.g.,][Theorem 2.29]{villani2021topics}, $T_0(x)\in\partial\phi_0(x)$ for $\mu$-a.e. $x$. The conjugacy relation for subdifferentials therefore yields $x\in\partial\phi_0^*(T_0(x))$ for $\mu$-a.e. $x$. Since $\phi_0^*$ is differentiable $\nu$-a.e., $\phi_0^*$ is differentiable at $T_0(x)$ for $\mu$-a.e. $x$. Thus we have $\nabla\phi_0^*(T_0(x))=x$ for $\mu$-a.e. $x$. Since $T_0$ is $(\theta, L)$-H\"older continuous, by \citet[Lemma 1]{yan2020stochastic}, it follows that $\phi_0^*$ is $(p, \rho)$-strongly convex (see Definition~\ref{def:stronglyconvex} and Proposition~\ref{prop:stronglyconvex}). That is to say, for any $x, y \in \Omega_2$
\begin{align*}
    \phi_0^*(y) \geq \phi_0^*(x) + \langle \nabla\phi_0^* (x) ,y - x \rangle + \frac{\rho}{2}\|x-y\|_2^{p},
\end{align*}
where $\rho = \frac{2\theta}{1+\theta}(1/L)^{\frac{1}{\theta}}$. Then applying $x = T_0(z)$ and $y = \hat T_n(z)$ for some $z \in \Omega_1$, we have 
\[
   \phi_0^*(\hat T_n(z)) - \phi_0^*(T_0(z))  \geq \langle z ,\hat T_n(z) - T_0(z) \rangle + \frac{\rho}{2}\|T_0(z) - \hat T_n(z)\|_2^{p}. 
\]
Integrating against $\hat \mu_n$ we obtain that 
\begin{align}
   \int \phi_0^*(\hat T_n(z)) &- \phi_0^*(T_0(z))\, \mathrm{d}\hat \mu_n(z)\notag\\
   & \geq \int \langle z ,\hat T_n(z) - T_0(z) \rangle  \, \mathrm{d}\hat \mu_n(z)+  \frac{\rho}{2} \int \|T_0(z) - \hat T_n(z)\|_2^{p}  \, \mathrm{d}\hat \mu_n(z) \label{ineq:IntStronglyConv}
\end{align}

By the optimality of $\hat T_n$ and $T_0$, we have 
\begin{align*} 
\mathcal{W}_2^2\bigl(\hat{\mu}_n, \hat{\nu}_n\bigr) & =\int\|z-\hat{T}_n(z)\|_2^2 d \hat{\mu}_n(z)=\int\|z\|_2^2 d \hat{\mu}_n+\int\|y\|_2^2 d \hat{\nu}_n-2 \int\langle z, \hat{T}_n(z)\rangle d \hat{\mu}_n, \\ 
\mathcal{W}_2^2\bigl(\hat{\mu}_n, \nu_n^{\sharp}\bigr) & \leq \int\left\|z-T_0(z)\right\|_2^2 d \hat{\mu}_n(z)=\int\|z\|_2^2 d \hat{\mu}_n+\int\|y\|_2^2 d \nu_n^{\sharp}-2 \int\langle z, T_0(z)\rangle d \hat{\mu}_n.
\end{align*}
Subtracting them gives
\begin{align}
\int \langle z ,\hat T_n(z) - T_0(z) \rangle  \, \mathrm{d}\hat \mu_n(z) \geq \frac{1}{2}\Bigl(\mathcal{W}_2^2\bigl(\hat{\mu}_n, \nu_n^{\sharp}\bigr)  - \mathcal{W}_2^2\bigl(\hat{\mu}_n, \hat{\nu}_n\bigr)\Bigr) + \frac{1}{2}\int \|y\|_2^2 \, \mathrm{d}(\hat\nu_n -\nu_n^{\sharp} ).  \label{ineq:InnerProdLowerBound}
\end{align}
Combining~\eqref{ineq:IntStronglyConv} with~\eqref{ineq:InnerProdLowerBound}, we get 
\begin{align}
\frac{\rho}{2} \int \|T_0(z) &- \hat T_n(z)\|_2^{p}  \, \mathrm{d}\hat \mu_n(z) \notag \\
\leq& \frac{1}{2}\Bigl(\mathcal{W}_2^2\bigl(\hat{\mu}_n, \hat{\nu}_n\bigr) - \mathcal{W}_2^2\bigl(\hat{\mu}_n, \nu_n^{\sharp}\bigr)\Bigr) + \int \bigl(\phi_0^*(y)-\frac{1}{2}\|y\|_2^2 \bigr) \, \mathrm{d}(\hat\nu_n - \nu_n^{\sharp}). \label{ineq:EmpOTUpperBound}
\end{align}

% Now for any measure $\eta$ and $\zeta$, we have the Kantorovich dual formula for the $2$-Wasserstein distance:
% \[
%  \frac{1}{2}\mathcal{W}_2^2(\eta, \zeta) = \frac{1}{2}\int \|y\|^2 \, \mathrm{d}\eta +  \frac{1}{2}\int \|y\|^2 \, \mathrm{d}\zeta - \inf_{f\in \mathcal{F}_\eta} \mathcal{S}_{\eta, \zeta}(f) 
% \]
% where $\mathcal{F}_\eta$ denotes spaces of convex functions on $\mathbb{R}^d$ that are also elements of $L^1(\eta)$ and $\mathcal{S}_{\eta, \zeta}$ denote the Kantorovich functional:
% \begin{align*}
%     S_{\eta, \zeta}(f) := \int f \, \mathrm{d}\eta + \int f^* \, \mathrm{d} \zeta. 
% \end{align*}
Hence for $\hat \Phi \in \argmin_f \mathcal{S}_{\hat \mu_n, \hat \nu_n}(f)$ and $\Phi^\sharp \in \argmin_f\mathcal{S}_{\hat \mu_n, \nu_n^\sharp}(f)$ we have 
\begin{align}
 \frac{1}{2}\Bigl(\mathcal{W}_2^2\bigl(\hat{\mu}_n, \hat{\nu}_n\bigr) - \mathcal{W}_2^2\bigl(\hat{\mu}_n, \nu_n^{\sharp}\bigr)\Bigr) &= \frac{1}{2}\int \|y\|_2^2 \, \mathrm{d}(\hat \nu_n - \nu_n^\sharp) + \mathcal{S}_{\hat \mu_n, \nu_n^\sharp}(\Phi^\sharp) - \mathcal{S}_{\hat \mu_n, \hat \nu_n}(\hat \Phi) \notag \\ 
 & \leq \frac{1}{2}\int \|y\|_2^2 \, \mathrm{d}(\hat \nu_n - \nu_n^\sharp) + \int \hat \Phi^* \, \mathrm{d}(\nu_n^\sharp - \hat \nu_n),\label{ineq:WassDiff}
\end{align}
where the final inequality is due to $\mathcal{S}_{\hat \mu_n, \nu_n^\sharp}(\Phi^\sharp) \leq \mathcal{S}_{\hat \mu_n, \nu_n^\sharp}(\hat\Phi)$. Taking~\eqref{ineq:WassDiff} back to~\eqref{ineq:EmpOTUpperBound} we obtain that 
\begin{align}
\int \|T_0(z) - \hat T_n(z)\|_2^{p}  \, \mathrm{d}\hat \mu_n(z) \leq \frac{2}{\rho} \biggl\{\int \phi_0^*\, \mathrm{d}(\hat \nu_n - \nu_n^\sharp) + \int \hat\Phi^* \, \mathrm{d}(\nu_n^\sharp - \hat \nu_n)\biggr\}. \label{ineq:EmpOTMapUpperBound}
\end{align}

Now, set $\hat{\mu}_n= n^{-1} \sum_{i=1}^n \delta_{X_i}$ and $\hat{\nu}_n=n^{-1} \sum_{j=1}^n \delta_{Y_j}$. We first show that the first term on the right-hand side of~\eqref{ineq:EmpOTMapUpperBound} is negligible. Because $\nu$ is compactly supported and $\phi_0^*$ is finite and convex on $\mathrm{Supp}(\nu)$, $\phi_0^*$ is bounded on $\mathrm{Supp}(\nu)$. Therefore, by the Law of Large Numbers we have
\[
\int \phi_0^*\, \mathrm{d}(\hat \nu_n - \nu_n^\sharp) = \int \phi_0^*\, \mathrm{d}(\hat \nu_n - \nu) + \int \phi_0^*(T_0(z)) \, \mathrm{d}(\mu- \mu_n) = O_p(n^{-1/2}). 
\]
As for the second term of~\eqref{ineq:EmpOTMapUpperBound}, by the triangle inequality
\[
\Bigl|\int \hat\Phi^* \, \mathrm{d}(\nu_n^\sharp - \hat \nu_n)\Bigr| \leq \Bigl|\int \hat\Phi^* \, \mathrm{d}(\nu_n^\sharp - \nu)\Bigr|  + \Bigl|\int \hat\Phi^* \, \mathrm{d}(\hat \nu_n - \nu)\Bigr| .
\]
By an identical argument that establishes (C.14) in \citet{Deb2021RatesOEsupp}, the two summands on the right-hand side above can be controlled to be $O_p(r_{n, d}\log^{c_{d, \alpha}} n)$. Thus establishing the desired inequality. 
\end{proof}

The following proposition addresses the remaining case when $d = 1$ via the classical Koml\'{o}s--Major--Tusn\'{a}dy Theorem \cite{komlos1975approximation}. 

\begin{prop}\label{prop:1DKMTMatching}
    Suppose $\mu$ is an absolutely continuous probability measure on $\mathbb{R}$ with distribution function $F$. Given $X_1, \ldots, X_n \stackrel{\mathrm{iid}}{\sim} \mu$ and $u_1, \ldots, u_n \in [0, 1]$. Define  
    \[
    \pi_n := \argmin_{\pi \in \mathcal{S}_n} \sum_{i = 1}^n |X_{i} - u_{\pi(i)}|^2.
    \]
    Then, if
    \begin{enumerate}[label = (\roman*), ref=\theprop(\roman*)]
        \item\label{prop:1DKMTMatching_random}  $u_i = U_i$, $i = 1, \ldots, n$, where $U_1, \ldots, U_n \stackrel{\mathrm{iid}}{\sim} \mathrm{Unif}([0, 1])$, or
        \item\label{prop:1DKMTMatching_deterministic}  $\bs{u}:=(u_1, \ldots, u_n)$ is a sequence of deterministic points such that the star discrepancy (see Definition~\ref{def:StarDiscrepancy}) $\mathcal{D}^*(\bs{u}) = O(\log n / n)$,
    \end{enumerate}
     we have 
    \[
    \frac{1}{n}\sum_{i = 1}^n (u_{\pi_n(i)} - F(X_i))^2 = O_p(n^{-1});
    \]

\end{prop}
\begin{proof}
(i) Suppose $X_{(1)} \leq \ldots \leq X_{(n)}$ and $U_{(1)} \leq \ldots \leq U_{(n)}$ are order statistics of $(X_i)_{i = 1}^n$ and $(U_i)_{i = 1}^n$, respectively. Note that
\begin{align}
\Bigl(\frac{1}{n}\sum_{i = 1}^n |U_{\pi_n(i)} - F(X_i)|^2\Bigr)^{1/2} &= \Bigl(\frac{1}{n} \sum_{i = 1}^n |U_{(i)} - F(X_{(i)})|^2 \Bigr)^{1/2} \notag \\
&\leq \biggl(\frac{1}{n} \sum_{i = 1}^n \Bigl|U_{(i)} - \frac{i}{n}\Bigr|^2 \biggr)^{1/2} + \biggl(\frac{1}{n} \sum_{i = 1}^n \Bigl|\frac{i}{n} - F(X_{(i)})\Bigr|^2\biggr)^{1/2}.  \label{eq:1DEmpiricalOTError}
\end{align}
It remains to bound each term in~\eqref{eq:1DEmpiricalOTError} separately. 

Let $F_n(t):= n^{-1}\sum_{i = 1}^n \bone_{[X_i, \infty)]}(t)$ be the empirical distribution function of $(X_i)_{i = 1}^n$. By the Koml\'os--Major--Tusn\'ady Theorem (\citet[Theorem~3]{komlos1975approximation}, see also \citet[Theorem~4.4.1]{csorgo1981strong}), there exists a sequence of Brownian bridges $\{B_n(y): 0 \leq y \leq 1\}_{n \geq 1}$ such that
\[
\sup_{t\in \mathbb{R}} |\sqrt{n}(F_n(t) - F(t)) - B_n(F(t))| \stackrel{\mathrm{a.s.}}{=} O(n^{-1/2} \log n).
\]
Since $\|B_n\|_{\infty} = O_p(1)$, we obtain that
\begin{align}
\biggl(\frac{1}{n} \sum_{i = 1}^n \Bigl|\frac{i}{n} - F(X_{(i)}) \Bigr|^2\biggr)^{1/2}
\leq \max_{i \in [n]} \Bigl|\frac{i}{n} - F(X_{(i)})\Bigr| = \max_{i \in [n]} \bigl|F_n(X_{(i)}) - F(X_{(i)})\bigr| = O_p(n^{-1/2}). \label{eq:KMTapplicationX}
\end{align}

Let $G_n(u):= n^{-1}\sum_{i = 1}^n \bone_{[U_i, \infty)} (u)$ be the empirical distribution of $(U_i)_{i = 1}^n$. Then an argument similar to~\eqref{eq:KMTapplicationX} can be applied to $G_n$ to obtain that 
\begin{align}
\Bigl(\frac{1}{n} \sum_{i = 1}^n \Bigl|U_{(i)} - \frac{i}{n}\Bigr|^2 \Bigr)^{1/2} = O_p(n^{-1/2}). \label{eq:KMTapplicationU}
\end{align}
Then the result follows by taking~\eqref{eq:KMTapplicationX} and~\eqref{eq:KMTapplicationU} back to~\eqref{eq:1DEmpiricalOTError}. 

(ii) Suppose $u_{(1)} \leq \ldots \leq u_{(n)}$, and by leveraging the similar decomposition as~\eqref{eq:1DEmpiricalOTError}, we have 
\begin{align*}
    \Bigl(\frac{1}{n}\sum_{i = 1}^n |u_{\pi_n(i)} - F(X_i)|^2\Bigr)^{1/2}  &\leq \biggl(\frac{1}{n} \sum_{i = 1}^n \Bigl|u_{(i)} - \frac{i}{n}\Bigr|^2 \biggr)^{1/2} + \biggl(\frac{1}{n} \sum_{i = 1}^n \Bigl|\frac{i}{n} - F(X_{(i)})\Bigr|^2\biggr)^{1/2}  \\ 
    & \leq \mathcal{D}^*(\bs{u}) + \biggl(\frac{1}{n} \sum_{i = 1}^n \Bigl|\frac{i}{n} - F(X_{(i)})\Bigr|^2\biggr)^{1/2}.
\end{align*}
Since $\mathcal{D}^*(\bs{u}) = O(\log n /n)$ and the second term can be bounded by using~\eqref{eq:KMTapplicationX}, we immediately get 
\[
\Bigl(\frac{1}{n}\sum_{i = 1}^n |u_{\pi_n(i)} - F(X_i)|^2\Bigr)^{1/2}  = O_p(n^{-1/2}),
\]
as desired. 
\end{proof}

Define 
\begin{align*}
\tilde{r}_{n, d, \theta} := \begin{cases}
    n^{-1}, & d = 1, \\
    r_{n, d}^{\frac{2\theta}{1+\theta}}  & d \geq 2,
\end{cases}
\end{align*}
where $c_{d, \alpha}$ is the same as in Proposition~\ref{prop:HolderOTConvergence}. We note that  Proposition~\ref{prop:HolderOTConvergence} and Proposition~\ref{prop:1DKMTMatching_random} can be combined to show that if $\mu \in \mathcal{P}_{t, \alpha}(\mathbb{R}^d)$ for some $t > 0$ and $\alpha \in (0, 2]$ is absolutely continuous with respect to the Lebesgue measure, and $T_0$, the optimal transport map from $\mu$ to $\mathrm{Unif}([0, 1]^d)$, is $(\theta, L)$-H\"older continuous for some $\theta\in (0, 1]$ and $L >0$, then 
\begin{equation}
    \label{Eq:OT_convergence}
\frac{1}{n}\sum_{i = 1}^n\|\hat T_n(X_i) - T_0(X_i)\|_2^2 = O_p(\tilde{r}_{n,d,\theta}\log^{\frac{2c_{d, \alpha} \theta}{1 + \theta}} n)
\end{equation}
where $c_{d, \alpha}$ is a constant depending only on $d$ and $\alpha$.  We will use this to establish the following key result, which shows that the vacancy computed from the empirical rank-transformed data can be approximated by the vacancy computed after applying the population optimal transport maps on the raw data.

\begin{prop}\label{prop: VonRawData}
Suppose $d_X, d_Y \geq 1$ are integers, write $d = d_X + d_Y$. Suppose $(X_1, Y_1), \ldots, (X_n, Y_n) \stackrel{\mathrm{iid}}{\sim} P^{(X, Y)}$, where $P^{(X, Y)}$ be a Borel probability measure on $\mathbb{R}^{d}$ with absolutely continuous marginals $P^X \in \mathcal{P}_{t_1, \alpha_1}(\mathbb{R}^{d_X})$ and $P^Y\in \mathcal{P}_{t_2, \alpha_2}(\mathbb{R}^{d_Y})$ for some $t_1, t_2 > 0$ and $\alpha_1, \alpha_2 \in (0, 2]$. Let $T_X$ and $T_Y$ be the optimal transport maps from $P^X$ and $P^Y$ to $\mathrm{Unif}([0, 1]^{d_X})$ and $\mathrm{Unif}([0, 1]^{d_Y})$, respectively, and we assume that $T_X$ is $(\theta_1, L_1)$-H\"older continuous and $T_Y$ is $(\theta_2, L_2)$-H\"older continuous for some $\theta_1, \theta_2 \in (0, 1]$ and $L_1, L_2 >0$. 

Given $U_1, \ldots, U_n \stackrel{\mathrm{iid}}{\sim} \mathrm{Unif}([0, 1]^{d_X})$ and $V_1, \ldots, V_n \stackrel{\mathrm{iid}}{\sim} \mathrm{Unif}([0, 1]^{d_Y})$, write $X_i^\natural  := T_X(X_i)$ and $Y_i^\natural  := T_Y(Y_i)$ for $i = 1, \ldots, n$, then when $(d_X, d_Y) \in \mathcal{D}(\theta_1, \theta_2)$, we have
\begin{align*}
    \Biggl|\mathcal{V}\biggl((X_i)_{i\in[n]}, (Y_i)_{i\in[n]}, &(U_i)_{i\in[n]},(V_i)_{i\in[n]}; \frac{1}{2n^{1/d}} \Bigr) \\
    &- \mathcal{V}\Bigl((X_i^\natural)_{i\in[n]}, (Y_i^\natural )_{i\in[n]}, ( X_i^\natural)_{i\in[n]}, ( Y_i^\natural )_{i\in[n]}; \frac{1}{2n^{1/d}}\biggr)\Biggr| \stackrel{\mathrm{p}}{\longrightarrow} 0,
\end{align*}
as $n \to \infty$.
\end{prop}
\begin{proof}
For notational simplicity, write $\bs{X}:=(X_i)_{i\in[n]}$, $\bs{Y}:=(Y_i)_{i\in[n]}$, ${\bs{X}^\natural}:=( X_i^\natural)_{i\in[n]}$, ${\bs Y^\natural} :=( Y_i^\natural)_{i\in[n]}$, $\bs{U}:=(U_i)_{i\in[n]}$, $\bs{V}:=(V_i)_{i\in[n]}$. Given reference points ${\bs U}$ and ${\bs V}$, let $R_i:= (R_i^X, R_i^Y)$ and $ R_i^\natural := ({R}_i^{X^\natural}, {R}_i^{Y^\natural})$ denote the joint rank of $\{(X_i, Y_i)\}_{i = 1}^n$ and $\{({X}_i^\natural, {Y}_i^\natural)\}_{i = 1}^n$, respectively, as defined in~\eqref{eq:R_iGen}. Also write $Z_i^\natural := ( X_i^\natural,  Y_i^\natural)$, for $i = 1, \ldots, n$. 

Define the coverage areas
\[
A_1 = \bigcup_{i = 1}^n B\Bigl(R_i, \frac{1}{2n^{1/d}}\Bigr), \quad {A}_2 = \bigcup_{i=1}^n B\Bigl({R}_i^\natural, \frac{1}{2n^{1/d}}\Bigr), \quad {A}_3 = \bigcup_{i=1}^n B\Bigl(Z_i^\natural, \frac{1}{2n^{1/d}}\Bigr).
\]
Note that  
\begin{align}
     &\Biggl|\mathcal{V}\biggl(\bs{X}, \bs{Y}, \bs{U}, \bs{V}; \frac{1}{2n^{1/d}} \Bigr) - \mathcal{V}\Bigl({\bs X^\natural}, {\bs Y^\natural}, {\bs X^\natural}, {\bs Y^\natural}; \frac{1}{2n^{1/d}}\biggr)\Biggr| \notag \\
     \leq & \biggl|\mathcal{V}\biggl(\bs{X}, \bs{Y}, \bs{U}, \bs{V}; \frac{1}{2n^{1/d}} \Bigr) - \mathcal{V}\Bigl({\bs X^\natural}, {\bs Y^\natural}, \bs{U}, \bs{V}; \frac{1}{2n^{1/d}}\biggr)\biggr| \notag \\ 
     &+ \biggl|\mathcal{V}\Bigl({\bs X^\natural}, {\bs Y^\natural}, \bs{U}, \bs{V}; \frac{1}{2n^{1/d}}\biggr) - \mathcal{V}\Bigl({\bs X^\natural}, {\bs Y^\natural}, {\bs X^\natural}, {\bs Y^\natural}; \frac{1}{2n^{1/d}}\biggr)\biggr|, \notag \\
     \leq & \mathrm{vol}(A_1\Delta {A}_2) + \mathrm{vol}(A_2\Delta {A}_3) \label{ineq: vacancydiff}.
\end{align}
Moreover, by Lemma \ref{lem: symdiff} we have
\begin{align}
\mathrm{vol}(A_1\Delta {A}_2) \leq \sum_{i = 1}^n \mathrm{vol}\biggl(B\Bigl(R_i, \frac{1}{2n^{1/d}}\Bigr) \Delta B\Bigl({R}_i^\natural, \frac{1}{2n^{1/d}}\Bigr)\biggr) 
\leq2 d n^{-\frac{d-1}{d}} \sum_{i = 1}^n  \|R_i - {R}_i^\natural\|_2 \label{ineq: coveragesymdiff_A12},  
\end{align}
and
\begin{align}
  \mathrm{vol}(A_2\Delta {A}_3) \leq \sum_{i = 1}^n \mathrm{vol}\biggl(B\Bigl( R_i^\natural, \frac{1}{2n^{1/d}}\Bigr) \Delta B\Bigl(Z_i^\natural, \frac{1}{2n^{1/d}}\Bigr)\biggr) 
\leq2 d n^{-\frac{d-1}{d}} \sum_{i = 1}^n  \| R_i^\natural - {Z}_i^\natural\|_2 \label{ineq: coveragesymdiff_A23}.
\end{align}
Now it remains to bound~\eqref{ineq: coveragesymdiff_A12} and~\eqref{ineq: coveragesymdiff_A23} separately. 

Set $ P_n^{X^\natural} := \frac{1}{n}\sum_{i = 1}^n \delta_{ X_i^\natural}$, $ P_n^{Y^\natural}:= \frac{1}{n}\sum_{i = 1}^n \delta_{ Y_i^\natural}$ and $ P_n^U := \frac{1}{n}\sum_{i = 1}^n \delta_{U_i}$, $P_n^V := \frac{1}{n}\sum_{i = 1}^n \delta_{V_i}$. 
Since the empirical rank matching is the optimal matching for the squared Euclidean cost, by the Cauchy--Schwarz inequality, we have
\begin{align*}
\frac{1}{n}\sum_{i = 1}^n \| X_i^\natural -  R_i^{X^\natural}\|_2
&\leq
\biggl(\frac{1}{n}\sum_{i = 1}^n \| X_i^\natural -  R_i^{X^\natural}\|_2^2\biggr)^{1/2}
= \mathcal{W}_2\bigl(P_n^{X^\natural}, P_n^U\bigr), \\
\frac{1}{n}\sum_{i = 1}^n \| Y_i^\natural -  R_i^{Y^\natural}\|_2
&\leq
\biggl(\frac{1}{n}\sum_{i = 1}^n \| Y_i^\natural -  R_i^{Y^\natural}\|_2^2\biggr)^{1/2}
= \mathcal{W}_2\bigl(P_n^{Y^\natural}, P_n^V\bigr).
\end{align*}
Since $X_i^\natural \stackrel{\mathrm{iid}}{\sim} \mathrm{Unif}([0,1]^{d_X})$ and $Y_i^\natural \stackrel{\mathrm{iid}}{\sim} \mathrm{Unif}([0,1]^{d_Y})$, the empirical $2$-Wasserstein convergence rate of \citet{fournier2015rate} yields
\begin{align}
\frac{1}{n}\sum_{i = 1}^n\|{R}_i^\natural - Z_i^\natural\|_2
&\leq \mathcal{W}_2\bigl(P_n^{X^\natural}, P_n^U\bigr) + \mathcal{W}_2\bigl(P_n^{Y^\natural}, P_n^V\bigr) \notag \\ 
&\leq \mathcal{W}_2\bigl(P_n^{X^\natural}, \mathrm{Unif}([0, 1]^{d_X})\bigr) + \mathcal{W}_2\bigl(\mathrm{Unif}([0, 1]^{d_X}), P_n^U\bigr) \notag \\ 
&\hspace{3cm}+\mathcal{W}_2\bigl(P_n^{Y^\natural}, \mathrm{Unif}([0, 1]^{d_Y})\bigr) + \mathcal{W}_2\bigl(\mathrm{Unif}([0, 1]^{d_Y}), P_n^V\bigr) \notag \\
&= O_p\bigl(n^{-\frac{1}{d_X \vee d_Y \vee 4}}\log n\bigr).\label{eq:EmpMap1}
\end{align}
Since $(d_X, d_Y)\in\mathcal{D}(\theta_1,\theta_2)$, we have $d=d_X+d_Y>4$. Hence $1/(d_X\vee d_Y\vee 4)>1/d$, and the right-hand side of~\eqref{eq:EmpMap1} is $o_p(n^{-1/d})$.

Write the empirical measures $P_n^X:= \frac{1}{n}\sum_{i = 1}^n \delta_{X_i}$ and $P_n^Y:= \frac{1}{n}\sum_{i = 1}^n \delta_{Y_i}$. To control~\eqref{ineq: coveragesymdiff_A12}, by Propositions~\ref{prop:HolderOTConvergence} and~\ref{prop:1DKMTMatching}, we have 
\begin{align}
    \frac{1}{n}\sum_{i = 1}^n \|R_i - Z_i^\natural \|_2 \leq & \biggl(\frac{1}{n}\sum_{i = 1}^n \|R_i^X -  X_i^\natural\|_2^2 \biggr)^{1/2}+  \biggl(\frac{1}{n}\sum_{i = 1}^n \|R_i^Y -  Y_i^\natural \|_2^2 \biggr)^{1/2} \notag \\ 
    = & O_p\bigl(\tilde r_{n, d_X, \theta_1}^{1/2}\log^{\gamma_1} n + \tilde r_{n, d_Y, \theta_2}^{1/2}\log^{\gamma_2} n\bigr) \label{eq:EmpMap2},
\end{align}
where $\gamma_1, \gamma_2>0$ are constant depending only on $\theta_1, \theta_2, \alpha_1, \alpha_2, d_X, d_Y$. The stochastic order in~\eqref{eq:EmpMap2} is $o_p(n^{-1/d})$ by a direct comparison of rates. 
Indeed, note that  
\begin{align}
    \tilde r_{n, d_X, \theta_1}^{1/2} = \begin{cases}
        n^{-\frac{1}{2}}, & d_X = 1 \\ 
        (n^{-\frac{1}{2}} \log n)^{\frac{\theta_1}{1+\theta_1}}, &d_X =2, 3, 4 \\ 
        n^{-\frac{2\theta_1}{d_X(1+\theta_1)}}, & d_X \geq 5 
    \end{cases}, \;
    \tilde r_{n, d_Y, \theta_2}^{1/2} = \begin{cases}
    n^{-\frac{1}{2}}, & d_Y = 1 \\ 
    (n^{-\frac{1}{2}} \log n)^{\frac{\theta_2}{1+\theta_2}}, & d_Y =2, 3, 4 \\ 
    n^{-\frac{2\theta_2}{d_Y(1+\theta_2)}}, & d_Y \geq 5 
\end{cases} \label{eq:RateCal}
\end{align}
and one can verify that when $(d_X, d_Y) \in \mathcal{D}(\theta_1, \theta_2)$, all the terms on the right-hand side of~\eqref{eq:RateCal} are $o(n^{-1/d})$.

Therefore, applying the triangle inequality and~\eqref{eq:EmpMap1}--\eqref{eq:RateCal}, we have
\begin{align}
    \frac{1}{n}\sum_{i = 1}^n  \|R_i - {R}_i^\natural\|_2  \leq \frac{1}{n}\sum_{i = 1}^n \|R_i - Z_i^\natural \|_2 + \frac{1}{n}\sum_{i = 1}^n \|R_i^\natural - Z_i^\natural \|_2 = o_p(n^{-1/d}). \label{eq:EmpMap3}
\end{align}

Consequently, taking~\eqref{eq:EmpMap1} and~\eqref{eq:EmpMap3} back to~\eqref{ineq: coveragesymdiff_A23} and~\eqref{ineq: coveragesymdiff_A12}, respectively, we obtain that $\mathrm{vol}(A_1\Delta {A}_2) = o_p(1)$ and $\mathrm{vol}(A_2 \Delta {A}_3) = o_p(1)$. The conclusion then follows from~\eqref{ineq: vacancydiff}.
\end{proof}

\subsection{Additional results used in the proof of Theorem~\ref{thm: CLT}}

The first proposition records the limiting behaviour of the conditional variance and conditional mean of the interior contribution. These two quantities provide the variance scale and the leading Gaussian fluctuation in the proof of Theorem~\ref{thm: CLT}.

\begin{prop}
\label{Prop:AsymptoticMeanVariance}
Let $M_n$ and $S_n$ be defined as in \eqref{eq:MS}, we have 
\begin{align*}
nS_n &\pto \beta_\lambda^2,\\
\sqrt{n} (M_n - \mathbb{E}(M_n)) &\dto \mathcal{N}(0,\alpha_\lambda^2),
\end{align*}
such that 
\begin{align}
    \alpha^2_\lambda &:= \frac{\lambda^{2d}}{e^2 2^d(\lambda+2)^d}\Bigl(e^{(\lambda/2+1)^{-d}} - 1\Bigr) -\frac{\lambda^{2d}}{e^2(\lambda+2)^{2d}}, \label{eq:alpha}\\
    % \beta^2_\lambda &:= \frac{\lambda^{2d}}{e^2 2^d(\lambda+2)^d} \int_{[0,1]^{2d}} \Bigl\{e^{(\lambda/2)^d \vol(\tilde C(w_1, w_2))} - e^{(\lambda/2+1)^{-d}}\Bigr\}\, dw_1\,dw_2,\label{eq:beta}
	    \beta^2_\lambda&:= \Big\{\frac{\lambda^d}{ (\lambda + 2)^de^{2}}C_d + \frac{\lambda^{2d}}{2^{d}(\lambda + 2)^de^2}(1 - e^{(\lambda/2 + 1)^{-d}})\Big\} + R_\lambda,\label{eq:beta}
	\end{align}
	where $C_d = \sum_{k\geq 1}\frac{2^d}{k!(k + 1)^d}$ and $R_\lambda=O(\lambda^{-1})$.
% where $\tilde C(w_1, w_2) := B(w_1, \lambda^{-1})\cap B(w_2, \lambda^{-1})\cap [-\lambda, \lambda]^d$.
\end{prop}
\begin{proof}
To prove that $nS_n\pto \beta_\lambda^2$ we show that $\E(nS_n)\to\beta_\lambda^2$, and $\var(nS_n) = o(1)$. Note that $(\mathcal{V}_{n,\ell}^{\mathrm{in}}:\ell\in[L])$ are conditionally independent given $(N_\ell:\ell\in[L])$. Therefore
    \begin{align*}
        S_n &= \var(\mathcal{V}_n^{\mathrm{in}}\mid N_1, \ldots, N_L) = \var\Big(\sum_{\ell = 1}^L \mathcal{V}_{n, \ell}^{\mathrm{in}}\mid N_1, \ldots, N_L\Big) = \sum_{\ell = 1}^L \var(\mathcal{V}_{n, \ell}^{\mathrm{in}}\mid N_\ell).
    \end{align*}
    First, note that when $X$ and $Y$ are independent, the optimal permutations $\pi^X$ and $\pi^Y$ are independent, and this implies that for $i \in [n]$ we have $R_i \iid \mathrm{Unif}[0,1]^d$. Also note that given $i\in\mathcal{I}_\ell$, $R_i$ follows a uniform distribution on $\mathcal{P}_\ell$. 

    Without loss of generality, we fix $\mathcal{Q}_\ell$ and introduce the following functions, which will be used throughout the remainder of the proof. For any $x_1, x_2\in[0, 1]^d$ define
    \[
	        C(x_1, x_2) := \frac{\vol(B(x_1, \gamma)\cap B(x_2, \gamma))}{\vol(B(x_1, \gamma))} = n \vol(B(x_1, \gamma)\cap B(x_2, \gamma)).
    \]
    For $x_1, x_2\in\mathcal{Q}_\ell$ and random variable $W\sim\mathrm{Unif}(\mathcal{P}_\ell)$ define
    \begin{align*}
        v(x_1) &:= \mathbb{P}\big(W\not\in B(x_1, \gamma)\big) = 1 - (\frac{\lambda}{2} + 1)^{-d},\\
        u(x_1, x_2) &:= \mathbb{P}\big(W\not\in B(x_1, \gamma)\cup B(x_2, \gamma)\big) = 1 - 2(\frac{\lambda}{2} + 1)^{-d} + C(x_1, x_2)(\frac{\lambda}{2} + 1)^{-d}.
    \end{align*}
    First, note that we have
    \begin{align*}
        \E[\mathcal{V}_{n, \ell}^{\mathrm{in}}\mid N_\ell] &= \E\Big[\int_{Q_\ell}\bone\{x\not\in \bigcup_{i\in\mathcal{I}_\ell} B(R_i, \gamma)\}dx\Bigm| N_\ell\Big] = \int_{\mQ_\ell} v(x)^{N_\ell} dx, \\
        \E[(\mathcal{V}_{n, \ell}^{\mathrm{in}})^2\mid N_\ell] &= \E\Bigl[ \int_{Q_\ell^2}\bone\{x_1, x_2\not\in \bigcup_{i\in\mathcal{I}_\ell} B(R_i, \gamma)\} dx_1dx_2\Bigm| N_\ell\Bigr] = \int_{\mQ_\ell^2} u(x_1, x_2)^{N_\ell}dx_1 dx_2.
    \end{align*}
    Since $N_\ell\sim\mathrm{Bin}(n, \vol(\mathcal{P}_\ell))$, for $a > 0$ constant, we have $\E a^{N_\ell} = (1 + \vol(\mathcal{P}_\ell)(a - 1))^n$. Using this equality, we have 
    \begin{align*}
        \E\bigl[&\var\bigl(\mathcal{V}_{n, \ell}^\mathrm{in}\mid N_\ell\bigr) \bigr] = \E\Bigl[ \E \bigl((\mathcal{V}_{n, \ell}^\mathrm{in})^2\mid N_\ell\bigr) - \E\bigl(\mathcal{V}_{n, \ell}^\mathrm{in} \mid N_\ell\bigr)^2 \Bigr] \notag \\
        &= \E\Big[\int_{\mQ_\ell^2} u(x_1, x_2)^{N_\ell} - \int_{\mQ_\ell^2}v(x_1)^{N_\ell}v(x_2)^{N_\ell}dx_1 dx_2\Big] \notag\\
        &= \int_{\mQ_\ell^2} \Big\{1 + \vol(\mathcal{P}_\ell)(u(x_1, x_2) - 1)\Big\}^n - \Big\{1 + \vol(\mathcal{P}_\ell)(v(x_1)v(x_2) - 1)\Big\}^n dx_1 dx_2\notag \\
        &= (1 + O(\frac{1}{n}))\int_{\mQ_\ell^2} \Big\{e^{(\frac{\lambda}{2} + 1)^d(u(x_1, x_2) - 1)} - e^{(\frac{\lambda}{2} + 1)^d(v(x_1)v(x_2) - 1)}\Big\}dx_1 dx_2 \notag \\
        &= (1 + O(\frac{1}{n}))e^{-2}\int_{\mQ_\ell^2} \Big\{\exp\big(C(x_1, x_2)\big) - \exp\big((\frac{\lambda}{2} + 1)^{-d}\big)\Big\}dx_1 dx_2 \notag \\
        &= (1 + O(\frac{1}{n}))\Big\{\frac{\lambda^d}{2^d e^{2}n^2}C_d + \frac{\lambda^{2d}}{2^{2d}e^2n^2}(1 - e^{(\lambda/2 + 1)^{-d}})\Big\} \notag\\
        &= \Big\{\frac{\lambda^d}{2^d e^{2}n^2}C_d + \frac{\lambda^{2d}}{2^{2d}e^2n^2}(1 - e^{(\lambda/2 + 1)^{-d}})\Big\} + O(\frac{\lambda^{2d}}{n^3} + \frac{\lambda^{d-1}}{2^d n^2}) 
    \end{align*}
    where in the last line we have used the following equality 
    \begin{align*}
        &\int_{\mQ_\ell^2}\exp(C(x_1, x_2)) dx_1 dx_2\\
        & \hspace{0.3cm}=   \vol(\mQ_\ell)^2 + \vol(\mQ_\ell)(2\gamma)^d \int_{[-1, 1]^d}(\exp(\prod_{i=1}^d\max\{(1 - \abs{u_i}), 0\}) - 1)du (1 + O(\lambda^{-1})),
    \end{align*}
    where
    \begin{align*}
       \int_{[-1, 1]^d}(\exp(\prod_{i=1}^d\max\{(1 - \abs{u_i}), 0\}) - 1)du = \sum_{k\geq 1}\frac{2^d}{k!(k + 1)^d} = C_d.
    \end{align*}
    Therefore
    \begin{align}\label{eq:ControlExp}
        \E[nS_n] &= n\E[\sum_{\ell = 1}^L \var(\mathcal{V}_{n, \ell}^{\mathrm{in}}\mid N_\ell)] \notag\\
        &= e^{-2}\Big\{C_d + \frac{\lambda^d}{2^d}(1 - e^{(\lambda/2 + 1)^{-d}})\Big\} + O(\lambda^{-1} + \frac{\lambda^d}{n})\notag\\
        &= e^{-2}(C_d - 1) + O(\lambda^{-1} + \lambda^{-d} + \frac{\lambda^d}{n}).
    \end{align}
    since $2^{-d}\lambda^d(1 - e^{(\lambda/2 + 1)^{-d}}) = -1 + O(\lambda^{-d})$.
    
    We then work out $\var(S_n)$.
    \begin{align*}
        \var(S_n) &= \var\Big(\sum_{\ell = 1}^L \var(\mathcal{V}_{n, \ell}^{\mathrm{in}}\mid N_\ell)\Big) = \sum_{\ell, k\in[L]} \cov\Big(\var(\mathcal{V}_{n, \ell}^{\mathrm{in}}\mid N_\ell), \var(\mathcal{V}_{n, k}^{\mathrm{in}}\mid N_k)\Big).
    \end{align*}
    Since 
    \begin{align*}
        \var(\mathcal{V}_{n, \ell}^{\mathrm{in}}\mid N_\ell) = \int_{\mQ_\ell^2} u(x_1, x_2)^{N_\ell} - v(x_1)^{N_\ell}v(x_2)^{N_\ell}dx_1 dx_2,
    \end{align*}
    applying Fubini's theorem, we have
    \begin{align*}
        \cov\Big(&\var(\mathcal{V}_{n, \ell}^{\mathrm{in}}\mid N_\ell), \var(\mathcal{V}_{n, k}^{\mathrm{in}}\mid N_k)\Big) \\
        &= \int_{\mQ_1^4}\Bigl\{ \cov(u(x_1, x_2)^{N_\ell}, u(x_3, x_4)^{N_k}) + \cov(v(x_1)^{N_\ell}v(x_2)^{N_\ell}, v(x_3)^{N_k}v(x_4)^{N_k}) \\
        &\hspace{2cm} - \cov(u(x_1, x_2)^{N_\ell}, v(x_3)^{N_k}v(x_4)^{N_k})\\
        &\hspace{2cm}- \cov(u(x_3, x_4)^{N_k}, v(x_1)^{N_\ell}v(x_2)^{N_\ell})\Bigr\} dx_1 dx_2 dx_3 dx_4 \\
        &= (\vol(\mQ_1))^4 O(n^{-2}) = O(n^{-6}),
    \end{align*}
where in the last line we have used Lemma~\ref{le: MultinCovariance} together with the following fact
\[
\max\{(u(x_1, x_2) - 1), (u(x_3, x_4) - 1), (v(x_1)v(x_2) - 1), (v(x_3)v(x_4) - 1)\} = O(1).
\]
As a result, we have
    \begin{align}\label{eq:ControlVariance}
        \var(nS_n) = O(n^{-2}).
    \end{align}
Combining~\eqref{eq:ControlExp} and~\eqref{eq:ControlVariance},  Markov's inequality implies that $nS_n \pto \beta_\lambda^2$

We now turn to proving that $\sqrt{n}(M_n - \E(M_n))\dto \mathcal{N}(0, \alpha_\lambda^2)$. For $w\in \{0, 1, \ldots, n\}$, let $f(w) := \E[\mathcal{V}_{n, 1}^{\mathrm{in}}\mid N_1 = w]$. Let $W\sim\mathrm{Poi}(n\vol(\mathcal{P}_1))$. Define
\[
\tau^2 := L\var(f(W)) - \frac{L^2}{n}\cov^2(W, f(W)).
\]
\citet[Theorem~1]{holst1972asymptotic} implies that as $n \to \infty$
\[
\frac{1}{\tau}(M_n - \E(M_n)) \dto \mathcal{N}(0, 1). 
\]
To finish the proof, it is therefore enough to show that $n\tau^2\to \alpha^2$. For $X\sim\mathrm{Unif}[\mQ_1]$ we have
\begin{align}
    n L \var \bigl(f(&W)\bigr) = n L \var(\int_{\mQ_1}v(x)^W dx)\notag\\
    &= n L \vol^2(\mathcal{Q}_1)\var\Bigl(\E\bigl(\{v(X)\}^{W}\mid W\bigr)\Bigr) \notag \\
    &= \frac{\lambda^{2d}}{(2\lambda + 4)^{d}}\var\Bigl(\bigl\{1 - (\lambda/2 + 1)^{-d}\bigr\}^{W}\Bigr) 
    = \frac{\lambda^{2d}}{(2\lambda + 4)^{d}} e^{-2} \Bigl( e^{(\lambda/2 + 1)^{-d}} - 1\Bigr),\label{eq:tauVar} 
\end{align}
where the final equality follows from the fact that $\E a^W = e^{n\vol(\mathcal{P}_1)(a - 1)}$ for any constant $a > 0$. Similarly we derive
\begin{align}
L &\mathrm{Cov}\bigl(W, f(W)\bigr)= L\vol(\mathcal{Q}_1)\biggl(\E \Bigl[W \E\bigl(v(X)^{W}\bigm|W\bigr)\Bigr] -  n\vol(\mathcal{P}_1)\E\Bigl[\E\bigl(v(X)^{W}\bigm|W\bigr)\Bigr]\biggr) \notag \\ 
&=\biggl(\frac{\lambda}{\lambda + 2}\biggr)^d \Bigl(\E\bigl\{W(1 - (\lambda/2 + 1)^{-d})^{W}\bigr\} - (\lambda/2 + 1)^d \E\bigl\{1 - (\lambda/2 + 1)^{-d}\bigr\}^W\Bigr) \notag \\
&= \biggl(\frac{\lambda}{\lambda + 2}\biggr)^d\Bigl(n\vol(\mathcal{P}_1)\bigl(1 - (\lambda/2 + 1)^{-d}\bigr)e^{-1} 
 - (\lambda/2 + 1)^d e^{-1}\Bigr) = -\biggl(\frac{\lambda}{\lambda + 2}\biggr)^d e^{-1} \label{eq:sigmanCov},
\end{align}
where we use the fact that  $\E (W a^W) =  n\vol(\mathcal{P}_1) a e^{n\vol(\mathcal{P}_1)(a - 1)}$ for constant $a > 0$ in the second to last equality. Combining (\ref{eq:tauVar}) and (\ref{eq:sigmanCov}) we have 
\[
    n \tau^2 = \frac{\lambda^{2d}}{2^d e^{2}(\lambda + 2)^{d}}\Bigl(e^{(\lambda/2 + 1)^{-d}} - 1\Bigr) - \frac{\lambda^{2d}}{e^{2}(\lambda + 2)^{2d}} = \alpha_\lambda^2,  
\]
which finishes the proof.
\end{proof}

The next proposition turns the conditional independence of the block contributions into a conditional central limit theorem. It is the step that converts the variance asymptotics from Proposition~\ref{Prop:AsymptoticMeanVariance} into a normal approximation for the interior coverage.

\begin{prop}
\label{Prop:ConditionalBerryEsseen}
    Let $M_n$ and $S_n$ be defined as in \eqref{eq:MS}, we have 
    \[  
    \sup_{t\in\mathbb{R}} \biggl|\mathbb{P}\biggl(\frac{\sqrt{n}(\mathcal{V}_n^{\mathrm{in}} - M_n)}{\sqrt{S_n}} \leq t\biggm| N_1,\ldots,N_L\biggr) - \Phi(t)\biggr| = O_p(n^{-1/2}),
    \]
    where $\Phi$ is the standard Gaussian distribution function.
\end{prop}

\begin{proof}
Since each point in $\mathcal{Q}_\ell$ lies at least $\gamma$ away from the boundary of $\mathcal{P}_\ell$, it follows that, conditional on $(N_\ell : \ell \in [L])$, the collections $(\mathcal{V}_{n, \ell}^{\mathrm{in}} : \ell \in [L])$ are independent. Then using Berry--Esseen Theorem \citep{berry1941accuracy, esseen1942liapunov} we have
\begin{align*}
    \sup_{t\in\mathbb{R}}\Big|\mathbb{P}\bigg(\frac{\sqrt{n}(\mathcal{V}_n^{\mathrm{in}} - M_n)}{\sqrt{S_n}}\leq t\bigg | N_1, \ldots, N_L\bigg) - \Phi(t)\Big| \leq C_n,
\end{align*}
where 
\begin{align*}
    C_n := C\frac{\sum_{\ell = 1}^L\E\big(\abs{\mathcal{V}_{n, \ell}^{\mathrm{in}} - \E(\mathcal{V}_{n, \ell}^{\mathrm{in}}\mid N_\ell)}^3\mid N_\ell\big)}{S_n^{3/2}},
\end{align*}
with $C$ a universal constant independent of $n$. Additionally note that for all $\ell\in[L]$
\begin{align*}
    \mathcal{V}_{n, \ell}^{\mathrm{in}} \leq \vol(\mathcal{P}_\ell) = \frac{(\lambda/2 + 1)^d}{n}.
\end{align*}
Therefore we have
\begin{align*}
    C_n\leq \frac{(\lambda/2 + 1)^{3d}}{n^2S_n^{3/2}}.
\end{align*}
Using Proposition~\ref{Prop:AsymptoticMeanVariance}, we have $S_n = O_p(n^{-1})$; thus it follows that $C_n = O_p(n^{-1/2})$ which completes the proof.
\end{proof}

It remains to show that the part of the vacancy lying outside the interior blocks is asymptotically negligible at the $\sqrt n$ scale. The following proposition gives the required variance bound for this exterior contribution.

\begin{prop}\label{Prop:VoutContribution}
    Let $\mathcal{V}_n^{\mathrm{out}}$ be defined as in \eqref{eq:VInOut}. We have 
    \[
    \mathcal{V}_n^{\mathrm{out}} - \mathbb{E}(\mathcal{V}_n^{\mathrm{out}}) = O_p(n^{-1/2}\lambda^{-1/2}).
    \]
\end{prop}

\begin{proof}
    First note that by Lemma~\ref{le: VacancyLimitRandomRef} we have
    \begin{align}\label{eq:VOutFirstMoment}
        \E(\mathcal{V}_n^{\mathrm{out}}) = \biggl\{1 - \biggl(\frac{\lambda}{\lambda +2 }\biggr)^d\biggr\}\biggl(1 - \frac{1}{n}\biggr)^n.
    \end{align}
    Take $Z_1, Z_2\iid\mathrm{Unif}\big[[0, 1]^d\setminus\bigcup_{\ell\in[L]}\mQ_\ell\big]$. Using Lemma~\ref{le: VacancyLimitRandomRef}, we have
    \begin{align}
        \E[(\mathcal{V}_n^{\mathrm{out}})^2] &=  \biggl\{ 1 - \biggl(\frac{\lambda}{\lambda +2 }\biggr)^d\biggr\}^2 \E\biggl\{1 - \frac{2}{n} + \vol(C(Z_1, Z_2))\biggr\}^n \notag \\ 
        & = \biggl\{ 1 - \biggl(\frac{\lambda}{\lambda + 2 }\biggr)^d\biggr\}^2  \sum_{k = 0}^n \binom{n}{k} \biggl(1 - \frac{2}{n}\biggr)^{n - k} \E(\vol(C(Z_1, Z_2)))^k \label{eq:VOutSecondMoment}.
    \end{align}
    Note that for $1\leq k \leq n$ we have
    \begin{align}
    \E(\vol(C(Z_1, &Z_2)))^k =  \biggl\{1 - \biggl(\frac{\lambda}{\lambda + 2}\biggr)^d\biggr\}^{-2} \int_{\big([0, 1]^d\setminus\bigcup_{\ell\in[L]}\mQ_\ell\big)^2}  \vol\bigl(C(z_1, z_2)\bigr)^k \, d z_1  d z_2 \notag \\ 
    & \leq \biggl\{1 - \biggl(\frac{\lambda}{\lambda + 2 }\biggr)^d\biggr\}^{-2}\int_{[0, 1]^d\setminus\bigcup_{\ell\in[L]}\mQ_\ell}\int_{[0,1]^d}  n^{-k}  \mathbbm{1}{\{z_1\in B^d(z_2,2\gamma)\}}\, d z_2 dz_1 \notag\\ 
    & = \biggl\{ 1 - \biggl(\frac{\lambda}{\lambda +2 }\biggr)^d\biggr\}^{-1}2^d n^{-(k + 1)}. \label{eq:volCmomentBound}
\end{align}
Therefore, putting together~\eqref{eq:VOutFirstMoment},\eqref{eq:VOutSecondMoment} and \eqref{eq:volCmomentBound} we have
\begin{align*}
    \var(\mathcal{V}_n^{\mathrm{out}}) &\leq \frac{2^d}{n} \biggl\{1 - \biggl(\frac{\lambda}{\lambda +2 }\biggr)^d\biggr\}\sum_{k = 1}^n\binom{n}{k}\biggl(1 - \frac{2}{n}\biggr)^{n - k}\frac{1}{n^{k}} \\
    &\hspace{4cm}+ \biggl\{1 - \biggl(\frac{\lambda}{\lambda +2 }\biggr)^d\biggr\}^2\bigg((1 - \frac{2}{n})^n - (1 - \frac{1}{n})^{2n}\bigg)\\
    &= \frac{2^d}{n} \biggl\{1 - \biggl(\frac{\lambda}{\lambda +2 }\biggr)^d\biggr\}\bigg(\big(1 - \frac{1}{n}\big)^n - \big(1 - \frac{2}{n}\big)^n\bigg) \\
    &\hspace{4cm}+ \biggl\{1 - \biggl(\frac{\lambda}{\lambda +2 }\biggr)^d\biggr\}^2\bigg(\big(1 - \frac{2}{n}\big)^n - \big(1 - \frac{1}{n}\big)^{2n}\bigg)\\
    &= \frac{2^d}{n} \biggl\{1 - \biggl(\frac{\lambda}{\lambda + 2 }\biggr)^d\biggr\}(e^{-1} - e^{-2}) -\biggl\{1 - \biggl(\frac{\lambda}{\lambda +2 }\biggr)^d\biggr\}^2 \frac{e^{-2}}{n} + o(\frac{1}{n}).
\end{align*}
Note that for large $\lambda$ we have
\begin{align*}
    1 - \biggl(\frac{\lambda}{\lambda + 2 }\biggr)^d = O(\frac{1}{\lambda}),
\end{align*}
therefore allowing $\lambda$ to grow to infinity as $n\rightarrow\infty$ we have
\[
\var(\mathcal{V}_n^{\mathrm{out}}) = O(\frac{1}{n\lambda}).
\]
This completes the proof. 
\end{proof}

\subsection{Additional results used in the proof of Theorem~\ref{Thm:PopulationLimitGrid}}

The first preliminary result controls the contribution to the coverage from a rectangular region in the domain where the density is constant.

\begin{prop}
\label{Prop:PiecewiseConstantFixed}
Suppose $P^{(X,Y)}$ is a probability measure on $[0,1]^2$ with $\mathrm{Unif}[0,1]$ marginals. Suppose that for some $0\leq a_1 < a_2\leq 1$ and $0\leq b_1 < b_2\leq 1$, $P^{(X,Y)}$ is equal to $q \cdot \mathrm{vol}$ when restricted to $[a_1,a_2]\times [b_1,b_2]$, where $\mathrm{vol}$ denotes the Lebesgue measure. Let $(X_1,Y_1),\ldots,(X_n,Y_n)\iid P^{(X,Y)}$ and let $R_i$ be defined as in~\eqref{eq:R_iGen} with respect to reference points $\boldsymbol{U} = \boldsymbol{V} = (1/n, \ldots, (n-1)/n, 1)$. Then
\[
   \vol\biggl( \bigcup_{i: (X_i,Y_i)\in(a_1,a_2]\times (b_1,b_2]} B\Bigl(R_i, \frac{1}{2\sqrt{n}}\Bigr)\biggr) = e^{-q}(a_2-a_1)(b_2-b_1)(1+o_p(1)).
\]
\end{prop}
\begin{proof}
Define $\mathcal{I} := \{i\in[n]: a_1 < X_i \leq a_2\}$, $\mathcal{I}_{-}:=\{i\in[n]:X_i \leq a_1\}$, $\mathcal{J} := \{j\in[n]: b_1 < Y_j\leq b_2\}$ and $\mathcal{J}_-:=\{j\in[n]:Y_j\leq b_1\}$. We write $S_0:= |\mathcal{I}_-|$, $S_1:=|\mathcal{I}_-\cup \mathcal{I}|$, $T_0:=|\mathcal{J}_-|$ and $T_1:=|\mathcal{J}_-\cup\mathcal{J}|$. Also, let $M := |\mathcal{I}\cap \mathcal{J}|$. Let 
\begin{align*}
    \mathcal{R}_X= \{(S_{0} + 1)/n, \ldots, S_1/n\} \quad \text{and} \quad  \mathcal{R}_Y= \{(T_{0} + 1)/n, \ldots, T_1/n \}.
\end{align*}  
For an arbitrary finite set $\mathcal{S}$ and integer $k > 0$, let $\mathcal{S}^{(k)}:=\{A\subseteq\mathcal{S}\mid \abs{A} = k\}$, e.g. the set of all subsets of size $k$ of $\mathcal{S}$. Define $\mathrm{Unif}(\mathcal{S}^{(k)})$ to be the uniform distribution over all subsets of size $k$ of $\mathcal{S}$. Let $\mathcal{F}$ be the $\sigma$-algebra generated by $S_0, S_1, T_0, T_1,$ and $M$. Then using \eqref{eq:BlockUniformXY} we have
\begin{align*}
\{R^X_i:i\in\mathcal{I} \cap \mathcal{J}\}\mid \mathcal{F} \sim \mathrm{Unif}(\mathcal{R}_X^{(M)}) \quad \text{and}\quad \{R^Y_j:j\in\mathcal{I} \cap \mathcal{J}\} \mid \mathcal{F}\sim \mathrm{Unif}(\mathcal{R}_Y^{(M)}).
\end{align*}
which are independent of each other, which gives us
\begin{align*}
    (R_i: i\in\mathcal{I}\cap\mathcal{J})\mid \mathcal{F}\sim \mathrm{Unif}(\mathcal{R}_X^{(M)})\otimes \mathrm{Unif}(\mathcal{R}_Y^{(M)}).
\end{align*}
This means that any matchings between any subset of size $M$ of $\mathcal{R}_X$ and any subset of size $M$ of $\mathcal{R}_Y$ are equally likely given $\mathcal{F}$.

By law of large numbers, there is an event $\Omega$ with probability 1 on which we have $M / n \to q(a_2-a_1)(b_2-b_1)$, $S_0 / n\to a_1$, $S_1/n\to a_2$, $T_0 / n\to b_1$, $T_1/n\to b_2$.  We will work on this event henceforth. 

As $n\to\infty$, the contribution of the covered area by points near the boundary of any rectangle is negligible (so we may ignore the periodic boundary condition), hence Lemma~\ref{le: VacancyLimit} and a linear rescaling, conditional on $\mathcal{F}$, we have 
\begin{align*}
\mathbb{E}\biggl\{\vol\biggl(\bigcup_{i\in\mathcal{I}\cap\mathcal{J}}B\Bigl(R_i, \frac{1}{2\sqrt{n}}\Bigr)\biggr) \biggm | \mathcal{F}\biggr\} &\to e^{-q}(a_2-a_1)(b_2-b_1).
\end{align*}
Then using Lemma~\ref{lem:VacancyDifference}
\begin{align*}
    &\var\biggl(\vol\biggl(\bigcup_{i\in\mathcal{I}\cap\mathcal{J}}B\Bigl(R_i, \frac{1}{2\sqrt{n}}\Bigr)\biggr) \biggm | \mathcal{F}\biggr) \\
    &= \E\biggl\{\bigg(\vol\biggl(\bigcup_{i\in\mathcal{I}\cap\mathcal{J}}B\Bigl(R_i, \frac{1}{2\sqrt{n}}\Bigr)\biggr) - \E\bigl\{\vol\biggl(\bigcup_{i\in\mathcal{I}\cap\mathcal{J}}B\Bigl(R_i, \frac{1}{2\sqrt{n}}\Bigr)\biggr)\bigm | \mathcal{F}\bigr\}\bigg)^2 \biggm | \mathcal{F}\biggr\} \\
    &= \int_0^\infty \mathbb{P}\biggl\{\bigg(\vol\biggl(\bigcup_{i\in\mathcal{I}\cap\mathcal{J}}B\Bigl(R_i, \frac{1}{2\sqrt{n}}\Bigr)\biggr) - \E\bigl\{\vol\biggl(\bigcup_{i\in\mathcal{I}\cap\mathcal{J}}B\Bigl(R_i, \frac{1}{2\sqrt{n}}\Bigr)\biggr)\bigm | \mathcal{F}\bigr\}\bigg)^2 \geq t \biggm | \mathcal{F}\biggr\} dt \\
    &\to 0.
\end{align*}
Using Lemma~\ref{lem:VacancyDifference} and McDiarmid inequality we have
\begin{align*}
\mathbb{P}\biggl\{\bigg(\vol\biggl(\bigcup_{i\in\mathcal{I}\cap\mathcal{J}}B\Bigl(R_i, \frac{1}{2\sqrt{n}}\Bigr)\biggr) - \E\bigl\{\vol\biggl(\bigcup_{i\in\mathcal{I}\cap\mathcal{J}}B\Bigl(R_i, \frac{1}{2\sqrt{n}}\Bigr)\biggr)\bigm | \mathcal{F}&\bigr\}\bigg)^2 \geq t  \biggm | \mathcal{F}\biggr\} \\
&\leq 2\exp(-Cnt),
\end{align*}
which gives
\[
\var\biggl(\vol\biggl(\bigcup_{i\in\mathcal{I}\cap\mathcal{J}}B\Bigl(R_i, \frac{1}{2\sqrt{n}}\Bigr)\biggr) \biggm | \mathcal{F}\biggr)\to 0.
\]
By the Dominated Convergence Theorem, the same result holds unconditionally, which implies the desired result by an application of Chebyshev's inequality.
\end{proof}

The next preliminary result shows that the coverage correlation coefficients of samples generated from two probability measures close in total variation distance are (stochastically) close to each other. 
\begin{prop}
\label{Prop:TVPerturbationFixedGrid}
Let $\mu$ and $\nu$ be two probability measures on $\mathbb{R}^2$ with $d_{\mathrm{TV}}(\mu, \nu)\leq \epsilon$. Let $\bs{U} = \bs{V} = (1/n,\ldots, (n-1)/n, 1)$. Then there is a coupling between $\mu$ and $\nu$ such that $(X_i,Y_i)_{i\in[n]}\iid \mu$, $(\tilde X_i,\tilde Y_i)\iid \nu$
\begin{align*}
    \mathbb{P}\Big(\biggl|\mathcal{V}\Bigl((X_i)_{i\in[n]}, (Y_i)_{i\in[n]}, \bs{U}, \bs{V}; \frac{1}{2\sqrt{n}}\Bigr) - \mathcal{V}\Bigl((\tilde X_i)_{i\in[n]}, (\tilde Y_i)_{i\in[n]}, \bs{U},\bs{V}; \frac{1}{2\sqrt{n}}\Bigr)\biggr| &\geq 12\epsilon\sqrt{n}\Big) \\
    &\leq e^{-n\epsilon/3}.
\end{align*}
\end{prop}
\begin{proof}
     Since $d_{\mathrm{TV}}(\mu, \nu)\leq \epsilon$ we can construct a coupling between $\mu$ and $\nu$ such that for two i.i.d. samples $\{(X_i, Y_i)\}_{i = 1}^n$ and $\{(\tilde{X}_i, \tilde{Y}_i)\}_{i = 1}^n$ and for each $i\in[n]$ we have (e.g. Theorem 5.2 in \citet{lindvall2002lectures})
     \[
     \mathbb{P}\big((X_i, Y_i)\neq (\tilde{X}_i, \tilde{Y}_i)\big) = d_{\mathrm{TV}}(\mu, \nu)/2\leq \epsilon/2.
     \]
     Let $\mathcal{I}$ be the set of all indices $i\in[n]$ where $(X_i, Y_i)$ and $(\tilde{X}_i, \tilde{Y}_i)$ are different and let $K = \abs{\mathcal{I}}$. The variable $K$ follows a binomial distribution $B(n, p)$ with $p\leq \epsilon/2$ which gives us $\E(K)\leq n\epsilon/2$. For any $i$ we have 
     \[
     R_i^X = n^{-1}\sum_{j = 1}^n\bone\{X_j\leq X_i\}, \qquad \tilde{R}_i^X = n^{-1}\sum_{j = 1}^n\bone\{\tilde{X}_j\leq \tilde{X}_i\}.
     \]
     For $i\not\in\mathcal{I}$, since $X_i = \tilde{X}_i$ and there are at most $K$ indices $j$ such that $\bone\{\tilde{X}_j\leq \tilde{X}_i\} \neq \bone\{X_j\leq \tilde{X}_i\}$ we have $\abs{R_i^X - \tilde{R}_i^X} \leq K/n$. Using the same argument, we have $\abs{R_i^Y - \tilde{R}_i^Y} \leq K/n$. Therefore $\|R_i - \tilde{R}_i\|_\infty\leq K/n$. Note that for $i\in\mathcal{I}$ we cannot provide any non-trivial bound. Let 
     \[
     S = \cup_{i = 1}^n (R_i + B), \qquad \tilde{S} = \cup_{i = 1}^n (\tilde{R}_i + B).
     \]
     Then note that $\abs{\vol(S) - \vol(\tilde{S})} \leq \vol(S\Delta\tilde{S})$. Also we have
     \begin{align*}
         \vol(S\Delta\tilde{S}) &\leq \vol\Big(\big(\cup_{i\not\in\mathcal{I}}(R_i + B)\big)\Delta\big(\cup_{i\not\in\mathcal{I}}(\tilde{R}_i + B)\big)\Big) \\
         &\hspace{4cm}+ \vol\Big(\cup_{i \in\mathcal{I}}(R_i + B)\Big) + \vol\Big(\cup_{i \in\mathcal{I}}(\tilde{R}_i + B)\Big).
     \end{align*}
     where 
     \begin{align*}
         \vol\Big(\cup_{i \in\mathcal{I}}(R_i + B)\Big) + \vol\Big(\cup_{i \in\mathcal{I}}(\tilde{R}_i + B)\Big) \leq \frac{2K}{n}.
     \end{align*}
    Then for the set of $i\not\in\mathcal{I}$ we have
    \begin{align*}
        \vol\Big(\big(\cup_{i\not\in\mathcal{I}}(R_i + B)\big)\Delta\big(\cup_{i\not\in\mathcal{I}}(\tilde{R}_i + B)\big)\Big) &\leq \sum_{i\not\in\mathcal{I}}\vol\Big((R_i + B)\Delta(\tilde{R}_i + B)\Big) \\
        &\leq \frac{4K(n - K)}{n\sqrt{n}}. 
    \end{align*}
    Putting these together, we get 
    \begin{align*}
        \abs{\vol(S) - \vol(\tilde{S})} \leq \frac{4K(n - K)}{n\sqrt{n}} + \frac{2K}{n} \leq \frac{6K}{\sqrt{n}}.
    \end{align*} 
    Therefore, using the Chernoff bound, we have
    \begin{align*}
        \mathbb{P}(\abs{\vol(S) - \vol(\tilde{S})} \geq 6(1 + \delta)\epsilon\sqrt{n}) \leq \mathbb{P}(K\geq (1 + \delta)n\epsilon) \leq \exp\big(-\frac{n\epsilon\delta^2}{3}\big).
    \end{align*}
    Therefore, by setting $\delta = 1$, we get the desired result.
\end{proof}

\subsection{Additional results used in the proof of Theorem~\ref{thm:RegularGridCLT}}

Throughout this subsection, we use the permutation representation and notation introduced in the proof of Theorem~\ref{thm:RegularGridCLT}: $\pi$, $X_{ij}$, $X_M$, $I_n(z)$, $J_n(z)$, $\mathcal M_{n,m}(z)$, $V_{n,m}(z)$, $U_n^{(L)}(z)$, $\mathscr C_n^{(L)}$ and $\Delta_n^{(L)}$. Also, $\mathrm{cum}(X_1,\ldots,X_r)$ denotes the joint cumulant, and $\mathrm{cum}_r(X):=\mathrm{cum}(X,\ldots,X)$. Now we are ready to introduce the main ingredient of the proof, namely a bound on joint cumulants. The following proposition is a direct consequence of the weighted dependency graph
for permutation-matrix entries~\citep{Feray2018WeightedDG}, its lift to bounded-degree monomials, and the
maximum-spanning-tree reordering lemma; see~\citet{Feray2018WeightedDG}, in particular
Definition~4.5, Lemma~3.3, Proposition~5.11, and Proposition~8.1 therein.

\begin{prop}\label{prop:imported}
Fix integers $r\ge2$ and $m\ge1$. Then there exists a constant $C_{r,m}<\infty$
such that for every $n\ge1$ and every choice of matchings
$M_1,\dots,M_r\subset\{1,\dots,n\}^2$ with $|M_i|\le m$,
\begin{equation}\label{eq:imported}
\bigl|\mathrm{cum}(X_{M_1},\dots,X_{M_r})\bigr|
\le
C_{r,m}\,
n^{-\abs{M_1\cup\cdots\cup M_r}}
\sum_{\sigma\in S_r}
\prod_{j=1}^{r-1}
W_n\bigl(M_{\sigma(j+1)};\,M_{\sigma(1)},\dots,M_{\sigma(j)}\bigr),
\end{equation}
where $S_r$ denotes the symmetric group on $\{1,\dots,r\}$.
\end{prop}

\begin{proof}
For $a=(i,j)\in \{1, \ldots, n\}^2$ write
\[
Y_a:=X_{ij}=\bone\{\pi(i)=j\}.
\]
We use the weighted dependency graph for permutation-matrix entries from \citet{Feray2018WeightedDG}. More precisely, by \citet[Proposition 8.1]{Feray2018WeightedDG}, the family $(Y_a)_{a\in \{1, \ldots, n\}^2}$ admits a weighted dependency graph $G_n$ with vertex set $\{1, \ldots, n\}^2$, edge weights
\[
w_n(a,a')=
\begin{cases}
1, & \text{if }a=(i,j),\ a'=(i',j')\text{ satisfy } i=i' \text{ or } j=j',\\[1mm]
n^{-1}, & \text{otherwise},
\end{cases}
\]
and with associated function
\[
\Psi_n(B)=n^{-\abs{B}}
\]
for every finite multiset $B$ of vertices, where $\abs{B}$ denotes the number of distinct elements of
$B$. The constants in the corresponding cumulant bound depend only on the order of the cumulant.

Fix integers $r\ge 2$ and $m\ge 1$, and let $M_1,\dots,M_r\subset \{1, \ldots, n\}^2$ be matchings with
$|M_i|\le m$. For a matching $M$, recall that
\[
X_M=\prod_{a\in M}Y_a.
\]

By~\citet[Proposition 5.11]{Feray2018WeightedDG}, applied to the family of monomials
\[
X_M = \prod_{a\in M}Y_a,\qquad \abs{M}\le m,
\]
the weighted dependency graph $G_n$ lifts to a weighted dependency graph for the
bounded-degree monomials $(X_M)_{\abs{M}\le m}$. In this lifted graph, the edge weight between
$X_M$ and $X_{M'}$ is
\[
\max_{a\in M,\,a'\in M'} w_n(a,a').
\]
Since the base weights $w_n(a,a')$ take only the values $1$ and $n^{-1}$, this lifted edge
weight is exactly
\[
W_n(M;M')
=
\begin{cases}
1, & \text{if } M \text{ shares a row or a column with } M',\\
n^{-1}, & \text{otherwise.}
\end{cases}
\]

The corresponding $\Psi$-factor is obtained by applying the original $\Psi_n$ to the multiset
of base vertices appearing in the product
\[
X_{M_1}\cdots X_{M_r}.
\]
Since $\Psi_n(B)=n^{-\abs{B}}$, where $\abs{B}$ denotes the
number of distinct elements in the multiset $B$, this gives
\[
\Psi_n(M_1\uplus\cdots\uplus M_r)
=
n^{-\abs{M_1\cup\cdots\cup M_r}}.
\]

More generally, for a family $M_1,\dots,M_\ell$,
\[
W_n(M;M_1,\dots,M_\ell)=
\begin{cases}
1, & \text{if }M\text{ shares a row/column with some }M_j,\\
n^{-1}, & \text{otherwise}.
\end{cases}
\]

Applying the defining cumulant bound for weighted dependency graphs to the lifted family, we obtain
\[
\bigl|\mathrm{cum}(X_{M_1},\dots,X_{M_r})\bigr|
\le C'_{r,m}\,
n^{-|M_1\cup\cdots\cup M_r|}
\,\mathcal{M}\bigl(\widetilde G_n[\{M_1,\dots,M_r\}]\bigr),
\]
where $\widetilde G_n$ denotes the lifted weighted graph and
$\mathcal{M}(\widetilde G_n[\{M_1,\dots,M_r\}])$ is the maximal weight of a spanning tree of the
induced weighted graph on the vertices $M_1,\dots,M_r$.

Finally, by the maximum-spanning-tree reordering lemma
\citep[Definition~4.5 and Lemma~3.3]{Feray2018WeightedDG}, for every weighted graph on
$\{1,\dots,r\}$ its maximal spanning-tree weight is bounded by
\[
\sum_{\sigma\in S_r}\prod_{j=1}^{r-1}
W_n\!\bigl(M_{\sigma(j+1)};M_{\sigma(1)},\dots,M_{\sigma(j)}\bigr).
\]
This gives
\[
\bigl|\mathrm{cum}(X_{M_1},\dots,X_{M_r})\bigr|
\le
C_{r,m}\,
n^{-|M_1\cup\cdots\cup M_r|}
\sum_{\sigma\in S_r}
\prod_{j=1}^{r-1}
W_n\!\bigl(M_{\sigma(j+1)};M_{\sigma(1)},\dots,M_{\sigma(j)}\bigr),
\]
which is exactly \eqref{eq:imported}.
\end{proof}

For $r\ge2$, let $\mathcal T_r$ denote the set of all spanning trees on the vertex set $\{1,\dots,r\}$. For $T\in\mathcal T_r$, we write $E(T)$ for its edge set. For $z=(x,y)$ and $z'=(x',y')$, define
\begin{equation}\label{eq:psi}
\begin{aligned}
\psi_n(z,z')
:=\; \bone\Bigl\{d_{\mathbb{T}}(x, x') \le n^{-1/2},\ d_{\mathbb{T}}(y, y') \le n^{-1/2}\Bigr\} &+ n^{-1/2}\bone\Bigl\{d_{\mathbb{T}}(x, x') \le n^{-1/2}\Bigr\} \\
&\hspace{-0.5cm}+ n^{-1/2}\bone\Bigl\{d_{\mathbb{T}}(y, y') \le n^{-1/2}\Bigr\}+ n^{-1}.
\end{aligned}
\end{equation}
The following proposition is the key locality estimate: fixed-degree cumulants of the fields $V_{n, m}(z)$ are bounded by a spanning-tree expression built from the interaction kernel $\psi_n$.
\begin{prop}
\label{prop:Vm-fixed-degrees}
Fix $r \ge 2$ and integers $m_1,\dots,m_r \ge 1$. Then there exists a constant
$C_{r, \mathbf{m}}<\infty$ depending only on $r$ and $\mathbf{m} := (m_1,\ldots, m_r)$ such that for all $z_1,\dots,z_r\in[0,1]^2$,
\begin{equation}\label{eq:8}
\left|
\mathrm{cum}\!\left(V_{n,m_1}(z_1),\dots,V_{n,m_r}(z_r)\right)
\right|
\le
C_{r, \mathbf{m}}
\sum_{T\in\mathcal T_r}
\prod_{\{a,b\}\in E(T)} \psi_n(z_a,z_b).
\end{equation}
\end{prop}

\begin{proof}
By multilinearity,
\[
\mathrm{cum}\!\left(V_{n,m_1}(z_1),\dots,V_{n,m_r}(z_r)\right)
=
\sum_{M_1\in\mathcal M_{n,m_1}(z_1)}
\cdots
\sum_{M_r\in\mathcal M_{n,m_r}(z_r)}
\mathrm{cum}(X_{M_1},\dots,X_{M_r}).
\]
Let $m:=\max(m_1,\dots,m_r)$. By Proposition~\ref{prop:imported},
\[ 
\left|
\mathrm{cum}\!\left(V_{n,m_1}(z_1),\dots,V_{n,m_r}(z_r)\right)
\right|
\le
C_{r,m}\sum_{\sigma\in S_r}\Sigma(\sigma),
\]
where
\[
\Sigma(\sigma):=
\sum_{M_1\in\mathcal M_{n,m_1}(z_1)}
\cdots
\sum_{M_r\in\mathcal M_{n,m_r}(z_r)}
n^{-\abs{M_1\cup\cdots\cup M_r}}
\prod_{j=1}^{r-1}
W_n\!\bigl(M_{\sigma(j+1)};M_{\sigma(1)},\dots,M_{\sigma(j)}\bigr).
\]

Fix $\sigma\in S_r$ and write $\beta_j:=\sigma(j)$. Let
\[
E_j:=M_{\beta_1}\cup\cdots\cup M_{\beta_j}.
\]
Since
\[
\abs{M_1\cup\cdots\cup M_r}
=
m_{\beta_1}
+
\sum_{j=1}^{r-1}\abs{M_{\beta_{j+1}}\setminus E_j},
\]
we obtain
\[
\Sigma(\sigma)
=
\sum_{M_{\beta_1}\in\mathcal M_{n,m_{\beta_1}}(z_{\beta_1})}\!\!
\cdots\!\!
\sum_{M_{\beta_r}\in\mathcal M_{n,m_{\beta_r}}(z_{\beta_r})}
n^{-m_{\beta_1}}
\prod_{j=1}^{r-1}
\Bigl[
W_n\!\bigl(M_{\beta_{j+1}};M_{\beta_1},\dots,M_{\beta_j}\bigr)
\,n^{-\abs{M_{\beta_{j+1}}\setminus E_j}}
\Bigr].
\]

The first factor is bounded using Lemma~\ref{lem:matchingcount}:
\[
\sum_{M_{\beta_1}\in\mathcal M_{n,m_{\beta_1}}(z_{\beta_1})}
n^{-m_{\beta_1}}
=
|\mathcal M_{n,m_{\beta_1}}(z_{\beta_1})|\,n^{-m_{\beta_1}}
\le
\frac{K^{m_{\beta_1}}}{m_{\beta_1}!}.
\]
For each $j=1,\dots,r-1$, Lemma~\ref{lem:onestep} gives
\[
\sum_{M_{\beta_{j+1}}\in\mathcal M_{n,m_{\beta_{j+1}}}(z_{\beta_{j+1}})}
W_n\!\bigl(M_{\beta_{j+1}};M_{\beta_1},\dots,M_{\beta_j}\bigr)
\,n^{-\abs{M_{\beta_{j+1}}\setminus E_j}}
\le
C_{r;m_1,\dots,m_r}^{(1)}
\sum_{t=1}^{j}\psi_n(z_{\beta_{j+1}},z_{\beta_t}),
\]
where $C_{r;m_1,\dots,m_r}^{(1)}<\infty$ depends only on $r,m_1,\dots,m_r$.

Therefore
\[
\Sigma(\sigma)
\le
C_{r;m_1,\dots,m_r}^{(2)}
\prod_{j=1}^{r-1}
\left(
\sum_{t=1}^{j}\psi_n(z_{\beta_{j+1}},z_{\beta_t})
\right).
\]

Expanding the product, for each $j=1,\dots,r-1$ one chooses an index
$p(j+1)\in\{1,\dots,j\}$. This produces a rooted increasing tree on the ordered vertex set
$(\beta_1,\dots,\beta_r)$, and hence
\[
\prod_{j=1}^{r-1}
\left(
\sum_{t=1}^{j}\psi_n(z_{\beta_{j+1}},z_{\beta_t})
\right)
\le
\sum_{T\in\mathcal T_r}
\prod_{\{a,b\}\in E(T)}\psi_n(z_a,z_b).
\]
Thus
\[
\Sigma(\sigma)
\le
C_{r;m_1,\dots,m_r}^{(2)}
\sum_{T\in\mathcal T_r}
\prod_{\{a,b\}\in E(T)}\psi_n(z_a,z_b).
\]

Finally, summing over $\sigma\in S_r$ and absorbing the factor $|S_r|=r!$ into the constant,
we obtain
\[
\left|
\mathrm{cum}\!\left(V_{n,m_1}(z_1),\dots,V_{n,m_r}(z_r)\right)
\right|
\le
C_{r, \mathbf{m}}
\sum_{T\in\mathcal T_r}
\prod_{\{a,b\}\in E(T)}\psi_n(z_a,z_b),
\]
which proves \eqref{eq:8}.
\end{proof}

In what follows, we work with the truncated versions of $U_n$ and $\mathscr{C}_n$. The need for truncation comes from the fact that the imported cumulant estimate in Proposition~\ref{prop:imported} is uniform in $n$ only for a \emph{fixed} degree bound $m$. Since the corresponding constants may depend arbitrarily on $m$, we cannot directly sum the inclusion--exclusion expansion of $U_n(z)$ over all $m$. To overcome this difficulty, we truncate at a level $L$, thereby restricting attention to finitely many degrees, for which the cumulant argument is valid, and then treat the remaining tail separately by a covariance estimate.

The following proposition is the bridge from the matching-count fields to the truncated area: it shows that the cumulants of $U_n^{(L)}$ satisfy the same tree bound as in Proposition~\ref{prop:Vm-fixed-degrees}, which can then be integrated over space.
\begin{prop}
\label{prop:truncated-field-cumulant}
Fix integers $r \ge 2$ and $L \ge 1$. Then there exists a constant $C_{r,L}<\infty$ such that
for all $z_1,\dots,z_r\in[0,1]^2$,
\begin{equation}
\left|
\mathrm{cum}\!\left(U_n^{(L)}(z_1),\dots,U_n^{(L)}(z_r)\right)
\right|
\le
C_{r,L}
\sum_{T\in\mathcal T_r}
\prod_{\{a,b\}\in E(T)} \psi_n(z_a,z_b).
\label{eq:11}
\end{equation}
\end{prop}

\begin{proof}
By definition,
\[
U_n^{(L)}(z)=\sum_{m = 1}^{L}(-1)^{m + 1}V_{n,m}(z),
\]
so multilinearity gives
\[
\mathrm{cum}\!\left(U_n^{(L)}(z_1),\dots,U_n^{(L)}(z_r)\right)
=
\sum_{m_1,\dots,m_r=1}^{L}
(-1)^{m_1+\cdots+m_r+r}
\mathrm{cum}\!\left(V_{n,m_1}(z_1),\dots,V_{n,m_r}(z_r)\right).
\]
For each fixed $r$ and each fixed $r$-tuple $\mathbf{m}:=(m_1,\dots,m_r)\in\{1,\dots,L\}^r$, Proposition~\ref{prop:Vm-fixed-degrees} yields
\[
\left|
\mathrm{cum}\!\left(V_{n,m_1}(z_1),\dots,V_{n,m_r}(z_r)\right)
\right|
\le
C_{r,\mathbf{m}}
\sum_{T\in\mathcal T_r}
\prod_{\{a,b\}\in E(T)} \psi_n(z_a,z_b).
\]
Since there are only finitely many such tuples, summing over them proves \eqref{eq:11}. 
\end{proof}

Proposition~\ref{prop:truncated-area-cumulants} shows that the tree bound from Proposition~\ref{prop:truncated-field-cumulant} integrates to the correct cumulant decay for $\mathscr{C}_n^{(L)}$.
\begin{prop}
\label{prop:truncated-area-cumulants}
Fix integers $r\ge2$ and $L\ge1$. Then
\begin{equation}
\mathrm{cum}_r(\mathscr{C}_n^{(L)}) = O_{r,L}\!\left(n^{-(r-1)}\right).
\label{eq:13}
\end{equation}
Consequently, if $(n\var(\mathscr{C}_n^{(L)}))_{n\ge1}$ converges along a subsequence to some $v\ge0$, then
along the same subsequence
\[
\sqrt n \bigl(\mathscr{C}_n^{(L)}-\mathbb{E} \mathscr{C}_n^{(L)}\bigr)
\;\xrightarrow{d}\;
\mathcal{N}(0,v).
\]
\end{prop}

\begin{proof}

We first claim that there exists $C<\infty$ such that
\begin{equation}
    \label{Eq:lmm:3}
\sup_{z\in[0,1]^2}\int_{[0,1]^2}\psi_n(z,z')\,dz' \le Cn^{-1}.
\end{equation}
Indeed, for fixed $z = (x,y) \in [0, 1]^2$, set $\{(x', y'): d_{\mathbb{T}}(x, x') \leq n^{-1/2}, \,d_{\mathbb{T}}(y, y') \leq n^{-1/2}\}$ has area at most $4n^{-1}$. The strip
$\{(x', y'): d_{\mathbb{T}}(x, x') \leq n^{-1/2}\}$ has area at most $2n^{-1/2}$, and similarly for $\{(x', y'): d_{\mathbb{T}}(y, y') \leq n^{-1/2}\}$. Hence by the definition of $\psi_n$,
\[
\int\psi_n(z,z')\,dz'
\le
4n^{-1} + 2n^{-1} + 2n^{-1} + n^{-1}
=
9n^{-1},
\]
as desired. 

Now, fix $r\ge2$, then for every tree $T\in\mathcal T_r$, we have 
\begin{equation}
    \label{Eq:lmm:tree}
\int_{([0,1]^2)^r}
\prod_{\{a,b\}\in E(T)} \psi_n(z_a,z_b)
\,dz_1\cdots dz_r
\le
C_r n^{-(r-1)},
\end{equation}
where the bound uses~\eqref{Eq:lmm:3} for the integral over each leaf variable when taking vertex $1$ as the root in the tree. 

Since cumulants are multilinear and $U_n^{(L)}$ is bounded for fixed $L$, we have by Fubini's Theorem together with Proposition~\ref{prop:truncated-field-cumulant} and~\eqref{Eq:lmm:tree} that
\[
|\mathrm{cum}_r(\mathscr{C}_n^{(L)})|
\le
C_{r,L}
\sum_{T\in\mathcal T_r}
\int_{([0,1]^2)^r}
\prod_{\{a,b\}\in E(T)} \psi_n(z_a,z_b)
\,dz_1\cdots dz_r
\le
C'_{r,L} n^{-(r-1)}.
\]
This proves \eqref{eq:13}. Therefore, for
\[
Z_n^{(L)}=\sqrt n \bigl(\mathscr{C}_n^{(L)}-\mathbb{E} \mathscr{C}_n^{(L)}\bigr),
\]
we have
\[
\mathrm{cum}_1(Z_n^{(L)})=0,
\qquad
\mathrm{cum}_r(Z_n^{(L)}) = n^{r/2}\mathrm{cum}_r(\mathscr{C}_n^{(L)}) = O\!\left(n^{1-r/2}\right)\to0
\quad (r\ge3).
\]
If along a subsequence $n_k$ we have $n_k\var(\mathscr{C}_{n_k}^{(L)})\to v$, then the second cumulants
of $Z_{n_k}^{(L)}$ converge to $v$ and all higher cumulants converge to $0$; hence
\[
Z_{n_k}^{(L)} \xrightarrow{d} \mathcal{N}(0,v)
\]
by the method of cumulants.
\end{proof}

For fixed $L\ge1$, let
\[
R_n^{(L)}(z) := U_n(z) - U_n^{(L)}(z),\qquad
\Delta_n^{(L)} := \mathscr{C}_n - \mathscr{C}_n^{(L)} = \int_{[0,1]^2} R_n^{(L)}(z)\,dz.
\]
Then we have
\[
R_n^{(L)}(z)=\sum_{m\ge L+1}(-1)^{m+1}V_{n,m}(z),
\]
and this is a finite sum because $\mathcal M_{n,m}(z)=\varnothing$ for all sufficiently large $m$.

% Since 
% \begin{eqnarray*}
%     R_n^{(L)}(z) = U_n(z) - U_n^{(L)}(z) = \bone{\{N_n(z)\geq 1\}} + \sum_{m = 1}^L(-1)^m \binom{N_n(z)}{m},
% \end{eqnarray*}
% we have
% \begin{eqnarray*}
%     R_n^{(L)}(z) = \sum_{m = 0}^L(-1)^m \binom{N_n(z)}{m}
% \end{eqnarray*}
% when $N_n(z) \geq 1$ and $R_n^{(L)}(z) = 0$ otherwise, and
% the binomial identity
% \[
% \sum_{m=0}^{L} (-1)^m \binom{q}{m} = (-1)^L \binom{q-1}{L}
% \qquad (q\ge1),
% \]
% yields the exact remainder formula
% \begin{equation}
% R_n^{(L)}(z)
% =
% (-1)^L \binom{N_n(z)-1}{L}\, \bone{\{N_n(z)\ge L+1\}}.
% \label{eq:9}
% \end{equation}
% In particular,
% \begin{equation}
% |R_n^{(L)}(z)|
% \le
% (L+1)\binom{N_n(z)}{L+1}
% =
% (L+1)V_{n,L+1}(z).
% \label{eq:10}
% \end{equation}

\begin{prop}
\label{prop:covariance-Vm}
Fix integers $m,m'\ge1$. Then there exist constants $B, C < \infty$ such that for all
$z=(x,y), w=(x',y')\in[0,1]^2$,
\begin{equation}
\left|
\cov\!\left(V_{n,m}(z),V_{n,m'}(w)\right)
\right|
\le
C\,\frac{B^{m+m'}}{\sqrt{m!\,m'!}}\,\psi_n(z,w).
\label{eq:14}
\end{equation}
\end{prop}

\begin{proof} 
Throughout this proof, $C$ and $B$ denote positive finite constants whose values may change from line to line. They are independent of $n,z,w,m,m'$. 

Write
\[
V_{n,m}(z)=\sum_{M\in\mathcal M_{n,m}(z)} X_M,
\qquad
V_{n,m'}(w)=\sum_{M'\in\mathcal M_{n,m'}(w)} X_{M'}.
\]
Hence
\begin{equation}
\cov(V_{n,m}(z),V_{n,m'}(w))
=
\sum_{M\in\mathcal M_{n,m}(z)}
\sum_{M'\in\mathcal M_{n,m'}(w)}
\cov(X_M,X_{M'}).
\label{eq:16}
\end{equation}

% We first record the following extension of Lemma~\ref{lem:matchingcount_a}: if $S\subset I_n(z)\times J_n(z)$ is a fixed
% matching of size $s\le m$, then
% \begin{equation}
% \left|\{M\in\mathcal M_{n,m}(z): S\subset M\}\right|
% \le
% \frac{K^m n^{m-s}}{(m-s)!}.
% \label{eq:173}
% \end{equation}
% Indeed, after prescribing the $s$ cells of $S$, one chooses the remaining $m-s$ rows and columns
% and a bijection, exactly as in the proof of Lemma~\ref{lem:matchingcount}.

Note that if $M\in \mathcal{M}_{n, m}(z)$, then 
\begin{align*}
    \mathbb{E}[X_M] = \frac{1}{(n)_m},
\end{align*}
and $m \leq \min\{\abs{I_n(z)}, \abs{J_n(z)}\} \leq C_0 \sqrt{n}$, hence for our purpose all relevant $m$ are $O(\sqrt{n})$. Then note that 
\[
(n)_m = n^m\prod_{\ell = 0}^{m - 1}\left(1 - \frac{\ell}{n}\right) \geq n^m\exp(-C\frac{m^2}{n}) \geq C n^m.
\]

We now split the sum in \eqref{eq:16} into four classes.

\smallskip
\noindent
\emph{Class 1: $M\cap M'\neq\varnothing$.}
Write
\[
s:=|M\cap M'|\in\{1,\dots,\min(m,m')\}.
\]
If $M\cup M'$ is a matching, then it has size $m+m'-s$, and hence
\[
\mathbb{E}[X_MX_{M'}]=\frac{1}{(n)_{m+m'-s}}.
\]
If $M\cup M'$ is not a matching, then $X_MX_{M'}=0$. In either case,
\[
\mathbb{E}[X_MX_{M'}]\le C n^{-(m+m'-s)}.
\]
Also
\[
\mathbb{E}[X_M]\mathbb{E}[X_{M'}]\le C n^{-m-m'}
\le C n^{-(m+m'-s)}.
\]
Hence
\[
|\cov(X_M,X_{M'})|\le C n^{-(m+m'-s)}.
\]

Fix $s$. For each fixed $M\in\mathcal M_{n,m}(z)$, the number of
$M'\in\mathcal M_{n,m'}(w)$ with $|M\cap M'|=s$ is at most
\[
\binom{m}{s}\frac{K^{m'} n^{m'-s}}{(m'-s)!}
\]
by Lemma~\ref{lem:matchingcount_a}, since one first chooses the $s$ common cells inside $M$ and then extends them
to a matching of size $m'$.
Using also
\[
|\mathcal M_{n,m}(z)|\le \frac{K^m n^m}{m!},
\]
the contribution of pairs with $|M\cap M'|=s$ is bounded by
\[
C\,
\frac{K^{m+m'} n^m}{m!}\,
\binom{m}{s}\frac{n^{m'-s}}{(m'-s)!}\,
n^{-(m+m'-s)}
=
C\,K^{m+m'}\frac{1}{s!(m-s)!(m'-s)!}.
\]
Summing over $s$ gives that the contribution of Class~1 is at most
\begin{equation}
C\,K^{m+m'}
\sum_{s=1}^{\min(m,m')}
\frac{1}{s!(m-s)!(m'-s)!}.
\end{equation}

To bound the sum, note that
\[
\frac{1}{s!(m-s)!}=\frac{\binom{m}{s}}{m!}\le \frac{2^m}{m!},
\]
hence
\[
\sum_{s=1}^{\min(m,m')}
\frac{1}{s!(m-s)!(m'-s)!}
\le
\frac{2^m}{m!}\sum_{t=0}^{m'-1}\frac{1}{t!}
\le
e\,\frac{2^m}{m!}.
\]
By symmetry, the same sum is also bounded by $e\,2^{m'}/m'!$. Therefore
\[
\sum_{s=1}^{\min(m,m')}
\frac{1}{s!(m-s)!(m'-s)!}
\le
e\,\min\!\left(\frac{2^m}{m!},\frac{2^{m'}}{m'!}\right)
\le
e\,\frac{2^{(m+m')/2}}{\sqrt{m!\,m'!}}.
\]
Thus the contribution of Class~1 is bounded by
\[
C\,\frac{B^{m+m'}}{\sqrt{m!\,m'!}}.
\]
Moreover this class is empty unless $I_n(z)\cap I_n(w)\neq\varnothing$ and $J_n(z)\cap J_n(w)\neq\varnothing$, that is,
\[
d_{\mathbb{T}}(x, x') \leq \frac{1}{\sqrt{n}}, \qquad d_{\mathbb{T}}(y, y') \leq \frac{1}{\sqrt{n}}.
\]

\smallskip
\noindent
\emph{Class 2: $M\cap M'=\varnothing$, but $\mathrm{row}(M)\cap\mathrm{row}(M')\neq\varnothing$.}
Then
\[
|\cov(X_M,X_{M'})|
\le
\mathbb{E}[X_MX_{M'}]+\mathbb{E}[X_M]\mathbb{E}[X_{M'}]
\le
2n^{-(m+m')}.
\]
For each fixed $M'\in\mathcal M_{n,m'}(w)$, Lemma~\ref{lem:matchingcount_b} implies that the number of
$M\in\mathcal M_{n,m}(z)$ with $\mathrm{row}(M)\cap\mathrm{row}(M')\neq\varnothing$ is at most
\[
m'\frac{K^m n^{m-1/2}}{(m-1)!}.
\]
Hence, using
\[
|\mathcal M_{n,m'}(w)|\le \frac{K^{m'} n^{m'}}{m'!},
\]
the contribution of Class~2 is bounded by
\[
2m'\frac{K^{m+m'} n^{m+m'-1/2}}{(m-1)!\,m'!}\,n^{-(m+m')}
\le
C\,\frac{B^{m+m'}}{\sqrt{m!\,m'!}}\,n^{-1/2}.
\]
This class is empty unless
\[
d_{\mathbb{T}}(x, x') \leq \frac{1}{\sqrt{n}}.
\]

\smallskip
\noindent
\emph{Class 3: $M\cap M'=\varnothing$, $\mathrm{row}(M)\cap\mathrm{row}(M')=\varnothing$, but
$\mathrm{col}(M)\cap\mathrm{col}(M')\neq\varnothing$.}
By the symmetric argument using Lemma~\ref{lem:matchingcount_c}, the contribution of Class~3 is bounded by
\[
C\,\frac{B^{m+m'}}{\sqrt{m!\,m'!}}\,n^{-1/2},
\]
and this class is empty unless
\[
d_{\mathbb{T}}(y, y') \leq \frac{1}{\sqrt{n}}.
\]

\smallskip
\noindent
\emph{Class 4: $\mathrm{row}(M)\cap\mathrm{row}(M')=\varnothing$ and $\mathrm{col}(M)\cap\mathrm{col}(M')=\varnothing$.}
Then $M\cup M'$ is a matching of size $m+m'$, so
\[
\mathbb{E}[X_MX_{M'}]=\frac{1}{(n)_{m+m'}},
\qquad
\mathbb{E}[X_M]\mathbb{E}[X_{M'}]=\frac{1}{(n)_m (n)_{m'}}.
\]
Note that 
\[
\frac{1}{(n)_{m+m'}}-\frac{1}{(n)_m (n)_{m'}} = \frac{1}{(n)_{m}(n)_{m'}}\left( \frac{(n)_{m}(n)_{m'}}{(n)_{m + m'}} - 1\right)
\]
where 
\[
\frac{(n)_{m}(n)_{m'}}{(n)_{m + m'}} = \frac{(n)_{m'}}{(n - m)_{m'}} = \prod_{\ell = 0}^{m' - 1}\frac{n - \ell}{n - m - \ell} = \prod_{\ell = 0}^{m' - 1}\left(1 + \frac{m}{n - m - \ell}\right).
\]
Since $m + m' = O(\sqrt{n})$, for large $n$ we have
\[
n - m - \ell \geq n/2.
\]
Therefore 
\[
\sum_{\ell = 0}^{m' - 1}\frac{m}{n - m - \ell} \leq C\frac{m\,m'}{n}.
\]
Since $m\, m' / n = O(1)$, we get
\[
\prod_{\ell = 0}^{m' - 1}\left(1 + \frac{m}{n - m - \ell}\right) - 1 \leq C\frac{m\,m'}{n}.
\]
Hence,
\[
\left|
\frac{1}{(n)_{m+m'}}-\frac{1}{(n)_m (n)_{m'}}
\right|
\le
C m m' n^{-(m+m'+1)},
\]

Using Lemma~\ref{lem:matchingcount} for both $M$ and $M'$, the total contribution of Class~4 is bounded by
\[
C m m'\frac{K^{m+m'} n^{m+m'}}{m!\,m'!}\,n^{-(m+m'+1)}
=
C\,K^{m+m'}\frac{m\,m'}{m!\,m'!}\,n^{-1}\leq C\,\frac{B^{m + m'}}{\sqrt{m!\,m'!}}n^{-1},
\]
where the last inequality holds since
\[
\frac{m\,m'}{m!\,m'!} \leq C\,\frac{2^{m + m'}}{\sqrt{m!\,m'!}}.
\]

Combining the four classes proves
\[
\left|
\cov\!\left(V_{n,m}(z),V_{n,m'}(w)\right)
\right|
\le
C\,\frac{B^{m+m'}}{\sqrt{m!\,m'!}}
\psi_n(z,w),
\]
which is exactly \eqref{eq:14}.
\end{proof}

Finally Proposition~\ref{prop:tail-covariance} justifies the truncation by showing that although the full inclusion--exclusion expansion cannot be handled directly, the remainder after truncation has covariances that satisfy the same spatial locality bound, now multiplied by a quantity $\eta_L$ tending to zero with $L$. This shows that the tail does not affect the $\sqrt{n}$-scale fluctuations.
\begin{prop}\label{prop:tail-covariance}
For every integer $L \ge 1$ there exists a constant $\eta_L \to 0$ as $L \to \infty$ such that for all $z,w \in [0,1]^2$,
\begin{equation}
\left| \cov\!\left(R_n^{(L)}(z),R_n^{(L)}(w)\right) \right|
\le \eta_L \,\psi_n(z,w),
\label{eq:17}
\end{equation}
where 
\[
\eta_L = O\left(\frac{B^{2L}}{L!}\right),
\]
for some constants $B < \infty$. Consequently,
\begin{equation}
n\,\var(\Delta_n^{(L)}) \le C\,\eta_L,
\label{eq:18}
\end{equation}
and therefore
\[
\lim_{L\to\infty}\sup_{n\ge 1} n\,\var(\Delta_n^{(L)}) = 0.
\]
\end{prop}

\begin{proof}
Since
\[
R_n^{(L)}(z)=\sum_{m\ge L+1}(-1)^{m+1}V_{n,m}(z),
\]
and this is a finite sum for each fixed $n$ and $z$, bilinearity of covariance gives
\[
\cov\!\left(R_n^{(L)}(z),R_n^{(L)}(w)\right)
=
\sum_{m\ge L+1}\sum_{m'\ge L+1}
(-1)^{m+m'}
\cov\!\left(V_{n,m}(z),V_{n,m'}(w)\right).
\]
Hence, by Proposition~\ref{prop:covariance-Vm},
\[
\left| \cov\!\left(R_n^{(L)}(z),R_n^{(L)}(w)\right) \right|
\le
\left(\sum_{m\ge L+1}\sum_{m'\ge L+1} C_{m,m'}\right)\psi_n(z,w).
\]
Define
\[
\eta_L:=\sum_{m\ge L+1}\sum_{m'\ge L+1} C_{m,m'}.
\]
Using the bound
\[
C_{m,m'} \le C\,\frac{B^{m+m'}}{\sqrt{m!\,m'!}},
\]
we obtain
\[
\eta_L
\le
C\sum_{m\ge L+1}\sum_{m'\ge L+1}
\frac{B^{m+m'}}{\sqrt{m!\,m'!}}
=
C\left(\sum_{m\ge L+1}\frac{B^m}{\sqrt{m!}}\right)^2.
\]
Now the series $\sum_{m\ge 0} B^m/\sqrt{m!}$ converges, since
\[
\frac{B^{m+1}/\sqrt{(m+1)!}}{B^m/\sqrt{m!}}
=
\frac{B}{\sqrt{m+1}}
\rightarrow 0.
\]
Therefore its tail tends to $0$, so $\eta_L \downarrow 0$ as $L\to\infty$. Additionally, once $m + 1\geq 4B^2$, we have
\[
\frac{B^{m+1}/\sqrt{(m+1)!}}{B^m/\sqrt{m!}} \leq \frac{1}{2}.
\]
Hence, for $L + 1 \geq 4B^2$, we have
\[
\sum_{m\geq L + 1}\frac{B^m}{\sqrt{m!}} \leq \frac{B^{L+1}}{\sqrt{(L+1)!}}\sum_{k\geq 0}2^{-k} = 2\frac{B^{L+1}}{\sqrt{(L+1)!}}.
\]
Therefore for sufficiently large $L$
\[
\eta_L \leq 2C \frac{B^{2L+2}}{(L+1)!},
\]
which using Stirling's formula $L!\sim \sqrt{2\pi L}(L/e)^L$, gives us
\[
\eta_L = O\left(\frac{B^{2L}}{L!}\right) = O\left(\frac{1}{\sqrt{L}}\left(\frac{eB^2}{L}\right)^L\right),
\]
which is super-exponentially fast. This proves \eqref{eq:17}.

Next,
\[
\Delta_n^{(L)}=\int_{[0,1]^2} R_n^{(L)}(z)\,dz,
\]
so by Fubini,
\[
\var(\Delta_n^{(L)})
=
\int_{[0,1]^2}\int_{[0,1]^2}
\cov\!\left(R_n^{(L)}(z),R_n^{(L)}(w)\right)\,dz\,dw.
\]
Using \eqref{eq:17} and~\eqref{Eq:lmm:3},
\[
\var(\Delta_n^{(L)})
\le
\eta_L
\int_{[0,1]^2}\int_{[0,1]^2}\psi_n(z,w)\,dz\,dw
\le C\,\eta_L\,n^{-1}.
\]
Multiplying by $n$ yields \eqref{eq:18}, and taking $\sup_n$ then letting $L\to\infty$ gives
\[
\lim_{L\to\infty}\sup_{n\ge 1} n\,\var(\Delta_n^{(L)}) = 0,
\]
as desired.
\end{proof}

\subsection{Additional results used in the proof of Theorem~\ref{Thm:ConsistencyGeneralDimGrid}}
We first generalise the convergence result~\eqref{Eq:OT_convergence} used in the proof of Proposition~\ref{prop: VonRawData} to the case where the empirical optimal transport map $\hat{T}_n$ is obtained by matching the empirical measure of the data to a sequence of deterministic points that lie in $[0, 1]^d$ instead of the random samples from $\mathrm{Unif}([0, 1]^d)$. This derandomisation of the uniform distribution approximation is useful to derive the consistency result of the coverage correlation coefficient when using regular grid as the reference points, and can also be applied to other deterministic reference points such as Halton sequence, Sobol' sequence, etc. 

\begin{prop}\label{prop:EmpiricalOTMapError_deterministic} 
Suppose $\mu$ is an absolutely continuous probability measure on $\mathbb{R}^d$, where $d \geq 2$, and $\mu\in \mathcal{P}_{t, \alpha}(\mathbb{R}^d)$ for some $t >0$ and $\alpha >0$. Suppose $T_0= \nabla \phi_0$ Is the (almost surely unique) optimal transport map from $\mu$ to $\mathrm{Unif}([0, 1]^d)$, and is $(\theta, L)$-H\"older continuous for some $\theta\in(0, 1]$ and $L > 0$. Additionally, we assume that the Brenier potential associated with the optimal transport map from $\mathrm{Unif}([0, 1]^d)$ to $\mu$ is Lipschitz over $[0, 1]^d$. 

Given data $X_1, \ldots, X_n \stackrel{\mathrm{iid}}{\sim} \mu$ and a sequence of distinct points $u_1, \ldots, u_n \in [0, 1]^d$, define ${\mu}_n= n^{-1} \sum_{i=1}^n \delta_{X_i}$ and $\nu_n = n^{-1} \sum_{i = 1}^n \delta_{u_i}$, and let $ T_n$ be the (almost surely unique) optimal transport map between $ \mu_n$ and $ \nu_n$. Then, if 
\begin{align*}
    \mathbb{E} \mathcal{W}_2( \nu_n, \mathrm{Unif}([0, 1]^d)) = O(n^{-1/d}\log^2 n),
\end{align*}
we have
\[
\frac{1}{n}\sum_{i = 1}^n \bigl\| T_n (X_i) - T_0(X_i)\bigr\|_2^2 = O_p(\tilde{s}_{n,d,\theta}\log^{\frac{2c_{ \alpha} \theta}{1 + \theta}} n),
\]
where $c_{\alpha} > 0$ is a finite constant depending only on $\alpha$, and
\begin{align*}
    \tilde s_{n, d, \theta} = \begin{cases}
        n^{-1}, \quad &\text{if $d = 1$}, \\
        n^{-\frac{\theta}{1+\theta}} \log^{\frac{2\theta}{1+\theta}} (1+n), \quad &\text{if $d = 2$}, \\
        n^{-\frac{2\theta}{d(1+\theta)}}, \quad &\text{if $d \geq 3$}.
    \end{cases}
\end{align*}
\end{prop}

\begin{proof}
Throughout the proof, we denote $\lambda^d$ as the Lebesgue measure on $[0, 1]^d$, and let $\nu_n^\sharp:=T_0 \# {\mu}_n$. Let $p=1+\theta^{-1}$, then by~\eqref{eq:OTEstimationErrors}, there exists a constant $C >0$ depending only on $\theta, L$ such that
\begin{align}
\int \|T_0(z) -  {T}_n(z)\|^{p}  \, \mathrm{d}{\mu}_n(z) \leq C \biggl\{\int \phi_0^*\, \mathrm{d}( \nu_n - \nu_n^\sharp) + \int \tilde{\Phi}^* \, \mathrm{d}(\nu_n^\sharp -  \nu_n)\biggr\}, \label{eq:OTEstimationErrors_deterministic}
\end{align}
where $\tilde{\Phi} = \argmin_f \mathcal{S}_{{\mu}_n, {\nu}_n}(f)$. We bound each of the two terms on the right-hand side of~\eqref{eq:OTEstimationErrors_deterministic}. 

Consider the decomposition
\[
\int \phi_0^*\, \mathrm{d}( \nu_n - \nu_n^\sharp) = \int \phi_0^*\, \mathrm{d}( \nu_n - \lambda^d) + \int \phi_0^*\, \mathrm{d}(\lambda^d - \nu_n^\sharp).
\]
Observing that $\phi_0^*$ is convex and finite on $[0, 1]^d$, $\phi_0^*$ is bounded on $[0, 1]^d$. Hence the second term on the right-hand side above is $O_p(n^{-1/2})$ by the law of large numbers. As for the first term, since $\phi_0^*$ is Lipschitz over $[0, 1]^d$, by the Kantorovich-Rubinstein duality \citep[Theorem~1.14]{villani2009optimal} we have 
\begin{align}
    \int \phi_0^* \, \mathrm{d} ({\nu}_n - \lambda^d) = M \int \frac{\phi_0^*}{M} \, \mathrm{d}({\nu}_n - \lambda^d) \leq M \mathcal{W}_1({\nu}_n, \lambda^d) = O_p(n^{-1/d} \log^2 n), \label{eq:FirstTermOTError}
\end{align}
where $M$ is the Lipschitz constant of $\phi_0^*$.

We turn to bound the second term on the right-hand side of~\eqref{eq:OTEstimationErrors_deterministic}. Note that $\tilde{\Phi}^*$ is only defined on $u_1, \ldots, u_n$, we extend $\tilde\Phi^*$ to the whole $\mathbb{R}^d$ by letting 
\[
\tilde\Phi^*(u) = \max_{i \in [n]} \bigl\{\langle X_i, u\rangle - \tilde\Phi(X_i)\bigr\},
\]
for $u \in \mathbb{R}^d$. It can be verified that $\tilde\Phi^*$ is convex and matches the original values of $\tilde\Phi^*$ on $u_1, \ldots, u_n$. Moreover, since $\tilde\Phi^*$ is the pointwise maximum of linear functions, $\tilde\Phi^*$ is Lipschitz with Lipschitz constant $\max_{i \in [n]} \|X_i\|_2$. Moreover, by the sub-Weibull tail assumption on $X_i$'s, we have for $K > (1/t)^{1/\alpha}$,
\begin{align*}
\mathbb{P}\Bigl(\max_{i \in [n]} \|X_i\|_2 \geq K (\log n)^{1/\alpha} \Bigr) &\leq n \mathbb{P}(\|X_1\|_2 \geq K (\log n)^{1/\alpha}) \\
&\leq 2n \exp(-t K^\alpha \log n) \mathbb{E}\exp(t\|X_1\|_2^\alpha) \to 0,
\end{align*}
as $n \to \infty$. Hence $\max_{i \in [n]} \|X_i\|_2 = O_p((\log n)^{1/\alpha})$. Then, by the convergence rate of the empirical 1-Wasserstein distance \citep[see, e.g.,][Theorem~1]{fournier2015rate} and the Kantorovich-Rubinstein duality again we obtain that
\begin{align}
    \int \tilde \Phi^* \, \mathrm{d}(\nu_n^\sharp -  \nu_n) &\leq \max_{i \in [n]} \|X_i\|_2 \mathcal{W}_1(\nu_n^\sharp, {\nu}_n) \notag \\ 
    &\leq \max_{i \in [n]} \|X_i\|_2 \Bigl(\mathcal{W}_1(\nu_n^\sharp, \lambda^d) + \mathcal{W}_1( \nu_n, \lambda^d)\Bigr) = O_p\Bigl(s_{n, d} \log^{c_\alpha'} n\Bigr), \label{eq:SecondTermOTError} 
\end{align}
where $s_{n, d} = n^{-1/2} \log (1+n)$ if $d = 2$ and $s_{n, d} = n^{-1/d}$ if $d \geq 3$, and $c_\alpha'>0$ is a constant depending only on $\alpha$. 

Consequently, taking~\eqref{eq:FirstTermOTError} and ~\eqref{eq:SecondTermOTError} back to~\eqref{eq:OTEstimationErrors_deterministic}, we have
\[
\frac{1}{n}\sum_{i = 1}^n \bigl\| T_n (X_i) - T_0(X_i)\bigr\|_2^p = O_p(s_{n, d} \log^{c_\alpha} n),
\]
for some $c_\alpha > 0$. Combining this with Proposition~\ref{prop:1DKMTMatching_deterministic} gives us that
\[
\frac{1}{n}\sum_{i = 1}^n \bigl\| T_n (X_i) - T_0(X_i)\bigr\|_2^2 = O_p(\tilde s_{n, d, \theta} \log^{\frac{2c_{\alpha} \theta}{1 + \theta}} n), 
\]

where $c_{\alpha} > 0$ is a finite constant depending only on $\alpha$. This completes the proof.
\end{proof}

\begin{prop}\label{prop: VonRawData_deterministic}
Suppose $(X_1, Y_1), \ldots, (X_n, Y_n) \stackrel{\mathrm{iid}}{\sim} P^{(X, Y)}$, where $P^{(X, Y)}$ is a Borel probability measure on $\mathbb{R}^{d_X+d_Y}$ with absolutely continuous marginals $P^X \in \mathcal{P}_{t_1, \alpha_1}(\mathbb{R}^{d_X})$ and $P^Y\in \mathcal{P}_{t_2, \alpha_2}(\mathbb{R}^{d_Y})$ for some $t_1, t_2 > 0$ and $\alpha_1, \alpha_2 \in (0, 2]$. Let $T_X$ and $T_Y$ be the optimal transport maps from $P^X$ and $P^Y$ to $\mathrm{Unif}([0, 1]^{d_X})$ and $\mathrm{Unif}([0, 1]^{d_Y})$, respectively, and assume that $T_X$ is $(\theta_1, L_1)$-H\"older continuous and $T_Y$ is $(\theta_2, L_2)$-H\"older continuous for some $\theta_1, \theta_2 \in (0, 1]$ and $L_1, L_2 >0$. Moreover, assume that the Brenier potentials associated with the optimal transport maps from $\mathrm{Unif}([0, 1]^{d_X})$ to $P^X$ and from $\mathrm{Unif}([0, 1]^{d_Y})$ to $P^Y$ are Lipschitz over $[0, 1]^{d_X}$ and $[0, 1]^{d_Y}$, respectively.

Given sequences $u_1, \ldots, u_n \in [0, 1]^d$ and $v_1, \ldots, v_n \in [0, 1]^d$ such that 
\begin{align}
\mathbb{E}\mathcal{W}_2\bigl(\mathrm{Unif}([0, 1]^{d_X}, P_n^u) = O(n^{-\frac{1}{d_X}}\log^2 n \bigr) \quad \text{and} \quad  \mathbb{E}\mathcal{W}_2\bigl(\mathrm{Unif}([0, 1]^{d_Y}, P_n^v) = O(n^{-\frac{1}{d_Y}}\log^2 n\bigr), \label{eq:QuasiUniformApproxRate}
\end{align}
where $P_n^u = n^{-1} \sum_{i = 1}^n \delta_{u_i}$ and $P_n^v = n^{-1} \sum_{i = 1}^n \delta_{v_i}$.  Write $d:= d_X + d_Y$, $X_i^\natural  := T_X(X_i)$ and $Y_i^\natural  := T_Y(Y_i)$ for $i = 1, \ldots, n$. Then when $(d_X, d_Y) \in \mathcal{D}(\theta_1, \theta_2)$, we have
\begin{align*}
    \Biggl|\mathcal{V}\biggl((X_i)_{i\in[n]}, (Y_i)_{i\in[n]}, &(u_i)_{i\in[n]},(v_i)_{i\in[n]}; \frac{1}{2n^{1/d}} \Bigr) \\
    &- \mathcal{V}\Bigl((X_i^\natural)_{i\in[n]}, (Y_i^\natural )_{i\in[n]}, ( X_i^\natural)_{i\in[n]}, ( Y_i^\natural )_{i\in[n]}; \frac{1}{2n^{1/d}}\biggr)\Biggr| \stackrel{\mathrm{p}}{\longrightarrow} 0,
\end{align*}
as $n \to \infty$.
\end{prop}
\begin{proof}
The proof follows the same argument as that of Proposition~\ref{prop: VonRawData}. Hence we only sketch the proof and highlight the differences. Write $\bs{X}:=(X_i)_{i\in[n]}$, $\bs{Y}:=(Y_i)_{i\in[n]}$, ${\bs{X}^\natural}:=( X_i^\natural)_{i\in[n]}$, ${\bs Y^\natural} :=( Y_i^\natural)_{i\in[n]}$, $\bs{U}:=(u_i)_{i\in[n]}$, $\bs{V}:=(v_i)_{i\in[n]}$. Given reference points ${\bs U}$ and ${\bs V}$, let $R_i:= (R_i^X, R_i^Y)$ and $ R_i^\natural := ({R}_i^{X^\natural}, {R}_i^{Y^\natural})$ denote the joint rank of $\{(X_i, Y_i)\}_{i = 1}^n$ and $\{({X}_i^\natural, {Y}_i^\natural)\}_{i = 1}^n$, respectively. Also write $Z_i^\natural := ( X_i^\natural,  Y_i^\natural)$, for $i = 1, \ldots, n$. 

Define the coverage areas
\[
A_1 = \bigcup_{i = 1}^n B\Bigl(R_i, \frac{1}{2n^{1/d}}\Bigr), \quad {A}_2 = \bigcup_{i=1}^n B\Bigl({R}_i^\natural, \frac{1}{2n^{1/d}}\Bigr), \quad {A}_3 = \bigcup_{i=1}^n B\Bigl(Z_i^\natural, \frac{1}{2n^{1/d}}\Bigr).
\]
Then similar to~\eqref{ineq: vacancydiff}, we have the following bound
\begin{align}
     \biggl|\mathcal{V}\Bigl(\bs{X}, \bs{Y}, \bs{U}, \bs{V}; \frac{1}{2n^{1/d}} \Bigr) - \mathcal{V}\Bigl({\bs X^\natural}, {\bs Y^\natural}, {\bs X^\natural}, {\bs Y^\natural}; \frac{1}{2n^{1/d}}\Bigr)\biggr|  \leq \mathrm{vol}(A_1\Delta {A}_2) + \mathrm{vol}(A_2\Delta {A}_3), \label{ineq: vacancydiff_deterministic}
\end{align}
and by Lemma \ref{lem: symdiff}, we obtain that
\begin{align}
\mathrm{vol}(A_1\Delta {A}_2) \leq 2 d n^{-\frac{d-1}{d}} \sum_{i = 1}^n  \|R_i - {R}_i^\natural\|_2 \quad \text{and} \quad \mathrm{vol}(A_2\Delta {A}_3) \leq 2 d n^{-\frac{d-1}{d}} \sum_{i = 1}^n  \| R_i^\natural - {Z}_i^\natural\|_2. \label{ineq: coveragesymdiff_deterministic}
\end{align} 

Now set $ P_n^{X^\natural} := \frac{1}{n}\sum_{i = 1}^n \delta_{ X_i^\natural}$, $ P_n^{Y^\natural}:= \frac{1}{n}\sum_{i = 1}^n \delta_{ Y_i^\natural}$ and $ P_n^u := \frac{1}{n}\sum_{i = 1}^n \delta_{u_i}$, $P_n^v := \frac{1}{n}\sum_{i = 1}^n \delta_{v_i}$. Then by the definition of $\{ R_i^{X^\natural}\}_{i = 1}^n$ and $\{ R_i^{Y^\natural}\}_{i = 1}^n$ we have by the Cauchy--Schwarz inequality that 
\begin{align*}
    \frac{1}{n}\sum_{i = 1}^n \| X_i^\natural -  R_i^{X^\natural}\|_2 &\leq \biggl(\frac{1}{n}\sum_{i=1}^n \| X_i^\natural -  R_i^{X^\natural}\|_2\biggr)^{1/2}  = \mathcal{W}_2
    \bigl({P}_n^{X^\natural}, P_n^u\bigr),\\
    \frac{1}{n} \sum_{i = 1}^n \|{Y}_i^\natural - {R}_i^{Y^\natural}\|_2 &\leq \biggl(\frac{1}{n}\sum_{i=1}^n \| Y_i^\natural -  R_i^{Y^\natural}\|_2\biggr)^{1/2} = \mathcal{W}_2\bigl({P}_n^{Y^\natural}, P_n^v\bigr).
\end{align*}
Since $X_i^\natural \stackrel{\mathrm{iid}}{\sim} \mathrm{Unif}([0, 1]^{d_X})$ and $Y_i^\natural \stackrel{\mathrm{iid}}{\sim} \mathrm{Unif}([0, 1]^{d_Y})$, then by the convergence rate of the empirical 2-Wasserstein distance \citep{fournier2015rate} and~\eqref{eq:QuasiUniformApproxRate}, we have
\begin{align}
\frac{1}{n}\sum_{i = 1}^n\|{R}_i^\natural - Z_i^\natural\|_2 \leq& \mathcal{W}_2\bigl({P}_n^{X^\natural}, P_n^u\bigr) + \mathcal{W}_2\bigl({P}_n^{Y^\natural}, P_n^v\bigr) \notag \\ 
\leq& \mathcal{W}_2\bigl({P}_n^{X^\natural}, \mathrm{Unif}([0, 1]^{d_X}\bigr) + \mathcal{W}_2\bigl(\mathrm{Unif}([0, 1]^{d_X}, P_n^u\bigr) \notag \\ 
&\hspace{3cm}+\mathcal{W}_2\bigl({P}_n^{Y^\natural}, \mathrm{Unif}([0, 1]^{d_Y})\bigr) + \mathcal{W}_2\bigl(\mathrm{Unif}([0, 1]^{d_Y}, P_n^v\bigr) \notag \\
=& O_p(n^{-\frac{1}{d_X \vee d_Y \vee 4}} \log^2 n).\label{eq:EmpMap1_deterministic}
\end{align}
Since $(d_X, d_Y) \in\mathcal{D}(\theta_1, \theta_2)$, we have $d_X \vee d_Y > 4$. Hence, the right-hand side of~\eqref{eq:EmpMap1_deterministic} is $o_p(n^{-1/d})$.

Let $P_n^X:= \frac{1}{n}\sum_{i = 1}^n \delta_{X_i}$ and $P_n^Y:= \frac{1}{n}\sum_{i = 1}^n \delta_{Y_i}$, by Proposition~\ref{prop:EmpiricalOTMapError_deterministic} we have 
\begin{align}
    \frac{1}{n}\sum_{i = 1}^n \|R_i - Z_i^\natural \|_2 \leq & \biggl(\frac{1}{n}\sum_{i = 1}^n \|R_i^X -  X_i^\natural\|_2^2 \biggr)^{1/2}+  \biggl(\frac{1}{n}\sum_{i = 1}^n \|R_i^Y -  Y_i^\natural \|_2^2 \biggr)^{1/2} \notag \\ 
    = & O_p\bigl(\tilde s_{n, d_X, \theta_1}^{1/2}\log^{\eta_1} n + \tilde s_{n, d_Y, \theta_2}^{1/2} \log^{\eta_2} n\bigr) \label{eq:EmpMap2_deterministic},
\end{align}
where $\eta_1, \eta_2>0$ are constants depending only on $\alpha_1, \alpha_2, \theta_1, \theta_2$. Again, we claim that the stochastic order in~\eqref{eq:EmpMap2_deterministic} is $o_p(n^{-1/d})$. In fact, note that 
\begin{align}
    \tilde s_{n, d_X, \theta_1}^{1/2} = \begin{cases}
        n^{-\frac{1}{2}}, \quad & d_X = 1 \\ 
        (n^{-\frac{1}{2}} \log n)^{\frac{\theta_1}{1+\theta_1}}, \quad &d_X =2 \\ 
        n^{-\frac{2\theta_1}{d_X(1+\theta_1)}}, \quad & d_X \geq 3
    \end{cases}, \quad 
    \tilde s_{n, d_Y, \theta_2}^{1/2} = \begin{cases}
    n^{-\frac{1}{2}}, \quad & d_Y = 1 \\ 
    (n^{-\frac{1}{2}} \log n)^{\frac{\theta_2}{1+\theta_2}}, \quad & d_Y =2 \\ 
    n^{-\frac{2\theta_2}{d_Y(1+\theta_2)}}, \quad & d_Y \geq 3
\end{cases}, \label{eq:RateCal_deterministic}
\end{align}
and one can verify that when $(d_X, d_Y) \in \mathcal{D}(\theta_1, \theta_2)$, all the terms on the right-hand side of~\eqref{eq:RateCal_deterministic} are $o(n^{-1/d})$.

Therefore, by the triangle inequality and~\eqref{eq:EmpMap1_deterministic}--\eqref{eq:RateCal_deterministic}, we have
\begin{align}
    \frac{1}{n}\sum_{i = 1}^n  \|R_i - {R}_i^\natural\|_2  \leq \frac{1}{n}\sum_{i = 1}^n \|R_i - Z_i^\natural \|_2 + \frac{1}{n}\sum_{i = 1}^n \|R_i^\natural - Z_i^\natural \|_2 = o_p(n^{-1/d}). \label{eq:EmpMap3_deterministic}
\end{align}

Consequently, taking~\eqref{eq:EmpMap1_deterministic} and~\eqref{eq:EmpMap3_deterministic} back to~\eqref{ineq: coveragesymdiff_deterministic}, we have $\mathrm{vol}(A_1\Delta {A}_2) = o_p(1)$ and $\mathrm{vol}(A_2 \Delta {A}_3) = o_p(1)$. The conclusion then follows from~\eqref{ineq: vacancydiff_deterministic}.
\end{proof}

\subsection{Additional results to substantiate claims used in the paper}

Star discrepancy has proven to be a useful tool to estimate the error of approximating the uniform distribution by a finite set of points, and has been widely used in the literature of quasi-Monte Carlo methods \citep[see, e.g.,][Chapter~2]{niederreiter1992random}. The following definition and proposition give the definition of star discrepancy and its connection to the Wasserstein distance.
\begin{defn}\label{def:StarDiscrepancy}
Let $\lambda^d$ be the Lebesgue measure on $[0, 1]^d$. For a collection of points ${\bs u_n} = \{u_1, \ldots, u_n\} \subset [0, 1]^d$, let $\mu_n = \frac{1}{n}\sum_{i = 1}^n \delta_{u_i}$ be the corresponding empirical measure. Then the star discrepancy of ${\bs u_n}$ is defined as
\begin{align*}
    \mathcal{D}^*({\bs u}_n) : = \sup_{J \in \mathbb{J}}| \mu_n(J) - \lambda^d(J)|,
\end{align*}
where $\mathbb{J}:= \bigl\{\prod_{j = 1}^d [0, y_j): (y_1, \ldots, y_d) \in [0, 1]^d\bigr\}$.
\end{defn}

\begin{prop}\label{prop:WassBoundbyStar}
Let $\lambda^d$ be the Lebesgue measure on $[0, 1]^d$. Suppose ${\bs u}_n:= \{u_1, \ldots, u_n\} \subset [0, 1]^d$ is a set of vectors and $\mu_n := \frac{1}{n}\sum_{i = 1}^n \delta_{u_i}$. Then there exists a constant $C_d>0$ only depending on $d$ such that 
\begin{align*}
\mathcal{W}_2^2(\mu_n, \lambda^d) \leq C_d \begin{cases}
\bigl(\mathcal{D}^*({\bs u}_n)\bigr)^{2/d}, \quad &\text{when $d \geq 3$}, \\ 
 \mathcal{D}^*({\bs u}_n) \log^2 \bigl(\mathcal{D}^*({\bs u}_n)\bigr), \quad &\text{when $d = 2$}, \\ 
 \bigl(\mathcal{D}^*({\bs u_n})\bigr)^2, \quad &\text{when $d =1$}.
  \end{cases}
\end{align*}
\end{prop}
\begin{proof}
Throughout the proof, we denote $C_d>0$ as a constant depending only on $d$ whose value may vary from line to line, and write $D_n:=\mathcal{D}^*({\bs u}_n)$. We define a sequence of \emph{dyadic partitions} of $[0, 1]^d$ as 
\begin{align*}
    \mathcal{Q}_{k} := \biggl\{\prod_{j = 1}^d \Bigl[\frac{m_j}{2^k}, \frac{m_{j}+1}{2^k}\Bigr): m_j \in \{0, 1, \ldots, 2^k-1\}, \ j = 1, \ldots, d \biggr\} , \quad \text{for $k = 1, 2\ldots $,}
\end{align*}
and $\mathcal{Q}_0 = [0, 1]^d$. For $k \geq 0$, define probability measure $\nu_k$ such that for any $Q \in \mathcal{Q}_k$
\begin{align*}
    \nu_k(Q) = \mu_n(Q), \quad \nu_k|_Q = \frac{\mu_n(Q)}{\lambda^d(Q)} \lambda^d|_{Q}. 
\end{align*}
Hence when transporting from $\nu_k$ to $\mu_n$, the mass will only move within each component of $\mathcal{Q}_k$, and it follows that $\mathcal{W}_2(\mu_n, \nu_k) \leq \sqrt{d} 2^{-k}$. Note $\mathcal{W}_2(\mu_n, \lambda^d) \leq \mathcal{W}_2(\mu_n, \nu_k) + \mathcal{W}_2(\nu_k, \lambda^d)$, it remains to control $\mathcal{W}_2(\nu_k, \lambda^d)$. To do this, we first decompose the mismatch between $\nu_k$ and $\lambda^d$ along the dyadic partition hierarchy and then bound the transportation cost at each scale. 

Specifically, by the triangle inequality we have
\begin{align}
    \mathcal{W}_2(\nu_k, \lambda^d) \leq \sum_{j = 0}^{k-1} \mathcal{W}_2(\nu_{j}, \nu_{j+1}). \label{eq:DyadicDecompose}
\end{align}
For any integer $0 \leq j \leq k-1$ and measurable set $A \subseteq [0, 1]^d$, set the conditional distribution $\nu_j^A:= \frac{\nu_j|_A}{\nu_j(A)}$. Since $\nu_j$ and $\nu_{j+1}$ coincide on all $Q \in \mathcal{Q}_j$ by construction, by Lemma~\ref{lem:2-WassersteinPartitionBound} we have 
\begin{align}
\mathcal{W}_2^2(\nu_j, \nu_{j+1}) \leq \sum_{Q \in \mathcal{Q}_j} \mathcal{W}_2^2(\nu_j^{Q}, \nu_{j+1}^{Q}) \nu_j (Q) = \sum_{Q \in \mathcal{Q}_j} \mathcal{W}_2^2(\nu_j^{Q}, \nu_{j+1}^{Q}) \mu_n (Q).  \label{eq:ScaleControl}
\end{align}

For any $Q \in \mathcal{Q}_j$, let $\mathcal{A}^Q_{j+1} := \{\widetilde{Q} \in \mathcal{Q}_{j+1}:\widetilde{Q}\subseteq Q\}$ denote the collection of dyadic children of $Q$ at level $j+1$. Note that $|\mathcal{A}_{j+1}^Q| = 2^d$, and that for every $\widetilde{Q} \in \mathcal{A}_{j+1}^Q$, we have $\mathrm{diam}(\widetilde{Q})= \sqrt{d}2^{-(j+1)}$. Writing $\Delta(J) = (\mu_n (J) - \lambda^d(J))$ for any $J \subseteq [0, 1]^d$, we have
\begin{align*}
    \mathcal{W}_2^2(\nu_j^{Q}, \nu_{j+1}^{Q}) \mu_n (Q) &\leq \sum_{\widetilde{Q} \in \mathcal{A}_{j+1}^Q} |\nu_j^{Q}(\widetilde{Q}) - \nu_{j+1}^{Q}(\widetilde{Q})| \mathrm{diam}^2(\widetilde{Q}) \mu_n (Q) \\ 
    &\leq d 2^{-2(j+1)} \sum_{\widetilde{Q} \in \mathcal{A}_{j+1}^Q} |\nu_j(\widetilde{Q}) - \nu_{j+1}(\widetilde{Q})|   \\
    &=d 2^{-2(j+1)} \sum_{\widetilde{Q} \in \mathcal{A}_{j+1}^Q} |2^{-d}\mu_n(Q) - \mu_n(\widetilde{Q})| \\
    & = d 2^{-2(j+1)} \sum_{\widetilde{Q} \in \mathcal{A}_{j+1}^Q} |2^{-d}\Delta(Q) - \Delta(\widetilde{Q})| \leq  d 2^{-2(j+1)+2d}\, D_n,
\end{align*}
where the final inequality is followed by Corollary~\ref{cor:BoxDiscByStarDisc}. Since $|\mathcal{Q}_j| = 2^{jd}$, plugging in~\eqref{eq:ScaleControl} we have 
\begin{align}
\mathcal{W}_2^2(\nu_j, \nu_{j+1}) \leq  d 2^{2d-2} 2^{j(d-2)}\, D_n,  \label{eq:ScaleBound}
\end{align}
and together with~\eqref{eq:DyadicDecompose} we have when $d \geq 3$
\begin{align*}
    \mathcal{W}_2(\nu_k, \lambda^d) \leq C_d D_n^{1/2} \frac{1 - 2^{(d-2)k/2}}{1 - 2^{(d-2)/2}}  \leq C_d  D_n^{1/2} 2^{(d-2)k/2}.
\end{align*}
Consequently, let $m = 2^k$, we have 
\begin{align*}
    \mathcal{W}_2(\mu_n, \lambda^d) \leq \sqrt{d}m^{-1} + C_d  D_n^{1/2} m^{(d-2)/2}.
\end{align*}
Then by optimising the choice of $m = C_d D_n^{-1/d}$, we obtain that $\mathcal{W}_2(\mu_n,\lambda^d) \leq C_d D_n^{1/d}$, as desired. 

When $d = 2$, \eqref{eq:ScaleBound} implies that $\mathcal{W}_2^2(\nu_j, \nu_{j+1}) \leq d 2^{2d-2} D_n$. Hence 
\[
\mathcal{W}_2(\mu_n, \lambda^d) \leq \sqrt{d}m^{-1} + \sqrt{d} 2^{d-1} D_n^{1/2}  \log_2 m.
\]
By setting $m =C_d D_n^{-1/2}$, we get $\mathcal{W}_2(\mu_n,\lambda^d) \leq C_d D_n^{1/2} \log D_n.$

When $d = 1$, let $F_n(x): = \frac{1}{n}\sum_{i = 1}^n \bone{\{u_i \leq x\}}$ be the empirical CDF of $\mu_n$. By the definition of $D_n$, we have for any $t \in [0, 1]$,
\begin{align*}
  t - D_n \leq F_n(t) \leq t + D_n.
\end{align*}
Hence for $0 \leq t \leq 1- D_n$,
\[
F_n\bigl(t + D_n\bigr)  \geq  t +D_n - D_n = t.
\]
Note for $t \in (1- D_n, 1]$, we have $F_n(1) = 1 \geq t$, hence we obtain that 
\begin{align}
F_n(\min\{1, t +D_n\}) \geq t, \quad \text{for $t \in [0, 1]$}. \label{eq:CDFLowerBound}
\end{align}

For any $u \in [0, 1]$, let $F_n^{-1}(u) : = \inf \{t: F_n(t) \geq u\}$ be the corresponding empirical quantile function. Then by~\eqref{eq:CDFLowerBound}, we have for all $ t \in [0, 1]$
\begin{align}
    t + D_n\geq \min\{1, t + D_n\} \geq F_n^{-1}(t). \label{eq:QuantileUpperBound}
\end{align}
A similar argument yields that $F_n^{-1}(t) \geq t - D_n$ for $t \in [0, 1]$. Together with~\eqref{eq:QuantileUpperBound}, we obtain that
\[
D_n \geq |F_n^{-1}(t) - t|, \quad t \in [0, 1].
\]
We can now obtain the final bound by noting
\begin{align*}
\mathcal{W}_2^2(\mu_n, \lambda^d) = \int_{0}^1 (F_n^{-1}(t) - t)^2 \, \mathrm{d}t \leq D_n^2,
\end{align*}
as desired. 
\end{proof}

\begin{lemma}\label{lem:2-WassersteinPartitionBound}
Let $(\mathcal{E}, d)$ be a Polish metric space, and suppose $\mu$ and $\nu$ are Borel probability measures supported on a compact set $S \subseteq \mathcal{E}$. Suppose that $\{S_l\}_{l = 1}^L$ is a partition of $S$ such that $\mu(S_l) = \nu(S_l)$ for all $l \in [L]$. Write $\mu_l := \frac{\mu|_{S_l}}{\mu(S_l)}$ and $\nu_l := \frac{\nu|_{S_l}}{\nu(S_l)}$ for $l \in [L]$, then we have $\mathcal{W}_2^2(\mu, \nu) \leq \sum_{l = 1}^L \mathcal{W}_2^2(\mu_l, \nu_l) \mu(S_l)$.   
\end{lemma}
\begin{proof}
For notational simplicity, write $p_l := \mu(S_l) = \nu(S_l)$. For every $l \in [L]$, let $\pi_l \in \mathcal{C}(\mu_l, \nu_l)$ be an optimal coupling, define $\pi = \sum_{l = 1}^L p_l \pi_l$. We claim that $\pi$ is a coupling between $\mu$ and $\nu$. In fact, for any $A \subseteq \mathcal{Y}$, we have 
\begin{align*}
    \pi(\mathcal{X} \times A) = \sum_{l = 1}^L p_l \pi_l(\mathcal{X} \times A) =\sum_{l = 1}^L p_l \nu_l(A) = \sum_{l = 1}^L \nu(A \cap S_l) = \nu(A). 
\end{align*}
The same argument yields that for any $B \subseteq \mathcal{X}$, $\pi(B \times \mathcal{Y}) = \mu(B)$. Then by construction, it follows that
\begin{align*}
    \mathcal{W}_2^2(\mu, \nu) \leq \int d^2(x, y) \, \mathrm{d}\pi(x, y) = \sum_{l = 1}^L \mathcal{W}_2^2(\mu_l, \nu_l) \mu(S_l), 
\end{align*}
as desired. 
\end{proof}

% \begin{lemma} \label{lem:BoxStarDiscrepancy}
% Let $\lambda^d$ be the Lebesgue measure on $[0, 1]^d$. Suppose ${\bs u}_n:= (u_1, \ldots, u_n)$ is a sequence of points fall in $[0, 1]^d$ and $\mu_n := \frac{1}{n}\sum_{i = 1}^n \delta_{u_i}$. 
% \[
% R(a, b):=\prod_{j=1}^d [a_j,b_j)\subseteq [0,1]^d.
% \]
% Then there exists a constant $C_d >0$ depending only on $d$ such that $|\mu_n(R) - \lambda^d(R)| \leq C_d \mathcal{D}^*({\bs u})$. 
% \end{lemma}
% \begin{proof}
% For each $\varepsilon=(\varepsilon_1,\ldots,\varepsilon_d)\in\{0,1\}^d$, set
% \[
% x_\varepsilon=x_\varepsilon(a,b):=
% \bigl((1-\varepsilon_1)b_1+\varepsilon_1 a_1,\ldots,(1-\varepsilon_d)b_d+\varepsilon_d a_d\bigr),
% \]

%     We only need to note that by the Inclusion-exclusion principle, we have 
%     \[
%     \mu_n(R) = 
%     \]
% \end{proof}

\begin{lemma}\label{lem:InclusionExclusionBox}
Let $\eta$ be a finite signed measure on $[0,1]^d$. For $y=(y_1,\ldots,y_d)\in[0,1]^d$, define
\begin{align}
F_\eta(y):=\eta\biggl(\prod_{j=1}^d [0,y_j)\biggr). \label{eq:MeasureofCube}
\end{align}
For any sequences $a= (a_j)_{j = 1}^d$ and $b = (b_j)_{j = 1}^d$ such that $0\leq a_j\leq b_j\leq 1$, $j=1,\ldots,d$, define the cube
\begin{align}
R(a,b):=\prod_{j=1}^d [a_j,b_j) \label{eq:halfopencube}
\end{align}
For each $\varepsilon=(\varepsilon_1,\ldots,\varepsilon_d)\in\{0,1\}^d$, set
\begin{align}
x_\varepsilon=x_\varepsilon(a,b):=\bigl((1-\varepsilon_1)b_1+\varepsilon_1 a_1,\ldots,(1-\varepsilon_d)b_d+\varepsilon_d a_d\bigr), \label{eq:vertex}
\end{align}
and write $|\varepsilon|:=\varepsilon_1+\cdots+\varepsilon_d$. Then
\[
\eta(R(a,b))=\sum_{\varepsilon\in\{0,1\}^d}(-1)^{|\varepsilon|} F_\eta(x_\varepsilon).
\]
\end{lemma}

\begin{proof}
For each $j=1,\ldots,d$, we have the elementary identity
\[
\bone_{[a_j,b_j)}(t)=\bone_{[0,b_j)}(t)-\bone_{[0,a_j)}(t), \quad t\in[0,1].
\]
Taking the product over \(j=1,\ldots,d\), we obtain
\[
\bone_{R(a,b)}(x)=\prod_{j=1}^d \bigl(\mathbf{1}_{[0,b_j)}(x_j)-\mathbf{1}_{[0,a_j)}(x_j)\bigr), \quad x=(x_1,\ldots,x_d)\in[0,1]^d.
\]
Expanding the product yields
\[
\bone_{R(a,b)}(x) = \sum_{\varepsilon\in\{0,1\}^d} (-1)^{|\varepsilon|}
\bone_{\prod_{j=1}^d [0,(1-\varepsilon_j)b_j+\varepsilon_j a_j)}(x). 
% =
% \sum_{\varepsilon\in\{0,1\}^d}
% (-1)^{|\varepsilon|} \mathbf{1}_{[0,x_\varepsilon)}(x).
\]
Integrating both sides with respect to $\eta$ gives
\[
\eta(R(a,b))=\sum_{\varepsilon\in\{0,1\}^d}(-1)^{|\varepsilon|} F_\eta(x_\varepsilon),
\]
as desired.
\end{proof}

\begin{cor}
\label{cor:BoxDiscByStarDisc}
Suppose ${\bs u_n} = \{u_1, \ldots, u_n\} \subset [0, 1]^d$ be a set of vectors and $\mu_n = \frac{1}{n}\sum_{i = 1}^n \delta_{u_i}$. For every half-open cube $R(a, b)$ defined as in~\eqref{eq:halfopencube} one has
\[
|\mu_n(R(a,b))-\lambda^d(R(a,b))|\leq 2^d\, \mathcal{D}^*({\bs u_n}). 
\]
\end{cor}

\begin{proof}
Apply Lemma~\ref{lem:InclusionExclusionBox} to the finite signed measure $\eta:=\mu_n-\lambda^d$. Then
\[
\eta(R(a,b)) =\sum_{\varepsilon\in\{0,1\}^d}(-1)^{|\varepsilon|}F_\eta(x_\varepsilon),
\]
where $x_\varepsilon$ and $F_\eta$ are defined as in~\eqref{eq:vertex} and~\eqref{eq:MeasureofCube} respectively. Hence
\[
|\eta(R(a,b))| \leq \sum_{\varepsilon\in\{0,1\}^d} |F_\eta(x_\varepsilon)| \leq \sum_{\varepsilon\in\{0,1\}^d} \mathcal{D}^*({\bs u_n})=2^d \mathcal{D}^*({\bs u_n}),
\]
as desired. 
\end{proof}

In the following, we show that condition~\ref{ass:subweibull} of Theorem~\ref{thm:ConsistencyGeneralDim} is implied by a curvature condition on the population Brenier potentials commonly used in the literature.  Let $\varphi_X:\Omega_X\to\mathbb{R}$, $\varphi_Y:\Omega_Y\to\mathbb{R}$ be Brenier potentials such that $\nabla\varphi_X{\#} P^X=\lambda^{d_X}$, $\nabla\varphi_Y{\#}P^Y=\lambda^{d_Y}$. If $\varphi_X$ and $\varphi_Y$ are \textit{$(\beta_i,K_i)$-strongly convex}, in the sense of Definition~\ref{def:stronglyconvex} below, for some $K_i>0$ and $\beta_i\in(1,2]$, $i=1,2$, then the argument of~\citet[Theorem~3]{balakrishnan2025stability} can be adapted to yield a convergence rate of the empirical optimal transport map with similar form as Proposition~\ref{prop:HolderOTConvergence}. Proposition~\ref{prop:stronglymonotone} below shows that this condition is, in fact, more restrictive than Assumption~\ref{ass:subweibull}. Indeed, by Proposition~\ref{prop:stronglymonotone}, $(\beta_i,K_i)$-strong convexity implies that, for all $x,x'\in\Omega_X$ and $y,y'\in\Omega_Y$,
\[
\|\nabla\varphi_X(x)-\nabla\varphi_X(x')\|_2
\geq
\frac{K_1}{2}\|x-x'\|_2^{\beta_1-1},
\qquad
\|\nabla\varphi_Y(y) - \nabla\varphi_Y(y')\|_2
\geq
\frac{K_2}{2}\|y-y'\|_2^{\beta_2-1}.
\]
Since the gradient of the Brenier maps take values in $[0,1]^{d_X}$ and $[0,1]^{d_Y}$, respectively, the left-hand sides are uniformly bounded. Consequently, the above inequalities imply that, in contrast to the tail condition~\ref{ass:subweibull}, the curvature condition effectively restricts the marginals to have bounded supports.

\begin{defn}\label{def:stronglyconvex}
Let $C$ be a nonempty convex set in $\mathbb{R}^n$. For $K >0$ and $\beta \in (1, 2]$, a function $f: C \to \mathbb{R}$ is said to be $(\beta, K)$-strongly convex if for any $x, y \in C$ and $\alpha \in (0, 1)$, it holds that 
\[
f(\alpha x + (1-\alpha)y) \leq \alpha f(x) + (1-\alpha) f(y) - \frac{K\alpha(1-\alpha)}{2}\|x-y\|_2^\beta.
\]
\end{defn}

\begin{prop}\label{prop:strongly-conv-equiv}
Let $f$ be a function differentiable on an open set $\Omega \subseteq \mathbb{R}^n$, and let $C$ be a convex subset of $\Omega$. Suppose $\beta \in (1, 2]$ and $K > 0$ are constants, then the following are equivalent: 
\begin{enumerate}[label = (\roman*), ref = \theprop(\roman*)]
    \item $f$ is $(\beta, K)$-strongly convex on $C$,
    \item\label{prop:stronglyconvex} for any $x, y \in C$
    \[
    f(x) \geq f(y) + \langle \nabla f(y), x- y \rangle + \frac{K}{2}\|x-y\|_2^\beta
    \]
    \item\label{prop:stronglymonotone} for any $x, y \in C$
    \[
    \langle \nabla f (x) - \nabla f(y), x- y\rangle \geq K \|x-y\|_2^{\beta}  
    \]
\end{enumerate}
\end{prop}
\begin{proof}
    $(i) \Rightarrow (ii)$: Let $h = x - y \in C$, then we have for any $\alpha \in (0, 1)$
    \[
    f(y + \alpha h) \leq \alpha f(x) + (1 - \alpha)f(y) - \frac{K\alpha (1-\alpha)}{2}\|x - y\|_2^\beta. 
    \]
    Rearranging gives 
    \[
    f(x) \geq f(y) + \frac{f(y + \alpha h)-f(y)}{\alpha} + \frac{K (1-\alpha)}{2}\|x - y\|_2^\beta.
    \]
    Since $f$ is differentiable, when $\alpha \to 0$, we have $\frac{f(y + \alpha h)-f(y)}{\alpha}  \to \langle \nabla f(y), x - y \rangle$, which implies (ii). 

    $(ii) \Rightarrow (i)$: For any $x, y \in C$ and $\alpha \in (0, 1)$, let $z = \alpha x + (1-\alpha)y$. Then we have 
    \begin{align}
        f(x) \geq f(z) + \langle \nabla f(z), x - z \rangle + \frac{K}{2}\|x-z\|_2^\beta \label{ineq:strongcon1}
    \end{align}
    and 
    \begin{align}
        f(y) \geq f(z) + \langle \nabla f(z), y - z \rangle + \frac{K}{2}\|y-z\|_2^\beta. \label{ineq:strongcon2}
    \end{align}
    Multiplying~\eqref{ineq:strongcon1} by $\alpha$ and~\eqref{ineq:strongcon2} by $(1 - \alpha)$ and adding them up, we obtain that 
    \[
    \alpha f(x)+(1-\alpha )f(y)\geq f(z)+  \frac{K}{2} (\alpha \|x-z\|_2^\beta+(1-\alpha)\|y-z\|_2^\beta).
    \]
    Note that $x - z = (1-\alpha)(x-y)$ and $y - z = \alpha (y - x)$, we have 
    \[
    \alpha f(x)+(1-\alpha )f(y)\geq f(z)+  \frac{K}{2} \alpha (1-\alpha)(\alpha^{\beta-1} +(1-\alpha)^{\beta-1})\|x-y\|_2^\beta.
    \]
    And the result follows by noting that $\alpha^{\beta-1} +(1-\alpha)^{\beta-1} \geq \alpha + 1-\alpha = 1$. 

    $(ii) \Rightarrow (iii)$: By symmetry, we have 
    \[
    f(x) \geq f(y) + \langle \nabla f(y), x- y \rangle + \frac{K}{2}\|x-y\|_2^\beta,
    \]
    and 
    \[
    f(y) \geq f(x) + \langle \nabla f(x), y - x \rangle + \frac{K}{2}\|x-y\|_2^\beta.
    \]
    Therefore, adding two inequalities gives 
    \[
    \langle \nabla f(x) - \nabla f(y), x - y \rangle \geq K \|x- y\|_2^\beta, 
    \]
    which is (iii). 

    $(iii) \Rightarrow (ii)$: for any $x, y \in C$, let $h = x-y \in C$. Write $\psi (\alpha) = f(y + \alpha h)$, then $\psi'(\alpha) = \langle \nabla f(y+ \alpha h), h \rangle$. By (iii), we have 
    \[
    \langle \nabla f(y + \alpha h) - \nabla f(y),  h\rangle \geq K \alpha^{\beta-1} \|h\|_2^\beta.
    \]
    Hence $\psi'(\alpha) - \psi'(0) \geq K \alpha^{\beta-1}\|h\|_2^\beta$. Then function $\gamma(\alpha):= \psi(\alpha) - \psi(0) - \alpha \psi'(0) - \frac{K}{\beta}\|h\|_2^{\beta}\alpha^\beta$ is increasing on $[0, 1]$. Hence $\gamma(1) \geq \gamma(0) = 0$ implies that 
    \[
    f(x)- f(y) - \langle \nabla f(y),x-y\rangle \geq \frac{K}{\beta}\|x-y\|_2^\beta \geq \frac{K}{2}\|x-y\|_2^\beta,
    \]
    as desired. 
\end{proof}

\section{Ancillary lemmas}
The following lemma shows that the $f$-divergence is between the joint distribution and product of marginal distributions is preserved under the Monge--Kantorovich rank transform.

\begin{lemma}
    \label{Lemma:DataProcessing}
    Let $(X,Y)$ be a pair of jointly distributed random variables on $\mathbb{R}^{d_X}\times\mathbb{R}^{d_Y}$. Let $U$ and $V$ be continuous random variables with distribution $P^U$ on $\mathbb{R}^{d_X}$ and $P^V$ on $\mathbb{R}^{d_Y}$ respectively, chosen such that $U\indep V\mid(X,Y)$ and that 
    \begin{equation}
    \label{Eq:DataProcessingCoupling}
    U\in\argmin_{\tilde U\sim P^U} \mathbb{E}\|X-\tilde U\|_2^2\quad \text{and}\quad V\in\argmin_{\tilde V\sim P^V}\mathbb{E}\|Y - \tilde V\|_2^2.
    \end{equation}
    Let $P^{(X,Y)}$ be the joint distribution of $(X,Y)$ with marginals $P^X$ and $P^Y$, and similarly $P^{(U,V)}$, $P^U$, $P^V$ the joint and marginal distributions of $(U,V)$. Then for any convex function $f:\mathbb{R}\to \mathbb{R}\cup\{\infty\}$ such that $f(1)=0$ (cf. Definition~\ref{def:f-divergence}, we have
    \[ 
    D_f(P^{(X,Y)}\,\|\,P^X\otimes P^Y) = D_f(P^{(U,V)}\,\|\,P^U\otimes P^V).
    \]
\end{lemma}
\begin{proof}
    Let $\pi$ be the joint coupling (joint distribution) of $(X,U)$ and $\gamma$ the coupling of $(Y,V)$ defined through the solution of the optimal transport problem in~\eqref{Eq:DataProcessingCoupling}. Let $P^{U\mid X=x}$ and $P^{V\mid Y=y}$ be the corresponding conditional distributions
of $U$ given $X=x$ and $V$ given $Y=y$ respectively. Note that these conditional distributions
are well-defined up to a $P^X$-measure 0 set of $x$-values and $P^Y$-measure 0 set of $y$-values. The fact that $U\indep V\mid (X,Y)$ means that $P^{U\mid X=x}\otimes P^{V\mid Y=y}$ is the conditional distribution of $(U,V)$ given $(X,Y)=(x,y)$. Therefore, we have
\begin{align*}
P^{(U,V)} &= \int_{\mathbb{R}^{d_X}\times \mathbb{R}^{d_Y}} P^{U\mid X=x}\otimes P^{V\mid Y=y}\, \mathrm{d}P^{(X,Y)}(x,y)\\
P^U\otimes P^V &= \int_{\mathbb{R}^{d_X}\times \mathbb{R}^{d_Y}} P^{U\mid X=x}\otimes P^{V\mid Y=y}\, \mathrm{d}(P^X\otimes P^Y)(x,y).
\end{align*}
By the data processing inequality \citep[Theorem~7.4]{polyanskiy2024information}, we thus have
\[  
D_f(P^{(X,Y)}\,\|\,P^X\otimes P^Y) \geq D_f(P^{(U,V)}\,\|\,P^U\otimes P^V).
\]

On the other hand, since $U$ is absolutely continuous with respect to the Lebesgue measure, by Brenier's Theorem \citep[see, e.g.][Theorem~2.12]{villani2021topics}, there exists a convex function $\phi: \mathbb{R}^{d_X}\to \mathbb{R}$ such that $d\pi(x,u) = dP^U(u)\delta_{\{x=\nabla \phi(u)\}}$. In other words, the optimal transport from $U$ to $X$ is the function $\nabla \phi$ (which is $P^U$-almost everywhere uniquely defined), and so $X = \nabla\phi(U)$. Similarly, we have $Y=\nabla \psi(V)$ for some
convex function $\psi:\mathbb{R}^{d_Y}\to\mathbb{R}$. Consequently, we have that conditional on $(U,V)$, $X$ and $Y$ are deterministic, so in particular, conditionally independent. This allows us to run a symmetric argument with the conditional distribution of $(X,Y)$ given $(U,V)$ to obtain a data processing inequality in the reverse direction, thus establishing the desired equality.
\end{proof}

The next lemma shows the stability of the $f$-divergence $D_f(P\,\|\,Q)$ with respect to total-variation perturbation of $P$ and $Q$ when the generator function $f$ is bounded Lipschitz in $[0,\infty)$. 
\begin{lemma}
\label{Lemma:DivergenceStability}
    Suppose $f:[0,\infty)\to [0,M]$ is convex and $L$-Lipschitz. Then for any probability measures $P, Q, P', Q'$ such that $Q'$ is absolutely continuous with respect to $Q$, we have 
    \[  
    |D_f(P\,\|\, Q) - D_f(P'\,\|\, Q')| \leq 2L\,d_{\mathrm{TV}}(P, P') + (4M+L) \sqrt{d_{\mathrm{TV}}(Q, Q')}.
    \]
\end{lemma}
\begin{proof}
    For notational convenience, we write $\epsilon_P := d_{\mathrm{TV}}(P,P')$ and $\epsilon_Q := d_{\mathrm{TV}}(Q,Q')$. Since $f$ is convex and bounded, it must be decreasing on $[0, \infty)$ with $\lim_{t\to\infty} t^{-1}f(t) \to 0$, so singular components of $P$ and $P'$ have no contribution in the $f$ divergence. Let $P_{\mathrm{ac}}$ and $P'_{\mathrm{ac}}$ be the absolutely continuous part of $P$ and $P'$ with respect to $Q$, and let $p, p', q, q'$ be densities of $P_{\mathrm{ac}}, P'_{\mathrm{ac}}, Q, Q'$ with respect to $Q$ (note $q\equiv 1$). Then we have
    \[
    D_f(P\,\|\, Q) = \int f\biggl(\frac{p(x)}{q(x)}\biggr)\,q(x)\, \mathrm{d}Q(x) \quad \text{and}\quad  D_f(P'\,\|\, Q') = \int f\biggl(\frac{p'(x)}{q'(x)}\biggr)\,q'(x)\, \mathrm{d}Q(x),
    \]
    where $f(\infty)$ is interpreted as $\lim_{t\to\infty}f(t)$, which exists since $f$ is decreasing and bounded from below. Thus, we have 
    \begin{align*}
        D_f(P\,\|\, Q) - D_f(P'\,\|\, Q') &= \int f\biggl(\frac{p(x)}{q(x)}\biggr)\,(q(x) - q'(x))\, \mathrm{d}Q(x)\\
        &\quad + \int \biggl\{f\biggl(\frac{p(x)}{q(x)}\biggr) - f\biggl(\frac{p(x)}{q'(x)}\biggr)\biggr\}\,q'(x)\, \mathrm{d}Q(x)\\
       & \quad  + \int \biggl\{f\biggl(\frac{p(x)}{q'(x)}\biggr) - f\biggl(\frac{p'(x)}{q'(x)}\biggr)\biggr\}\,q'(x)\, \mathrm{d}Q(x) =: I_1+ I_2+ I_3.
    \end{align*}
    For the first terms on the right-hand side, we have
    \[
        |I_1| \leq M\int |q(x)-q'(x)|\, \mathrm{d}Q(x) = 2M\epsilon_Q.
    \]
    For the last term, using the fact that $f$ is Lipschitz, we have
    \[  
     |I_3| \leq L\int \biggl|\frac{p(x)}{q'(x)}-\frac{p'(x)}{q'(x)}\biggr| \, q'(x) \, \mathrm{d}Q(x) \leq 2L\epsilon_P.
    \]
    For the second term, writing $\mathcal{A}:=\{x:|q'(x)-q(x)|\leq \epsilon_Q^{1/2}\}$ and using the fact that $q(x) \equiv 1$, we have 
    \begin{align*}
    |I_2| &\leq L \int_{\mathcal{A}} \biggl|\frac{p(x)}{q(x)} - \frac{p(x)}{q'(x)}\biggr|\,q'(x)\, \mathrm{d}Q(x) + M \int_{\mathcal{A}^{\mathrm{c}}} q'(x)\, \mathrm{d}Q(x)\\
    &\leq L\int_{\mathcal{A}} |q(x)-q'(x)|p(x)\, \mathrm{d}Q(x) + M\int |q'(x)-q(x)|\, \mathrm{d}Q(x) + M\int_{\mathcal{A}^{\mathrm{c}}} q(x) \, \mathrm{d}Q(x)\\
    &\leq L\epsilon_Q^{1/2} + M\epsilon_Q + M\epsilon_Q^{1/2},
    \end{align*}
    where we used Markov's inequality in the final step. The desired result is obtained by combining the bounds for $I_1$, $I_2$ and $I_3$.
\end{proof}
The following lemma studies locations of order statistics of a sample drawn from a distribution close in total variation to $\mathrm{Unif}[0,1]$.
\begin{lemma}
    \label{Lemma:DKW}
    Suppose $P$ is a distribution on $[0,1]$ such that $d_{\mathrm{TV}}(P, \mathrm{Unif}[0,1]) \leq \epsilon$. For observations $X_1,\ldots,X_n\iid P$, let $X_{(1)}\leq \cdots\leq X_{(n)}$ denote their order statistics. Then for any $t>0$, we have
\[  
\mathbb{P}\biggl(\max_{i\in[n]} \biggl|X_{(i)} - \frac{i}{n}\biggr| \geq \epsilon + t\biggr) \leq 2e^{-2nt^2}.
\]
\end{lemma}
\begin{proof}
    Let $F$ be the distribution function of $P$ and $F_n$ the empirical c.d.f.\ of $X_1,\ldots,X_n$. We then have
    \begin{align*}
    \max_{i\in[n]} \biggl| X_{(i)} - \frac{i}{n}\biggr| &\leq \max_{i\in[n]} \biggl| X_{(i)} - F(X_{(i)})\biggr| + \max_{i\in[n]} \biggl| F(X_{(i)}) - F_n(X_{(i)})\biggr| \\
    &\leq \sup_{x\in[0,1]}|x-F(x)| + \sup_{x\in[0,1]} |F(x)-F_n(x)| \\
    &\leq d_{\mathrm{TV}}(P, \mathrm{Unif}[0,1]) + \sup_{x\in[0,1]} |F(x)-F_n(x)|.
    \end{align*}
    Therefore, we have 
    \[
    \mathbb{P}\biggl(\max_{i\in[n]} \biggl|X_{(i)} - \frac{i}{n}\biggr| \geq \epsilon + t\biggr) \leq \mathbb{P}\biggl(\sup_{x\in[0,1]} \biggl|F(x) - F_n(x)\biggr| \geq t\biggr)  \leq 2e^{-2nt^2},
    \]
    where the final bound uses the Dvoretzky--Kiefer--Wolfowitz--Massart--Reeve inequality \citep{dvoretzky1956asymptotic,massart1990tight,reeve2024short}.
\end{proof}
The following lemma shows that we can construct disjoint (randomised) intervals around each order statistic of a uniform sample $U_1,\ldots, U_n$ on $[0,1]$, such that after deleting a subset of these intervals, the remaining $U_i$'s are uniformly distributed in the carved out set.
\begin{lemma}
\label{Lemma:CarveOut}
    Suppose $U_1,\ldots,U_n\iid \mathrm{Unif}([0,1])$ and $P_1,\ldots,P_{n+1}\iid \mathrm{Beta}(1/2,1/2)$ are independent. Let $U_{(1)},\ldots,U_{(n)}$ be order statistics of $(U_i)_{i\in[n]}$ with the convention that $U_{(0)}:=0$ and $U_{(n+1)}:=1$ and set $G_i := U_{(i)} - (U_{(i)} - U_{(i-1)})P_i$. Given $\mathcal{I}\subseteq[n]$, define a mapping 
    \[
    g:[0,1]\setminus \bigcup_{i\in \mathcal{I}} [G_i, G_{i+1}) \to[0,1] 
    \]
    by 
    \[
    g(x) = \frac{x - \sum_{i\in\mathcal{I}}(G_{i+1}-G_i)\mathbbm{1}\{x\geq G_{i+1}\}}{1-\sum_{i\in\mathcal{I}}(G_{i+1}-G_i)}.
    \]
    Then, $\{g(U_{(i)}):i\notin \mathcal{I}\}$ are order statistics of an independent and identically distributed sample from $\mathrm{Unif}[0,1]$.
\end{lemma}
\begin{proof}
    Let $M\sim \mathrm{Gamma}(n+1,1)$ be independent from other randomness in the problem. Write $S_i := U_{(i)} M$ and $T_i:=G_i M$, then 
    \[  
    T_1 - S_0, S_1 - T_1, T_2 - S_1, S_2 - T_2, \ldots, T_{n+1}-S_n, S_{n+1}-T_{n+1}
    \]
    are $2n+2$ independent $\mathrm{Gamma}(1/2,1)$ random variables. Moreover,
    \[  
    D := [S_{(0)}, S_{(n+1)}] \setminus \bigcup_{i\in\mathcal{I}} [T_i, T_{i+1}),
    \]
    is exactly the domain of $g$ scaled by $M$. For every $i\in[n+1]$, define its predecessor $\mathrm{pred}(i):=\max\{i'\leq i: i'\notin I\}$. Let $L_i$ denote the Lebesgue measure of $[T_{(\mathrm{pred}(i)+1)}, T_{(i+1)}]\cap D$. We note that each $L_i = (S_i - T_i) + (T_{\mathrm{pred}(i)+1} - S_{\mathrm{pred}(i)})$ is the sum of two independent $\mathrm{Gamma}(1/2,1)$ increments and distinct increments are used to compute $L_i$ for different $i$. Hence,  $L_i: i\notin \mathcal{I}$ are independent $\mathrm{Exp}(1)$ random variables, which is equivalent to the desired result after rescaling.
\end{proof}

The next two lemmas show how the covered area changes under simple operations. First, we show that the covered area decreases by a small amount when we delete narrow horizontal and vertical strips in $[0,1]^2$.

\begin{lemma}
    \label{Lemma:VolumeCompress}
    For $(x_1,y_1),\ldots,(x_n,y_n) \in [0,1]^2$, let $[a,b), [c,d)\subseteq[0,1]$ be intervals such that $[a,b)\cap \{x_1,\ldots,x_n\} = \{x_n\}$ and $[c,d)\cap \{y_1,\ldots,y_n\} = \{y_n\}$. Define $g:[0,1]\setminus [a,b) \to [0, 1-(b-a)]$ by $g(x):= x- (b-a)\mathbbm{1}\{x\geq b\}$ and $h:[0,1]\setminus [c,d) \to [0,1-(d-c)]$ by $h(y):=y-(d-c)\mathbbm{1}\{y\geq d\}$. Then
    \[
    0\leq \mathrm{vol}\biggl(\bigcup_{i=1}^n B((x_i,y_i), r)\biggr) - \mathrm{vol}\biggl(\bigcup_{i=1}^{n-1} B((g(x_i),h(y_i)), r)\biggr)\leq (b-a)+(d-c) + 4r^2.
    \]
\end{lemma}
\begin{proof}
    We define four sets of unions $\mathcal{A}_1 := \cup_{i=1}^n B((x_i,y_i), r)$, $\mathcal{A}_2 := \cup_{i=1}^{n-1} B((x_i,y_i), r)$, $\mathcal{A}_3 := \cup_{i=1}^{n-1} B((g(x_i),y_i), r)$ and $\mathcal{A}_4 := \cup_{i=1}^{n-1} B((g(x_i),h(y_i)), r)$. It is easy to see that $\vol(\mathcal{A}_1)-\vol(\mathcal{A}_2)\in [0, 4r^2]$, hence it suffices to show that $\vol(\mathcal{A}_2)-\vol(\mathcal{A}_3)\in [0, b-a]$ and $\vol(\mathcal{A}_3)-\vol(\mathcal{A}_4)\in [0, d-c]$. We will prove the former, and the latter follows by an essentially identical argument. 

    Let $S(y):= \{x: (x,y)\in \mathcal{A}_2\}$ and $\tilde S(y):=\{x:(x,y)\in\mathcal{A}_3\}$ be horizontal `slices' of $\mathcal{A}_2$ and $\mathcal{A}_3$ respectively. By Fubini's theorem, it suffices to check that $\lambda(S(y))-\lambda(\tilde{S}(y)) \in[0,b-a]$ for all $y$, where $\lambda$ denotes the Lebesgue measure on the real line. Observe that 
    \[  
    S(y) = \bigcup_{i\in[n-1]: |y-y_i|\leq r, x_i < a} [x_i-r,x_i+r] \cup \bigcup_{i\in[n-1]: |y-y_i|\leq r, x_i \geq b} [x_i-r,x_i+r] =: S_1(y) \cup S_2(y).
    \]
    and 
    \begin{align*}
    \tilde{S}(y) &= \bigcup_{i\in[n-1]: |y-y_i|\leq r, x_i < a} [x_i-r,x_i+r] \cup \bigcup_{i\in[n-1]: |y-y_i|\leq r, x_i \geq b} [x_i-(b-a)-r,x_i-(b-a)+r]\\
    &=S_1(y)\cup \{S_2(y) + (a-b)\},
    \end{align*}
    where $S_2(y)+(a-b)$ denotes translation of the set $S_2(y)$ by $a-b$.
    If $S_1(y)\cap S_2(y) = \varnothing$, then $\lambda(S_1(y)\cap \{S_2(y)+(a-b)\}) \leq b-a$, so $\lambda(S(y)) - \lambda(\tilde S(y)) \in [0, b-a]$. On the other hand, if $S_1(y)\cap S_2(y)\neq \varnothing$, let $i_1 = \argmax_{i\in[n-1],|y-y_i|\leq r, x_i < a} x_i$ and $i_2 = \argmin_{i\in[n-1],|y-y_i|\leq r, x_i \geq b} x_i$. We must have $x_{i_1} + r \geq x_{i_2}-r$. Observe that $x_{i_2} - (b-a) - r \geq x_{i_1} - r$ and $x_{i_1}+r \leq x_{i_2}-(b-a) + r$, together with $[x_{i_1}-r,x_{i_1}+r]\subseteq S_1(y)$ and $[x_{i_2}-r,x_{i_2}+r]\subseteq S_2(y)$, we deduce that $S_2(y) + (a-b) \cap [0, a] \subseteq S_1(y)$ and $S_1(y) \cap [a,1-(b-a)] \subseteq S_2(y)+(a-b)$. In particular, we have $\lambda(\tilde S(y)) = \lambda(S(y)) - (b-a)$. This establishes the desired result.
\end{proof}

The following lemma controls the extent of change in the coverage area when we scale both the domain and the radius of each $\ell_\infty$ ball. 
\begin{lemma}
    \label{Lemma:VolumeScale} 
    Fix $a,b\in(0,1)$ and let $f:[0,a]\times [0,b]\to[0,1]^2$ be defined such that $f(x,y) = (x/a,y/b)$.  
    For $(x_1,y_1),\ldots,(x_n,y_n)\in [0,a]\times [0,b]$ and $r,r'\in(0,1/2)$, we have 
    \begin{align*}
    -\biggl(\frac{1}{ab}-1\biggr)nr^2&- 4n\biggl\{(r')^2 - \min\biggl(\frac{r}{a}, r'\biggr)\min\biggl(\frac{r}{b}, r'\biggr) \biggr\}\\
    &\leq \vol\biggl(\bigcup_{i\in[n]}B((x_i,y_i), r)\biggr) - \vol\biggl(\bigcup_{i\in[n]}B(f(x_i,y_i), r')\biggr)\\
    &\leq 4n\biggl\{\max\biggl(\frac{r}{a}, r'\biggr)\max\biggl(\frac{r}{b}, r'\biggr) - (r')^2\biggr\}.
    \end{align*}
\end{lemma}
\begin{proof}
    Let $\mathcal{A}_1 := \cup_{i\in[n]} B((x_i,y_i),r)$, $\mathcal{A}_2 := \cup_{i\in[n]} [x_i/a - r/a, x_i/a+r/a]\times [y_i/b-r/b, y_i/b+r/b]$ and $\mathcal{A}_3:=\cup_{i\in[n]} B(f(x_i,y_i), r')$. We will control $\vol(\mathcal{A}_1)-\vol(\mathcal{A}_3)$ by controlling separately $\vol(\mathcal{A}_1)-\vol(\mathcal{A}_2)$ and $\vol(\mathcal{A}_2)-\vol(\mathcal{A}_3)$. For the former, we observe that $\mathcal{A}_2$ is simply $f(\mathcal{A}_1)$, so 
    \[
    0\leq \vol(\mathcal{A}_2)-\vol(\mathcal{A}_1)\leq \biggl(\frac{1}{ab}-1\biggr)\vol(\mathcal{A}_1) \leq \biggl(\frac{1}{ab}-1\biggr)nr^2.
    \]
    For the latter, we have 
    \[
    \vol(\mathcal{A}_2) - \vol(\mathcal{A}_3) \leq \vol(\mathcal{A}_2 \setminus \mathcal{A}_3)\leq 4n\bigl\{\max(r/a, r')\max(r/b, r') - (r')^2\bigr\}
    \]
    and 
    \[  
    \vol(\mathcal{A}_3) - \vol(\mathcal{A}_2) \leq \vol(\mathcal{A}_3 \setminus \mathcal{A}_2) \leq 4n\bigl\{(r')^2 - \min(r/a, r')\min(r/b, r') \bigr\}
    \]
    The desired result follows by combining the two bounds.
\end{proof}

The following lemma provides an upper bound on the change in vacancy area when we delete a few points from a sample of size $n$. We recall the definition of the vacancy volume in~\eqref{eq:VnGen}.
\begin{lemma}
    \label{Lemma:DeletePoints}
Given $n,m\in\mathbb{N}$ with $m\leq n/2$, let $\bs{X}:=(X_i)_{i\in[n]}$ and $\bs{Y}:=(Y_i)_{i\in[n]}$ be fixed and suppose $(U_i, V_i)_{i\in[n]} \iid \mathrm{Unif}[0,1]^2$ and $(\tilde U_i, \tilde V_i)_{i\in[n-m]} \iid \mathrm{Unif}[0,1]^2$. Write $\bs{U}:=(U_i)_{i\in[n]}$, $\bs{V}:=(V_i)_{i\in[n]}$, $\tilde{\bs{U}}:=(\tilde U_i)_{i\in[n-m]}$, $\tilde{\bs{V}}:=(\tilde V_i)_{i\in[n-m]}$, $\tilde{\bs{X}}:=(X_i)_{i\in[n-m]}$ and $\tilde{\bs{Y}}:=(Y_i)_{i\in[n-m]}$. There exists a coupling between $(\bs{U},\bs{V})$ and $(\tilde{\bs{U}},\tilde{\bs{V}})$ such that for every $t\in[0, \frac{n}{2m}-1]$ the following holds with probability at least $1-2e^{-t^2m/(2+4t/3)}$:
    \[
    \frac{(5+3t)m}{n}\leq \mathcal{V}\Bigl(\bs{X},\bs{Y},\bs{U},\bs{V}; \frac{1}{2\sqrt{n}}\Bigr)  -  \mathcal{V}\Bigl(\tilde{\bs{X}},\tilde{\bs{Y}},\tilde{\bs{U}},\tilde{\bs{V}};\frac{1}{2\sqrt{n-m}}\Bigr) \leq \frac{(5+5t)m}{n},
    \]
    where, by convention, $U_{(0)} = 0$ and $U_{(n+1)} = 1$.
\end{lemma}
\begin{proof}
Let $P_1,\ldots,P_{n+1},Q_1,\ldots,Q_{n+1}\iid \mathrm{Beta}(1/2,1/2)$ be independent from other randomness in the problem and define $G_i:=U_{(i)}-(U_{(i)}-U_{(i-1)})P_i$ and $H_i:=V_{(i)}-(V_{(i)}-V_{(i-1)})Q_i$. Let $r^X_i:=\sum_{i'=1}^n\mathbbm{1}\{X_{i'}\leq X_i\}$ and $r^Y_i:=\sum_{i'=1}^n\mathbbm{1}\{Y_{i'}\leq Y_i\}$ for $i\in[n]$. For $r\in[m]$, define
\begin{align*}
g_r: {}&[0,1]\setminus \bigcup_{i=n-r+1}^{n} [G_{r^X_i}, G_{r^X_i+1}) \to \biggl[0,1 - \sum_{i=n-r+1}^{n} (G_{r^X_i+1}-G_{r^X_i})\biggr]\\
h_r:{}&[0,1]\setminus \bigcup_{i=n-r+1}^{n} [H_{r^Y_i}, H_{r^Y_i+1}) \to \biggl[0,1 - \sum_{i=n-r+1}^{n} (H_{r^Y_i+1}-H_{r^Y_i})\biggr]
\end{align*}
by
\begin{align*}
g_r(x) &:= x-\sum_{i=n-r+1}^{n} (G_{r^X_i+1}-G_{r^X_i})\mathbbm{1}\{x\geq G_{r^X_i+1}\},\\
h_r(y) &:= y - \sum_{i=n-r+1}^{n} (H_{r^Y_i+1}-H_{r^Y_i})\mathbbm{1}\{y\geq H_{r^Y_i+1}\}.
\end{align*}
Now, define 
\begin{align*}
g:[0,1]\setminus\bigcup_{i=n-m+1}^{n} [G_{r^X_i}, G_{r^X_i+1}) \to [0,1]\quad\text{and}\quad h:[0,1]\setminus\bigcup_{i=n-m+1}^{n} [H_{r^Y_i}, H_{r^Y_i+1}) \to [0,1]
\end{align*}
via
\begin{align*}
    g(x):=g_m(x) / g_m(1) \quad \text{and}\quad 
    h(y):=h_m(y) / h_m(1).
\end{align*}
Intuitively, $g_r$ can be seen as the bijection that compresses the carved out interval $[0,1]\setminus \cup_{i=n-r+1}^{n} [G_{r^X_i}, G_{r^X_i+1})$ to a contiguous interval and $g$ further rescales the compressed interval to $[0,1]$ after deleting intervals associated with $X_{n-m+1},\ldots, X_{n}$. Similarly, $h_r$ and $h$ represent compression of the carved out $y$-interval and its rescaled version. 

By Lemma~\ref{Lemma:CarveOut}, $\tilde U_i := g(U_i)$ and $\tilde V_i:=h(V_i)$ for $i\in[n-m]$ satisfies $(\tilde U_i,\tilde V_i)_{i\in[n-m]} \iid\mathrm{Unif}[0,1]^2$. We will establish the desired bound under this coupling.

Let $R_i \equiv (R^X_i, R^Y_i) := (U_{r^X_i}, V_{r^Y_i})$ for $i\in[n]$, then 
\[  
1- \mathcal{V}\Bigl(\bs{X},\bs{Y},\bs{U},\bs{V};\frac{1}{2\sqrt{n}}\Bigr) = \mathrm{vol} \biggl(\bigcup_{i\in[n]}B\Bigl(R_i, \frac{1}{2\sqrt{n}}\Bigr)\biggr).
\]
Under the present coupling, we also have that 
\[  
1 - \mathcal{V}\Bigl(\tilde{\bs{X}}, \tilde{\bs{Y}}, \tilde{\bs{U}}, \tilde{\bs{V}}; \frac{1}{2\sqrt{n-m}}\Bigr)  = \mathrm{vol}\biggl(\bigcup_{i\in[n-m]} B\Bigl((g(R^X_i), h(R^Y_i)), \frac{1}{2\sqrt{n-m}}\Bigr)\biggr).
\]
To establish the desired result, we will construct a few sets whose volume interpolates the two volumes on the right-hand side of the two previous displays. Specifically, define
\begin{align*}
\mathcal{A}_0&:=\bigcup_{i\in[n]}B\Bigl(R_i, \frac{1}{2\sqrt{n}}\Bigr),\\
\mathcal{A}_r&:=\bigcup_{i\in[n-r]}B\Bigl((g_r(R_i^X), h_r(R_i^Y)), \frac{1}{2\sqrt{n}}\Bigr)\\
\mathcal{B} &:= \bigcup_{i\in[n-m]} B\Bigl((g(R^X_i), h(R^Y_i)), \frac{1}{2\sqrt{n-m}}\Bigr).
\end{align*}
By Lemma~\ref{Lemma:VolumeCompress}, 
\[
0\leq \vol(\mathcal{A}_0)-\vol(\mathcal{A}_m) = \sum_{r=1}^m \bigl\{\mathrm{vol}(\mathcal{A}_{r-1}) - \mathrm{vol}(\mathcal{A}_r)\bigr\}\leq (1-g_m(1)) + (1-h_m(1)) +\frac{m}{n}.
\]
By Lemma~\ref{Lemma:VolumeScale},
\begin{align*}
\min\biggl\{\frac{1}{g_m(1)}&, \sqrt\frac{n}{n-m}\biggr\}\min\biggl\{ \frac{1}{h_m(1)},\sqrt\frac{n}{n-m}\biggr\}-\frac{n}{n-m}-\biggl(\frac{1}{g_m(1)h_m(1)}-1\biggr)\nonumber\\
&\leq \mathrm{vol}(\mathcal{A}_{m}) - \mathrm{vol}(\mathcal{B}) \leq \max\biggl\{\frac{1}{g_m(1)}, \sqrt\frac{n}{n-m}\biggr\}\max\biggl\{ \frac{1}{h_m(1)},\sqrt\frac{n}{n-m}\biggr\}-\frac{n}{n-m}.
\end{align*}
Since $1-g_m(1), 1-h_m(1)\iid \mathrm{Beta}(m,n+1-m)$,
by Lemma~\ref{Lemma:Beta}, there is an event $\Omega$ with probability at least $1 - 2e^{-t^2m/(2+4t/3)}$ on which 
\[
\max\{1-g_m(1), 1-h_m(1)\} \leq \frac{(1+t)m}{n} =: \delta.
\]
Writing $\delta' := m/n$, so that $\delta'\leq \delta \leq 1/2$ under the assumption. We have the event $\Omega$ that
\[  
\vol(\mathcal{A}_0)-\vol(\mathcal{B}) = \vol(\mathcal{A}_0) - \vol(\mathcal{A}_m) + \vol(\mathcal{A}_m) - \vol(\mathcal{B}) \leq 2\delta + \delta' + \frac{1}{(1-\delta)^2} - \frac{1}{1-\delta'}\leq 5\delta.
\]
Similarly, on $\Omega$, we also have
\[  
\vol(\mathcal{B}) - \vol(\mathcal{A}_0) =  \vol(\mathcal{A}_m) - \vol(\mathcal{A}_0) +  \vol(\mathcal{B}) -  \vol(\mathcal{A}_m) \leq \frac{1}{(1-\delta)^2}-1+\frac{1}{1-\delta'}-1 \leq 3\delta+2\delta'.
\]
The desired result follows from combining the two bounds above.
\end{proof}

We will often employ the `Poissonisation trick' to analyse the vacancy volume, i.e.\ studying the vacancy associated with $N\sim \mathrm{Poi}(n)$ instead of $n$ data points. The following lemma studies the difference in vacancy due to Poissonisation. 

\begin{lemma}
    \label{Lemma:Poissonisation}
    Let $P^{(X,Y)}$ be a probability measure on $\mathbb{R}^2$. Let $(X_i,Y_i)_{i\in\mathbb{N}} \iid P^{(X,Y)}$ and $N\sim \mathrm{Poi}(n)$.  Suppose we have $(U_i,V_i)_{i\in\mathbb{N}}\stackrel{\mathrm{iid}}\sim \mathrm{Unif}[0,1]^2$ and $(\tilde U_i, \tilde V_i)_{i\in\mathbb{N}}\stackrel{\mathrm{iid}}\sim \mathrm{Unif}[0,1]^2$.
    There is a coupling between $(U_i,V_i)_{i\in [n]}$ and $(N, (\tilde U_i, \tilde V_i)_{i\in[N]})$ such that 
    \begin{align*}
     \mathcal{V}\Bigl((X_i)_{i\in[n]}, (Y_i)_{i\in[n]}, &(U_i)_{i\in[n]}, (V_i)_{i\in[n]}; \frac{1}{2\sqrt{n}}\Bigr) \\
     & - \mathcal{V}\Bigl((X_i)_{i\in[N]}, (Y_i)_{i\in[N]}, (\tilde U_i)_{i\in[N]}, (\tilde V_i)_{i\in[N]}; \frac{1}{2\sqrt{N}}\Bigr) = O_p(n^{-1/2})
    \end{align*}
\end{lemma}
\begin{proof}
We define the coupling conditionally on $N$ as follows. If $N \geq n$, we can view $(X_i,Y_i)_{i\in[n]}$ as obtained from $(X_i,Y_i)_{i\in N}$ by deleting the last $N-n$ points, and $(U_i, V_i)_{i\in[n]}$ can be defined conditionally on $(\tilde U_i,\tilde V_i)_{i\in[N]}$ by Lemma~\ref{Lemma:DeletePoints}. If $N < n$, we view $(X_i,Y_i)_{i\in N}$ as obtained from $(X_i,Y_i)_{i\in[n]}$ by deleting the last $n-N$ points and again obtain a conditional joint distribution between $(\tilde U_i,\tilde V_i)_{i\in[N]}$ and $(U_i,V_i)_{i\in[n]}$ by Lemma~\ref{Lemma:DeletePoints}.

Now, given any $\epsilon > 0$, we can find $C>0$ such that $\mathbb{P}(|N-n|\geq Cn^{1/2}) \leq \epsilon/2$. Conditional on $N=n'$ for $n'\in [n-Cn^{1/2}, n+Cn^{1/2}]$, by Lemma~\ref{Lemma:DeletePoints} under the current coupling, we have with (conditional) probability at least $1-2e^{-3Cn^{1/2}/10}$ that 
\begin{align*}
    \biggl|\mathcal{V}\Bigl((X_i)_{i\in[n]}, (Y_i)_{i\in[n]}, &(U_i)_{i\in[n]}, (V_i)_{i\in[n]}; \frac{1}{2\sqrt{n}}\Bigr) \\
    & - \mathcal{V}\Bigl((X_i)_{i\in[N]}, (Y_i)_{i\in[N]}, (\tilde U_i)_{i\in[N]}, (\tilde V_i)_{i\in[N]}; \frac{1}{2\sqrt{N}}\Bigr)\biggr| \leq \frac{10}{n^{1/2}}.
\end{align*}
Choosing $n$ sufficiently large so that $2e^{-3Cn^{1/2}/10}<\epsilon/2$, and integrating over $\mathbb{N}$, we have with probability at least $1-\epsilon$ that the desired vacancy volumes have difference bounded by $10n^{-1/2}$.
\end{proof}

For any finite measure $\lambda$ on $\mathbb{R}^2$, we use $\mathrm{PP}(\lambda)$ to denote the Poisson point process with intensity $\lambda$. Recall that for finite measure, the Poisson point process can be identified as $N\sim \mathrm{Poi}(\lambda(\mathbb{R}^2))$ points drawn independently (conditionally on $N$) from the probability measure $\lambda / \lambda(\mathbb{R}^2)$. The following lemma studies the effect on vacancy due to thinning the Poisson point process.

\begin{lemma}
    \label{Lemma:Thinning} Given two finite measures $\mu$ and $\nu$ on $\mathbb{R}^2$ such that $\mu\leq \nu$, $\nu(\mathbb{R}^2)=n$ and $\mu(\mathbb{R}^2)=(1-\epsilon)n$. There
    exists a coupling between $((X_i,Y_i)_{i\in[M]}, (U_i,V_i)_{i\in\mathbb{N}}) \sim \mathrm{PP}(\mu)\otimes \mathrm{Unif}[0,1]^{\otimes\mathbb{N}}$ and $((\tilde X_i,\tilde Y_i)_{i\in[N]}, (\tilde U_i,\tilde V_i)_{i\in\mathbb{N}}) \sim \mathrm{PP}(\nu)\otimes \mathrm{Unif}[0,1]^{\otimes\mathbb{N}}$ such that 
    \begin{align*}
    \biggl|\mathcal{V}\Bigl((X_i)_{i\in[M]}, &(Y_i)_{i\in[M]}, (U_i)_{i\in[M]}, (V_i)_{i\in[M]}; \frac{1}{2\sqrt{M}}\Bigr) \\
    & - \mathcal{V}\Bigl((\tilde X_i)_{i\in[N]}, (\tilde Y_i)_{i\in[N]}, (\tilde U_i)_{i\in[N]}, (\tilde V_i)_{i\in[N]}; \frac{1}{2\sqrt{N}}\Bigr)\biggr| \leq 10\epsilon + O_p(n^{-1/2}).
    \end{align*}
\end{lemma}
\begin{proof}
    Since $\mu\leq \nu$, we have that $\mu$ is absolutely continuous with respect to $\nu$ and the Radon--Nikodym derivative satisfies $0\leq d\mu/d\nu\leq 1$ $\nu$-almost everywhere. We can define a coupling between $\mathrm{PP}(\mu)$ and $\mathrm{PP}(\nu)$ via Poisson thinning, i.e.\ we can define $(X_1,Y_1),\ldots,(X_M,Y_M)$ to be a subset of $(\tilde X_1,\tilde Y_1),\ldots,(\tilde X_N,\tilde Y_N)$ such that conditional on $(\tilde X_i,\tilde Y_i)_{i\in[N]}$, each point $(\tilde X_i, \tilde Y_i)$ is selected into the subset with probability $d\mu/d\nu(\tilde X_i, \tilde Y_i)$. Note in particular that under such a coupling, $M\leq N$ and $N-M\sim \mathrm{Poi}(\epsilon n)$. 
    We can then define a conditional joint distribution (conditional on $(X_i,Y_i)_{i\in[M]}$ and $(\tilde X_i,\tilde Y_i)_{i\in[M]}$) between $(U_i,V_i)_{i\in\mathbb{N}}$ and $(\tilde U_i,\tilde V_i)_{i\in\mathbb{N}}$ by first constructing a conditional coupling between $(U_i, V_i)_{i\in [M]}$ and $(\tilde U_i,\tilde V_i)_{i\in [N]}$ as in Lemma~\ref{Lemma:DeletePoints} and then set $(U_i,V_i)_{i > M}$ and $(\tilde U_i, \tilde V_i)_{i > N}$ to be independent.

    Under this coupling, we have $N\sim \mathrm{Poi}(n)$ and $N-M\sim\mathrm{Poi}(n\epsilon)$. So given any $\delta > 0$, there exists $C > 0$ (depending on $\delta$ and $\epsilon$) and an event with probability at least $1-\delta/2$ on which 
    \[
    N \in [n-C\sqrt{n}, n+C\sqrt{n}] \quad \text{and}\quad  N-M \in [\epsilon n-C\sqrt{n}, \epsilon n + C\sqrt{n}].
    \]
    Conditional on any realisation of $(X_i,Y_i)_{i\in[M]}$ and $(\tilde X_i,\tilde Y_i)_{i\in [N]}$ on this event, we can apply Lemma~\ref{Lemma:DeletePoints} to obtain that under the current coupling, we have with (conditional) probability at least $1-2e^{-3Cn^{1/2}/10}$ that 
    \begin{align*}
    \biggl|\mathcal{V}\Bigl((X_i)_{i\in[M]}, (Y_i)_{i\in[M]}, &(U_i)_{i\in[M]}, (V_i)_{i\in[M]}; \frac{1}{2\sqrt{M}}\Bigr) \\
    & \hspace{-2cm}- \mathcal{V}\Bigl((\tilde X_i)_{i\in[N]}, (\tilde Y_i)_{i\in[N]}, (\tilde U_i)_{i\in[N]}, (\tilde V_i)_{i\in[N]}; \frac{1}{2\sqrt{N}}\Bigr)\biggr| \leq \frac{10(\epsilon n + C\sqrt{n})}{n-C\sqrt{n}},
\end{align*}
which can be further upper bounded by $10\epsilon + 2C n^{-1/2}$ for sufficiently large $n$. Also, for large $n$, we have $2e^{-3Cn^{1/2}/10}\leq \delta/2$, the desired result follows after integrating over realisation of $(X_i,Y_i)_{i\in[M]}$ and $(\tilde X_i,\tilde Y_i)_{i\in [N]}$.
\end{proof}

The following lemma bounds the symmetric difference between two cubes in terms of the $\ell_\infty$ distance between their centers. This will be useful in Proposition~\ref{prop: VonRawData} and Proposition~\ref{prop: VonRawData_deterministic} for controlling the change in vacancy volume when we translate the centers of the cubes by a small amount.

\begin{lemma}\label{lem: symdiff}
Let $B(0, \gamma)$ and $B(h, \gamma)$ be two $d$-dimensional cubes where $h \in \mathbb{R}^d$ and $\gamma >0 $, we have
\begin{align}
    \mathrm{vol}\bigl(B(0, \gamma) \Delta B(h, \gamma)\bigr) \leq 2(2\gamma)^{d-1} \|h\|_{1}. \label{ineq: symdiffbound}
\end{align}
\end{lemma}
\begin{proof}
We may assume without loss of generality that $h = (h_1,\ldots,h_d)^\top \in [0,\infty)^d$. Define $S_k = \{(x_1,\ldots,x_d)^\top\in \mathbb{R}^d: x_k \in [\gamma, \gamma + h_k],  x_j \in [-\gamma + h_j, \gamma + h_j],\; \forall\, j\neq k\}$. Then we have $B(h,\gamma)\setminus B(0,\gamma) \subseteq \cup_{k\in[d]} S_k$. Hence 
\[
\mathrm{vol}(B(h,\gamma) \setminus B(0,\gamma)) \leq \sum_{k\in[d]} (2\gamma)^{d-1}h_k = \|h\|_1 (2\gamma)^{d-1}.
\]
The same bounds hold for $\mathrm{vol}(B(0,\gamma)\setminus B(h,\gamma))$ by symmetry. Combining the two gives the desired inequality. 
\end{proof}

The following lemma provides a characterisation of the expectation and the variance of the vacancy volume when constructed from random samples from the uniform distribution on unit cube.

\begin{lemma}
\label{le: VacancyLimitRandomRef}
Suppose $R_1,\ldots,R_n \stackrel{\mathrm{iid}}{\sim}\mathrm{Unif}([0,1]^d)$. For $\delta \in [0, 1/2)$, let $B$ be a compact set symmetric around $0$ with $\vol(B) = \delta$ and define
    \[
    \mathcal{V}:=\vol\Bigl([0,1]^d\setminus \bigcup_{i=1}^n (R_i + B)\Bigr),
    \]
    where $R_i + B$ denotes the Minkowski sum of the sets. 
    If $n\delta\to q$, then as $n\to\infty$, we have 
    \[
    \mathbb{E}(\mathcal{V})  = (1-\delta)^n \to e^{-q} \quad\text{and}\quad \var(\mathcal{V}) \to 0.
    \]
    If in addition, $B = B(0, \gamma)$ (i.e.\ $\delta = (2\gamma)^d$) for some $\gamma < 1/4$, then we have 
    \[  
    \var(\mathcal{V}) = \sum_{r=2}^n \binom{n}{r}(1-2\delta)^{n-r}\biggl\{\frac{2^d\delta^{r+1}}{(r+1)^d} - \delta^{2r}\biggr\} = (1+o(1))\frac{qe^{-2q}}{n}\sum_{r=2}^\infty \frac{q^r}{r!}\biggl(\frac{2}{r+1}\biggr)^d.
    \]
\end{lemma}
\begin{proof}
    Draw $W_1, W_2 \iid \mathrm{Unif}[0,1]^d$ independent of all other randomness. We have
    \begin{align*}
        \E (\mathcal{V}) =& \Prob\Bigl(W_1 \not\in \bigcup_{i = 1}^n (R_i + B)\Bigr) =\E \biggl\{\Prob\Bigl(\bigcap_{i = 1}^n\{R_i \not\in (W_1 + B) \} \Bigm| W_1\Bigr)\biggr\} = (1 - \delta)^n \to e^{-q}.
    \end{align*}
    To control the variance, we start with
    \begin{align}
        \E (\mathcal{V}^2) =& \Prob\Bigl(W_j \not\in \bigcup_{i = 1}^n (R_i+B) \; \forall \, j\in\{1,2\}\Bigr) = \E \biggl[\Prob\Bigl(\bigcap_{i = 1}^n\{R_i \not\in (W_1+B) \cup (W_2+B)\} \Bigm| W_1, W_2\Bigr)\biggr] \nonumber\\ 
        =& \E \Bigl[\bigl\{1 -2\delta+ \vol\bigl((W_1+B)\cap (W_2+B)\bigr)\bigr\}^n\Bigr].\label{Eq:VarianceIntermediate}
    \end{align}
    Since $\vol\bigl((W_1+B)\cap (W_2+B)\bigr) \xrightarrow{\mathrm{a.s.}} 0$, by the Dominated Convergence Theorem, we have  
    \[
    \lim_{n\to\infty} \E(\mathcal{V}^2) = \lim_{n\to\infty}\{1 - 2\delta\}^n = e^{-2q}.
    \]
    Consequently $\var(\mathcal{V}) = \E(\mathcal{V}^2) - (\E \mathcal{V})^2\to 0$. 

    When $B = B(0,\gamma)$ is a cube, we can obtain an explicit expression of the variance by examining moments of $H:=\vol\bigl((W_1+B)\cap (W_2+B)\bigr)$ in detail. Specifically, noting that $H$ measures the volume of a hypercube with independent side lengths each having distribution $4\gamma\mathrm{Unif}[0,2\gamma] + (1-4\gamma)\delta_0$, we get for any $r\in\mathbb{N}$ that 
    \[  
    \mathbb{E}(H^r) =  \biggl\{2\int_0^{2\gamma} t^r\,dt\biggr\}^d = \biggl\{\frac{2(2\gamma)^{r+1}}{r+1}\biggr\}^d = \frac{2^d \delta^{r+1}}{(r+1)^d}.
    \]
    Consequently, from~\eqref{Eq:VarianceIntermediate}, we have 
    \begin{equation}
    \mathbb{E}(\mathcal{V}^2) = \sum_{r=0}^n \binom{n}{r}(1-2\delta)^{n-r}\mathbb{E}(H^r) = (1-2\delta)^n + \sum_{r=1}^n\binom{n}{r}\frac{(1-2\delta)^{n-r}2^d \delta^{r+1}}{(r+1)^d}.\label{Eq:VacancySecondMoment}
    \end{equation}
    On the other hand, we have 
    \begin{equation}
    \label{Eq:VacancySqrFirstMoment}
    \{\mathbb{E}(\mathcal{V})\}^2 = (1-\delta)^{2n} = \sum_{r=0}^n\binom{n}{r}(1-2\delta)^{n-r}\delta^{2r}.
    \end{equation}
    Combining~\eqref{Eq:VacancySecondMoment} and~\eqref{Eq:VacancySqrFirstMoment}, we have 
    \[  
    \var(\mathcal{V}) = \sum_{r=2}^n \binom{n}{r}(1-2\delta)^{n-r}\biggl\{\frac{2^d\delta^{r+1}}{(r+1)^d} - \delta^{2r}\biggr\}.
    \]
We next compute the asymptotic variance. Since $\sum_{r=2}^n\binom{n}{r}(1-2\delta)^{n-r}\delta^{2r} \leq \sum_{r=2}^\infty (n\delta^2)^r = O(n^{-2})$, we have 
\begin{align*}
    \delta^{-1}\var(\mathcal{V}) &= \sum_{r=2}^n \binom{n}{r}(1-2\delta)^{n-r}\biggl(\frac{2}{r+1}\biggr)^d \delta^{r} + O(n^{-1})\\
    &=(1-2\delta)^n\sum_{r=2}^n\binom{n}{r} \frac{\delta^r}{(1-2\delta)^r}\biggl(\frac{2}{r+1}\biggr)^d + O(n^{-1})\\
    &=(1+o(1))e^{-2q} \sum_{r=2}^n \frac{q^r}{r!}\biggl(\frac{2}{r+1}\biggr)^d = (1+o(1))e^{-2q} \sum_{r=2}^\infty \frac{q^r}{r!}\biggl(\frac{2}{r+1}\biggr)^d
\end{align*}
which implies the desired limit since $n\delta\to q$.
\end{proof}

\begin{lemma}\label{lmm:ConditionalDistributionCLT}
Let $(M_n)_n$ and $(L)_n$ be sequences of random variables such that $M_n \dto \mathcal{N}(\mu, \alpha^2)$ and $L_n \pto \beta^2$. Let $\mathcal{F}_n$ be the sigma-algebra generated by $(M_i)_{i \leq n}$ and $(L_i)_{i \leq n}$. If $(X_n)_n$ is a sequence of random variables such that 
\begin{align}
    \E\sup_{-\infty < x < \infty} \biggl|\Prob\biggl(\frac{X_n - M_n}{\sqrt{L_n}} \leq x \biggm| \mathcal{F}_n\biggr) - \Phi(x)\biggr| \to 0, \quad n \to +\infty. \label{limit: ConditionalDistributionCLT} 
\end{align}
Then we have $X_n \dto \mathcal{N}(\mu, \alpha^2 + \beta^2)$.
\end{lemma}

\begin{proof}
    For any $x \in \mathbb{R}$, we have 
    \begin{align}
        \sup_{x \in \mathbb{R}}\biggl|\Prob\Bigl(X_n & \leq \mu + x \sqrt{\alpha^2 + \beta^2}\Bigr) - \Phi(x)\biggr| \notag \\ 
        =& \sup_{x \in \mathbb{R}}\biggl|\E\biggl\{ \Prob\biggl(\frac{X_n - M_n}{\sqrt{L_n}}\leq \frac{\mu - M_n + x \sqrt{\alpha^2 + \beta^2}}{\sqrt{L_n}} \biggm| \mathcal{F}_n\biggr)\biggr\} - \Phi(x)\biggr| \notag \\
        = & \sup_{x \in \mathbb{R}}\biggl|\E\biggl\{\Phi\biggl(\frac{\mu -M_n + x \sqrt{\alpha^2 + \beta^2}}{\sqrt{L_n}}\biggr)\biggr\}-\Phi(x)\biggr| + o(1) \label{ineq: MeanVarAsymptotic}, 
    \end{align}
    where we use condition~\eqref{limit: ConditionalDistributionCLT} in the  final step. By Slutsky's theorem, we have for each $x\in\mathbb{R}$ 
    \[
    \frac{\mu - M_n + x \sqrt{\alpha^2 + \beta^2} }{\sqrt{ L_n}} \dto \mathcal{N}\biggl(\frac{x \sqrt{\alpha^2 + \beta^2}}{\beta}, \frac{\alpha^2}{\beta^2}\biggr).
    \]
    Consequently, we have for $Z\sim \mathcal{N}(0,1)$ independent of all other randomness in the lemma that
    \begin{align}
        \E\biggl\{\Phi\biggl(\frac{\mu - M_n + x \sqrt{\alpha^2 + \beta^2} }{\sqrt{ L_n}}\biggr)\biggr\} \notag &= \Prob\biggl(Z - \frac{\mu- M_n + x \sqrt{\alpha^2 + \beta^2} }{\sqrt{L_n}} \leq 0\biggr) = \Phi(x) + o(1)\notag.
    \end{align}
    The conclusion holds by combining the above with~\eqref{ineq: MeanVarAsymptotic}, and using \citet[Lemma 3, pp.265]{chow1988probability}.
\end{proof}

\begin{lemma}
\label{lmm:AlphaBetaLambda}
    For $\alpha^2_\lambda$ and $\beta^2_\lambda$ defined in~\eqref{eq:alpha} and \eqref{eq:beta}, as $\lambda\to\infty$, we have
    \begin{align*}
        \alpha_\lambda^2 \to 0, \qquad \beta_\lambda^2 \to \beta^2.
    \end{align*}
    where $\beta^2$ is a positive constant. 
\end{lemma}

\begin{proof}
    Note that $\alpha_\lambda^2 = \alpha^2_{1, \lambda} - \alpha^2_{2, \lambda}$ where 
    \begin{align*}
        \alpha^2_{1, \lambda} = \frac{\lambda^{2d}}{e^2 2^d(\lambda+2)^d}\Bigl(e^{(\lambda/2+1)^{-d}} - 1\Bigr),\qquad \alpha^2_{2, \lambda} = \frac{\lambda^{2d}}{e^2(\lambda+2)^{2d}}.
    \end{align*}
    We show that as $\lambda\to\infty$, $\alpha_\lambda^2$ converges to 0. First note that $\alpha^2_{1, \lambda} = e^{-2} + O(\lambda^{-d})$ and $\alpha^2_{2, \lambda} = e^{-2} + O(\lambda^{-1})$. Therefore $\alpha_\lambda^2\to 0$ as $\lambda\to\infty$.

    For $\beta_\lambda^2$, note that as $\lambda\to\infty$ we have
    \[
        \beta_\lambda^2\rightarrow e^{-2}(C_d - 1) = \beta^2 > 0
    \]
    as desired.
\end{proof}

\begin{lemma}
\label{le: MultinCovariance}
    For $n, L \in\mathbb{N}$, let $p:=1/L$ and suppose $(N_1,\ldots,N_L)\sim \mathrm{Multin}(n; (p,\ldots,p))$. Consider the asymptotic regime where $n\to\infty$ and $L$ is fixed. Suppose $a,b\geq 1$ satisfies $p(a-1) = O(1/n)$ and $p(b-1) = O(1/n)$, then for any and $\ell,k\in[L]$, we have
    \[
    \mathrm{Cov}(a^{N_\ell}, b^{N_k}) = O(n^{-2}).
    \]
\end{lemma}
\begin{proof}
We first assume that $\ell\neq k$. We write $\alpha = p(a-1)$ and $\beta = p(b-1)$ for simplicity.
Using the moment generating function of the Multinomial distribution, we observe that,
\begin{align*}
\mathbb{E}(a^{N_\ell}) &= (1 + \alpha)^n\\
\mathbb{E}(a^{N_\ell}b^{N_k}) &= (1 + \alpha + \beta)^n.
\end{align*}
Using the above identities and the Taylor expansion 
\begin{align*}
    |\mathrm{Cov}(a^{N_\ell}, b^{N_k})| &= |(1 + \alpha + \beta)^n - (1 + \alpha+\beta + \alpha\beta)^n|\\
    &= \alpha\beta\sum_{i=0}^{n-1}(1+\alpha+\beta)^{n-1-i}(\alpha\beta)^i \\
    &\leq \alpha\beta (1+\alpha+\beta)^n \sum_{i=0}^{n-1}(\alpha\beta)^i = O(\alpha\beta) = O(n^{-2}).
\end{align*}
It remains to check the case where $\ell = k$. For this, define $\eta = p(ab-1)$,
\begin{align*}
    |\mathrm{Cov}(a^{N_\ell}, b^{N_k})| &= \{1 + p(ab-1)\}^n -\{1 + p(a-1)\}^n \{1 + p(b-1)\}^n\\
    &=(1+\alpha\beta/p + \alpha + \beta)^n - (1+\alpha+\beta+\alpha\beta)^n\\
    &=\alpha\beta(1/p-1)\sum_{i=0}^{n-1}(1+\alpha+\beta+\alpha\beta)^{n-1-i}(\alpha\beta/p-\alpha\beta)^i\\
    &\leq \alpha\beta(1/p-1)(1+\alpha)^n(1+\beta)^n\sum_{i=0}^{n-1}(\alpha\beta/p-\alpha\beta)^i=O(\alpha\beta) = O(n^{-2}),
\end{align*}
as desired.
\end{proof}

The following lemma provides a multiplicative Chernoff bound for a Beta distribution. 
\begin{lemma}
    \label{Lemma:Beta}
    For $0 < m \leq n/2$ and  $X\sim \mathrm{Beta}(m, n-m)$, we have for any $t \geq 0$ that 
    \[
    \mathbb{P}\biggl(X \geq \frac{(1+t)m}{n} \biggr) \leq \exp\biggl\{ -\frac{t^2m}{2+4t/3}\biggr\}.
    \]
\end{lemma}
\begin{proof}
Set 
\[
v=\frac{m(n-m)}{n^2(n+1)}\quad c=\frac{2(n-2m)}{n(n+2)}.
\]
By \citet[Theorem~1]{skorski2023bernstein}, we have
\begin{align*}
\mathbb{P}\biggl(X \geq \frac{(1+t)m}{n}\biggr) &\leq \exp\biggl\{ -\frac{t^2m^2/n^2}{2v+2ctm/(3n)}\biggr\} \leq \exp\biggl\{ -\frac{t^2m}{2+4t/3}\biggr\}
\end{align*}
as desired.
\end{proof}

The following two lemmas compute the variance of the vacancy in the grid reference setting. 

\begin{lemma}
\label{Lemma:Combinatorics}
For $m\in\mathbb{N}$ and $r\in\{0,\ldots,\lfloor \log m\rfloor \}$, we have 
\[  
\frac{\binom{m}{r}\binom{m^2-2m}{m-r}\binom{m^2-2m+r}{m}}{\binom{m^2}{m}\binom{m^2-m}{m}} = \frac{1}{e^3 r!}\biggl\{1 - \frac{5-5r+r^2}{m} + \frac{27-112r+105r^2-32r^3+3r^4}{6m^2}+O(r^4m^{-3}) \biggr\},
\]
where the implied constant in $O(r^4m^{-3})$ does not depend on $r$ or $m$.
\end{lemma}
\begin{proof}
We write $A_r$ for the left-hand side of the desired result. 
    For any $u, v \in \mathbb{N}$ and $u \geq v$, let $(u)_v := u!/(u-v)!$ denote the falling factorial, then we have
\[
A_r = \frac{[(m)_r]^2 (m^2 -2m)_{m-r}(m^2 - 2m +r)_m}{(m^2)_m (m^2 - m)_m (r!)}.
\]
By Taylor's Theorem, there exists a sequence of $\xi_k \in (0, k/m)$, where $k = 1, \ldots, r-1$, such that
\begin{align}
2\log \bigl((m)_r\bigr) &= \log(m^{2r}) + 2 \sum_{k = 0}^{r-1}\bigl(-\frac{k}{m} - \frac{k^2}{2m^2} - \frac{k^3}{3m^3 (1 - \xi_k)^3}\bigr) \notag \\ 
& =  \log(m^{2r}) - \frac{-r + r^2}{m} - \frac{r(r-1)(2r-1)}{6m^2} - R_1{(m, r)},\label{eq:m_rLagrangian}
\end{align}
where $R_1{(m, r)} =  \frac{2}{3m^3}\sum_{k = 0}^{r-1}\frac{k^3}{(1-\xi_k)^3 }$. In addition, for $r \leq \log m$, we have
\begin{align}
R_1{(m, r)} \leq \frac{r^2(r-1)^2}{6m^3(1 - \frac{r-1}{m})^3} \leq C_1 \frac{r^4}{m^3},\label{eq:boundR_1}
\end{align}
for some universal constant $C_1>0$.

Similarly, there exists $\eta_k \in (0, (2m - r + k)/m^2)$ and $\zeta_k \in (0, (2m + k) / m^2)$, $k = 0, \ldots, m - r - 1$, such that 
\begin{align}
    \log \Bigl(\frac{(m^2 - 2m)_{m-r}}{m^{2m}}\Bigr) &= -\log \bigl(m^{2r}\bigr) + \sum_{k = 0}^{m - r - 1} \Bigl(-\frac{2m + k}{m^2}- \frac{(2m + k)^2}{2m^4} -\frac{(2m + k)^3}{3m^6} - \frac{(2m + k)^4}{4m^8(1 - \zeta_k)^4}\Bigr) \notag \\ 
    &= - \log(m^{2r}) - \frac{5}{2} + \frac{-8 + 9r}{3m} + \frac{-25 + 24 r - 3r^2}{6m^2} - R_2(m, r), \label{eq:R2Lagrangian}\\ 
    \log \Bigl(\frac{(m^2 - 2m + r)_m}{m^{2m}}\Bigr) &= \sum_{k = 0}^{m -1}\Bigl(-\frac{2m - r + k}{m^2} - \frac{(2m - r+k)^2}{2m^4} - \frac{(2m - r + k)^3}{3m^6} - \frac{(2m - r + k)^4}{4m^8(1 - \eta_k)^4}\Bigr) \notag \\ 
    &= -\frac{5}{2} + \frac{-8 + 3r}{3m} + \frac{-25 + 15 r}{6m^2} - R_3(m, r),\label{eq:R3Lagrangian}
\end{align}
where
\begin{align*}
    R_2(m, r) = \frac{1}{4m^8}\sum_{k = 0}^{m - r - 1}\frac{(2m + k)^4}{(1 -\zeta_k)^4} &+ \frac{37 + 90 r - 18 r^2}{12m^3} + \frac{-5 - 53 r - 51r^2 + 2r^3}{12 m^4} \\ 
    &+\frac{r+3r^2 + r^3}{2m^5} + \frac{-r^2 -2r^3-r^4}{12 m^6} 
\end{align*}
and 
\[
    R_3(m, r) = \frac{1}{4m^8}\sum_{k = 0}^{m - 1}\frac{(2m - r + k)^4}{(1 - \eta_k)^4} + \frac{37 + 70r- 6r^2}{12m^3} + \frac{-5-30 r-30 r^2}{12m^4} + \frac{r + 3r^2 + 2r^3}{6m^5}.
\]
Moreover, for $r \leq \log m$, there exists universal constants $C_2', C_2, C_3', C_3> 0$ such that
\begin{align}
    R_2(m, r) &\leq \frac{1}{4m^8(1 - \frac{3m -r -1}{m^2})^4}\sum_{k = 0}^{m - r -1}(2m + k)^4 + C_2'\frac{r^4}{m^3} \leq C_2 \frac{r^4}{m^3},  \label{eq:boundR_2} \\ 
        R_3(m, r) &\leq \frac{1}{4m^8(1 - \frac{3m - r - 1}{m^2})^4} \sum_{k = 0}^{m - 1} (2m -r + k)^4 + C_3'\frac{r^3}{m^3} \leq C_3\frac{r^4}{m^3}. \label{eq:boundR_3}
\end{align}

Finally, by noting 
\begin{align*}
\log\Bigl(\frac{(m^2)_m}{m^{2m}}\Bigr) & = -\frac{1}{2} + \frac{1}{3m} + \frac{1}{6m^2} + O(\frac{1}{m^3})\\ 
\log \Bigl(\frac{(m^2 -m)_{m}}{m^{2m}}\Bigr) &= -\frac{3}{2} - \frac{2}{3m} - \frac{1}{2m^2} + O(\frac{1}{m^3}),
\end{align*}
and combining~\eqref{eq:m_rLagrangian}--\eqref{eq:boundR_3}, we deduce that 
\begin{align*}
    \log(A_r) =& \log\Bigl(\frac{[(m)_r]^2}{r!}\Bigr)+\log\Bigl(\frac{(m^2 - 2m)_{m-r}}{m^{2m}}\Bigr) + \log \Bigl(\frac{(m^2 - 2m + r)_{m}}{m^{2m}}\Bigr)  \\
    &- \log\Bigl(\frac{(m^2)_m}{m^{2m}}\Bigr)-\log\Bigl(\frac{(m^2 -m)_{m}}{m^{2m}}\Bigr) \\
    =& -3 + \frac{-5 +5r - r^2}{m} + \frac{-48 +38r  -2r^3}{6m^2} - \log(r!) + O(\frac{r^4}{m^3}).
\end{align*}
Finally exponentiating both sides, we have 
\[ 
A_r = \frac{1}{e^3 r!}\bigl(1 + \frac{-5+5r-r^2}{m} + \frac{27-112r+105r^2-32r^3+3r^4}{6m^2}+O(\frac{r^4}{m^3}) \bigr),
\]
as desired.
\end{proof}

\begin{lemma}\label{lem:varVacancyReg}
Suppose $X, Y \in \mathbb{R}$ are independent random variables and $(X_1, Y_1), \ldots, (X_n, Y_n)$ are independent copies of $(X, Y)$. Let $\mathcal{V}_n^{\mathrm{grid}}$ be the vacancy defined as~\eqref{eq:RegularGridVacancy}. Then we have
\begin{align*}
    \lim_{n \to \infty} n \mathrm{Var}(\mathcal{V}_n^{\mathrm{grid}}) =  \frac{4 \mathrm{Ei}(1) - 4\gamma_0 - 5}{e^2}.
\end{align*}
\end{lemma}

\begin{proof}
For notation simplicity, we use $V$ to denote $\mathcal{V}_n^{\mathrm{grid}}$ throughout the proof. Also, for simplicity of presentation, we prove the result when $n=m^2$ is a square. If $n$ is not a square, one may write $q=\lfloor\sqrt n\rfloor$ and carry the same argument with one-dimensional window counts equal to either $q$ or $q+1$; these $O(1)$ rounding changes do not affect the limiting value of $n\mathrm{Var}(\mathcal{V}_n^{\mathrm{grid}})$.  Let $W = (W_1, W_2)$ and $\widetilde W = (\widetilde W_1, \widetilde W_2)$ be two independent random vectors from $\mathrm{Unif}[0,1]^2$. Let $N_{xy}:= \#\{i\in [n]: R_i \in B(W, \gamma_n)\}$. Given $W$, $N_{xy}$ follows a Hypergeometric distribution $\mathrm{Hyper}(m^2, m, m)$\footnote{We denote $\mathrm{Hyper}(N,k,n)$ for the hypergeometric distribution giving the number of successes in $n$ draws without replacement from a population of size $N$ containing $k$ successful elements.}. Using Stirling's formula
\[  
x! \sim \sqrt{2\pi x}\biggl(\frac{x}{e}\biggr)^x\biggl(1+\frac{1}{12x}+\frac{1}{288x^2}+ O(x^{-3})\biggr),
\]
we have 
\begin{align}
    \E V = \E \Bigl[\Prob(N_{xy} = 0 \mid W)\Bigr] = \frac{{m^2 - m \choose m}}{{m^2 \choose m}} =e^{-1}\biggl\{1-\frac{1}{m}-\frac{1}{6m^2}+O(m^{-3})\biggr\}. \label{eq:EV}
\end{align}

In the rest of the proof, we aim to control $\mathbb{E}(V^2)$. Let $M_{xy}$ be the number of $R_i$ covered by $B(W, \gamma_n)\cup B(\widetilde W, \gamma_n)$, i.e. $M_{xy}:= \#\{i \in [n]: R_i \in B(W, \gamma_n) \cup B(\widetilde W, \gamma_n)\}$, we have 
\[
  \mathbb{E}(V^2) = \mathbb{P}(M_{xy} = 0)
\]

Let $\widetilde S / m$ and $\widetilde T / m$ be the length of intersection of the $x$- and $y$- projections of $B(W,\gamma_n)$ and $B(\widetilde W,\gamma_n)$ and let $Sm$ and $Tm$ denote the number of lattice points $\{1/m, 2/m, \ldots,1\}$ that fall inside the intersections. Then 
\[
  \widetilde T = \biggl(1-\frac{2}{m}\biggr)\delta_0 + \frac{2}{m}\mathrm{Unif}[0, 1],
\]
and 
\[
  Tm\mid \widetilde{T} \sim (1-\{\widetilde T m\}) \delta_{\lfloor \widetilde T m\rfloor } + \{\widetilde T m\}\delta_{\lceil \widetilde T m\rceil},
\]
where $\{a\} := a-\lfloor a\rfloor$ denotes the fractional part of a real number. Consequently, we have 
\[
  T \sim \biggl(1-\frac{2}{m}\biggr)\delta_0 + \frac{2}{m}\nu,
\]
where $\nu = \frac{1}{2m}\delta_0 + \frac{1}{m}\delta_{1/m} + \cdots + \frac{1}{m}\delta_{1-1/m}+\frac{1}{2m}\delta_{1}$. 

Since $S$ is independent and identically distributed with $T$, we have
\begin{align}
  \mathbb{E}(V^2) &= \mathbb{E}[\mathbb{P}(M_{xy} = 0\mid W, \widetilde W)]=\mathbb{E}[\mathbb{P}(M_{xy} = 0\mid S, T)]\nonumber\\
  &=\biggl(1-\frac{2}{m}\biggr)^2 \mathbb{P}(M_{xy} = 0 \mid S=T=0) \nonumber\\
  &\hspace{3cm} + \frac{4}{m}\biggl(1-\frac{2}{m}\biggr)\int_{t=0}^1 \mathbb{P}(M_{xy} = 0 \mid S=0, T=t)\,d\nu(t) \nonumber\\
&\hspace{3cm} +  \biggl(\frac{2}{m}\biggr)^2\int_{t=0}^1\int_{s=0}^1\mathbb{P}(M_{xy} = 0\mid S=s, T=t)\,d\nu(s) d\nu(t) .\label{Eq:Combined}
\end{align}

We now analyse each term of~\eqref{Eq:Combined}. Suppose we have $m^2$ balls randomly ordered, $m$ of which are red and $m$ are blue and the rest are colourless. Let $\mathcal{A}$ be the event that the first $m$ balls are not red and the second $m$ balls are not blue, and let $\mathcal{A}_r$ be the event that $r$ of the first $m$ are blue. Then 
\begin{align}
  \mathbb{P}(M_{xy} = 0\mid S = T= 0) &= \mathbb{P}(\mathcal{A}) = \sum_{r = 0}^m \Prob(\mathcal{A} \cap \mathcal{A}_r) \notag \\
&=\sum_{r=0}^m\frac{\binom{m}{r}\binom{m^2-2m}{m-r}\binom{m^2-2m+r}{m}}{\binom{m^2}{m}\binom{m^2-m}{m}} =: \sum_{r=0}^m A_r. \label{eq:well-seperate}
\end{align}
When $r \leq \lfloor \log m\rfloor$, we have by Lemma~\ref{Lemma:Combinatorics} that 
\begin{align*}
A_r = \frac{1}{e^3 r!}\biggl\{1 + \frac{-5+5r-r^2}{m} + \frac{27-112r+105r^2-32r^3+3r^4}{6m^2}+O(\frac{r^4}{m^3}) \biggr\}, 
\end{align*}
which implies that 
\begin{align}
    \sum_{r = 0}^{\lfloor\log m\rfloor} A_r  =  e^{-2}\Bigl\{1-\frac{2}{m}+\frac{5}{3m^2}+O(m^{-3})\Bigr\}.\label{eq:S1}
\end{align}

When $r > \lfloor \log m\rfloor$, there exists a universal constant $C > 0$ such that
\begin{align*}
    A_r \leq \frac{m^{2r} m^{2m - 2r} m^{2m}}{(m^2-2m)^{2m}  \lceil \log m\rceil !} \leq \frac{C}{\lceil \log m\rceil !},
\end{align*}
thus 
\begin{align*}
    \sum_{r = \lfloor \log m \rfloor +1}^{m} A_r = O\biggl(\frac{m}{\lceil \log m\rceil !}\biggr) = O(m^{-3}).
\end{align*}
Combining this with~\eqref{eq:S1} and~\eqref{eq:well-seperate}, we have
\begin{equation}
    \mathbb{P}(M_{xy} = 0\mid S = T= 0) = e^{-2}\Bigl\{1-\frac{2}{m}+\frac{5}{3m^2}+O(m^{-3})\Bigr\}.
    \label{Eq:FirstTerm}
\end{equation}

In a similar way, for $t\in[0,1]$ such that $tm$ is an integer, $\mathbb{P}(M_{xy}=0\mid S=0, T=t)$ computes the probability that in the same $m^2$ balls as described above, the first $tm$ balls are colourless, the next $(1-t)m$ contain no red and the next $(1-t)m$ contain no blue. Hence, by a very similar argument as above, we have
\begin{align*}
\mathbb{P}(M_{xy} = 0 \mid S=0, T=t) &=\frac{\sum_{r=0}^{(1-t)m}\binom{(1-t)m}{r}\binom{m^2-(2-t)m}{m-r}\binom{m^2-2m+r}{m}}{\binom{m^2}{m}\binom{m^2-m}{m}} \\
&=e^{-2}\biggl\{1-\frac{2+t}{m}+O(m^{-2})\biggr\},
\end{align*}
where the implied constant does not depend on $t$. Consequently, we have 
\begin{equation}
  \int_{t=0}^1\mathbb{P}(M_{xy}=0 \mid S=0, T=t)\,d\nu(t) = e^{-2}\biggl\{1-\frac{5}{2m}+O(m^{-2})\biggr\}.\label{Eq:SecondTerm}
\end{equation}

For the final term on the right-hand side of~\eqref{Eq:Combined}, fixed $s,t\in[0,1]$ such that $sm$ and $tm$ are integers. Suppose there are $m^2=n$ balls randomly permuted, $sm$ of which are red-blue striped, $(1-s)m$ are red, $(1-s)m$ are blue and the rest colourless. We are after the probability that the first $tm$ balls are colourless, next $(1-t)m$ contain no red tint, next $(1-t)m$ contain no blue tint. We again partition the $m^2$ slots into four blocks of $tm$, $(1-t)m$, $(1-t)m$ and $m^2-(2-t)m$ slots each. Red-blue striped balls can only go to last block, red balls can only go to blocks 3 and 4, blue balls can only go to blocks 2 and 4. Suppose there are $r$ red balls in block 3. Then there are $\binom{m^2-(2-t)m}{sm}$ ways to place the red-blue balls, $\binom{(1-t)m}{r}\binom{m^2-(2-t)m-sm}{(1-s)m-r}$ ways to place the red balls, $\binom{m^2-m-sm-(1-s)m+r}{(1-s)m}$ ways to place the blue balls. Thus, we have
\begin{align*}
  \mathbb{P}(M_{xy} = 0 \mid S=s,T=t) &= \frac{\sum_{r=0}^{(1-t)m} \binom{m^2-(2-t)m}{sm}\binom{(1-t)m}{r}\binom{m^2-(2-t+s)m}{(1-s)m-r}\binom{m^2-2m+r}{(1-s)m}}{\binom{m^2}{sm}\binom{m^2-sm}{(1-s)m}\binom{m^2-m}{(1-s)m}}\\
  &=e^{-2+st}+O(m^{-1}),
\end{align*}
where the final step again follows from a similar calculation to Lemma~\ref{Lemma:Combinatorics}. Hence,
\begin{equation}
  \int_{t=0}^1\int_{s=0}^1 \mathbb{P}(M_{xy} = 0 \mid S=s,T=t)\,d\nu(s)d\nu(t) = e^{-2}(\mathrm{Ei}(1)-\gamma_0) + O(m^{-1}).\label{Eq:ThirdTerm}
\end{equation} 
Combining~\eqref{Eq:FirstTerm},~\eqref{Eq:SecondTerm},~\eqref{Eq:ThirdTerm} with~\eqref{Eq:Combined}, we have
\[
  \mathbb{E}(V^2) = e^{-2}\biggl\{1-\frac{2}{m}+\frac{12\mathrm{Ei}(1)-12\gamma_0 -13}{3m^2}+O(m^{-3})\biggr\}.
\]
Together with the square-grid first-moment expansion in~\eqref{eq:EV}, we have
\[
  \mathrm{Var}(V) = \frac{4\mathrm{Ei}(1)-4\gamma_0-5}{e^2 m^2}+O(m^{-3}),
\]
which establishes the claim since $4\mathrm{Ei}(1)-4\gamma_0-5 \approx 0.036 > 0$.
\end{proof}

\begin{lemma}
\label{le: VacancyLimit}
Suppose $\sigma\in\mathcal{S}_n$ is a random permutation on $[n]$. Write $R_i = (i/n, \sigma(i)/n)$ for $i\in[n]$. For $\delta\in[0, 1/2)$, let $B$ be a compact set symmetric around 0 with $\vol(B) = \delta$ and define
\[
    \mathcal{V}:=\vol\Bigl([0,1]^2\setminus \bigcup_{i=1}^n (R_i + B)\Bigr),
\]
where $R_i + B$ denotes the Minkowski sum of the sets. If $n\delta\to q$, then as $n\to\infty$, we have 
\[
    \mathbb{E}(\mathcal{V}) \to e^{-q}.
\]
\end{lemma}
\begin{proof}
    For any $x\in[0, 1]^2$, let $N_x := \abs{\{i: (i/n, \sigma(i)/n)\in x+B\}}$. Note that $N_x$ follows a Hypergeometric distribution $\mathrm{Hyper}(n^2, K, n)$ where $K = \abs{\{(i,j):(i/n, j/n)\in x + B\}}$. Note that $K = n^2\delta + o(n)$. Therefore $\lambda_n := \E[N_x] = n\delta + o(1)$. Then note that we have 
    \[
    d_{TV}(\mathrm{Hyper}(n^2, K, n), \mathrm{Pois}(\lambda_n)) = O(\frac{1}{n})
    \]
    Then note that 
    \begin{align*}
        \mathcal{V} = \int_{[0,1]^2}\bone\{N_x = 0\}dx,
    \end{align*}
    therefore
    \begin{align*}
        \E[\mathcal{V}] = \int_{[0,1]^2}\mathbb{P}(N_x = 0)dx = e^{-n\delta} + O(\frac{1}{n}) \to e^{-q}.
    \end{align*}
    as desired.
\end{proof}

Finally, we introduce some combinatorial lemmas that will be used in the proofs of Propositions~\ref{prop:Vm-fixed-degrees}--\ref{prop:tail-covariance}. For a finite set of cells $E\subset\{1,\dots,n\}^2$, define
\[
\mathrm{row}(E):=\{i:\exists j,\ (i,j)\in E\},
\qquad
\mathrm{col}(E):=\{j:\exists i,\ (i,j)\in E\}.
\]

\begin{lemma}\label{lem:matchingcount}
There exists a constant $K<\infty$ such that for all $n\ge1$, all $z\in[0,1]^2$,
and all $m\ge1$,
\[
|\cM_{n,m}(z)|\le \frac{K^m n^m}{m!}.
\]
Moreover:
\begin{enumerate}[label=\rm(\alph*), ref = \thelemma(\alph*), itemsep=0.2em]
% \item for any fixed $e\in I_n(z)\times J_n(z)$,
% \[
% \bigl|\{M\in\cM_{n,m}(z): e\in M\}\bigr|
% \le \frac{K^m n^{m-1}}{(m-1)!};
% \]
\item \label{lem:matchingcount_a} for any fixed matching $S\subset I_n(z)\times J_n(z)$ of size $q\le m$,
\[
\bigl|\{M\in\cM_{n,m}(z): S\subset M\}\bigr|
\le \frac{K^m n^{m-q}}{(m-q)!};
\]

\item \label{lem:matchingcount_b} for any fixed row $i_0\in I_n(z)$,
\[
\bigl|\{M\in\cM_{n,m}(z): \exists j,\ (i_0,j)\in M\}\bigr|
\le \frac{K^m n^{m-1/2}}{(m-1)!};
\]
\item \label{lem:matchingcount_c} for any fixed column $j_0\in J_n(z)$,
\[
\bigl|\{M\in\cM_{n,m}(z): \exists i,\ (i,j_0)\in M\}\bigr|
\le \frac{K^m n^{m-1/2}}{(m-1)!}.
\]
\end{enumerate}
\end{lemma}

\begin{proof}
Let $a:=|I_n(z)|$ and $b:=|J_n(z)|$. By \eqref{eq:IJsize}, we have
$a,b\le C_0\sqrt n$.

The number of matchings of size $m$ in $I_n(z)\times J_n(z)$ is
\[
\binom{a}{m}\binom{b}{m}m!
\le \frac{a^m b^m}{m!}
\le \frac{(C_0^2)^m n^m}{m!}.
\]
This proves the first estimate.

% If a fixed cell $e$ is prescribed, then the number of possibilities is
% \[
% \binom{a-1}{m-1}\binom{b-1}{m-1}(m-1)!
% \le \frac{a^{m-1}b^{m-1}}{(m-1)!}
% \le \frac{(C_0^2)^m n^{m-1}}{(m-1)!}.
% \]

Now let $S\subset I_n(z)\times J_n(z)$ be a fixed matching of size $q\le m$. Note that $|\mathrm{row}(S)|=|\mathrm{col}(S)|=q$ because $S$ is a matching.

Any $M\in\cM_{n,m}(z)$ with $S\subset M$ is obtained by choosing the remaining $m-q$ rows from $I_n(z)\setminus \mathrm{row}(S)$, choosing the remaining $m-q$ columns from $J_n(z)\setminus \mathrm{col}(S)$, and choosing a bijection between these two $(m-q)$-element sets.

Hence
\[
\bigl|\{M\in\cM_{n,m}(z): S\subset M\}\bigr|
=
\binom{a-q}{m-q}\binom{b-q}{m-q}(m-q)!.
\]
Therefore
\[
\bigl|\{M\in\cM_{n,m}(z): S\subset M\}\bigr|
\le
\frac{a^{m-q}b^{m-q}}{(m-q)!}
\le
\frac{(C_0^2)^{\,m-q} n^{m-q}}{(m-q)!}.
\]
After enlarging the constant, this gives
\[
\bigl|\{M\in\cM_{n,m}(z): S\subset M\}\bigr|
\le \frac{K^m n^{m-q}}{(m-q)!}.
\]

If a fixed row $i_0$ is prescribed, then one first chooses the column paired with
$i_0$, then the remaining $m-1$ rows and columns and a bijection:
\[
b\,\binom{a-1}{m-1}\binom{b-1}{m-1}(m-1)!
\le \frac{a^{m-1}b^m}{(m-1)!}
\le \frac{(C_0^2)^m n^{m-1/2}}{(m-1)!}.
\]
The column estimate is symmetric.
\end{proof}

\begin{lemma}[One-step summation bound]\label{lem:onestep}
Fix integers $r\ge2$ and $m_1,\dots,m_r\ge1$.
There exists a constant $C_{r,\mathbf m}<\infty$, depending only on
$\mathbf m:=(m_1,\dots,m_r)$ and $r$, such that the following holds.

Let $\ell\in\{1,\dots,r-1\}$ and let $a_1,\dots,a_{\ell+1}$ be distinct indices in
$\{1,\dots,r\}$. For $t=1,\dots,\ell$, let
\[
M_{a_t}\in \cM_{n,m_{a_t}}(z_{a_t}).
\]
Set
\[
E_*:=M_{a_1}\cup\cdots\cup M_{a_\ell},
\qquad
R_*:=\mathrm{row}(E_*),
\qquad
C_*:=\mathrm{col}(E_*).
\]
Then
\begin{align}
&\sum_{M\in \cM_{n,m_{a_{\ell+1}}}(z_{a_{\ell+1}})}
W_n(M;\,M_{a_1},\dots,M_{a_\ell})\,n^{-\abs{M\setminus E_*}} \le
C_{r,\mathbf m}\,\frac{1}{(m_{a_{\ell+1}}-1)!}
\sum_{t=1}^{\ell}\psi_n(z_{a_{\ell+1}},z_{a_t}).
\label{eq:onestep}
\end{align}
\end{lemma}

\begin{proof}
We split the sum into four classes.

\smallskip
\noindent

\emph{Class 1: $M\cap E_*\neq\varnothing$.} Then $W_n(M;M_{a_1},\dots,M_{a_\ell})=1$. Using Lemma~\ref{lem:matchingcount_a},
\begin{align*}
\sum_{\substack{M\in \cM_{n,m_{a_{\ell+1}}}(z_{a_{\ell+1}})\\ M\cap E_*\neq\varnothing}}
n^{-\abs{M\setminus E_*}} &= \sum_{q = 1}^{ |E_*|}\sum_{\substack{M\in \cM_{n,m_{a_{\ell+1}}}(z_{a_{\ell+1}})\\ |M\cap E_*|=q}}
n^{-(m_{a_{\ell+1}} - q)} \\ 
& \leq \sum_{q = 1}^{|E_*|} {|E_*| \choose q}\,\frac{K^{m_{a_{\ell+1}}}n^{m_{a_{\ell+1}}-q}}{(m_{a_{\ell+1}}-q)!} n^{-(m_{a_{\ell+1}} - q)}.
\end{align*}
Since $|E_*|\le m_{a_1}+\cdots+m_{a_\ell}$, this is at most 
\[
C_{r,\mathbf m}\,\frac{1}{(m_{a_{\ell+1}}-1)!}.
\]

% \emph{Class 1: $M\cap E_*\neq\varnothing$.}
% Then $W_n(M;M_{a_1},\dots,M_{a_\ell})=1$ and
% $\abs{M\setminus E_*}\ge m_{a_{\ell+1}}-1$.
% Using Lemma~\ref{lem:matchingcount}(a),
% \[
% \sum_{\substack{M\in \cM_{n,m_{a_{\ell+1}}}(z_{a_{\ell+1}})\\ M\cap E_*\neq\varnothing}}
% n^{-\abs{M\setminus E_*}}
% \le
% |E_*|\,\frac{K^{m_{a_{\ell+1}}}n^{m_{a_{\ell+1}}-1}}{(m_{a_{\ell+1}}-1)!}\,
% n^{-(m_{a_{\ell+1}}-1)}.
% \]
% Since $|E_*|\le m_{a_1}+\cdots+m_{a_\ell}$, this is at most
% \[
% C_{r,\mathbf m}\,\frac{1}{(m_{a_{\ell+1}}-1)!}.
% \]

Moreover this class is empty unless $I_n(z_{a_{\ell+1}})$ meets
$I_n(z_{a_t})$ and $J_n(z_{a_{\ell+1}})$ meets $J_n(z_{a_t})$ for some $t$,
which implies
\[
d_{\mathbb{T}}(x_{a_{\ell + 1}}, x_{a_t}) \leq \frac{1}{\sqrt{n}}, \qquad d_{\mathbb{T}}(y_{a_{\ell + 1}}, y_{a_t}) \leq \frac{1}{\sqrt{n}}
\]
Therefore the contribution of Class 1 is bounded by
\[
C_{r,\mathbf m}\,\frac{1}{(m_{a_{\ell+1}}-1)!}
\sum_{t=1}^{\ell}
\bone{\Bigl\{d_{\mathbb{T}}(x_{a_{\ell + 1}}, x_{a_t}) \leq \frac{1}{\sqrt{n}}, \, d_{\mathbb{T}}(y_{a_{\ell + 1}}, y_{a_t}) \leq \frac{1}{\sqrt{n}}\Bigr\}}.
\]

\smallskip
\noindent
\emph{Class 2: $M\cap E_*=\varnothing$ but $\mathrm{row}(M)\cap R_*\neq\varnothing$.}
Then $W_n(M;M_{a_1},\dots,M_{a_\ell})=1$ and $\abs{M\setminus E_*} = m_{a_{\ell+1}}$.
Using Lemma~\ref{lem:matchingcount_b},
\[
\sum_{\substack{M\in \cM_{n,m_{a_{\ell+1}}}(z_{a_{\ell+1}})\\
M\cap E_*=\varnothing,\ \mathrm{row}(M)\cap R_*\neq\varnothing}}
n^{-m_{a_{\ell+1}}}
\le
|R_*|\,\frac{K^{m_{a_{\ell+1}}}n^{m_{a_{\ell+1}}-1/2}}{(m_{a_{\ell+1}}-1)!}\,
n^{-m_{a_{\ell+1}}}.
\]
Since $|R_*|\le m_{a_1}+\cdots+m_{a_\ell}$, this is at most
\[
C_{r,\mathbf m}\,\frac{n^{-1/2}}{(m_{a_{\ell+1}}-1)!}.
\]
This class is empty unless $I_n(z_{a_{\ell+1}})$ meets $I_n(z_{a_t})$ for some $t$,
hence its contribution is bounded by
\[
C_{r,\mathbf m}\,\frac{1}{(m_{a_{\ell+1}}-1)!}
\sum_{t=1}^{\ell} n^{-1/2}\bone{\Bigl\{d_{\mathbb{T}}(x_{a_{\ell + 1}}, x_{a_t}) \leq \frac{1}{\sqrt{n}}\Bigr\}}.
\]

\smallskip
\noindent
\emph{Class 3: $M\cap E_*=\varnothing$, $\mathrm{row}(M)\cap R_*=\varnothing$, but
$\mathrm{col}(M)\cap C_*\neq\varnothing$.}
The same argument using Lemma~\ref{lem:matchingcount_c} gives the bound
\[
C_{r,\mathbf m}\,\frac{1}{(m_{a_{\ell+1}}-1)!}
\sum_{t=1}^{\ell} n^{-1/2}\bone{\Bigl\{d_{\mathbb{T}}(y_{a_{\ell + 1}}, y_{a_t}) \leq \frac{1}{\sqrt{n}}\Bigr\}}.
\]

\smallskip
\noindent
\emph{Class 4: $\mathrm{row}(M)\cap R_*=\varnothing$ and
$\mathrm{col}(M)\cap C_*=\varnothing$.}
Then $W_n(M;M_{a_1},\dots,M_{a_\ell})=n^{-1}$ and
$\abs{M\setminus E_*} = m_{a_{\ell+1}}$.
Using Lemma~\ref{lem:matchingcount},
\[
\sum_{\substack{M\in \cM_{n,m_{a_{\ell+1}}}(z_{a_{\ell+1}})\\
\mathrm{row}(M)\cap R_*=\varnothing,\ \mathrm{col}(M)\cap C_*=\varnothing}}
n^{-1}n^{-m_{a_{\ell+1}}}
\le
n^{-1}\frac{K^{m_{a_{\ell+1}}}n^{m_{a_{\ell+1}}}}{m_{a_{\ell+1}}!}n^{-m_{a_{\ell+1}}}
\le
C_{r,\mathbf m}\,\frac{n^{-1}}{(m_{a_{\ell+1}}-1)!}.
\]

Adding the four classes yields \eqref{eq:onestep}.
\end{proof}

\begin{lemma}\label{lem:VacancyDifference}
Let $\{(x_i,y_i)\}_{i=1}^n$ be $n$ points in $\mathbb{R}^2$ and let $x = (x_1, \ldots, x_n)$ and $y = (y_1, \ldots, y_n)$. Assume that $x_i\neq x_j$ and $y_i\neq y_j$ for all $i\neq j$. For any $z\in\mathbb{R}^n$ let $r^z_i = n^{-1}\sum_{j = 1}^n \bone\{z_j \leq z_i\}$. Define $r_i = (r^x_i, r^y_i)$ for $i \in [n]$. Let $B = [-\frac{1}{2\sqrt{n}}, \frac{1}{2\sqrt{n}}]^2$ and define
\[
    \mathcal{V}:=\vol\Bigl([0,1]^2\setminus \bigcup_{i=1}^n (r_i + B)\Bigr),
\]
where $r_i + B$ denotes the Minkowski sum of the sets. Let $\{(x_i', y_i')\}_{i=1}^n$ be another set of points in $\mathbb{R}^2$ and define $x^k, y^k\in\mathbb{R}^n$ where $x^k_i = x_i$ and $y^k_i = y_i$ for all $i\in[n]\setminus\{k\}$ and $x^k_k = x_k^\prime$ and $y^k_k = y_k^\prime$. 
\[
    \mathcal{V}^k := \vol\Bigl([0,1]^2\setminus \bigcup_{i=1}^n (r_i^k + B)\Bigr),
\]
where $r_i^k = (r_i^{x^k}, r_i^{y^k})$. Then
\[
\abs{\mathcal{V} - \mathcal{V}^k} \leq 10/n.
\]
\end{lemma}
\begin{proof}
First note that $|\mathcal{V} - \mathcal{V}^k| = 0$ if $x_k = x_k'$ and $y_k = y_k'$, thus the result holds automatically. Assume $x_k\neq x_k'$ or $y_k \neq y_k'$. Define
\begin{align*}
    \mathcal{I}_0 &:= \{i \in [n]\setminus\{k\}: \text{$(x_i - x_k)(x_i - x'_k) \geq 0$ and $(y_i - y_k)(y_i - y'_k)\geq 0$}\}, \\  
    \mathcal{I}_1 &:= \{i \in [n]\setminus\{k\}: \text{$(x_i - x_k)(x_i - x'_k)  < 0$ and $(y_i - y_k)(y_i - y'_k) < 0$}\},\\  
    \mathcal{I}_2 &:= \{i \in [n]\setminus\{k\}: \text{$(x_i - x_k)(x_i - x'_k)  < 0$ and $(y_i - y_k)(y_i - y'_k) \geq 0$}\},\\
    \mathcal{I}_3 &:= \{i \in [n]\setminus\{k\}: \text{$(x_i - x_k)(x_i - x'_k)  \geq 0$ and $(y_i - y_k)(y_i - y'_k) < 0$}\}.
\end{align*}

Note that when replacing $x_k$ by $x_k'$ and $y_k$ by $y_k'$, both $\{x_i: i \in \mathcal{I}_0\}$ and $\{y_i: i \in \mathcal{I}_0\}$ maintain their original ranks, and both $\{x_i: i \in \mathcal{I}_1\}$ and $\{y_i: i \in \mathcal{I}_1\}$ have a shift of $1/n$ in either direction. The ranks of $\{x_i: i \in \mathcal{I}_2\}$ have a similar shift of $1/n$ while the ranks of $\{y_i: i \in \mathcal{I}_2\}$ remain the same. Conversely, the ranks of $\{x_i: i \in \mathcal{I}_3\}$ remains the same while ranks of $\{y_i: i \in \mathcal{I}_3\}$ have a shift of $1/n$ in either direction. Specifically, we have 
\begin{align}
    r_i^{k} = \begin{cases}
        r_i,  &\text{if $i \in \mathcal{I}_0$} \\ 
        r_i + (\pm \frac{1}{n}, \pm \frac{1}{n}),  &\text{if $i \in \mathcal{I}_1$} \\ 
        r_i + (\pm \frac{1}{n}, 0),  &\text{if $ i \in \mathcal{I}_2$} \\
        r_i + (0, \pm \frac{1}{n}),  &\text{if $ i \in \mathcal{I}_3$}.
    \end{cases} \label{eq:rankchange}
\end{align}

Let $\mathcal{U}_j = \bigcup_{i \in \mathcal{I}_j}(r_i + B)$ and $\mathcal{U}_j^k = \bigcup_{i \in \mathcal{I}_j}(r_i^k + B)$ for $j = 0, 1, 2, 3$, we have the following decompositions
\begin{align*}
\mathcal{C}&:=\bigcup_{i \in [n]} (r_i + B) = \Bigl(\bigcup_{j = 0}^3 \mathcal{U}_j\Bigr) \cup (r_k + B), \\ 
\mathcal{C}^k&:=\bigcup_{i \in [n]} (r_i^k + B) = \Bigl(\bigcup_{j = 0}^3 \mathcal{U}_j^{k}\Bigr) \cup (r_k^k + B).
\end{align*}
By \eqref{eq:rankchange} we have $\vol\bigl(\mathcal{U}_0 \Delta \mathcal{U}_0^k\bigr) = 0$. For $i\in\mathcal{I}_2$, $\mathcal{U}_2^k$ is a shift of $\mathcal{U}_2$ by $1/n$ to the right (or left). Therefore, for almost every $x\in \mathcal{U}_2\setminus \mathcal{U}_2^k$ there is $y\in \mathcal{U}_2^k\setminus \mathcal{U}_2$ except for those $x\in[0, 1]^2$ with $x\in\big([0, 1/n]\cup[1-1/n, 1]\big)\times[0, 1]$. A similar argument holds for $i\in\mathcal{I}_3$, with the difference that the shift is up/down. Hence, for $j = 2, 3$
\[
\abs{\vol\big(\mathcal{U}_j\big) - \vol\big(\mathcal{U}_j^k\big)} \leq 2/n.
\]
For $i\in\mathcal{I}_1$, boxes shift both right/left and up/down and hence we have
\[
\abs{\vol\big(\mathcal{U}_1\big) - \vol\big(\mathcal{U}_1^k\big)} \leq 4/n.
\]
Therefore, we have
\begin{align*}
    \abs{\mathcal{V} - \mathcal{V}^k} &\leq \sum_{j = 0}^3 \abs{\vol\bigl(\mathcal{U}_j\bigr) - \vol\bigl(\mathcal{U}_j^k\bigr)} + \vol\Bigl((r_k + B)\Delta (r_k^{k} + B)\Bigr) \leq 10/n.
\end{align*}
as claimed.
\end{proof}

% \begin{defn}\citep{weed2019sharp}
% A dyadic partition of a compact set $S \subseteq \mathbb{R}^d$ with parameter $0< \delta <1$ is a sequence $\{\mathcal{Q}_k\}_{k \geq 1}$ such that 
% \begin{itemize}
%     \item the sets in $\mathcal{Q}_k$ form a partition of $S$; 
%     \item if $Q \in \mathcal{Q}_k$ then $\mathrm{diam}(Q) \leq \delta^k$
%     \item if $Q^{k+1} \in \mathcal{Q}_{k+1}$ and $Q^k \in \mathcal{Q}_{k}$, then either $Q_{k + 1} \subseteq Q_k$ or $Q_{k+1} \cap Q_k = \varnothing$
% \end{itemize}
% \end{defn}

\section{Algorithm and simulation settings}

\subsection{Algorithmic implementation details}
\label{Sec:Algorithm}
The computation of the coverage correlation coefficient involves evaluating the volume of the union of $n$ axis-aligned hypercubes, each of volume $1/n$, in the unit cube $[0,1]^d$, with edge wrapping. This is a special case of Klee's measure problem \citep{klee1977can}, which concerns computing the volume of the union of arbitrary axis-aligned hyperrectangles. When $d=2$, Bentley's algorithm solves this in $O(n\log n)$ time by sweeping along one axis and maintaining the union of intervals along the other using a segment tree \citep{ben1983lower}. In higher dimensions, the time complexity of Bentley's algorithm becomes $O(n^{d-1} \log n)$, while the best known theoretical bound is Chan's $O(n^{d/2})$ algorithm \citep{chan2013klee}. However, for moderate dimensions (e.g., $d \leq 10$), Bentley's approach remains more practical due to its smaller constant factors, better memory behaviour, and simpler implementation.  Moreover, in our setting, we are able to exploit the uniform size of the input small hypercubes to make computational gains using Bentley's algorithm alone. Specifically, we partition $[0,1]^d$ into $m^d$ grid blocks with $m \approx n^{1/d}$, and compute the union volume by summing contributions from individual blocks.
In cases where the small hypercubes are spread out uniformly at random, we expect $O_p(\log n)$ small hypercubes intersecting each block, thus making the entire algorithm run in an average-case complexity of $O(n\log^{d-1} n)$. We outline the recursive union volume computation using Bentley's algorithm in Algorithm~\ref{Algo:Bentley} and the full coverage correlation computation algorithm in Algorithm~\ref{Algo:CoverCorr}.

\begin{algorithm}[htbp]
\caption{\label{Algo:Bentley}\texttt{UnionVolume}($\mathcal{R}, d$): Bentley's algorithm for union volume for $d \geq 2$}
\SetAlgoLined
\KwIn{List $\mathcal{R}$ of axis-aligned rectangles in $\mathbb{R}^d$, $d\geq 2$}
\KwOut{Volume of the union of rectangles}

\uIf{$d = 2$}{
  Initialize an empty event list $E$\;
  \ForEach{rectangle $(x_{\min}, x_{\max}, y_{\min}, y_{\max}) \in \mathcal{R}$}{
    Add events $(x_{\min}, +1, [y_{\min}, y_{\max}])$  and $(x_{\max}, -1, [y_{\min}, y_{\max}])$ to $E$\;
  }
  Sort $E$ by x-coordinate\;

  Initialize active multiset $A \leftarrow \varnothing$, $\texttt{total\_area} \leftarrow 0$, $x_{\text{prev}} \leftarrow \text{undefined}$\;

  \ForEach{event $(x, \text{type}, [y_{\min}, y_{\max}]) \in E$}{
    \If{$x_{\text{prev}}$ is defined}{
      Let $\texttt{height} \leftarrow \text{total length of union of intervals in } A$\;
      $\texttt{total\_area} \leftarrow \texttt{total\_area} + (x - x_{\text{prev}}) \cdot \texttt{height}$\;
    }

    \uIf{$\text{type} = +1$}{
      Insert interval $[y_{\min}, y_{\max}]$ into $A$\;
    }
    \Else{
      Remove interval $[y_{\min}, y_{\max}]$ from $A$\;
    }

    $x_{\text{prev}} \leftarrow x$\;
  }

  \Return $\texttt{total\_area}$\;
}
\uElse{
  Extract all unique coordinates along the first axis and sort them as $x_1 < \dots < x_k$\;
  Initialize $\texttt{total\_volume} \leftarrow 0$\;
  \ForEach{interval $[x_i, x_{i+1}]$}{
    Let $\mathcal{S} \leftarrow \{ \text{rectangles in } \mathcal{R} \text{ that span } [x_i, x_{i+1}] \text{ in axis 1} \}$\;
    Project each rectangle in $\mathcal{S}$ to the remaining $d-1$ dimensions to obtain $\mathcal{S}'$\;
    $\texttt{total\_volume} \leftarrow \texttt{total\_volume} + (x_{i+1} - x_i) \cdot \texttt{UnionVolume}(\mathcal{S}', d-1)$\;
  }
  \Return $\texttt{total\_volume}$\;
}
\end{algorithm}

\begin{algorithm}[htbp]
\caption{\label{Algo:CoverCorr}Pseudocode for computing the coverage correlation coefficient}
\SetAlgoLined
\KwIn{Two samples $X_1,\ldots,X_n \in \mathbb{R}^{d_X}$ and $Y_1,\ldots,Y_n\in \mathbb{R}^{d_Y}$}
\KwOut{coverage correlation $\kappa_n^{X,Y}$ and the corresponding p-value $p_{\kappa}$}

Draw $U_1,\ldots,U_n \iid \mathrm{Unif}([0,1]^{d_X})$ and $V_1,\ldots,V_n\iid \mathrm{Unif}[0,1]^{d_Y}$\;
Compute Monge--Kantorovich ranks $R_1,\ldots,R_n \in [0,1]^d$ for $d= d_X+d_Y$ as in~\eqref{eq:R_iGen}\;

Initialize empty list $\mathcal{R}$\;

\For{$i$ in $1,\ldots,n$}{
  Split $B(R_i, \frac{1}{2n^{1/d}})$ along wrapped axes to get up to $2^d$ axis-aligned hyperrectangles within $[0,1]^d$\;
  Add all resulting non-wrapping rectangles to $\mathcal{R}$\;
}

Partition $[0,1]^d$ into $m^d$ grid blocks $G_1,\ldots,G_{m^d}$ for $m := \lfloor n^{1/d}\rfloor$. \;
\uIf{$d=2$}{
$\mathcal{V}_n\leftarrow 1 - \mathrm{CoveredVolume}(\mathcal{R})$
}
\Else{
Initialise $\mathcal{V}_n \leftarrow 1$\;
\ForEach{grid block $G_k$, $k\in[m^d]$}{
Define $\mathcal{R}'_k := \{A\cap G_k: A\in\mathcal{R}\}$
$\mathcal{V}_n\leftarrow \mathcal{V}_n - \mathrm{CoveredVolume}(\mathcal{R}'_k)$
}
}

Compute 
\[
\sigma^2 := \sum_{k=2}^n\binom{n}{k}\Bigl(1-\frac{2}{n}\Bigr)^{n-k}\biggl\{\Bigl(\frac{2}{k+1}\Bigr)^dn^{-k-1}-n^{-2k}\biggl\}.
\]

\Return $\kappa_n^{X,Y}:= (\mathcal{V}_n-e^{-1})/(1-e^{-1})$ and $p_\kappa := 1 - \Phi(\sqrt{n}(\mathcal{V}_n - e^{-1})/\sigma)$.
\end{algorithm}

\subsection{Simulation supplement}
\label{Subsec:SimulationSetting}
We show in Figure~\ref{Fig:simsetting} scatter plots of the six simulation settings used in Section~\ref{Sec:Numerics} at different noise levels. Also, Table~\ref{Tab:Timing} shows the running time of the six algorithms under comparison for $n\in\{125,250,500,\ldots,8000\}$ and $d_X=d_Y\in\{1,2\}$. Algorithm timing was performed on an 8-core 3.2 GHz laptop CPU, averaged over 10 repetitions. We see that when $d_X=d_Y=1$, both the coverage correlation and Chatterjee's correlation scale approximately linearly and the other algorithms scale quadratically in $n$. When $d_X=d_Y=2$, the coverage correlation has a quadratic scaling in $n$, mostly driven by the Monge--Kantorovich rank computation. 
\begin{figure}[]
\begin{tabular}{ccccc} 
  & $\gamma = 0$ & $\gamma = 0.2$ & $\gamma = 0.4$ & $\gamma = 1.0$\\
\raisebox{5.5\height}{sinusoidal} & \includegraphics[width=0.18\textwidth]{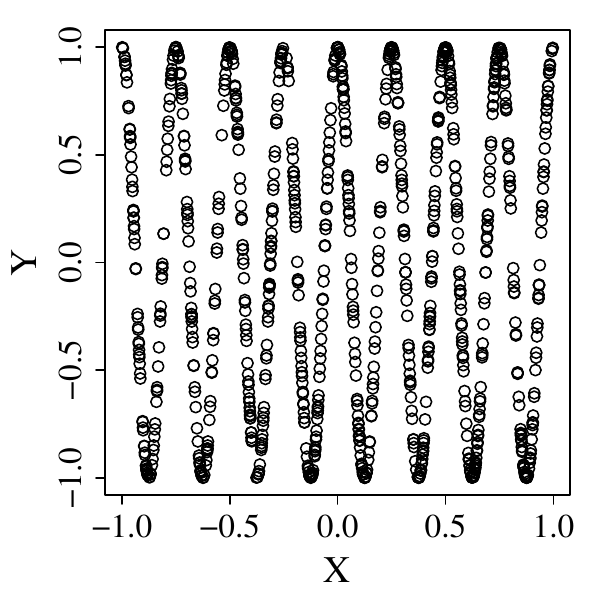} & \includegraphics[width=0.18\textwidth]{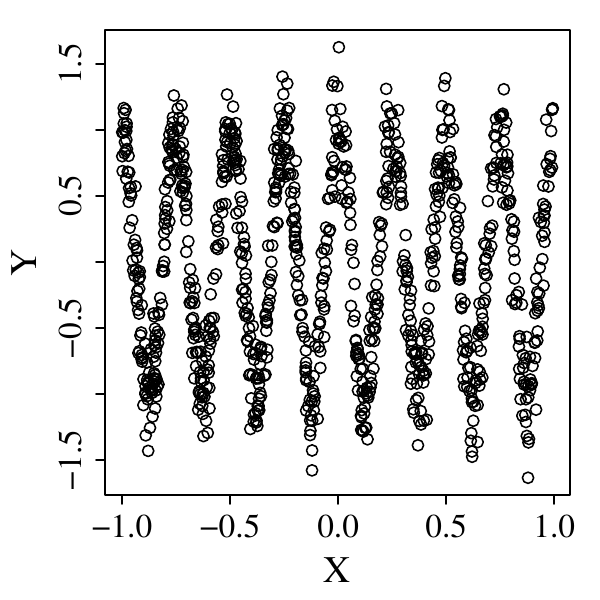} & \includegraphics[width=0.18\textwidth]{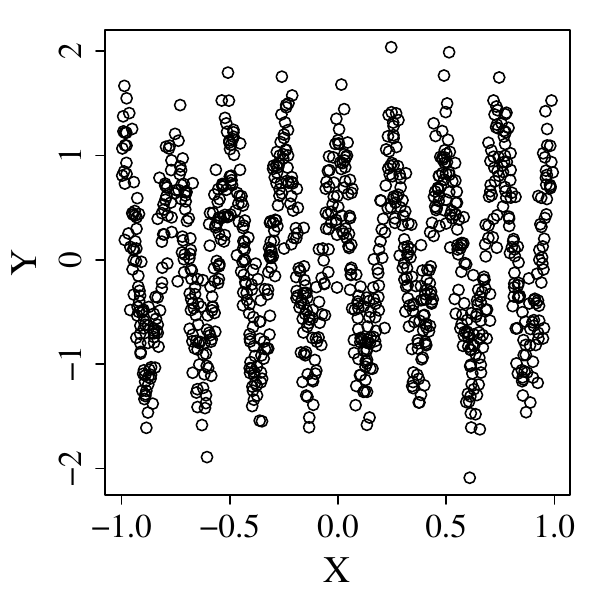} & \includegraphics[width=0.18\textwidth]{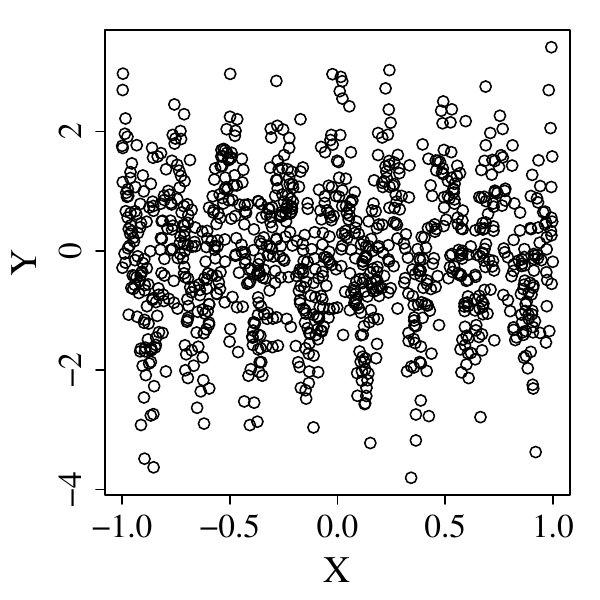}\\
\raisebox{5.5\height}{zigzag} & \includegraphics[width=0.18\textwidth]{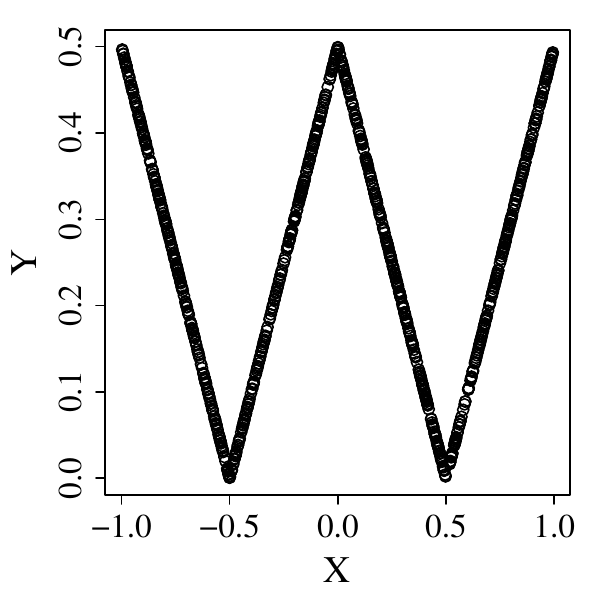} & \includegraphics[width=0.18\textwidth]{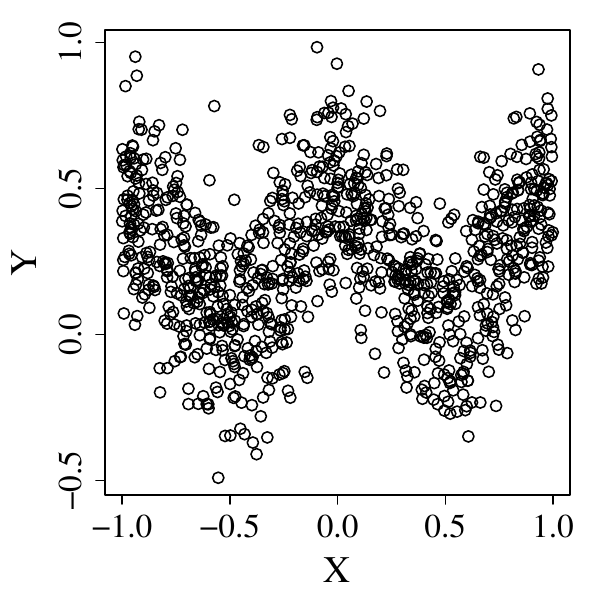} & \includegraphics[width=0.18\textwidth]{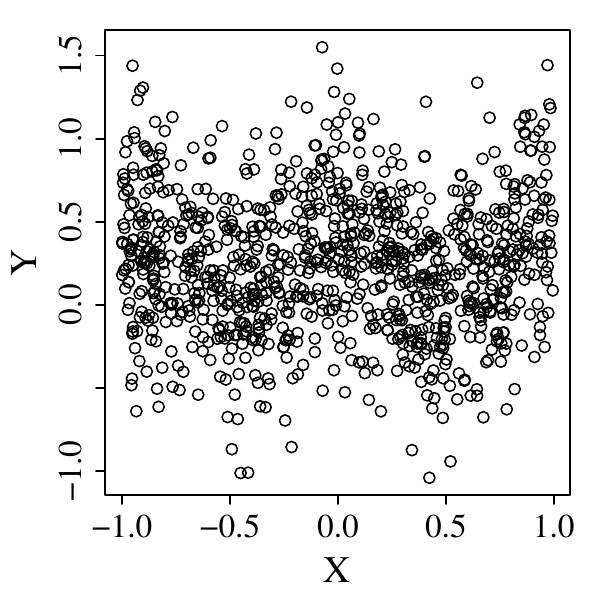} & \includegraphics[width=0.18\textwidth]{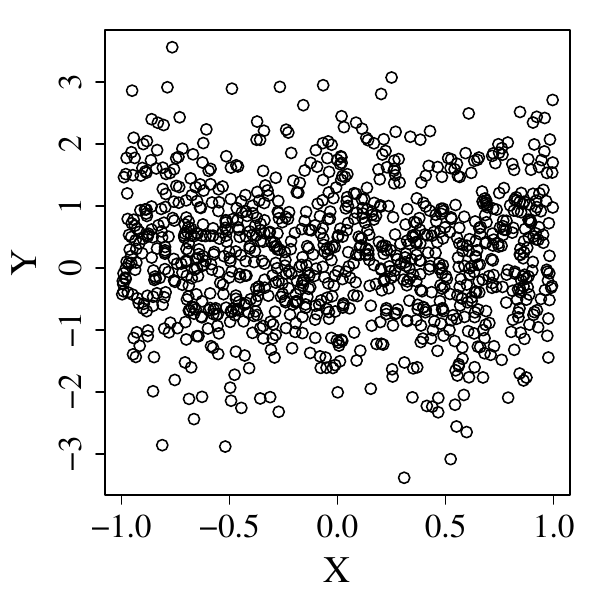}\\
\raisebox{5.5\height}{circle} & \includegraphics[width=0.18\textwidth]{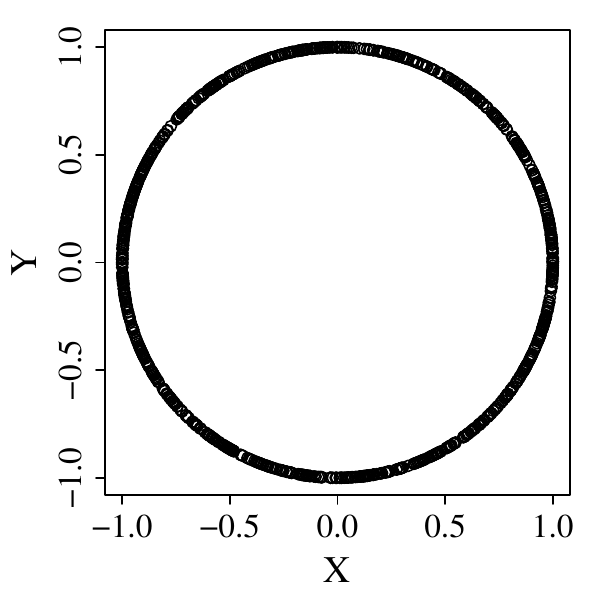} & \includegraphics[width=0.18\textwidth]{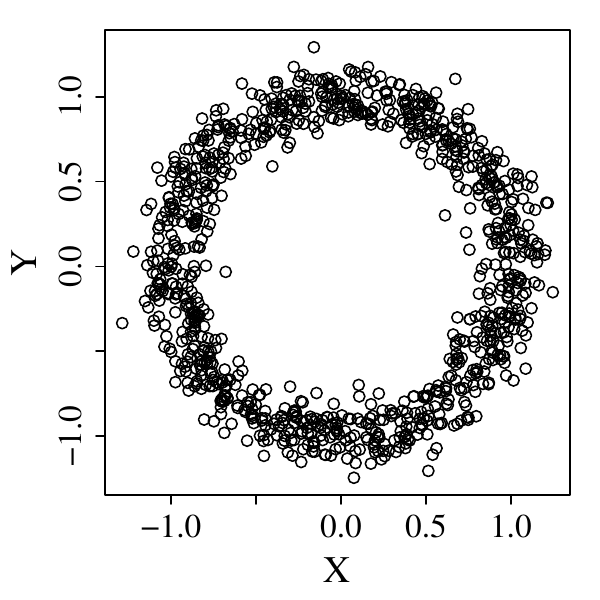} & \includegraphics[width=0.18\textwidth]{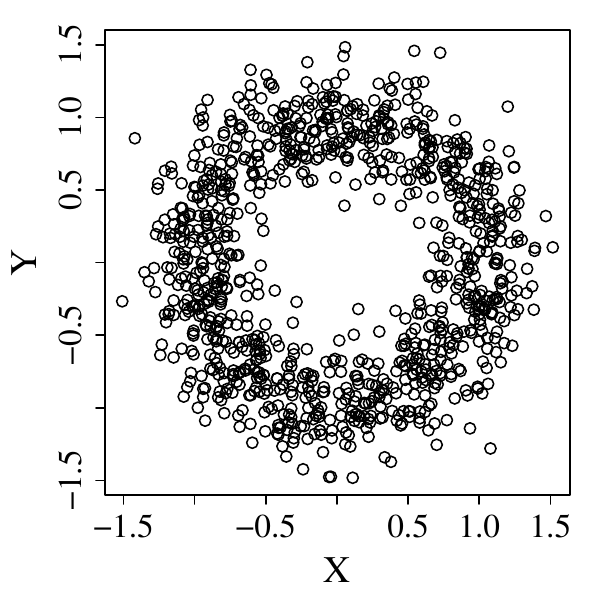} & \includegraphics[width=0.18\textwidth]{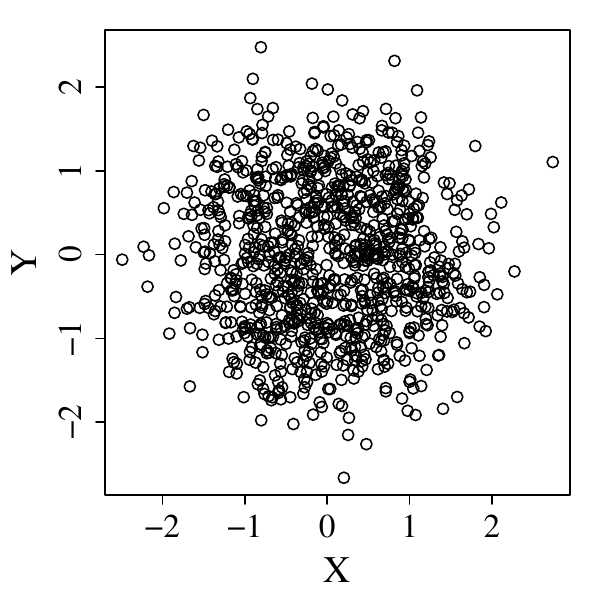}\\
\raisebox{5.5\height}{spiral} & \includegraphics[width=0.18\textwidth]{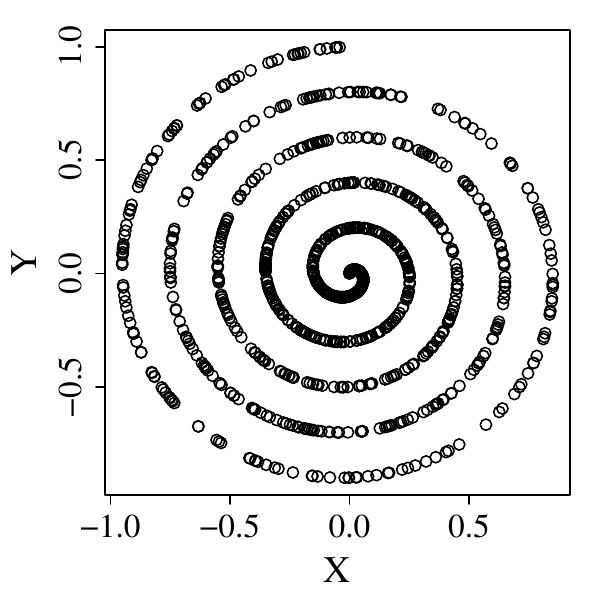} & \includegraphics[width=0.18\textwidth]{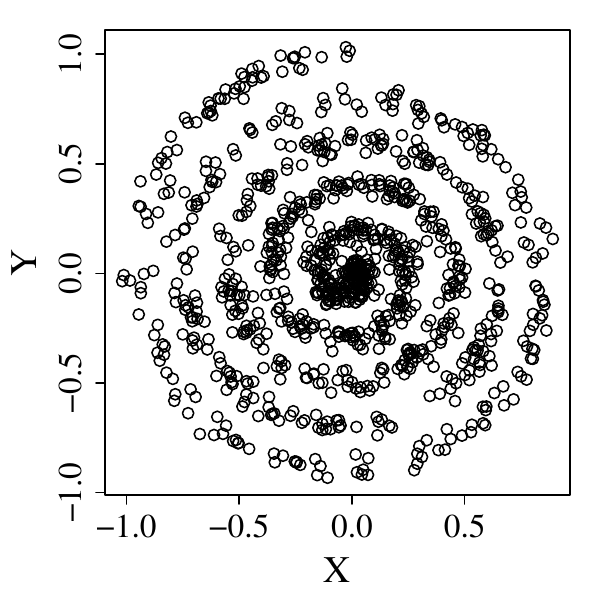} & \includegraphics[width=0.18\textwidth]{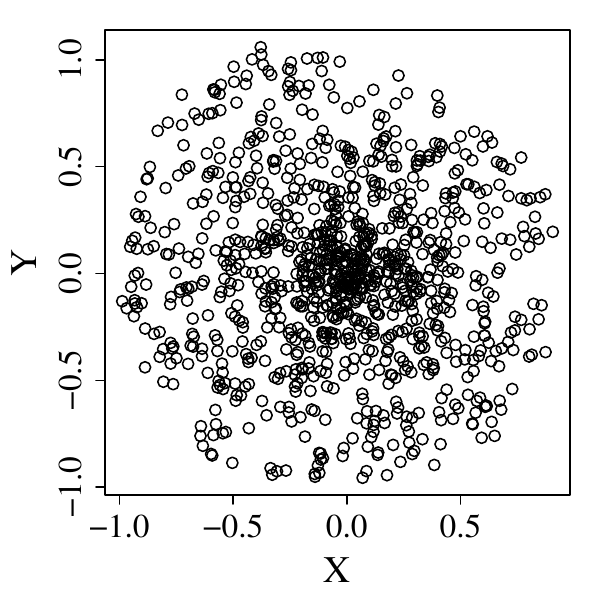} & \includegraphics[width=0.18\textwidth]{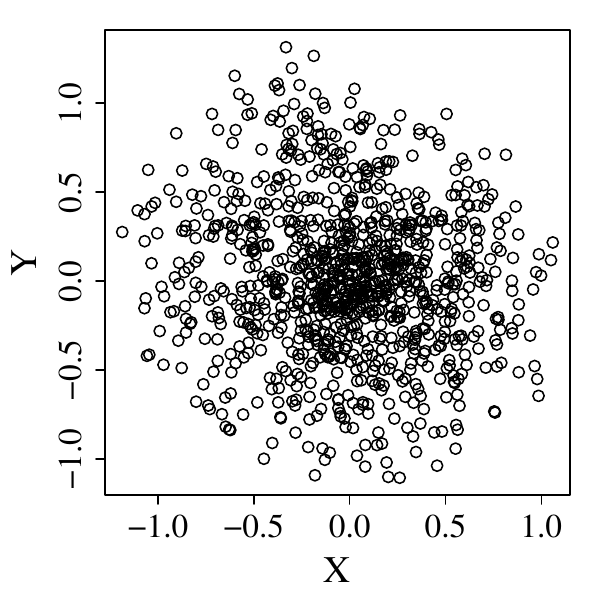}\\
\raisebox{5.5\height}{Lissajous} & \includegraphics[width=0.18\textwidth]{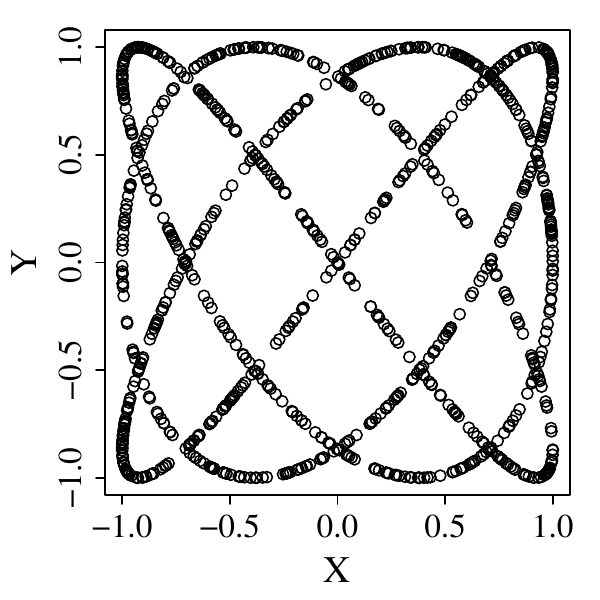} & \includegraphics[width=0.18\textwidth]{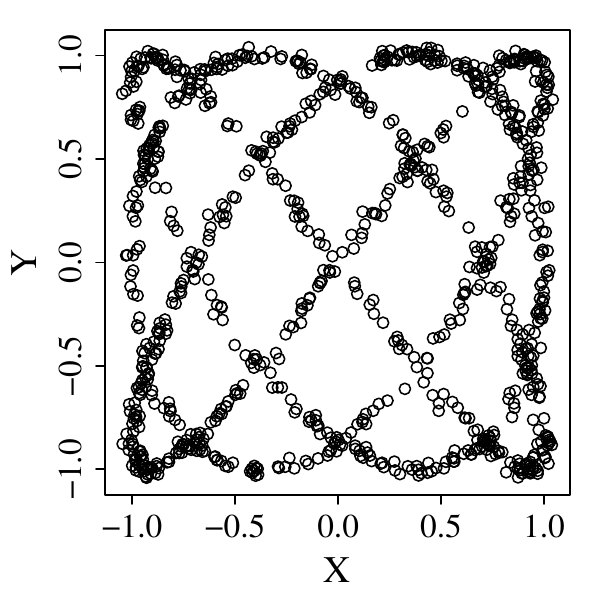} & \includegraphics[width=0.18\textwidth]{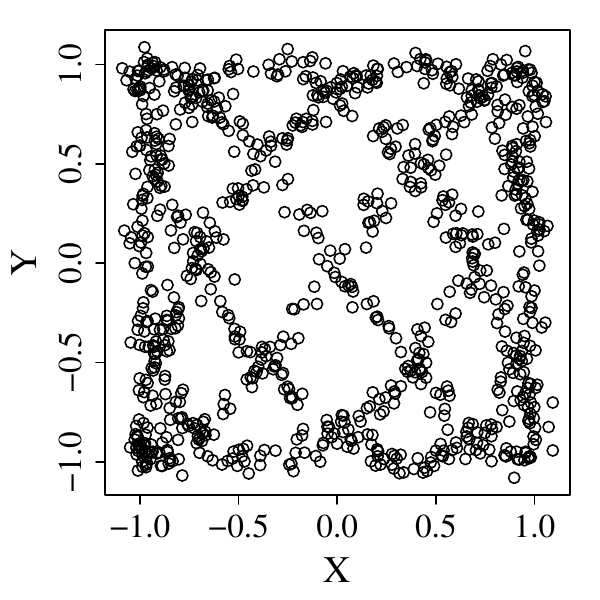} & \includegraphics[width=0.18\textwidth]{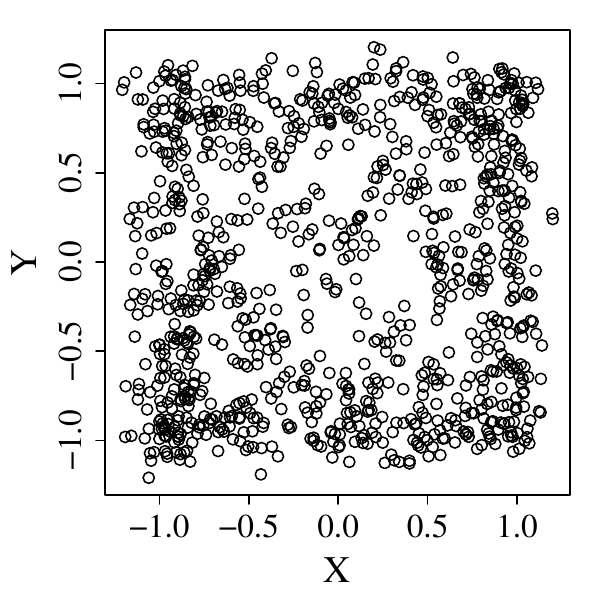}\\
\raisebox{5.5\height}{local} & \includegraphics[width=0.18\textwidth]{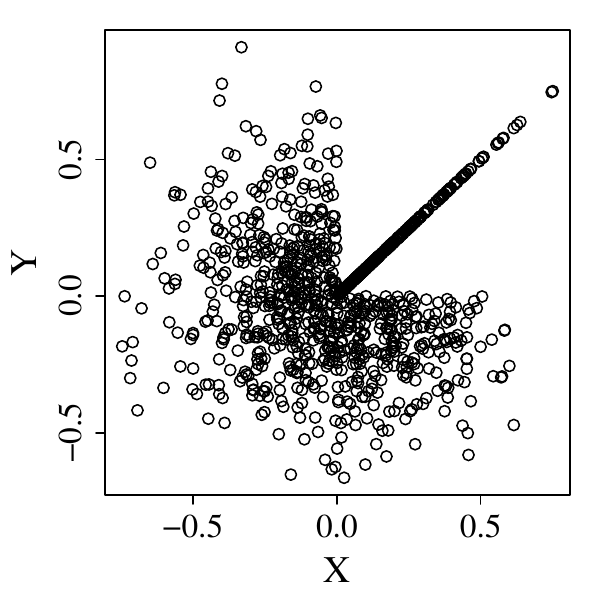} & \includegraphics[width=0.18\textwidth]{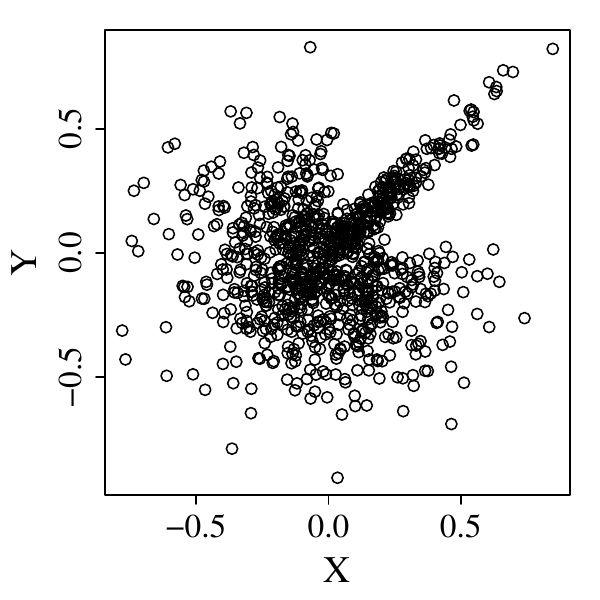} & \includegraphics[width=0.18\textwidth]{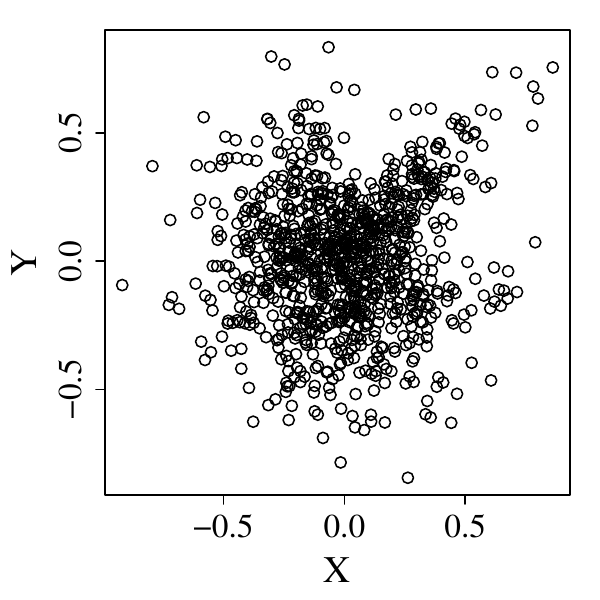} & \includegraphics[width=0.18\textwidth]{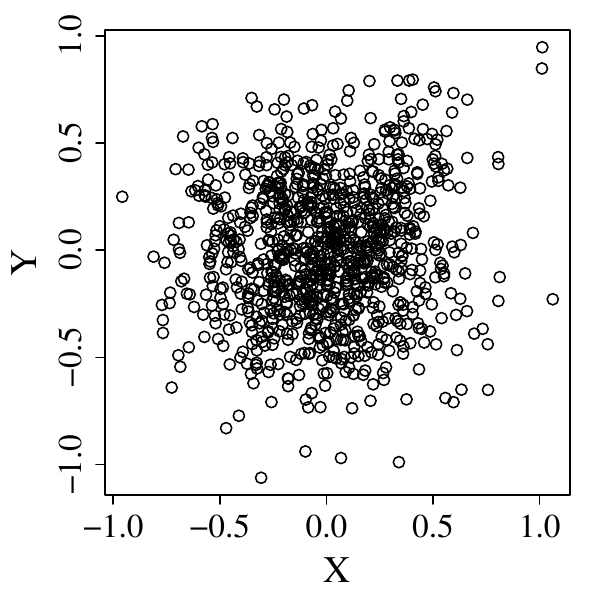}
\end{tabular}
\caption{\label{Fig:simsetting}Scatter plots of data from different simulation settings at different noise levels.}
\end{figure}

\begin{table}
\begin{center}
    \begin{tabular}{cccccccc}
    $n$ & $d_X$ & $\kappa_n^{X,Y}$ & $\xi_n^{X,Y}$ & dCor & HSIC & KMAc & USP\\
    \hline
    125 & 1 & 0.001 & 0.001 & 0.008 & 0.014 & 0.553 & 0.010\\
    250 & 1 & 0.001 & 0.001 & 0.010 & 0.047 & 1.06 & 0.043\\
    500 & 1 & 0.002 & 0.001 & 0.037 & 0.192 & 2.50 & 0.182\\
    1000 & 1 & 0.003 & 0.001 & 0.130 & 1.01 & 7.35 & 0.781 \\
    2000 & 1 & 0.005 & 0.001 & 0.498 & 4.23 & 26.3 & 2.98 \\
    4000 & 1 & 0.010 & 0.002 & 2.01 & 21.6 & - & 10.8 \\
    8000 & 1 & 0.019 & 0.003 & 7.95 & - & - & - \\
    125 & 2 & 0.034 & - & 0.004 & 0.011 & 0.514 & 0.042\\
    250 & 2 & 0.076 & - & 0.014 & 0.042 & 1.05 & 0.164 \\
    500 & 2 & 0.177 & - & 0.052 & 0.186 & 2.52 & 0.720\\
    1000 & 2 & 0.567 & - & 0.176 & 0.975 & 7.50 & 3.17\\
    2000 & 2 & 1.93 & - & 0.694 & 4.45 & 27.5 & 10.8 \\
    4000 & 2 & 6.16 & - & 2.77 & 21.5 & - & 43.9\\
    8000 & 2 & 24.4 & - & 11.4 & - & - & -
    \end{tabular}
    \caption{\label{Tab:Timing}Average running time of coverage correlation, Chatterjee's correlation, distance correlation, HSIC, KMAc and USP for various $n$ and $d_X=d_Y$ values under the null. Chatterjee's correlation is only computed for $d_X=d_Y=1$. Time is shown in seconds, and time larger than 60 seconds is not displayed.}
    \end{center}
\end{table}

\subsection{On the periodic boundary condition}\label{sec:no-wrapping}

In the definition of the coverage correlation, we employ a periodic boundary condition, which effectively wraps the edges of cubes centred at the boundary of the unit cube $[0, 1]^d$. We remark that this approach helps to avoid the boundary inhomogeneity that can lead to biased estimates. 

In fact, when $X$ and $Y$ are independent, the vacancy estimates can be overestimated without the boundary wrapping, leading to a positive excess vacancy bias, i.e. $\mathcal{V}_n - e^{-1}>0$. The excess vacancy bias translates to anti-conservative testing behaviour. In Table~\ref{Tab:no-wrapping}, we repeated the size experiment from Table~\ref{Tab:Size}, but computed the covered volume without wrapping around the boundary. The results show that the no-wrapping version is not well calibrated. This supports the use of the flat torus distance, which removes the boundary inhomogeneity and gives the well-calibrated sizes reported in Table~\ref{Tab:Size}.

\begin{table}
\centering
\begin{tabular}{cccccc}
$n$ & $d_X$ &  $\alpha = 1\%$ & $\alpha = 2.5\%$ &  $\alpha = 5\%$ &  $\alpha = 10\%$\\
\hline
$10$ & $1$ & $8.42_{(0.11)}$ & $14.94_{(0.14)}$ & $22.70_{(0.16)}$ & $34.31_{(0.18)}$\\
$100$ & $1$ & $10.11_{(0.12)}$ & $18.12_{(0.15)}$ & $27.28_{(0.17)}$ & $40.31_{(0.19)}$\\
$1000$ & $1$ & $11.22_{(0.12)}$ & $19.73_{(0.15)}$ & $29.56_{(0.18)}$ & $43.02_{(0.19)}$\\
$10$ & $2$ & $100.0_{(0.00)}$ & $100.00_{(0.00)}$ & $100.00_{(0.00)}$ & $100.00_{(0.00)}$\\
$100$ & $2$ & $100.00_{(0.00)}$ & $100.00_{(0.00)}$ & $100.00_{(0.00)}$ & $100.00_{(0.00)}$\\
$1000$ & $2$ & $100.00_{(0.00)}$ & $100.00_{(0.00)}$ & $100.00_{(0.00)}$ & $100.00_{(0.00)}$\\
\end{tabular}

\caption{\label{Tab:no-wrapping}Empirical sizes (in percentage) of independence test based on coverage correlation coefficient without periodic condition, estimated over 100000 Monte Carlo repetitions (with Monte Carlo standard errors in parentheses)}
\end{table}

\bibliographystyle{custom2author}
\bibliography{refs}
\end{document}